\documentclass[11pt]{amsart} 
\usepackage{amscd,amssymb,amsxtra}
\usepackage[mathscr]{eucal}
\usepackage{mathtools}
\usepackage{mathabx}
\usepackage{booktabs}
\usepackage[usenames, dvipsnames]{color}
\usepackage{enumerate}
\usepackage{comment}

\newcounter{thmItem}

\newenvironment{thmList}{\begin{list}%
{\rm \roman{thmItem})}{\usecounter{thmItem}
\setlength{\labelwidth}{2em}
\setlength{\itemindent}{2em}
\setlength{\leftmargin}{0pt}
\setlength{\listparindent}{0pt}
\setlength{\parsep}{0pt}
\setlength{\partopsep}{0pt}
\setlength{\itemsep}{\medskipamount}
\setlength{\topsep}{\medskipamount}
}}{\end{list}}

\makeatletter
\let\@secnumfont\bfseries
\def\section{\@startsection{section}{1}%
  \z@{4\linespacing\@plus\linespacing}{\linespacing}%
  {\bfseries\centering}}
\def\introsection{\@startsection{section}{1}%
  \z@{3\linespacing\@plus\linespacing}{\linespacing}%
  {\bfseries\centering}}
\def\subsection{\@startsection{subsection}{2}%
   \z@{1.25\linespacing\@plus.7\linespacing}{.5\linespacing}%
   {\normalfont\bfseries}}
\def\subsectionsinline{\def\subsection{\@startsection{subsection}{2}%
  \z@{1\linespacing\@plus.7\linespacing}{-.5em}%
  {\normalfont\bfseries}}}

\def\@makefnmark{%
  \leavevmode
  \raise.9ex\hbox{\fontsize\sf@size\z@\normalfont\tiny\@thefnmark}} 

\def\footref#1{$\textsuperscript{\tiny\ref{#1}}$}

\makeatother

\theoremstyle{definition}
\newtheorem{definition}[equation]{Definition}
\newtheorem{example}[equation]{Example}

\newtheorem{ansatz}[equation]{Ansatz}

\newtheorem*{definition*}{Definition}
\newtheorem*{example*}{Example}
\newtheorem*{problem*}{Problem}
\newtheorem*{exercise*}{Exercise}
\newtheorem*{question*}{\color{blue}Question}
\newtheorem*{construction*}{Construction}

\theoremstyle{remark}

\newtheorem{remark}[equation]{Remark}

\newtheorem*{note*}{Note}
\newtheorem*{notation*}{Notation}
\newtheorem*{remark*}{Remark}
\newtheorem*{data*}{Data}

\theoremstyle{plain}
\newtheorem{theorem}[equation]{Theorem}
\newtheorem{corollary}[equation]{Corollary}
\newtheorem{lemma}[equation]{Lemma}
\newtheorem{proposition}[equation]{Proposition}
\newtheorem{conjecture}[equation]{Conjecture}

\newtheorem*{theorem*}{Theorem}
\newtheorem*{corollary*}{Corollary}
\newtheorem*{lemma*}{Lemma}
\newtheorem*{proposition*}{Proposition}
\newtheorem*{conjecture*}{Conjecture}
\newtheorem*{claim*}{Claim}
\newtheorem*{proposal*}{Proposal}
\newtheorem*{conclusion*}{Conclusion}
\newtheorem*{hypothesis*}{Hypothesis}
\newtheorem*{assumption*}{Assumption}

\numberwithin{equation}{section}

\definecolor{refkey}{rgb}{0,.6,.4}

\newcommand{\bmuu}{\mbox{$\raisebox{-.07em}{\rotatebox{9.9}
  {\tiny {\bf /}
  }}\hspace{-0.53em}\mu\hspace{-0.88em}\raisebox{-0.98ex}{\scalebox{2} 
  {$\color{white}\phantom{.}$}}\hspace{-0.416em}\raisebox{+0.88ex}
  {$\color{white}\phantom{.}$}\hspace{0.46em}$}} 
\newcommand{\bmut}{\bmu 2}
\newcommand{\bmu}[1]{\bmuu _{#1}}

\renewcommand{\:}{\colon}
\renewcommand{\AA}{{\mathbb A}}

\DeclareMathOperator{\Aut}{Aut}
\newcommand{\CC}{{\mathbb C}}

\newcommand{\EE}{\mathbb E}

\DeclareMathOperator{\End}{End}

\newcommand{\HH}{{\mathbb H}}
\DeclareMathOperator{\Hom}{Hom}
\DeclareMathOperator{\id}{id}
\DeclareMathOperator{\Map}{Map}

\DeclareMathOperator{\pt}{pt}

\newcommand{\RP}{{\mathbb R\mathbb P}}
\newcommand{\RR}{{\mathbb R}}
\newcommand{\TT}{\mathbb T}
\DeclareMathOperator{\Spin}{Spin}

\DeclareMathOperator{\tr}{tr}

\newcommand{\ZZ}{{\mathbb Z}}

\newcommand{\chiup}{\raise.5ex\hbox{$\chi$}}
\newcommand{\cir}{S^1}
\DeclareMathOperator{\coker}{coker}
\newcommand{\dbar}{{\bar\partial}}

\newcommand{\inv}{^{-1}}
\DeclareRobustCommand{\mstrut}{^{\vphantom{1*\prime y\vee M}}}

\newcommand{\res}[1]{\negmedspace\bigm|\mstrut_{#1}}

\newcommand{\temsquare}{\raise3.5pt\hbox{\boxed{ }}}

\newcommand{\zmod}[1]{\ZZ/#1\ZZ}

\newcommand{\zt}{\zmod2}

\DeclareMathOperator{\SO}{SO}

\let\O\relax

\DeclareMathOperator{\U}{U}

\usepackage[all,2cell]{xy}\renewcommand{\cir}{\ensuremath{S^1}}
\usepackage{graphicx}
\usepackage[colorlinks,backref=page,citecolor=refkey]{hyperref}

 \definecolor{refkey}{rgb}{0,.8,.2}\definecolor{labelkey}{rgb}{1,0,0} 

\theoremstyle{remark}
\newtheorem*{lnote*}{Literature Note}

\newcommand{\bhath}[2]{B\widehat{H}_{#1}^{(#2)}}
\DeclareMathOperator{\pr}{pr}
\DeclareMathOperator*{\colim}{colim}

\DeclareMathOperator{\Bord}{Bord}
\DeclareMathOperator{\Cat}{Cat}
\DeclareMathOperator{\Cliff}{Cliff}
\DeclareMathOperator{\Euc}{Euc}
\DeclareMathOperator{\Euler}{Euler}
\DeclareMathOperator{\Ext}{Ext}

\DeclareMathOperator{\Line}{Line}
\DeclareMathOperator{\Pfaff}{Pfaff}
\DeclareMathOperator{\Pin}{Pin}
\DeclareMathOperator{\Sign}{Sign}

\DeclareMathOperator{\Thom}{Thom}
\DeclareMathOperator{\Trace}{Trace}
\DeclareMathOperator{\Vect}{Vect}
\DeclareMathOperator{\cliff}{Cliff}
\DeclareMathOperator{\diag}{diag}
\DeclareMathOperator{\ext}{Ext}
\DeclareMathOperator{\grass}{Gr}
\DeclareMathOperator{\haut}{hAut}
\DeclareMathOperator{\ho}{ho}
\DeclareMathOperator{\image}{image}
\DeclareMathOperator{\inde}{index}
\DeclareMathOperator{\ko}{ko}

\DeclareMathOperator{\pin}{pin}

\DeclareMathOperator{\spinc}{Spin^{c}}
\DeclareMathOperator{\spin}{Spin}
\DeclareMathOperator{\sq}{Sq}
\DeclareMathOperator{\thom}{Thom}

\newcommand{\Clm}[1]{\Cliff_{-#1}}
\newcommand{\Clp}[1]{\Cliff_{+#1}}
\newcommand{\Cxu}{\underline{\CC}^\times}
\newcommand{\Cx}{\CC^\times }
\newcommand{\Dirac}{D\hskip-.65em /} 
\newcommand{\FFF}[2]{FF_{#1}(#2)}

\newcommand{\GG}[2]{H_{#1}(#2)}
\newcommand{\Gm}{G^-}
\newcommand{\Gp}{G^+}
\newcommand{\HCn}{H\CC^{\npn}}
\newcommand{\HRon}{\HRo^{\npn}}
\newcommand{\HRo}{H\RR(1)}
\newcommand{\HR}{H\RR}
\newcommand{\HRp}{H\Rp}
\newcommand{\HM}{\sH\mstrut _{1,n-1}}
\newcommand{\Hm}{\mathbb{E}^n_-}
\newcommand{\Hp}{\mathbb{E}^n_+}

\newcommand{\ICxn}{(\ICx)^{\npn}}
\newcommand{\ITn}{\IT^{\npn}}
\newcommand{\ICx}{I\Cx} 
\newcommand{\IMt}{\IM^{\uparrow}}
\newcommand{\IM}{\sI _{1,n-1}}

\newcommand{\IT}{I\TT}
\newcommand{\IZn}{\IZ^\gamma }

\newcommand{\IZ}{I\Zo}

\newcommand{\K}{\mathbf{P}}
\newcommand{\MH}{H\mstrut _{1,n-1}}

\newcommand{\MSpin}{M\!\Spin}

\newcommand{\OMmt}{O_{1,n-1}^{\downarrow}}
\newcommand{\OMt}{O_{1,n-1}^{\uparrow}}

\newcommand{\Pcm}{\Pin^{\tc-}}
\newcommand{\Pcpm}{\Pin^{\tc\pm}}
\newcommand{\Pcp}{\Pin^{\tc+}}
\newcommand{\Pc}{\Pin^c}
\newcommand{\Pmp}{\Pin^{\mp}}
\newcommand{\Pm}{\Pin^{-}}

\newcommand{\Ppm}{\Pin^{\pm}}
\newcommand{\Ppn}{\Pin^{+}_n}
\newcommand{\Pp}{\Pin^{+}}

\newcommand{\RC}{R_{\CC}}
\newcommand{\RM}{\RR^{1,n-1}}
\newcommand{\RMm}{\RR^{1,n-2}}

\newcommand{\Rp}{\RR^{>0}}
\newcommand{\R}{\RR}
\newcommand{\SCxg}{\Sigma ^n\ICxn}
\newcommand{\SCx}{\Sigma ^n\ICx}
\newcommand{\SC}{\SS_{\CC}}
\newcommand{\SIZn}{\Sigma ^{n+1}\IZn}

\newcommand{\SIZ}{\Sigma ^{n+1}\IZ}

\newcommand{\SOmt}{SO_{1,n-1}^{\downarrow}}

\newcommand{\Si}{\relax}
\newcommand{\Snizn}{\Sigma ^{n+1}\IZn}
\newcommand{\Sniz}{\Sigma ^{n+1}\IZ}

\newcommand{\TP}[2]{TP_{#1}(#2)}
\newcommand{\Up}{U^{\perp}}
\newcommand{\VectPos}{\fVC^{\text{pos}}}
\newcommand{\WC}{W_{\CC}}
\newcommand{\Zo}{\ZZ(1)}

\newcommand{\Z}{\mathbb{Z}}
\newcommand{\ab}{^{\alpha /\beta }}
\newcommand{\aone}{\mathbf{A}_{1}}

\newcommand{\bB}{\beta _{\sB}}
\newcommand{\bC}{\beta _{\sC}}

\newcommand{\bX}{\partial{X}}

\newcommand{\be}{\Bord_n(H_n)}
\newcommand{\bgspectra}[1]{\spectra^{h#1}}
\newcommand{\bne}{\Bord_{\langle n-1,n\rangle}(H_n)}
\newcommand{\bone}{\mathbf{1}}

\newcommand{\bz}{\overline{\zeta }}

\newcommand{\cat}[1]{\mathcal #1}

\newcommand{\dspectra}{\ho\uspectra}
\newcommand{\dual}{^\vee}
\newcommand{\fVC}{f\!\Vect_{\CC}}
\newcommand{\fieldSymbol}{\mathscr I}
\newcommand{\fieldsHermd}[2]{\fieldSymbol^\delta _{#2}(#1)_{\textnormal{Hermitian}}}
\newcommand{\fieldsHerm}[2]{\fieldSymbol_{#2}(#1)_{\textnormal{Hermitian}}}
\newcommand{\fieldsPosd}[2]{\fieldSymbol^\delta _{#2}(#1)_{\textnormal{positive}}}
\newcommand{\fieldsPos}[2]{\fieldSymbol_{#2}(#1)_{\textnormal{positive}}}
\newcommand{\fieldsReal}[2]{\fieldSymbol^{\mathbb R}_{#2}(#1)}
\newcommand{\fieldsReflPosd}[2]{\fieldSymbol^\delta_{#2}(#1)_{\substack{\textnormal{reflection}\\\textnormal{positive}}}} 
\newcommand{\fieldsReflPos}[2]{\fieldSymbol_{#2}(#1)_{\substack{\textnormal{reflection}\\\textnormal{positive}}}}
\newcommand{\fieldsRefld}[2]{\fieldSymbol^\delta _{#2}(#1)_{\textnormal{reflection}}}
\newcommand{\fieldsRefl}[2]{\fieldSymbol_{#2}(#1)_{\textnormal{reflection}}}

\newcommand{\fieldsStable}[2]{\fieldSymbol_{#2}(#1)_{\textnormal{stable}}}
\newcommand{\fields}[2]{\fieldSymbol_{#2}(#1)}
\newcommand{\gMap}[1]{\Map^{#1}}

\newcommand{\hH}{\widehat{H}}
\newcommand{\hPin}{\widehat{\Pin}}

\newcommand{\hSH}{\widehat{SH}}
\newcommand{\ha}{\hat{\alpha }}
\newcommand{\hin}{\hat\imath_n}
\newcommand{\hhat}{\hat{h}}
\newcommand{\hh}[2]{H_{#1}(#2)}
\newcommand{\hsH}{\widehat{\sH}}
\newcommand{\hs}{\hat\sigma }
\newcommand{\infn}{(\infty ,n)}
\newcommand{\mth}[1]{\Sigma ^{#1}MTH_{#1}}
\newcommand{\noo}{_{n-1,1}}
\newcommand{\npn}{\nu '_0}
\newcommand{\onoa}{\ono^{\alpha\phantom{\beta }}}
\newcommand{\ono}{_{1,n-1}}
\newcommand{\op}{^{\textnormal{op}}}

\newcommand{\pinm}{\pin^{-}}
\newcommand{\pinp}{\pin^{+}}

\newcommand{\pmo}{\{\pm1\}}

\newcommand{\ppm}{\pin^{\pm}}

\newcommand{\sA}{\mathscr{A}}
\newcommand{\sBO}{\mathscr{B}_O}
\newcommand{\sB}{\mathscr{B}}
\newcommand{\sCx}{\sC^\times }
\newcommand{\sC}{\mathscr{C}}
\newcommand{\sD}{\mathscr{D}}
\newcommand{\sG}{\mathscr{G}}
\newcommand{\sH}{\mathscr{H}}
\newcommand{\sI}{\mathcal{I}}
\newcommand{\sLC}{s\!\Line_{\CC}}
\newcommand{\sLine}{s\!\Line_{\CC}}
\newcommand{\sL}{\mathscr{L}}

\newcommand{\sO}{\mathcal{O}}

\newcommand{\sQQ}{\mathscr{Q}}
\newcommand{\sQ}{P}
\newcommand{\sS}{\mathscr{S}}
\newcommand{\sT}{\mathscr{T}}

\newcommand{\sVect}{s\!\Vect_{\CC}}

\newcommand{\sX}{\mathscr{X}}

\newcommand{\slot}{\,-\,}

\newcommand{\spaces}{\mathcal T}
\newcommand{\spectra}{\EuScript S}
\newcommand{\sqmo}{\sqrt{-1}}

\newcommand{\tF}{\widetilde{F}}

\newcommand{\tH}{\widetilde{H}}
\newcommand{\tI}{\widetilde{I}}
\newcommand{\tO}{\widetilde{\hbox to 11pt{O}}}
\newcommand{\tSH}{\widetilde{SH}}
\newcommand{\tSO}{\widetilde{\SO}}
\newcommand{\tc}{\tilde{c}}
\newcommand{\te}{f}
\newcommand{\tih}{\tilde{h}}
\newcommand{\tors}{\mstrut _{\textnormal{tor}}}
\newcommand{\tphi}{\tilde{\varphi }}
\newcommand{\triv}[1]{\underline{\RR^{#1}}}

\newcommand{\ts}{\tilde\sigma }
\newcommand{\uH}{\underline{H}}
\newcommand{\uPp}{\underline{\Pin}^+}
\newcommand{\uSpin}{\underline{\Spin}}
\newcommand{\uhH}{\widehat{\underline{H}}}

\newcommand{\uspaces}{\underline{\spaces}}
\newcommand{\uspectra}{\underline{\spectra}}
\newcommand{\ztmap}{\gMap{\Z/2}}
\renewcommand{\O}{\mathcal O}
\renewcommand{\SS}{\mathbb{S}}
\newcommand{\hrn}{\hat{\rho }_n}

\begin{document}

\abovedisplayskip18pt plus4.5pt minus9pt
\belowdisplayskip \abovedisplayskip
\abovedisplayshortskip0pt plus4.5pt
\belowdisplayshortskip10.5pt plus4.5pt minus6pt
\baselineskip=15 truept
\marginparwidth=55pt

\makeatletter
\renewcommand{\tocsection}[3]{%
  \indentlabel{\@ifempty{#2}{\hskip1.5em}{\ignorespaces#1 #2.\;\;}}#3}
\renewcommand{\tocsubsection}[3]{%
  \indentlabel{\@ifempty{#2}{\hskip 2.5em}{\hskip 2.5em\ignorespaces#1%
    #2.\;\;}}#3} 
\renewcommand{\tocsubsubsection}[3]{%
  \indentlabel{\@ifempty{#2}{\hskip 5.3em}{\hskip 5.3em\ignorespaces#1%
    #2.\;\;}}#3} 
\makeatother

\setcounter{tocdepth}{2}



 \title[Reflection Positivity and Invertible Topological Phases]{Reflection Positivity and Invertible Topological Phases} 
 \author[D. S. Freed]{Daniel S.~Freed}
 \address{Department of Mathematics \\ University of Texas \\ Austin, TX
78712} 
 \email{dafr@math.utexas.edu}
 \author[M. J. Hopkins]{Michael J.~Hopkins}
 \address{Department of Mathematics \\ Harvard University \\ Cambridge, MA
02138} 
 \email{mjh@math.harvard.edu}
 \thanks{This material is based upon work supported by the National Science
Foundation under Grant Numbers DMS-1207817, DMS-1160461, DMS-1510417, and
DMS-1158983.  Any opinions, findings, and conclusions or recommendations
expressed in this material are those of the authors and do not necessarily
reflect the views of the National Science Foundation.  We thank the Aspen
Center for Physics for hosting our very productive working group.  The first
author also thanks the Institute for Advanced Study through the Wolfensohn
Fund for providing a stimulating environment for parts of this work.}
 \date{December 18, 2022}
 \begin{abstract} 
 We implement an extended version of reflection positivity (Wick-rotated
unitarity) for invertible topological quantum field theories and compute the
abelian group of deformation classes using stable homotopy theory. We apply
these field theory considerations to lattice systems, assuming the existence
and validity of low energy effective field theory approximations, and thereby
produce a general formula for the group of Symmetry Protected Topological
(SPT) phases in terms of Thom's bordism spectra; the only input is the
dimension and symmetry type.  We provide computations for fermionic systems
in physically relevant dimensions.  Other topics include symmetry in quantum
field theories, a relativistic 10-fold way, the homotopy theory of
relativistic free fermions, and a topological spin-statistics theorem.
 \end{abstract}
\maketitle

{\small
\def\reftext{References}
\renewcommand{\tocsection}[3]{%
  \begingroup 
   \def\tmp{#3}%
   \ifx\tmp\reftext#3%
  \else\indentlabel{\ignorespaces#1 #2.\;\;}#3%
  \fi\endgroup}
\tableofcontents}

   \section{Introduction}\label{sec:1}

The moduli space, or stack, of a geometric object with fixed discrete
invariants is a central object of interest in geometry.  A typical example is
the moduli stack of Riemann surfaces of fixed genus.  Here the underlying
topological space is connected, but moving up to complex dimension two the
moduli stack of complex surfaces of general type with fixed Euler number and
signature is not necessarily connected.  It has finitely many
components~\cite{Ca}, so there are finitely many \emph{deformation types}.
If singular objects are permitted, then sometimes connectivity can be
restored.  For example, Reid~\cite{Re} speculates that the moduli stack of
three-dimensional Calabi-Yau varieties is connected if one allows certain
singularities.  To illustrate further, consider the moduli stack of
one-dimensional Riemannian manifolds.  If we allow simple singularities, such
as the figure eight, then we can connect a single circle to two circles by a
path (standard Morse function on a two-dimensional torus).  We can also
connect one circle to two circles if we allow noncompact smooth manifolds:
elongate a circle to an ellipse to two lines and then each line to a circle.
On the other hand, the set of path components of the moduli stack of smooth
closed Riemannian 1-manifolds is isomorphic to~$\ZZ^{\ge0}$; the isomorphism
maps a 1-manifold to the cardinality of~$\pi _0$.

In theoretical physics one contemplates moduli stacks of quantum systems with
fixed discrete invariants, such as dimension and symmetry type.  If we remove
the singular locus of \emph{phase transitions}, then path components of the
moduli stack are identified with \emph{phases} of the quantum
system.\footnote{There is a tight analogy with the example of Riemannian
1-manifolds above: a figure eight corresponds to a first-order phase
transition, while a noncompact manifold corresponds to a higher-order phase
transition.}  In condensed matter physics the quantum systems are modeled
discretely, using lattices, and the classification of phases is an active
topic of current interest.  As far as we know there is not a robust
mathematical theory of lattice systems and their moduli which leads to
rigorous computations of sets of phases.  Quantum field theories also exhibit
phases and phase transitions, and those too are topical.  Physicists often
pass back and forth between lattice models and field theories using various
mechanisms.  In this paper we envision passing from a lattice system to an
effective low-energy field theory using two heuristic principles to argue
that the set of phases is conserved:

\smallskip
      \begin{enumerate}[{\textnormal(}i{\textnormal)}]

 \item the deformation class of a quantum system is determined by its low
energy behavior;
	
 \item the low energy physics of a \emph{gapped}\footnote{A quantum
mechanical system is gapped if its minimum energy is an eigenvalue of finite
multiplicity of the Hamiltonian, assumed bounded below, and is an isolated
point of the spectrum.  For quantum field theory `spectrum' means the
spectrum of representations of the translation group of Minkowski spacetime.
For lattice systems the spectral gap must be bounded below independent of the
lattice size.} system is well-approximated by a
\emph{topological}\footnote{We allow a topological field theory tensored with
a non-topological invertible field theory; see~\S\ref{subsec:4.4}.  A field
theory is \emph{topological} if it does not depend on any continuously
varying (background) fields, such as a metric or conformal structure.  We
give a precise definition of a topological field theory
in~\S\ref{subsec:12.2}.} field theory.

      \end{enumerate}
\smallskip

\noindent
 A stronger version of~(i) asserts that the entire homotopy type of the
moduli stack is determined by the low energy behavior.  These two principles
are applied by physicists to quantum systems of all kinds: condensed matter
systems, quantum field theories, string theories.  For discrete lattice
systems we also assume an emergent low energy relativistic symmetry.  We
remark that fracton models~\cite{NH} are thought not to satisfy~(ii), nor to
have any sort of emergent relativistic symmetry, but those are not relevant
here.  The lattice models that motivate this paper belong to a special class,
often called \emph{short-range entangled}, for which the long-range effective
topological field theory is \emph{invertible}.  In particular, there is a
unique ground state for the lattice model on any compact manifold.  Early
discussions of this property may be found in~\cite{CGW,K1}.  (Now
`invertible' is used in place of `short-range entangled' to describe the
lattice model.)

One reason to pass to continuum models is that there \emph{is} a mathematical
Axiom System for Wick-rotated quantum field theory; it encodes the structural
properties of correlation functions and linear spaces of quantum states.  It
was first introduced in the mid 1980's for scale-independent theories: by
Segal~\cite{Se1} for 2-dimensional conformal field theories and later by
Atiyah~\cite{A1} for topological field theories.  With
modifications these axioms are now believed to be relevant to scale-dependent
theories as well.  In this framework a quantum field theory is a linear
representation of a \emph{bordism category}.  The latter categorifies Thom's
bordism groups~\cite{T}, and a field theory categorifies integer-valued
bordism invariants, such as the signature of a compact oriented manifold.

The twin pillars of quantum field theory are \emph{locality} and
\emph{unitarity}.  These fundamental properties persist after Wick rotation:
locality manifests as factorization laws for correlation functions and
unitarity manifests as reflection positivity.  Locality is encoded in the
Axiom System using composition of morphisms: gluing bordisms along
codimension one submanifolds.  In the early 1990's, especially motivated by
3-dimensional Chern-Simons theory, an \emph{extended} notion of locality was
introduced by gluing bordisms with corners along higher codimension
submanifolds, and this led naturally to formulations involving higher
categories; see~\cite{F1,La,BD,L}, for example.  Extended locality is a
characteristic feature of both physical and mathematical applications of
field theory, whereas unitarity is often not present in purely mathematical
contexts.  Unitarity in field theory, or rather its Wick-rotated
manifestation---reflection positivity---is the first main subject of this
paper.  It is straightforward to implement reflection positivity in the
non-extended Axiom System.  A natural question arises: What is the
\emph{extended} notion of reflection positivity that goes with extended
locality?  We offer a solution in a very special case: \emph{invertible
topological} field theories.  These theories can be studied using stable
homotopy theory~\cite{FHT1}, and indeed we define\footnote{A better starting
point is the topological version of the Axiom System, and then
Theorem~\ref{thm:54} brings us to stable homotopy theory.  But as the
literature is still in flux we opt for Ansatz~\ref{thm:145} instead; see the
remarks following Theorem~\ref{thm:54}.} a theory of this type as a map of
spectra.  Spectra are the main characters in stable homotopy theory, a
mathematical field that partly grew out of Thom's work.  The domain of an
invertible topological field theory is a Madsen-Tillmann bordism spectrum,
and our main result tells that extended reflection positivity brings us full
circle to the bordism spectra introduced by Thom in his thesis~\cite{T}.

  \begin{theorem}[]\label{thm:110}
 There is a 1:1 correspondence 
  \begin{equation}\label{eq:128}
     \left\{\normalfont \vcenter{\hbox{deformation classes of reflection
     positive 
     }\hbox{invertible $n$-dimensional extended topological}\hbox{field
     theories with 
     symmetry group~$H_n$}} \right\} \cong [MTH,\SIZ]\tors.  
  \end{equation}
  \end{theorem}

\noindent 
 The right hand side is the torsion subgroup of homotopy classes of maps from
a Thom spectrum to a shift of the Anderson dual to the sphere spectrum.
There are standard computational techniques which we employ in the latter
part of this paper to illustrate the efficacy of the theorem.  Often field
theories are classified by enumerating lagrangians with specified background
and fluctuating fields that are consistent with a given symmetry group.  By
contrast, Theorem~\ref{thm:110} is a direct \emph{quantum} classification of
correlation functions and state spaces, as encoded by the Axiom System.  The
only inputs are the discrete invariants: the spacetime dimension~$n$ and the
Wick-rotated vector symmetry group\footnote{The basic case is~$H_n=SO_n$.  In
general there is a homomorphism $\rho _n\:H_n\to O_n$ whose image
includes~$SO_n$; the kernel consists of internal global symmetries.  There is
a unique associated stable symmetry group~$H$ independent of dimension, as we
prove in Theorem~\ref{thm:6}.}~$H_n$.  We prove Theorem~\ref{thm:110}
in~\S\ref{sec:7} as a corollary of a more general result
(Theorem~\ref{thm:184}).  There is a related assertion which remains
conjectural in this paper: the abelian group of deformation classes of
\emph{all} reflection positive invertible field theories, including those
that are not topological, is obtained by simply omitting `tor' on the right
hand side of~\eqref{eq:128}.  We make some comments about this generalization
in~\S\ref{subsec:4.4} and Remark~\ref{thm:144}; we use it in the computations
of~\S\ref{sec:8}.  More to the point, we introduce ``continuous invertible
topological field theories'' as a substitute for invertible non-topological
theories, and prove theorems for those.\footnote{We thank Peter Teichner for
his encouragement to adopt this point of view.}  We remark that for general
reasons nontorsion only arises if the spacetime dimension~$n$ is odd.

We apply Theorem~\ref{thm:110} to compute the abelian group of phases of
invertible lattice systems with fixed dimension and symmetry type.  This
implicitly assumes that every possible deformation class of invertible
topological theory can be realized by a lattice model, something not implied
by the heuristic principles~(i) and~(ii) above.  We emphasize the algorithmic
nature of our classification: given a spacetime dimension~$n$ and a symmetry
group~$H_n$ the right hand side of~\eqref{eq:128} is the group of topological
phases and is computable.  We provide concrete evidence for this application
of Theorem~\ref{thm:110}: in~\S\ref{subsec:8.2} we undertake detailed
computations for some fermionic systems and compare to results in the physics
literature, the latter derived by means of physical arguments.  Some readers
may wish to examine our tables of computations before tackling the more
theoretical parts of the paper.  In unpublished work Kitaev~\cite{K1,K2,K3}
develops a classification of invertible phases based on microscopic
considerations, and he too is led to stable homotopy theory and results
consonant with our effective field theory classification.
Kapustin~\cite{Ka1} initiated computations of topological phases via
character groups of bordism groups, and he used them and subsequent
computations, for example~\cite{KTTW}, as phenomenological evidence for a
general classification along these lines.  Gaiotto-Kapustin~\cite{GK},
following on Gu-Wen~\cite{GW}, show that some invertible fermionic phases
defined by lattice models are characterized by spin bordism groups; see also
Brumfiel-Morgan~\cite{BrMo}.  Campbell~\cite{C} and
Guo-Putrov-Wang~\cite{GPW} carry out computations for other bosonic and
fermionic cases of interest, providing further affirmative checks against the
condensed matter literature.

A second subject of this paper, after extended reflection positivity, is the
study of symmetry groups in relativistic quantum field theory, and that is
where we begin in~\S\ref{sec:12}.  Our starting point is a theory on
$n$-dimensional Minkowski spacetime with global symmetry group~$\MH$, after
dividing out by translations.  The analytic continuations of correlation
functions, which exist as a consequence of positivity of energy, are
invariant under the complex Lie group~$H_n(\CC)$, and the entire Wick-rotated
theory is symmetric under the compact real form~$H_n\subset H_n(\CC)$ that
appears on the left hand side of~\eqref{eq:128}.  In an
appendix~\S\ref{subsec:a3.2} we discuss Wick rotation and the CRT
theorem\footnote{There is a subtlety concerning double covers of the Lorentz
signature isometry group, uncovered in~\cite{GT}, which we explicate in the
context of Wightman quantum field theory for general symmetry types;
see~\S\ref{subsec:10.1}.} for general symmetry types.  We use the rigidity of
compact Lie groups to constrain possible symmetry groups
(Theorem~\ref{thm:5}) \emph{\`a la} Coleman-Mandula~\cite{CM}.  One key
result in this section (Theorem~\ref{thm:6}) is the existence and uniqueness
of a stabilization~$H$, which is the group in the Thom spectrum on the right
hand side of~\eqref{eq:128}.  When we move to curved Riemannian
manifolds---i.e., couple the theory to background gravity---the symmetry
becomes infinitesimal in the sense of Cartan: an $H_n$-structure on the
tangent bundle.  In~\S\ref{sec:2} we formulate reflection symmetry in terms
of a group extension 
  \begin{equation}\label{eq:282}
     1\longrightarrow
     H_n\longrightarrow\hH_n\longrightarrow\pmo\longrightarrow 1;  
  \end{equation}
elements in ~$\hH_n\setminus H_n$ are a Wick-rotated analog of anti-unitary
symmetries in quantum mechanics.  (We give a topological account of~$\hH_n$
and the group extension~\eqref{eq:282} in Appendix~\ref{sec:16}.)  We use
this extension in~\S\ref{subsec:3.1} to define an involution on the bordism
category of $H_n$-manifolds.  In the basic case $H_n=SO_n$ the involution is
orientation-reversal; our uniform treatment gives analogs for any symmetry
group.  For example, fermionic theories with time-reversal symmetry (and no
other symmetry) have~$H_n=\Ppm_n$: the involution takes a pin structure to
its ``$w_1$-flipped'' pin structure.  Topological field theories are
independent of the Riemannian metric, so we can replace~$H_n$ by a noncompact
analog, which we construct in Appendix~\ref{sec:14}.

Three basic lessons we learned about reflection positivity: (i)~`reflection'
and `positivity' are distinct; (ii)~`reflection' is a structure whereas
`positivity' is a condition; and (iii)~`extended positivity' is a structure,
not a condition.  In the Axiom System a field theory is defined to be a
homomorphism---a symmetric monoidal functor---
  \begin{equation}\label{eq:158}
     F\:\bne\longrightarrow \Vect_{\CC} 
  \end{equation}
from the bordism category to the category of complex vector spaces and linear
maps.  A \emph{reflection structure}~(\S\ref{subsec:3.3}) is equivariance
data for~$F$ with respect to the generalized orientation-reversal involution
on~$\bne$ and the involution of complex conjugation on~$\Vect_{\CC}$.  (We
briefly review involutions on categories and equivariant functors in
Appendix~\ref{sec:11}.)  A reflection structure induces a hermitian metric on
the vector space of states attached to an $(n-1)$-manifold, and
\emph{positivity} is the condition that these hermitian structures be
positive definite.  Analogous to reflection positivity in Euclidean
space~(\S\ref{subsec:2.2}) we see that the partition function of the
\emph{double} of a manifold with boundary must be positive in order that a
reflection structure be positive.  Our treatment of this material using
general symmetry groups means it applies to all theories, including those
with time-reversal symmetry and fermions which, after Wick rotation, involve
nonorientable manifolds with pin structure.
 
To proceed to \emph{extended} field theories we specialize in~\S\ref{sec:4}
to the invertible case.  (Invertible field theories were first singled out
in~\cite{FM2} in an application to string theory.)  In~\S\ref{subsec:4.2} we
review how invertibility catalyzes a transition to stable homotopy theory:
the analog of~\eqref{eq:158} for an invertible topological field theory is a
map of spectra 
  \begin{equation}\label{eq:281}
     F\:\mth n\longrightarrow \sI. 
  \end{equation}
The domain is the invertible quotient of a higher bordism category, a
Madsen-Tillmann spectrum.  There is freedom to choose the codomain spectrum,
and in~\S\ref{subsec:4.3} we introduce two universal choices.  The first is
(a shift of)~$\ICx$, a ``character dual'' to the sphere spectrum, which is
used to track topological theories on the nose: theories with unequal
partition functions are distinct.  The second universal target spectrum is (a
shift of) the \emph{Anderson dual}~$\IZ$ to the sphere spectrum.  It tracks
deformation classes of invertible theories rather than individual theories.
Significantly, in the spirit of ``derived geometry'', maps into~$\IZ$
classify deformation classes of invertible theories that are not necessarily
topological; the topological theories have finite order in the abelian group
of homotopy classes of maps.  For the application to topological phases one
should include the non-topological theories, as they incorporate nonzero
thermal Hall response.  An example is Kitaev's $E_8$~phase~\cite{K5}.
See~\S\ref{subsec:4.4} for a general discussion, including an interpretation
of maps into~$\IZ$ as a continuous invertible topological field theory.  In
this paper we only use non-topological field theories heuristically and posit
that their deformation classes are encoded in continuous topological field
theories, which we treat rigorously.
 
The main arguments about \emph{extended positivity} occur
in~\S\ref{sec:5}--\S\ref{sec:7}.  Madsen-Tillmann spectra filter Thom
spectra, which leads to a notion of a \emph{stable} invertible topological
field theory: a map out of a Thom spectrum.  For invertible theories a
reflection structure is a lift of~\eqref{eq:281} to an equivariant map of
$\Z/2$-equivariant spectra.  Section~\ref{sec:5} begins with a brief
exposition of spectra and Borel equivariant stable homotopy theory,
sufficient for the considerations in this paper.  The involution on the
domain that models generalized orientation-reversal is straightforward to
construct from the group extension~\eqref{eq:281}.  On the other hand, it is
not clear \emph{a priori} how to model complex conjugation on the codomain,
so in~\S\ref{sec:mssreal-structures} we give an extended discussion
motivating our choice, Definition~\ref{def:2a}.  We conclude~\S\ref{sec:5} by
introducing spectra and spaces of ``higher super lines'', including Hermitian
structures and a higher notion of positivity (Definition~\ref{thm:178},
Definition~\ref{thm:179}).  There is a basic link between \emph{non-extended}
positivity and stability, which we establish in Theorem~\ref{thm:78} and
Theorem~\ref{thm:82} using obstruction theory arguments.  This results in an
intermediate classification (Corollary~\ref{thm:86}) of invertible
topological theories with reflection structure satisfying non-extended
positivity.  We undertake a more systematic study in~\S\ref{sec:8}.  There we
define extended positivity for invertible field theories in terms of higher
super lines and their embellishments.  We give an intuitive construction of
the space of invertible reflection positive theories, and then we identify
its homotopy type in Theorem~\ref{thm:184}, whose proof occupies the second
half of~\S\ref{sec:5}.  Theorem~\ref{thm:110} is a corollary.
 
The third main subject of this paper is what might be called the homotopy
theory of relativistic free fermions.\footnote{A free fermion field theory is
neither topological nor invertible, but it has an associated invertible field
theory.}  There are two distinct scenarios in which a free fermion field
theory gives rise to a deformation class of $n$-dimensional reflection
positive invertible theories.  First scenario: an $(n-1)$-dimensional free
fermion theory has an associated $n$-dimensional invertible anomaly theory,
which is not necessarily topological; our concern here is its deformation
class.\footnote{The anomaly theory lies in \emph{differential} $KO$-theory,
whereas its deformation class lies in topological $KO$-theory.}  Second
scenario: an $n$-dimensional \emph{massive} free fermion theory has a
long-range effective invertible topological field theory approximation,
according to the general principle~(ii) invoked above, applied to a quantum
field theory rather than a lattice system.  We sketch the first scenario in
some detail in~\S\ref{subsec:8.4}, culminating in a formula
(Conjecture~\ref{thm:124}) for the deformation class of the anomaly theory.
Since massive free fermions have trivial anomaly, the starting point is the
group of free fermionic data under direct sum modulo massive free fermionic
data.  The existence of a mass term has a meaning in terms of Clifford
modules (Lemma~\ref{thm:119}), and this produces an identification of the
quotient as a homotopy group of the $KO$-theory spectrum
(Theorem~\ref{thm:120}).  The formula for the deformation class of the
associated anomaly theory is, conjecturally, a product of the
Atiyah-Bott-Shapiro map~\cite{ABS} with the $KO$-theory class of the spinor
data, followed by a Pfaffian map (Conjecture~\ref{thm:124}).  In this paper
we provide a detailed sketch of these ideas; we hope to give a thorough
mathematical treatment in the future.  There is a huge literature on
relativistic free fermion field theories and associated anomalies; the recent
paper~\cite{W1}, which describes several particular cases in detail, provided
motivation and guidance for the general story here.  By contrast, we only
comment briefly (\S\ref{subsubsec:8.4.6}) on the second scenario, beginning
from a massive $n$-dimensional free fermion theory, enough to show that the
starting and ending data match those in the first scenario.  In fact, it is
this second scenario that is relevant to this paper, and in particular the
conjecture~\eqref{eq:150} about its low energy effective field theory is used
in the computations which follow.

To enable detailed comparisons with the physics literature we carry out the
discussion of relativistic free fermions for 10~cases simultaneously.  To
enumerate them we resume group theoretical arguments in~\S\ref{subsec:8.1} to
classify relativistic symmetry groups whose internal subgroup is the unit
reals~$\pmo$, unit complexes~$\TT$, or unit quaternions~$SU_2$.  Restricting
to fermionic theories in which $(-1)^F$~embeds in this internal
subgroup---which implements the ``spin/charge relation''~\cite{SeWi}---we
obtain the 10~groups in question.  They include~$\Spin$, $\Ppm$, and
semidirect products with the various unit scalars.  This ``relativistic
10-fold way'' is a variation on the nonrelativistic case, which is described
in many works: a sample includes~\cite{D,AZ,HHZ,K6,SRFL,FM1,KZ,WS}.
Remark~\ref{thm:136} provides a link to this condensed matter literature: we
compute a group~$I$ of symmetries that preserve points of \emph{space} in a
nonrelativistic setting.  It is this group~$I$ which acts at each lattice
site in a discrete model, and it can be used to compare to the ubiquitous
symmetry tables for fermion lattice systems.  Our uniform treatment is based
on Lemma~\ref{thm:90}, which embeds each symmetry group in a Clifford
algebra.  Usual constructions with Clifford modules---the
Atiyah-Bott-Shapiro-Thom class, Dirac operators and their indices---then
generalize easily.  There is a purely geometric application that we do not
pursue here: index theory on pin and $\textnormal{pin}^c$~manifolds is
straightforward using this embedding.
 
The results of the homotopy theory computations are reported
in~\S\ref{subsec:8.2}.  We provide a table for each of the 10~fermionic
symmetry groups.  In each spacetime dimension~$n\le5$ we compute the group of
free fermion theories (Theorem~\ref{thm:120}), the group of deformation
classes of interacting theories (Theorem~\ref{thm:110}), and the map between
them (Conjecture~\ref{thm:124}).  We make comparisons with the condensed
matter literature where available and find almost total agreement; in the few
cases with a discrepancy we motivate a reexamination of the physics
assertions.  In~\S\ref{sec:13} we outline how the calculations are done and
supply Ext charts that encode the $E_2$-term of the relevant Adams spectral
sequences.  The Ext charts also encode the map to $KO$-theory; in fact, one
of the main tasks in this section is to rewrite the ``twisted''
Atiyah-Bott-Shapiro maps in a more accessible form.  We provide more
explanation of the charts in Appendix~D.  In that appendix we also illustrate
the use of Margolis homology to derive information from the Adams spectral
sequence.  Papers by Campbell~\cite{C} and Beaudry-Campbell~\cite{BeC} give
pedagogical introductions to the Adams spectral sequence and flesh out the
details of our computations.  Notice that whereas Theorem~\ref{thm:110}
computes the group of interacting phases for any symmetry type, the
10~fermionic symmetry types are special in that there is a notion of a free
fermionic phase which does not exist in general.  This leads to a richer
application of homotopy theory and a more stringent test against the
condensed matter literature.

The sections of the paper not yet mentioned contain complements or background
material.  An analog of the spin-statistics theorem in relativistic quantum
field theory holds for reflection positive invertible topological theories,
as we explain in~\S\ref{sec:9}.  Section~\ref{subsec:a3.1} contains a review
of pin groups and Clifford algebras, background for the discussion of the CRT
theorem later in Appendix~\ref{sec:10} and for some of the material
in~\S\ref{sec:8}. 

Beyond the immediate relevance to the study of topological phases, the
successful application of bordism computations to quantum systems is
evidence---perhaps the first substantial test against physics---that the
sparse Axiom System initiated by Segal and Atiyah captures essential features
of quantum field theory.

The lecture series~\cite{F4} provides additional background and discussion on
many of the topics treated here.

\bigskip
 We warmly thank David Ben-Zvi, Jonathan Campbell, Jacques Distler, Mike
Freedman, Davide Gaiotto, Zheng-Cheng Gu, Meng Guo, Matt Hastings, Andre
Henriques, Anton Kapustin, Alexei Kitaev, Max Metlitski, Greg Moore, Andy
Neitzke, Graeme Segal, Nathan Seiberg, Peter Teichner, Constantin Teleman,
Ulrike Tillmann, Senthil Todadri, Kevin Walker, Xiao-Gang Wen, Edward Witten,
and the anonymous referees for many illuminating conversations,
correspondence, comments, and feedback on the first version of this paper.

   \section{Symmetry groups in relativistic quantum field theory}\label{sec:12}

The analytic extension of correlation functions, a consequence of positivity
of energy, provides a powerful constraint on symmetry groups.  We explore the
general structure in~\S\ref{subsec:12.1} from the Wick-rotated point of view.
The rigidity of compact Lie groups is the key idea that underlies our proofs
of structure theorems, such as Theorem~\ref{thm:5}.  One important result is
Theorem~\ref{thm:6}, which constructs a stable group~$H$ from an
$n$-dimensional symmetry group~$H_n$, assuming the spacetime dimension
satisfies~$n\ge3$.  In the expository~\S\ref{subsec:12.2} we recall the
axiomatization of a field theory as a categorified bordism invariant.  We
accommodate general symmetry groups on curved manifolds using reductions of
frame bundles, an analog of the passage from Klein's Erlangen
Programm~\cite{BB} to Cartan's $H$-structures~\cite{S}.

  \subsection{Stabilization of Wick-rotated symmetry groups}\label{subsec:12.1}

The Poincar\'e group is the connected double cover of the identity component
of the isometry group~$\IM$ of $n$-dimensional Minkowski spacetime~$M^n$.
Minkowski spacetime~$M^n$ is assumed equipped with a time orientation, a
choice of component of timelike vectors in the inner product space~$\RM$ of
translations.  Let $\IMt\subset \sI \mstrut _{1,n-1}$ denote the subgroup of
isometries that preserve the time orientation.  Assume~$n\ge2$.  Many
treatments of quantum field theory, for example those based on S-matrix
theory, begin with the assumption that the Poincar\'e group is a
\emph{subgroup} of the (unbroken) global symmetry group~$\HM$ of the theory.
Then the Coleman-Mandula theorem~\cite{CM} asserts that on the level of Lie
algebras there is a splitting as a direct sum of the Lie algebra of
Poincar\'e with the Lie algebra of a \emph{compact} Lie group~$K$.  We find
it more natural to posit from the beginning a homomorphism~$\rho
_n\:\HM\to\IMt$.  After all, $g\in \HM$ acts on the operators in the theory,
and so on the supports of those operators.  For a single point operator, or
local operator, that action is~$\rho _n(g)$.  The relativistic invariance of
the theory is the hypothesis that the image of~$\rho _n$ contains the
identity component of~$\IMt$.  Therefore, the image is either the identity
component or the entire two-component group~$\IMt$.  The kernel of~$\rho _n$
is the group~$K$ of {\it internal\/} symmetries---symmetries that fix the
points of spacetime.  Note that $K$~contains the central element of the
Lorentz group~$\Spin_{1,n-1}$ if that element acts effectively, which by the
spin-statistics theorem happens if and only if the theory contains fermionic
states.  (That element is often denoted~`$(-1)^F$'.  Below we deduce in
general a central element~$k_0\in K$ with $(k_0)^2=1$, and it is identified
with either the central element of~$\Spin$ or the identity element.)  The
internal symmetry group~ $K$ is assumed to be a \emph{compact} Lie
group.\footnote{The global symmetry group of a ``noncompact field theory'',
such as for a free massless $\RR$-valued scalar field theory, may be
noncompact.  Our discussion does not include supersymmetries or higher
symmetries.\label{SYM}}

Assume the translation subgroup~$\RR^{1,n-1}\subset \IMt$ lifts to a normal
subgroup of~$\HM$; see~\cite[Remark~2.13]{FM1} for a justification of this
hypothesis.  Let $H\mstrut _{1,n-1}$~denote the quotient of~$\HM$ by this normal
subgroup of translations.  There is a short exact sequence\footnote{We
overload the symbol~`$\rho _n$'.  Here it denotes the homomorphism induced
from the previous~$\rho _n$ after modding out translations.  Below we use it
for the complexification, restriction to the Euclidean real form, and various
lifts.}
  \begin{equation}\label{eq:12}
     1\longrightarrow K\longrightarrow H\mstrut _{1,n-1}\xrightarrow{\;\;\rho
     _n\;\;} \OMt 
  \end{equation}
where the image of~$\rho _n$ contains the identity component of~$\OMt\subset
O\mstrut _{1,n-1}$, by the relativistic invariance of the theory.  The CRT
theorem, reviewed in~\S\ref{subsec:a3.2}, gives a larger symmetry group.  A
fundamental consequence of the positivity of energy\footnote{The dual to the
cone of forward timelike vectors determines the notion of positive energy.}
in quantum field theory, also reviewed in~\S\ref{subsec:a3.2}, is a
holomorphic extension\footnote{See~\cite{KS} for a geometric version on
curved manifolds.} of correlation functions on which the
complexification~$H_n(\CC)$ of~$H_{1,n-1}$ acts as symmetries.  There is an
exact sequence
  \begin{equation}\label{eq:13}
     1\longrightarrow K(\CC)\longrightarrow H_n(\CC)\xrightarrow{\;\;\rho
     _n\;\;} O_n(\CC)
  \end{equation}
of complex Lie groups.  The \emph{Wick-rotated theory} has a \emph{compact}
real form~$H_n$ of~$H_n(\CC)$ as symmetry group such that $H_n$~ fits
into the exact sequence
  \begin{equation}\label{eq:14}
     1\longrightarrow K\longrightarrow H_n\xrightarrow{\;\;\rho _n\;\;}
     O_n
  \end{equation}
of compact Lie groups with the \emph{same} compact kernel~$K$ as
in~\eqref{eq:12}.  The image of this~$\rho _n$ is either~$O_n$ or~$SO_n$,
depending on whether the relativistic theory has spatial reflections or not;
equivalently, by the CRT theorem, whether it has time-reversal symmetry or
not.  

  \begin{definition}[]\label{thm:153}
 The \emph{symmetry type} of a quantum field theory is a pair~$(H_n,\rho _n)$
of a compact Lie group~$H_n$ and a homomorphism $\rho _n\:H_n\to O_n$ whose
image contains~$SO_n\subset O_n$.  The kernel~$K$ of~$\rho _n$ is called the
\emph{group of internal symmetries}.  We require that the anti-Wick rotation
to Minkowski spacetime has a Lorentzian real form~\eqref{eq:12} with compact
internal symmetry group~$K=\ker\rho _n$.
  \end{definition}

\noindent
 The caveats in footnote~\footref{SYM} apply.  See Remark~\ref{thm:32} below
for an example of a pair~$(H_n,\rho _n)$ that does not satisfy the anti-Wick
rotation condition.  The symmetry type is a basic structure in a quantum
field theory, useful to articulate explicitly in any example.

Define $SH_n=\rho _n\inv (SO_n)$ and let $\tSH_n$~be the double cover
of~$SH_n$ constructed from the spin double cover of~$SO_n$.  These compact
Lie groups are usefully encoded in the pullback diagram
  \begin{equation}\label{eq:15}
     \begin{gathered}
     \xymatrix{
     1\ar[r]&K\ar[r]\ar@{=}[d]&\tSH_n\ar@{->>}[d]^{2:1}\ar[r]^{\rho _n}&
     \Spin_n\ar@{->>}[d]^{2:1} \ar[r] & 1\\ 
     1\ar[r]&K\ar[r]\ar@{=}[d]&SH_n\ar@{^{(}->}[d]^{1:2}\ar[r]^{\rho _n}&
     SO_n\ar@{^{(}->}[d]^{1:2} \ar[r] & 1\\ 1\ar[r] &K \ar[r]
     &H_n\ar[r]^{\rho _n} &O_n} 
     \end{gathered} 
  \end{equation}
If $\rho _n\:H_n\to O_n$ is surjective, define $\tH_n$ as the
pullback\footnote{See~\S\ref{subsec:a3.1} for a review of pin groups.}
  \begin{equation}\label{eq:19}
     \begin{gathered} \xymatrix{
     1\ar[r]&K\ar[r]\ar@{=}[d]&\tH_n\ar@{->>}[d]\ar[r]^{\rho _n}&
     \Pp_n\ar@{->>}[d] \ar[r] & 1\\
     1\ar[r]&K\ar[r]&H_n\ar[r]^{\rho _n}&
     O_n\ar[r] & 1} \end{gathered} 
  \end{equation}
The restriction of~$\tH_n$ over $\Spin_n\subset \Pp_n$ is~$\tSH_n$.  Let
$\mathfrak{k},\mathfrak{h}_n,\mathfrak{o}_n$ denote the Lie algebras of
$K,H_n,O_n$, respectively.  The following theorem makes precise the sense in
which the entire symmetry group is nearly the product of (Wick-rotated)
spacetime symmetries and internal symmetries.  In our approach to symmetry it
plays the role of the Coleman-Mandula theorem.

  \begin{theorem}[]\label{thm:5}
 \ \begin{enumerate}[{\textnormal(}1{\textnormal)}]

 \item There is a splitting $\mathfrak{h}_n\cong \mathfrak{o}'_n\oplus
\mathfrak{k}$, and $\rho _n$~induces an isomorphism of Lie algebras
$\mathfrak{o}_n'\xrightarrow{\;\cong \;} \mathfrak{o}_n$.

 \item If~$n\ge3$ there is an isomorphism $\tSH_n\cong \Spin_n\times K$.
Hence there exists a central element~$k_0\in K$ with $(k_0)^2=1$ and an
isomorphism
  \begin{equation}\label{eq:16}
     SH_n\cong \Spin_n\times K\bigm / \langle(-1,k_0)\rangle,
  \end{equation}
where $\langle (-1,k_0) \rangle$ is the cyclic group generated by~$(-1,k_0)$.

 \item If $n\ge3$ and $\rho _n\:H_n\to O_n$ is surjective, then there exists
a group extension 
  \begin{equation}\label{eq:225}
     1\longrightarrow K\longrightarrow J\longrightarrow \pmo\longrightarrow 1
  \end{equation}
and a pullback diagram of group extensions 
  \begin{equation}\label{eq:226}
     \begin{gathered} \xymatrix{
     1\ar[r]&K\ar[r]\ar@{=}[d]&\tH_n\ar@{->>}[d]\ar[r]^{\rho _n}&
     \Pp_n\ar@{->>}[d] \ar[r] & 1\\ 1\ar[r]&K\ar[r]&J\ar[r]&
     \pmo\ar[r] & 1} \end{gathered} 
  \end{equation}
There is an isomorphism
  \begin{equation}\label{eq:227}
     H_n\cong \tH_n \bigm / \langle(-1,k_0)\rangle. 
  \end{equation}

 \end{enumerate} 
  \end{theorem}

\noindent
 The pullback~\eqref{eq:226} shows that the failure of~$\tH_n$ to be a
product is encoded in the group extension~\eqref{eq:225}, which is
independent of~$n$.

  \begin{corollary}[]\label{thm:113}
 There is a canonical homomorphism $\Spin_n\to H_n$ under which the image of
the central element~$-1\in \Spin_n$ is $k_0\in K$. 
  \end{corollary}

\noindent 
 This homomorphism anti-Wick rotates back to a homomorphism of the
Poincar\'e group into the total symmetry group~$\HM$ of the relativistic
theory, the traditional starting point for discussions of symmetry in quantum
field theory.

  \begin{remark}[]\label{thm:32}
 For~$n=2$ we can only conclude that $\tSH_2$~is isomorphic to a semidirect
product of~$\Spin_2$ and~$K$.  An example is $SH_2=SO_2\ltimes O_2$, where a
rotation~$R\in SO_2$ acts on~$O_2$ by the automorphism that is the identity
on~$SO_2\subset O_2$ and composes a reflection with~$R$.  Alternatively,
$SH_2\cong \zt\ltimes(\TT\times \TT)$ where the involution on~$\TT\times \TT$
is $(\lambda\mstrut _1,\lambda\mstrut _2)\mapsto(\lambda \mstrut _1,\lambda
_1\inv \lambda _2\inv )$.
  \end{remark}

  \begin{proof}[Proof of Theorem~\ref{thm:5}]
 Split the Lie algebra $\mathfrak{h}_n=[\mathfrak{h}_n,\mathfrak{h}_n]\oplus
\mathfrak{z}$, where $\mathfrak{z}\subset \mathfrak{h}_n$ is the center, and
let $\mathfrak{o}_n'$ be the orthogonal complement of the
ideal~$\mathfrak{k}\cap [\mathfrak{h}_n,\mathfrak{h}_n]\subset
[\mathfrak{h}_n,\mathfrak{h}_n]$ with respect to the nondegenerate Killing
form on the semisimple Lie algebra~$[\mathfrak{h}_n,\mathfrak{h}_n]$.  Then
$\rho _n$~induces an isomorphism $\mathfrak{o}'_n\to\mathfrak{o}_n$, which
proves~(1).  The exponential of~$\mathfrak{o}'_n$ is a closed Lie subgroup
$S\subset \tSH_n$ which locally projects diffeomorphically onto~$\Spin_n$
under~$\rho _n$, so is isomorphic to~$\Spin_n$.  It follows that $\tSH_n\cong
S\ltimes K$.

We claim this semidirect product is a direct product if~$n\ge 3$.  To see
this observe that conjugation by~$s\in S$ induces an automorphism~$\alpha
(s)$ of~$K$ which is the identity on the identity component $K^0\subset K$,
since the Lie algebra of~$S$ commutes with the Lie algebra of~$K$.  Since
$S$~is connected, the induced automorphism of~$\pi _0K$ is also trivial.
Hence on each component of~$K$ the automorphism~$\alpha (s)$ is left
multiplication by an element~$z(s)\in Z^0$ in the center of~$K^0$.  (Proof:
Write~$\alpha =\alpha (s)$ and suppose $\alpha (k)=zk$ for some~$k$ in that
component and $z\in K^0$.  Any other element of that component has the
form~$kk_0$ for~$k_0\in K^0$, and $\alpha (kk_0)=z(kk_0)$.  But we can also
write any element in the component as~$k_0'k$ for some $k_0'\in K^0$, and
$\alpha (k_0'k)=k_0'zk=(k_0'z{k_0'}\inv )(k_0'k)$ from which $k_0'z{k_0'}\inv
=z$.  This holds for every~$k_0'\in K^0$, from which we deduce $z\in Z^0$.)
Next, $\Spin_n$~acts trivially on~$Z^0$; this follows since the outer
automorphism group of a compact Lie group is discrete, every inner
automorphism of the abelian group~$Z^0$ is trivial, and $\Spin_n$~is
connected.  Hence the map $s\mapsto z(s)$ is a homomorphism $S\to Z^0$.  But
if~$n\ge 3$ the Lie group $S\cong \Spin_n$ has no nontrivial homomorphisms to
an abelian Lie group.
 
Assume $\rho _n\:H_n\to O_n$ is surjective.  We claim $\Spin_n\subset
\tSH_n\subset \tH_n$ is a normal subgroup.  Fix~$\tih\in \tH_n$ such that
$\rho _n(\tih)=e_2\in \Pp_n$.  Conjugation by~$e_2$ induces an involution
$\alpha :\Spin_n\to\Spin_n$.  It lifts to an automorphism of~$\tSH_n\cong
\Spin_n\times K$ defined as conjugation by~$\tih$, so there is an induced
automorphism $\beta \:K\to K$ and a homomorphism $\gamma \:\Spin_n\to K$.

  \begin{lemma}[]\label{thm:7}
 If~$n\ge3$, then the homomorphism~$\gamma $ is trivial. 
  \end{lemma}

  \begin{proof}
 Define $\tH_n(\CC)$ by pulling back as in~\eqref{eq:19} using the
complexified groups~\eqref{eq:13}; pullback over the Lorentzian real forms to
obtain the first of the pair of real forms $\tH_{1,n-1}\subset
\tH_n(\CC)\supset \tH_n$.  Note that $\tih$~lies in each of these groups, and
conjugation by~$\tih$ preserves both real forms.  Thus we obtain a
homomorphism $\Spin_n(\CC)\to K(\CC)$ that restricts to $\gamma \:\Spin_n\to
K$ and to a homomorphism $\Spin_{1,n-1}\to K$.  Now if $\gamma $~is
nontrivial, then so is the induced map on Lie algebras, and since
$\mathfrak{o}_n$~is simple, $\dot\gamma \:\mathfrak{o}_n\to\mathfrak{k}$ is
injective.  It follows that the Lie algebra map $\mathfrak{o}_{1,n-1}\to
\mathfrak{k}$ is also injective.  Hence $\mathfrak{k}$~contains a subalgebra
isomorphic to~$\mathfrak{o}_{1,2}\cong \mathfrak{s}\mathfrak{l}_2\RR$.  The
Killing form on~$\mathfrak{k}$ induces a nonzero semidefinite invariant
symmetric bilinear form on the simple Lie
algebra~$\mathfrak{s}\mathfrak{l}_2\RR$, which is impossible since every
invariant symmetric form on~$\mathfrak{s}\mathfrak{l}_2\RR$ is a multiple of
the Killing form, which is indefinite and nondegenerate.
  \end{proof}

It follows that $\Spin_n\subset \tH_n$ is a normal subgroup.  Set
$J=\tH_n/\Spin_n$.  Then \eqref{eq:226}~follows from~\eqref{eq:19} and
\eqref{eq:227}~follows from the fact that the kernel of $\tH_n\to H_n$ equals
the kernel of $\tSH_n\to SH_n$.  This completes the proof of
Theorem~\ref{thm:5}.
  \end{proof}

  \begin{remark}[]\label{thm:8}
 Lemma~\ref{thm:7} is not true without using the anti-Wick rotation back to
Lorentzian signature.  Namely, let~$n=3$ and $H_3=\zt\ltimes(SO_3\times
SO_3)$, where the nontrivial element of~$\zt$ acts by shearing
$(g_1,g_2)\mapsto (g_1,g_1g_2)$; the homomorphism~$\rho _3$ that kills the
last factor~$K=SO_3$ maps $H_3\to O_3$ and sends the generator of~$\zt$ to
the central element $-1\in O_3$.  The reader can check that $\gamma
\:\Spin_3\to SO_3$ is surjective.  But $H_3$~is not a possible symmetry group
because of the anti-Wick rotation, as in the proof of Lemma~\ref{thm:7}.
  \end{remark}

If we restrict the internal symmetry group to only include the image of the
central element $-1\in \Spin_n$ under $\Spin_n\to H_n$, then there are five
possibilities.  In these cases $K$~is trivial or~$K\cong \pmo$.  Let $\mu
_4=\{\pm1,\pm\sqrt{-1}\}$ be the multiplicative group of fourth roots of
unity, and define~$E_n\subset O_n\times \mu _4$ as the subgroup
of~$(A,\lambda )$ such that $\det A=\lambda ^2$.

  \begin{proposition}[]\label{thm:65}
 Assume~$n\ge3$.  If the internal symmetry group~$K$ is trivial, then
$H_n\cong SO_n$ or $H_n\cong O_n$.  If $K\cong \pmo$ is cyclic of order two,
then there are six possibilities for~$H_n$ up to isomorphism: $SO_n\times
\pmo$, $\Spin_n$, $O_n\times \pmo$, $E_n$, $\Pp_n$, and~$\Pm_n$.
  \end{proposition}

  \begin{proof}
 The first statement is clear from the fact that the image of~$\rho _n$
in~\eqref{eq:14} is either~$SO_n$ or~$O_n$.  The group extensions by~$\pmo$
are central and are classified up to isomorphism by the cohomology group
$H^2(BSO_n;\pmo)\cong \zt$ or $H^2(BO_n;\pmo)\cong \zt\times \zt$, depending
on the image of~$\rho _n$, and it is not difficult to work out what the
groups~$H_n$ are.
  \end{proof}

\noindent
 The non-identity element of~$K$ in $SO_n\times \pmo$, $O_n\times \pmo$,
and~$E_n$ is not the image of the central element~$-1\in \Spin_n$.  This
leaves the five basic symmetry types listed in the following table:
  \begin{equation}\label{eq:10}
     \begin{tabular}{ c@{\hspace{2em}} c@{\hspace{2em}} c@{\hspace{2em}}
     c@{\hspace{2em}} } 
     \toprule 
     states/symmetry&$H_n$&$K$&$\phantom{j}k_0$\\ \midrule \\[-8pt] 
     bosons only & $SO_n$&$\{1\}$&$\phantom{+}1$\\
     fermions allowed & $\Spin_n$&$\pmo$&$-1$\\
     bosons, time-reversal~($T$) & $O_n$&$\{1\}$&$\phantom{+}1$\\
     fermions, $T^2=(-1)^F$ & $\Pp_n$&$\pmo$&$-1$\\
     fermions, $T^2=\id$ & $\Pm_n$&$\pmo$&$-1$\\ \bottomrule
     \end{tabular} 
  \end{equation}
Appendix~\ref{sec:10} reviews the pin groups and justifies the Wick rotation
of time-reversal that leads to the last three lines in the first column of
the table.
 
The main result in this section is a stabilization of~$H_n$ for increasing
dimensions, as needed in Theorem~\ref{thm:110}.  Throughout this
paper for $k<\ell $ we use the embedding 
  \begin{equation}\label{eq:51}
     \begin{aligned} O_k&\longrightarrow O_\ell \\ A&\longmapsto
      \begin{pmatrix} I_{\ell-k }\\&A \end{pmatrix}\end{aligned} 
  \end{equation}
of orthogonal groups, where $I$~denotes the identity matrix.

  \begin{theorem}[]\label{thm:6}
 Assume~$n\ge 3$.  There exist compact Lie groups~$H_m$, $m>n$, and
homomorphisms~$i_n,\rho _n$ which fit into the commutative diagram
  \begin{equation}\label{eq:17}
     \begin{gathered} \xymatrix{H_n\ar@{^{(}->}[r]^{i_n}\ar[d]^{\rho _n}
     &H_{n+1}\ar@{^{(}->}[r]^{i_{n+1}}\ar[d]^{\rho _{n+1}}
     &H_{n+2}\ar@{^{(}->}[r]\ar[d]^{\rho _{n+2}} & \dots \\ O_n
     \ar@{^{(}->}[r]& O_{n+1} \ar@{^{(}->}[r]& O_{n+2} \ar@{^{(}->}[r]&
     \dots} \end{gathered} 
  \end{equation}
in which squares are pullbacks.  
  \end{theorem}

\noindent
  The stabilization is usually apparent, even when~$n=2$ and
Theorem~\ref{thm:6} does not apply.  For example, if $H_n=\Pp_n\ltimes
\TT\bigm/\langle (-1,-1) \rangle$, where $\Pp_n$~acts on~$\TT=U_1$ through
its components by conjugation, then $H_m=\Pp_m\ltimes \TT\bigm/\langle
(-1,-1) \rangle$.  (We encounter this and related groups in~\S\ref{sec:8}.)

  \begin{remark}[]\label{thm:30}
 For $m<n$, define $H_m$ and the homomorphism $\rho _m\:H_m\to O_m$ by a
pullback square:  
  \begin{equation}\label{eq:50}
     \begin{gathered} \xymatrix{H_m\ar@{-->}[r]^{} \ar@{-->}[d]_{\rho _m} &
     H_n\ar[d]^{\rho _n} \\ O_m\ar@{^{(}->}[r] & O_n} \end{gathered} 
  \end{equation}
  \end{remark}

  \begin{remark}[]\label{thm:33}
 The pullback diagram~\eqref{eq:17} and the fact that $\rho _{m+1}(H_{m+1})$
acts transitively on the $m$-sphere imply diffeomorphisms
  \begin{equation}\label{eq:52}
     H_{m+1}/H_m\cong O_{m+1}/O_m\cong S^m 
  \end{equation}
  \end{remark}

  \begin{proof}[Proof of Theorem~\ref{thm:6}]
 In view of~\eqref{eq:16}, for~$m>n$ define $SH_m:=\Spin_m\times K\bigm/
\langle (-1,k_0) \rangle$, and so obtain a stabilization over~$SO_m$.  If
$\rho _n(H_n)=SO_n$ this completes the proof.  If not, define~$\tH_m$ as the
pullback 
  \begin{equation}\label{eq:228}
     \begin{gathered} \xymatrix{
     1\ar[r]&K\ar[r]\ar@{=}[d]&\tH_m\ar@{->>}[d]\ar[r]&
     \Pp_m\ar@{->>}[d] \ar[r] & 1\\ 1\ar[r]&K\ar[r]&J\ar[r]& \pmo\ar[r] & 1}
     \end{gathered} 
  \end{equation}
and 
  \begin{equation}\label{eq:229}
     H_m\cong \tH_m \bigm / \langle(-1,k_0)\rangle. 
  \end{equation}
\vskip-2.8em
 \end{proof} 

\bigskip

Theorem~\ref{thm:6} allows us to speak about symmetry types in quantum field
theory independent of dimension.  Set 
  \begin{equation}\label{eq:136}
     H=\colim\limits_{n\to\infty }H_n. 
  \end{equation}
For $H_n=SO_n$ we obtain $H=SO_\infty =SO$.  Thus we can speak of `oriented
theories'=`$SO$ theories', `Spin theories', `$\Pp$ theories', etc.  The
colimit of~\eqref{eq:17} is a homomorphism 
  \begin{equation}\label{eq:278}
     \rho \:H\longrightarrow  O. 
  \end{equation}
The \emph{symmetry type} of a theory (Definition~\ref{thm:153}) can be taken
to be the pair~$(H,\rho )$ in place of~$(H_n,\rho _n)$.

  \subsection{Curved manifolds and bordism categories with $H_n$-structure}\label{subsec:12.2}

Fix an $n$-dimensional relativistic quantum field theory with symmetry
type~$(H_n,\rho _n)$.  A ``coupling to background gravity'' means that we
define the theory on each $n$-dimensional smooth Riemannian manifold~$X$.
The $H_n$-symmetry is no longer global; it is tangential and encoded in a
reduction of the orthonormal frame bundle to~$H_n$.  Let $\sBO(X)\to X$
denote the principal $O_n$-bundle of frames: a point of~$\sBO(X)$ is an
orthonormal basis of the tangent space at a point of~$X$.  If $P\to X$ is a
principal $H_n$-bundle, define the principal $O_n$-bundle $\rho _n(P)=P\times
\mstrut _{H_n}O_n\to X$ via mixing: $[ph,g]=[p,\rho _n(h)g]$ for all $p\in
P$, $g\in O_n$, and $h\in H_n$.

  \begin{definition}[]\label{thm:45}
 An \emph{$H_n$-structure} is a pair $(P,\theta )$ consisting of a principal
$H_n$-bundle $P\to X$ equipped with an isomorphism of principal $O_n$-bundles
$\sBO(X)\xrightarrow{\;\theta \;}\rho _n(P)$.  An \emph{$H_n$-manifold} is a
Riemannian $n$-manifold endowed with an $H_n$-structure.  A
\emph{differential $H_n$-structure} is a connection~$\Theta $ on $P\to X$
with the property that $\theta $~maps the Levi-Civita connection to~$\rho
_n(\Theta )$. 
  \end{definition}

\noindent
 It also makes sense to have an $H_n$-structure on a Riemannian manifold of
dimension~$\ell >n$, via the composition $H_n\xrightarrow{\rho
_n}O_n\hookrightarrow O_{\ell }$, and on a manifold of dimension~$k<n$ by
stabilizing the $O_k$-frame bundle to a principal $O_n$-bundle via the
inclusion $O_k\hookrightarrow O_n$.  The Stability Theorem~\ref{thm:6}
implies that an $H_n$-manifold has an induced $H_m$-structure for all~$m\ge
n$.  The same applies to the differential refinements.

  \begin{example}[]\label{thm:9}
 In bosonic theories of electromagnetism, $K=\TT$ is the group~$U_1$ of unit
norm complex numbers, at least in the absence of further global symmetries.
If there is no time-reversal symmetry, then $H_n=SO_n\times \TT$.  Thus $P\to
X$ is the fiber product of the frame bundle with a principal $\TT$-bundle,
which is usually equipped with a connection, or gauge field.  In theories of
electromagnetism with fermions we still have~$K=\TT$, but now the center
$-1\in \Spin_n$ of the spin group is identified\footnote{This assumes the
spin/charge relation that particles of even electromagnetic charge are bosons
while those of odd electromagnetic charge are fermions; see~\cite{SeWi} for
more discussion.} with~$-1\in \TT$ and so
  \begin{equation}\label{eq:134}
     H_n=\Spin^c_n = \Spin_n\times \TT\bigm/\pmo
  \end{equation}
is the group introduced in~\cite{ABS}.  In other words, the Riemannian
manifold~$X$ has a $\Spin^c$-structure.  If, in addition, there is
time-reversal symmetry, then there are several different extensions,
including the Atiyah-Bott-Singer group $\Pin^c_n$; see
Proposition~\ref{thm:60} for the complete classification.
  \end{example}

  \begin{example}[]\label{thm:141}
 For $H_n=O_n\times K$ an $H_n$-structure on a Riemannian manifold is an
auxiliary principal $K$-bundle, and a \emph{differential} $H_n$-structure is
a connection on that bundle.  For $H_n=\Spin^c_n$ the differential structure
is usually called a \emph{$\textnormal{spin}^c$~connection}.
  \end{example}

The basic properties of Wick-rotated correlation functions on all
\emph{compact} manifolds simultaneously are encoded in the powerful framework
of \emph{bordism categories}, following the fundamental work of
Segal~\cite{Se1} and Atiyah~\cite{A1}.  Topological field theories do not
depend on the metric, nor do they require differential structures, and for
the most part we focus on topological theories and so on topological bordism
categories.  The geometric case is used as motivation; we make some comments
in Remark~\ref{thm:142}.

For the topological bordism category~$\Bord_{\langle n-1,n \rangle}(H_n)$
defined in the next paragraph, we drop the connection.  We can also drop the
Riemannian metric, as just mentioned, and to do so we would replace the
compact Lie group~$H_n$ and homomorphism $\rho _n\:H_n\to O_n$ with a
canonically associated noncompact real Lie group~$\uH_n$ and homomorphism
$\uH_n\to GL_n\RR$.  We give the construction in Appendix~\ref{sec:14}.  Our
field theories are \emph{discrete} in the sense that the partition function
in $\CC$-valued and $\CC$~has the discrete topology.  Hence the theories
factor through the topological bordism category built with $\uH_n$-manifolds
in place of $H_n$-manifolds.  So we follow standard usage (``spin theories'',
etc.) and use the compact Lie group~$H_n$, but no connections.

Define a \emph{topological bordism category} $\Bord_{\langle n-1,n
\rangle}(H_n)$ as follows.  An object is a compact $(n-1)$-manifold~$Y$
without boundary, equipped with an $H_n$-structure $Q\to Y$ and an ``arrow of
time''.  To make sense of an $H_n$-structure on an $(n-1)$-manifold we
stabilize the tangent bundle of~$Y$ to a rank~$n$ bundle
$\underline{\RR}\oplus TY\to Y$ by summing with a trivial line bundle,
thought of as a normal direction into $n$~dimensions.  In this topological
setting the Riemannian metric is not present; in the geometric setting of
Remark~\ref{thm:142}, an object in a geometric bordism category is an
$(n-1)$-manifold with a germ of an embedding in an $n$-manifold.  The arrow
of time is a normal orientation.  In the topological setting only the
tangential information is relevant---we can drop the germ---and the arrow of
time is an orientation of the trivial subbundle $\underline{\RR}\to Y$ of
$\underline{\RR}\oplus TY\to Y$.  Nonetheless, even in this topological case
it is illuminating to use the product germ $(-\epsilon ,\epsilon )\times Y$
for some $\epsilon >0$ and replace $\underline{\RR}\oplus TY\to Y$ by the
tangent bundle to the germ.  A morphism $X\:Y_0\to Y_1$ is an equivalence
class of compact $n$-manifolds~$X$ with $H_n$-structure $P\to X$ and an
isomorphism of the boundary $\partial X\xrightarrow{\cong }Y_0\amalg Y_1$
with the disjoint union of the incoming~$Y_0$ and the outgoing~$Y_1$; the
equivalence relation is diffeomorphism commuting with all of the data.  The
isomorphisms include the $H_n$-structures and under those isomorphisms the
orientation of the trivial bundle~$\underline{\RR}\to Y_i$ must line up with
the incoming normal to the boundary for~$i=0$ and with the outgoing normal to
the boundary for~$i=1$.  In other words, the arrow of time is used to
distinguish incoming and outgoing boundary components of morphisms.
Composition of morphisms is gluing of bordisms.  There is a additional
commutative composition law on the category---disjoint union---and with this
structure $\bne$~is a \emph{symmetric monoidal category}.  See~\cite{L,CS}
for detailed accounts.

A Wick-rotated field theory is a linear representation of a bordism category.

  \begin{definition}[]\label{thm:42}
 A \emph{topological field theory with Wick-rotated vector symmetry
group~$H_n$} is a symmetric monoidal functor
  \begin{equation}\label{eq:25}
     F\:\bne\longrightarrow \Vect_{\CC} 
  \end{equation}
to the symmetric monoidal category of complex vector spaces under tensor
product.   
  \end{definition}

\noindent
 Much has been written about this definition, and we defer to previous
accounts---such as the original~\cite{A1} and the recent
survey~\cite[\S\S2--4]{F2} and the references therein---for more exposition.
Here we simply make the connection to point operators\footnote{These are
usually called `local operators' in the physical literature, but we use
`point' rather than `local' to distinguish point operators from line
operators and higher dimensional analogs, since those too are local.} and
their correlation functions.

  \begin{remark}[Vector spaces of point operators]\label{thm:40}
 The sphere~ $S^{n-1}$ is the link of a point in $n$~dimensions, i.e., it is
the boundary of a small ball about the point.  Therefore, the vector
space~$V:=F(S^{n-1})$ is the space of point operators in a topological field
theory; in a geometric theory we take a limit as the radius of the sphere
shrinks to zero.  If the theory has total symmetry group~$H_n$, then the
sphere has an $H_n$-structure and the vector space of point operators depends
on it.  If $H_n=SO_n\times K$ or $H_n=O_n\times K$, the extra data is a
principal $K$-bundle ~$Q\to S^{n-1}$ (with connection).  So there is a vector
space~$V_Q$ of point operators for each~$Q$.  The group~$\Aut Q$ of global
gauge transformations acts on~$V_Q$.  For the trivial $K$-bundle this is the
familiar representation of the global symmetry group~$K$ on local operators.
If $K$~is finite, then the ``twist operators'' for $Q\to\cir$ nontrivial are
familiar in~$n=2$.  They are also familiar when $H_2=\Spin_2$, in which case
the operators associated to the nonbounding spin circle create a defect at
the excised point which changes the spin structure on the punctured surface.
In $n=3$~dimensions, if $H_3$~is a Cartesian product of~ $SO_3$ and~ $K=\TT$,
then the twist operators in some sense create a magnetically charged
instanton for the global symmetry group~$K$; the $\ZZ$-grading from the
action of~$K$ on the point operators measures the electric charge.
  \end{remark}

  \begin{remark}[Correlation functions of point operators]\label{thm:41}
 Let $M$~be a closed $n$-manifold. Fix points $x_1,\dots ,x_k$ of~$M$ at
which we place local operators.  Let $X$~be the compact manifold with
boundary obtained from~$M$ by removing small open balls about each~$x_i$;
regard~$X$ as a bordism  
  \begin{equation}\label{eq:135}
     X\:\bigsqcup\limits_i S^{n-1}(x_i)\longrightarrow \emptyset ^{n-1} 
  \end{equation}
from the disjoint union of the $k$~boundary spheres to the empty manifold.
Equip the manifold~$X$ with an $H_n$-structure~$P$, and let $Q_i$~denote
its restriction to the $i^{\textnormal{th}}$~sphere.  Applying the
theory~\eqref{eq:25} we obtain a homomorphism
  \begin{equation}\label{eq:26}
     F(X;P)\:\underbrace{V_{Q_1}\otimes \cdots\otimes
     V_{Q_k}}_{\textnormal{$k$~times}}\longrightarrow \CC  
  \end{equation}
which, evaluated on operators~$\mathcal{O}_1,\dots ,\mathcal{O}_k$, is usually
written $\langle \mathcal{O}_1(x_1)\cdots \mathcal{O}_k(x_k)  \rangle\mstrut
_M$.   
  \end{remark}

  \begin{figure}[ht]
  \centering
  \includegraphics[scale=1.1]{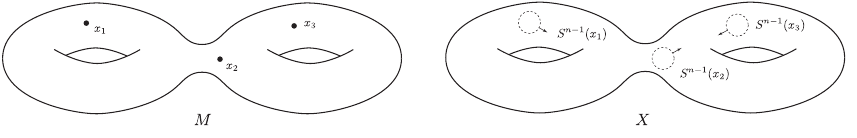}
  \caption{Correlation functions}\label{fig:6}
  \end{figure}

  \begin{remark}[Remark about non-topological theories]\label{thm:142}
 Wick-rotated field theories which are not topological can also be formulated
as functors on bordism categories, but now the objects and morphisms have a
geometric structure.  The references~\cite{Se2,KS,ST} develop this idea in
various directions.  We confine ourselves here to a few heuristic formal
remarks.  Analogous to the topological bordism category~ $\Bord_{\langle
n-1,n \rangle}(H_n)$ we envision a \emph{geometric bordism category}
$\Bord^{\nabla }_{\langle n-1,n \rangle}(H_n)$ whose objects and morphisms
are smooth manifolds with \emph{differential} $H_n$-structures
(Definition~\ref{thm:45}).  An object is a closed $(n-1)$-manifold equipped
with an infinite jet of an embedding into an $n$-dimensional manifold with
differential $H_n$-structure and an arrow of time.  A morphism is a compact
$n$-manifold with differential $H_n$-structure together with a partition of
the boundary and boundary isomorphisms as in the topological case.  As in the
topological case~\eqref{eq:25}, a field theory is a functor with domain
$\Bord^{\nabla }_{\langle n-1,n \rangle}(H_n)$ and codomain a suitable
symmetric monoidal category of topological vector spaces.  We want the
correlation functions and vector spaces to vary smoothly in smooth families,
so the whole structure must be ``sheafified'' over the category of smooth
manifolds and smooth maps~\cite[\S2]{ST}.
  \end{remark}

   \section{Unitarity and Wick rotation}\label{sec:2}

We recall in~\S\ref{subsec:2.1} how positivity of energy leads to Wick
rotation in quantum mechanics, and describe reflection positivity in that
context.  The usual quantum mechanical context for reflection positivity is
recollected in~\S\ref{subsec:2.2}, with attention paid to nontrivial internal
symmetry groups.  These preliminaries are motivation for~\S\ref{subsec:2.3},
where we encode the reflection structure in a novel way via a coextension of
the Wick-rotated vector symmetry group to a $\zt$-graded group, constructed
from a hyperplane reflection.  (We give a topological account of the
construction in Appendix~\ref{sec:16}.)  The new components act antilinearly
on the Hilbert space of states.  It is this formulation that we use in the
rest of the paper.

  \subsection{Wick rotation in quantum mechanics}\label{subsec:2.1}

A quantum mechanical system, according to basic axioms, consists of a complex
separable Hilbert space~$\sH$ equipped with a self-adjoint operator~$H$, the
Hamiltonian.  The group~$\RR$ of time translations is represented unitarily
on~$\sH$:  
  \begin{equation}\label{eq:1}
     \begin{aligned} \RR&\longrightarrow U(\sH) \\ t&\longmapsto
      e^{-itH/\hbar},\end{aligned} 
  \end{equation}
where $i$~is a choice of complex number such that $i^2=-1$.  If we assume
positivity of energy---that $H$~is a nonnegative self-adjoint operator---then
real time evolution~\eqref{eq:1} is the boundary value of a holomorphic
semigroup of bounded operators defined on the lower half plane $\sT=\RR -
\sqmo\, \RR^{>0}\subset \CC$.  The semigroup of \emph{imaginary time
evolution} is the restriction to $-\sqmo\,\RR^{>0}$, which is the semigroup
  \begin{equation}\label{eq:3}
     \tau \longmapsto e^{-\tau H/\hbar},\qquad \tau >0. 
  \end{equation}
The transition from~\eqref{eq:1} to~\eqref{eq:3} is called \emph{Wick
rotation}.   
 
The unitarity of time evolution manifests in the reality of the
semigroup~\eqref{eq:3}.

  \begin{example}[Particle on the circle]\label{thm:1}
 Let $\AA^1$~denote the affine\footnote{We (pedantically) distinguish the
affine time line~$\AA^1$ from the group~$\RR$ of translations of time, which
appears in~\eqref{eq:1}: after all, a 1-hour seminar and a seminar ending
at~1:00 can be quite different.} time line.  The trajectory of a particle on
the circle is a function $\lambda (s)=e^{ix(s)},\;s\in \AA^1$; the lagrangian
density is $L=\frac 12 \dot x^2\,|ds|$.  The ensuing quantum mechanical
system has Hilbert space $\sH=L^2(S^1;\CC)$, Hamiltonian the Laplace
operator~$H=\Delta $ (up to a constant), and imaginary time evolution the
heat operator $\tau \mapsto e^{-\tau \Delta }$.

It is illuminating to add a ``$\theta $-angle'' to this system;
see~\cite[Appendix~D]{GKKS}, for example.
Orient~$\cir$ and fix $\omega \in \Omega ^1(\cir)$ with $\int_{\cir}\omega
=1$.  Then for a fixed constant~$\theta \in \RR$ define the lagrangian
  \begin{equation}\label{eq:4}
     L = \frac 12\dot x^2\,|ds| - \theta \lambda ^*(\omega ). 
  \end{equation}
In this classical theory we must orient time in order to integrate~$L$;
time-reversal exchanges the theories labeled by~$\theta $ and~$-\theta $.
Upon quantization we obtain the Hilbert space $\sH=L^2(\cir;\sL_{e^{i\theta
}})$ of sections of the complex line bundle~$\sL_{e^{i\theta }}$ with
holonomy~$e^{i\theta }$.  The Hamiltonian is the Laplace operator on this
space, and imaginary time evolution is by the associated heat operator.  Now
time-reversal ($\theta \mapsto -\theta $) acts as complex conjugation:
  \begin{align}
     \sH&\longmapsto \overline{\sH}\label{eq:5}\\
     e^{-\tau \Delta }&\longmapsto \overline{e^{-\tau
      \Delta }}\label{eq:8}
  \end{align} 
  \end{example}

We encode the formal structure in terms of oriented compact Riemannian
1-manifolds, as described in~\S\ref{subsec:12.2}, though we emphasize that
this is not a topological theory.  The interval of length~$\tau >0$ maps to
the imaginary time evolution $e^{-\tau H/\hbar}\:\sH\to \sH$.  The semigroup
law is manifest by gluing intervals.  The circle of length~$\tau $ maps to
$\Trace(e^{-\tau H/\hbar})\in \CC$.  We interpret these oriented Riemannian
1-manifolds as morphisms in a geometric bordism category whose objects are,
roughly, compact oriented 0-manifolds.  More precisely, they are 0-manifolds
embedded in the germ of an oriented Riemannian 1-manifold, and there is an
arrow of time, or orientation of the normal bundle.  The simplest object is a
single point, which we can view as~$0\in \RR$ embedded in a small
interval~$(-\epsilon ,\epsilon )$ with its standard orientation; in the
quantum mechanics it maps to the Hilbert space~$\sH$.  According
to~\eqref{eq:5} we have
  \begin{equation}\label{eq:6}
     \textnormal{orientation-reversal}\longmapsto \textnormal{complex
     conjugation} 
  \end{equation}
More precisely, the orientation-reversal on objects in the geometric bordism
category reverses the orientation and reverses the arrow of time.  This is
the `reflection' part of `reflection positivity'; the positivity is the
positive definiteness of the Hilbert space~$\sH$.

  \subsection{Reflection positivity in Euclidean quantum field theory}\label{subsec:2.2}

Positivity of energy in a relativistic quantum field theory also results in
an analytic continuation and restriction to Euclidean space, as we review
in~\S\ref{subsec:a3.2}.  Here we focus on the Wick rotation of correlation
functions and the Wick rotation of unitarity as manifested in reflection
positivity.  (See~\cite[\S6]{GJ}, \cite[\S2.2]{Kaz} for an account.)  Let
$n$~be the spacetime dimension and $\EE^n$~Euclidean $n$-space.  In this
subsection we restrict to the basic symmetry type~$H_n=SO_n$; we take up
general symmetry types in the next subsection (see Remark~\ref{thm:154}).
Fix an affine hyperplane~$\Pi \subset \EE^n$ and let $\sigma $~denote
(affine) reflection about~$\Pi $.  Let $\sO$~denote an operator, or product
of operators, in the quantum theory which is supported in the open
half-space~$\Hp$ on one side of~$\Pi $; the reflected operator~$\sigma (\sO)$
has support in the complementary half-space~$\Hm$.  Let $\langle \sO
\rangle\mstrut _{\Hp}\in \sH$ denote the half-space correlation function,
which is a vector in the Hilbert space of the theory.  In a lagrangian field
theory it is the functional integral over the half-space~$\Hp$.  Then the
\emph{reflection} part of `reflection positivity' is
  \begin{equation}\label{eq:7}
     \langle \sigma (\sO) \rangle\mstrut _{\Hm} = \overline{\langle \sO
     \rangle\mstrut _{\Hp}} 
  \end{equation}

  \begin{figure}[ht]
  \centering
  \includegraphics[scale=.75]{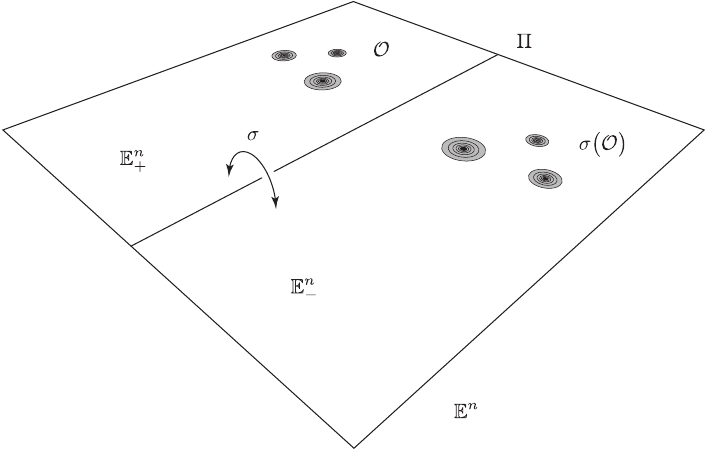}
  \caption{Reflection positivity in Euclidean space}\label{fig:2}
  \end{figure}

\noindent
 in accordance with~\eqref{eq:6}; see~\eqref{eq:8} for the analog in quantum
mechanics.  The Hilbert space~$\sH$ is associated to~$(\Pi ,\mathfrak{o})$,
where $\mathfrak{o}$~is an orientation of the normal line to~$\Pi $, the
arrow of time in~\S\ref{subsec:12.2}.  The reflection~$\sigma $
reverses~$\mathfrak{o}$, and the Hilbert space associated to~$(\Pi
,-\mathfrak{o})$ is the complex conjugate
  \begin{equation}\label{eq:42}
     \sH_{(\Pi ,-\mathfrak{o})} \xrightarrow{\;\;\cong \;\;}
     \overline{\sH_{(\Pi ,\mathfrak{o})}}, 
  \end{equation}
according to the dictum~\eqref{eq:6}; cf.~\eqref{eq:5}.  Therefore, $\langle
\sigma (\sO) \rangle\mstrut _{\Hm} \!\!\in \overline{\sH}$ and \eqref{eq:7}~
is an equation in the complex conjugate Hilbert space~$\overline{\sH}$.  The
\emph{positivity} part of `reflection positivity' is the positive
definiteness of~$\sH$, which implies that the norm square of the
vector~$\langle \sO \rangle\mstrut _{\Hp}$ is nonnegative:
  \begin{equation}\label{eq:9}
     \langle \sigma (\sO)\,\sO \rangle\mstrut _{\EE^n}\ge0 
  \end{equation}
A theorem of Osterwalder-Schrader~\cite{OS} reconstructs the relativistic
theory in Minkowski spacetime from the Euclidean theory; reflection
positivity is an important ingredient. 

  \begin{remark}[]\label{thm:160}
 In theories with fermionic states the Hilbert space~$\sH$ is $\zt$-graded.
The norm square of an odd vector is then purely imaginary~\cite[\S4.4]{DM}
and positive definiteness requires a sign choice; see Example~\ref{thm:159}
for details in the invertible case.
  \end{remark}

  \begin{remark}[Internal symmetry and reflection positivity]\label{thm:43}
 Suppose the full Wick-rotated vector symmetry group~$H_n$ has a nontrivial
internal symmetry group~$K$, and for simplicity take $H_n=SO_n\times K$.  Let
$X$~be Euclidean space with an open neighborhood of the support of the
operators~$\sO$, $\sigma (\sO)$ removed.  Let $Y=\bX\cap \HH_+$ and assume
$\sigma (Y)=\bX\cap\HH_-$.  In general there are twist operators that are
defined by a principal $K$-bundle $\sQ\to X$, as in Remark~\ref{thm:40}.  The
reflection~$\sigma $ must account for the $K$-bundle, and it might seem at
first that $\sigma $~should ``reverse'' it by an involution on~$K$.  But that
does not happen; rather $\sigma $ lifts to $\sQ\to X$.  We give three
arguments.

 \begin{enumerate}

\item If $\sO$~is a point operator, then $Y$~is a sphere.  Identifying
$\sigma (Y)$~with $Y$ via a translation, $\sigma $~acts on~$Y$ as reflection
in the equatorial plane parallel to~$\Pi $.  If we one-point compactify~$X$
to~$S^n$ minus the two balls and assume $\sQ $~extends over the
compactification, then the restrictions of $\sQ $ to~$Y$ and~$\sigma (Y)$ are
isomorphic, since the compactification is diffeomorphic to~$[0,1]\times
S^{n-1}$.

 \item Continuing, suppose $\sQ \to X$ is the trivial bundle and $V$~is the
vector space of local operators attached to~$Y$.  (In a geometric theory we
take a limit as the radius of the removed ball shrinks to zero.)  The
automorphism group~$K$ of the trivial bundle over~$Y$ acts on~$V$, producing
$K$-multiplets of point operators.  The hyperplane reflection~$\sigma $
induces an isomorphism $V\to\overline{V}$ that commutes with the $K$-action,
since geometrically the lift of reflection to the trivial bundle commutes
with the global gauge transformations.  So a $K$-multiplet in~$V$ is mapped
to a $K$-multiplet in~$\overline{V}$ that transforms in the complex
conjugate representation.

 \item Let $n=1$ and $H_1=SO_1\times \zmod3$.  Let $\alpha \:\Bord_{\langle
0,1 \rangle}(H_1)\to\Vect_{\CC}$ be the invertible theory which attaches a
nontrivial character $\chi \:\zmod3\to\TT$ to the positively oriented point
with its trivial $\zmod3$~bundle.  (That object~$Y$ of the bordism category
has automorphism group~$\zmod3$, which then acts on the vector space $\alpha
(Y)$.)  This theory is unitary.  Now $\alpha (\sQ \to\cir)$~is $\chi $~applied
to the holonomy of the principal $\zmod3$-bundle $\sQ \to S^1$.  Reflection
reverses the orientation of~$\cir$, and if the bundle stays the same under
reflection, then the holonomy complex conjugates, which is precisely what it
should do in a reflection positive theory.

 \end{enumerate}
  \end{remark}

  \subsection{The extended symmetry group~$\hH_n$}\label{subsec:2.3}

Let $(H_n,\rho _n)$ be a symmetry type (Definition~\ref{thm:153}).  We use
reflection symmetry~\eqref{eq:7} to construct a larger symmetry group~$\hH_n$
from~$H_n$ by adjoining an involution.  In the special case~$H_n=\Spin_n$, we
define $\hH_n=\Pp_n$; the general case is a bootstrap from this, following
the proof of Theorem~\ref{thm:6}.  The arguments in Remark~\ref{thm:43}
motivate the triviality of the hyperplane reflection automorphism of~$K$ in
our construction.  We view~$\hH_n$ as a symmetry group of the Euclidean
quantum field theory; the action of an element in $\hH_n\setminus H_n$ on the
Hilbert space~$\sH$ is by an anti-unitary transformation.

  \begin{proposition}[]\label{thm:10}
 There exists a canonical group extension
  \begin{equation}\label{eq:30}
     1\longrightarrow H_n\xrightarrow{\;\;j_n\;\;} \hH_n\longrightarrow
     \pmo\longrightarrow 1 ,
  \end{equation}
split (noncanonically) by a choice of hyperplane reflection~$\sigma \in
O_n$, such that the splitting induces the automorphism of~$\tSH_n\cong
\Spin_n\times K$ that is the product of conjugation by~$\sigma $
on~$\Spin_n$ and the identity automorphism of~$K$.  There is a
homomorphism~$\hat\rho _n$ that fits into the pullback diagram
  \begin{equation}\label{eq:35}
     \begin{gathered} \xymatrix{H_n\ar[r]^{j_n} \ar[d]_{\rho _n} &
     \hH_n\ar[d]^{\hat\rho _n} \\ O_n\ar[r]^{} & \pmo\times O_n} \end{gathered} 
  \end{equation}
Finally, there are inclusions $\hat\imath_n\:\hH_n\to\hH_{n+1}$ which, together
with the inclusions $i_n\:H_n\to H_{n+1}$, induce a commutative diagram linking
~\eqref{eq:35} for varying~$n$.
  \end{proposition}

\noindent
 A hyperplane reflection~$\sigma \in O_n$ induces an automorphism of~$SO_n$
by conjugation in~$O_n$, and it lifts uniquely to an automorphism
of~$\Spin_n$, which is realized as conjugation by~$\ts\in \Pp_n$, where
$\ts$~is a lift of~$\sigma $.  However, it is the \emph{twisted}
conjugation~\cite[\S3]{ABS} by~$\ts$ in~$\Pp_n$ that lifts conjugation
by~$\sigma$ in~$O_n$, where the twist is multiplication by the nontrivial
character
  \begin{equation}\label{eq:32}
     \Pp_n\longrightarrow \pi _0\Pp_n\xrightarrow{\;\;\cong \;\;}
     \pmo.
  \end{equation}
Note $\ts$~is only determined up to sign; the splitting of~\eqref{eq:30}
associated to~$\sigma $ is determined up to multiplication by~$k_0$.

See Appendix~\ref{sec:16} for an alternative approach to Theorem~\ref{thm:10}
using homotopy theory.

  \begin{proof}\hskip-4pt\footnote{Our original proof had an error.  We thank
Peter Teichner for bringing it to our attention.}  \;Define
  \begin{equation}\label{eq:31}
     \hSH_n = \Pp_n\times K\bigm/ \langle (-1,k_0) \rangle 
  \end{equation}
and project onto~$\pi _0\Pp_n$ to define the quotient map in the extension
  \begin{equation}\label{eq:33}
     1\longrightarrow SH_n\longrightarrow \hSH_n\longrightarrow
     \pmo\longrightarrow 1 
  \end{equation}
If $\rho _n(H_n)=SO_n$, then set $\hH_n=\hSH_n$.  In general, define~$\hH_n$
as the pullback
  \begin{equation}\label{eq:e12}
     \begin{gathered} \xymatrix{\hH_n\ar@{-->}[r]^{\phi _n} \ar@{-->}[d]_{\hrn} &
     H_{n+3}\ar[d]^{\rho _{n+3}} \\ \bmut\times O\mstrut _n\;\ar@{^{(}->}[r]^{}
     & \;O\mstrut _{n+3}} \end{gathered} 
  \end{equation}
where the bottom map is 
  \begin{equation}\label{eq:e13}
     (\epsilon ,A)\longmapsto \begin{pmatrix} A&0\\0&\epsilon I_3
     \end{pmatrix} 
  \end{equation}
in an $(n+3)\times (n+3)$ block decomposition.  Here $\bmu \ell \subset
\TT$~denotes the group of $\ell ^{\textnormal{th}}$~roots of unity.  Use this
embedding to define the subgroup $S\bigl(\bmut\times O_n \bigr)$ of
$\bmut\times O_n $ as the intersection
  \begin{equation}\label{eq:e23}
     S\bigl(\bmut\times O_n \bigr) := \bigl(\bmut\times O_n\bigr)\cap
     SO_{n+3}. 
  \end{equation}
Then define
  \begin{equation}\label{eq:e22}
     \hSH_n' := \hrn\inv \left(S\bigl(\bmut\times SO_n\bigr)\right)\,\subset
     \,\hH_n.  
  \end{equation}

Construct the homomorphism~$j_n$ by mapping into the pullback
square~\eqref{eq:e12}:
  \begin{equation}\label{eq:e14}
     \begin{gathered} \xymatrix{H_n\ar@/^1.7pc/[rr] 
     \ar@{-->}[r]^{j_n} \ar[d]_{\rho _n} &
     \hH_n\ar[r]^{\phi _n} \ar[d]_{\hrn} &
     H_{n+3}\ar[d]^{\rho _{n+3}} \\ O_n\ar@{^{(}->}[r]^{}&
     \bmut\times O\mstrut _n\;\ar@{^{(}->}[r]^{}
     & \;O\mstrut _{n+3}} \end{gathered} 
  \end{equation}
The bottom left map is $A\mapsto (1,A)$, and the top curved map is the
inclusion in~\eqref{eq:17}.  The group extension~\eqref{eq:30} can be read off
from~\eqref{eq:e14}.

We claim that $\hSH_n$, as defined in~\eqref{eq:31}, is isomorphic
to~$\hSH_n'$, as defined in~\eqref{eq:e22}.  To prove this, pull
back~\eqref{eq:e12} over the homomorphism $\Spin_{n+3}\to O_{n+3}$ into the
southeast corner to obtain
  \begin{equation}\label{eq:e17}
     \begin{gathered} \xymatrix{\Pp_n\times K\ar@{-->}[r]^{}
     \ar@{-->}[d]_{\pr_1} & 
     \widetilde{SH}_{n+3}\rlap{$=\Spin_{n+3}\times K$}\ar[d]^{\pr_1} \\
     \Pp_n\;\ar@{^{(}->}[r]^{} 
     & \;\Spin _{n+3}} \end{gathered} 
  \end{equation}
To check that the pullback has the indicated form, embed $\Pp_n\subset
\Cliff_n^+$ and $\Spin_{n+3}\subset \Cliff_{n+3}^+$ as usual.  Then
from~\eqref{eq:e13} we find that the bottom map sends~$e_i$ to~$\pm
e_ie_{n+1}e_{n+2}e_{n+3}\in \Spin_{n+3}$, and we observe that
$(e_ie_{n+1}e_{n+2}e_{n+3})^2=+1$, which explains why the southwest corner
is~$\Pp_n$ rather than~$\Pm_n$.  The northeast corner is computed by
Theorem~\ref{thm:5}(2).  Finally, \eqref{eq:e17}~is a pullback diagram, since
\eqref{eq:e12}~is, and this determines the northwest corner.  To recover
~\eqref{eq:31} from~\eqref{eq:e22}, observe that $\widehat{SH}_n'$~is the
northwest corner of the pullback of~\eqref{eq:e12} over the homomorphism
$\SO_{n+3}\to O_{n+3}$.  The isomorphism~\eqref{eq:16} implies that the
northeast corner of this pullback is obtain by dividing the northeast corner
of~\eqref{eq:e17} by the order~2 subgroup~$\langle (-1,k_0) \rangle$.  The
top homomorphism of~\eqref{eq:e17} is an injection, and now
\eqref{eq:31}~follows.
 
If $\sigma \in O_n$ is a hyperplane reflection, and $\xi \in \Pp_n$ is a
lift (which is determined up to a sign), then define the splitting
of~\eqref{eq:30} to send $-1\in \{\pm1\}$ to the element ~$\hhat\in \hH_n$
which in~\eqref{eq:e12} satisfies
  \begin{equation}\label{eq:e18}
     \begin{aligned} \hrn(\hhat)&=(-1,\sigma ) \\ \phi _n(\hhat)&=[\xi
      e_{n+1}e_{n+2}e_{n+3},1]\,\in SH_{n+3}=\Spin_{n+3}\times
     K\!\bigm/\!\langle (-1,k_0) 
      \rangle\end{aligned} 
  \end{equation}
Note that $\hhat^2=1$.  Conjugation by~$\hhat$ on~$\widetilde{SH}_n$ induces the
automorphism stated in Proposition~\ref{thm:10}.

Finally, define~$\hin$ using the definition of~$\hH_{n+1}$ as a pullback: 
  \begin{equation}\label{eq:e19}
     \begin{gathered} \xymatrix@C+1pc{&H_{n+3}\ar[dr]^{i_{n+3}}\\
     \hH_n\ar[ur]^{\phi 
     _n}\ar[d]_{\hrn}\ar@{-->}[r]^{\hin} & \hH_{n+1} \ar[r]^{\phi
     _{n+1}}\ar[d]^{\hat{\rho }_{n+1}} & H_{n+4}\ar[d]^{\rho _{n+4}} \\
     \bmut\times O\mstrut _n\;\ar@{^{(}->}[r]^{} &\bmut\times O\mstrut
     _{n+1}\;\ar@{^{(}->}[r]^{} &O\mstrut _{n+4}} \end{gathered} 
  \end{equation}
Recall our convention~\eqref{eq:51} for the embedding $O_n\hookrightarrow
O_{n+1}$.  This, together with~\eqref{eq:e13}, makes
  \begin{equation}\label{eq:e24}
     \begin{gathered} \xymatrix{\bmut\times O_n\ar[r]^{} \ar[d]_{} &
     O_{n+3}\ar[d]^{} \\ \bmut\times O_{n+1}\ar[r]^{} & O_{n+4}}
     \end{gathered} 
  \end{equation}
a commutative diagram, and this is the key observation necessary to
construct the commutative diagram indicated in the final sentence of
Proposition~\ref{thm:10}, whose proof is now complete.
  \end{proof}

  \begin{remark}[]\label{thm:154}
 Now we formulate reflection positivity on Euclidean space for a theory with
symmetry type~$(H_n,\rho _n)$.  Adjoining translations via the pullback
  \begin{equation}\label{eq:c97}
     \begin{gathered} \xymatrix{
     1\ar[r]&K\ar[r]\ar@{=}[d]&\sH_n\ar@{-->}[d]^{}\ar@{-->}[r]^{}&
     \Euc_n\ar[d]^{ 
     } \ar[r] & 1\\ 1\ar[r]&K\ar[r]&H_n\ar[r]^{\rho _n}& O_n\ar[r] & 1}
     \end{gathered} 
  \end{equation}
we obtain a larger group~$\sH_n$ and a homomorphism $\sH_n\to\Euc_n$ to the
Euclidean group.  The complex point observables form a vector bundle
$\sO\to\EE^n$, and the action of~$\Euc_n$ on~$\EE^n$ lifts to an action
of~$\sH_n$ on~$\sO$.  Proposition~\ref{thm:10} gives a coextension~$\hsH_n$
of~$\sH_n$ and a homomorphism $\hsH_n\to\pmo\times \Euc_n$.  As before fix a
hyperplane reflection~$\sigma $ and now fix a lift $\hs\in \hsH_n$ of
$(-1,\sigma )\in \pmo\times \Euc_n$.  Then part of the data of a reflection
structure is a lift of~$\hs$ to an antilinear map of the complex vector
bundle $\sO\to\EE^n$.  Therefore \eqref{eq:7}--\eqref{eq:9} apply, with $\hs$
replacing~$\sigma $.
  \end{remark}

  \begin{proposition}[]\label{thm:148}
 For each~$n\ge 1$ there is an inclusion of group extensions 
  \begin{equation}\label{eq:192}
     \begin{gathered} \xymatrix{1\ar[r]&H_n\ar[r]\ar@{^{(}->}[d]^{i_n}
     &\pmo\times H_{n}\ar[r]\ar@{^{(}->}[d]^{s_{n+1}\ast j_{n+1}i_n} & \pmo
     \ar@{=}[d]\ar[r] & 1\\ 1
     \ar[r]& H_{n+1} \ar[r]^{j_{n+1}}& \hH_{n+1} \ar[r]&
     \pmo\ar[r]&1} \end{gathered} 
  \end{equation}
in which $i_n$~is the inclusion in~\eqref{eq:17} and $j_n$~the inclusion
in~\eqref{eq:30}.  Furthermore, the inclusions~$i_n$ and~$\hat\imath_n$ induce
a commutative diagram linking~\eqref{eq:192} for varying~$n$. 
  \end{proposition}

  \begin{proof}
First, define a particular splitting $s_n\:\bmut\to \hH_n$ for~$\sigma \in
O_n$ reflection in the hyperplane orthogonal to~$e_1$.  Namely, as
in~\eqref{eq:e18}, characterize the element $\hhat_n=s_n(-1)$ by
  \begin{equation}\label{eq:e20}
     \begin{aligned} \hrn(\hhat_n)&=(-1,\sigma ) \\ \phi
     _n(\hhat_n)&=[e_1e_2e_3e_4 
      ,1]\,\in SH_{n+3}=\Spin_{n+3}\times K\!\bigm/\!\langle (-1,k_0)
      \rangle\end{aligned} 
  \end{equation}
Then $\hhat_n$~ has order~2, and $\hhat_{n+1}\in \hH_{n+1}$ centralizes the image
of $H_n\xrightarrow{\;i_n\;}H_{n+1}\xrightarrow{\;j_{n+1}\;}\hH_{n+1}$.  It
is easy to verify that \eqref{eq:192}~commutes.
  \end{proof}

For the basic symmetry groups in~\eqref{eq:10} the extended symmetry groups
are listed here:
  \begin{equation}\label{eq:11}
     \begin{tabular}{ c@{\hspace{2em}} c@{\hspace{2em}} c@{\hspace{2em}} }
     \toprule 
     states/symmetry&$H_n$&$\hH_n$\\ \midrule \\[-8pt] bosons only &
     $SO_n$&$O_n$\\ 
     fermions allowed & $\Spin_n$ & $\Pp_n$\\ bosons, time-reversal~($T$) &
     $O_n$ & $\pmo\times O_n$\\ 
     fermions, $T^2=(-1)^F$ & $\Pp_n$&$\hPin^+_n$\\ fermions, $T^2=\id$ &
     $\Pm_n$&$\hPin^-_n$\\ 
     \bottomrule \end{tabular} 
  \end{equation}
The splitting of~$\widehat{O}_n$ is a consequence of the fact that hyperplane
reflections are inner in~${O}_n$.  As in Remark~\ref{thm:a17}, let $\ha$~be
the automorphism of $\Ppm_n$ that is the identity on~$\Spin_n$ and
multiplication by the central element~$-1$ on the complement; it covers the
identity automorphism of~$O_n$.  Then a computation of the
pullback~\eqref{eq:e12} yields isomorphisms
  \begin{equation}\label{eq:e21}
     \begin{aligned} \hPin^+_n&\cong (\bmu4\ltimes\Ppn)\!\bigm/\!\langle
      (-1,-1) \rangle \\ \hPin^-_n&\cong \;\,\bmut\ltimes\Pm_n\\ \end{aligned} 
  \end{equation}
where the generators of~$\bmu4,\bmut$ act on the pin group via~$\ha$.

   \section{Reflection symmetry on manifolds}\label{sec:3}

The enhanced symmetry group~$\hH_n$ produces an
involution~(\S\ref{subsec:3.1}) on $H_n$-manifolds that generalizes
orientation-reversal for~$H=SO$.  In the field theory context it induces an
involution on bordism categories that we call `bar'.  (See
Appendix~\ref{sec:11} for a general discussion of involutions on categories
and other relevant background.)  In~\S\ref{subsec:3.2} we prove that the dual
of an object in a bordism category is isomorphic to its bar.  The definitions
of \emph{reflection structure} and \emph{positive reflection structure} for
non-extended field theories are in~\S\ref{subsec:3.3}.  In a reflection
positive theory the partition function of any double is nonnegative, as we
prove in~\S\ref{subsec:3.4}.  We work as always with arbitrary symmetry
groups.\footnote{The definition of the double of a (s)pin manifold is
somewhat tricky, for example; the general setting is clarifying.}

Kevin Walker has introduced theories with more general reflection structures
in which, possibly, the group extension~\eqref{eq:30} that controls
anti-unitarity is not split.  In particular, he allows $\hH_n=\Pm_n$ when
$H_n=\Spin_n$.  This leads to exotic hermitian structures.  Our more
restrictive framework is based on Wick rotation of relativistic theories.

  \subsection{An involution on $H_n$-manifolds}\label{subsec:3.1}

Recall from~\S\ref{subsec:12.2} that an $H_n$-manifold is a Riemannian
$n$-manifold equipped with a reduction $(P,\theta )$ of its orthonormal frame
bundle $\sBO(X)\to X$ to~$H_n$.  Extend the principal $H_n$-bundle $P\to X$
to a principal $\hH_n$-bundle $j_n(P)\to X$, where $j_n$~is the inclusion of
groups in~\eqref{eq:30}.  Using~\eqref{eq:35} extend the isomorphism $\theta
\:\sBO(X)\to \rho _n(P)$ to an isomorphism $\hat\theta \:\pmo\times
\sBO(X)\to \hat\rho _n\bigl(j_n(P)\bigr)$.

  \begin{definition}[]\label{thm:11}
 The \emph{opposite $H_n$-structure}~$(P',\theta ')$ is the principal
$H_n$-bundle $P':=j_n(P)\setminus P\to X$ and the restriction~$\theta '$
of~$\hat\theta $ to $\{-1\}\times \sBO(X)$.
  \end{definition}

\noindent 
 Taking opposites is involutive: there is a canonical isomorphism $(P,\theta
)\xrightarrow{\;\cong \;}(P'',\theta '')$.

  \begin{remark}[]\label{thm:13}
 Let $\sigma \in O_n$ be a hyperplane reflection and $\phi _{\sigma }$~the
automorphism of~$H_n$ resulting from the splitting of~\eqref{eq:30}.  Then we
can identify the principal $H_n$-bundle $P'\to X$ as the projection $P\to X$
of manifolds with the original $H_n$-action on~$P$ precomposed with the
automorphism~$\phi _\sigma $.  For if $\ts \in \hH_n$ is the splitting
element, then we map $P\to j_n(P)\setminus P$ by $p\mapsto p\cdot \ts$. 
  \end{remark}

  \begin{example}[]\label{thm:12}
 An $SO_n$-structure is an orientation, and the opposite $SO_n$-structure is
the reverse orientation.  In this case $P\to X$~is the bundle of oriented
orthonormal frames, $j_n(P)\to X$ the bundle $\sBO(X)\to X$ of all
orthonormal frames, and $j_n(P)\setminus P\to X$ the bundle of oppositely
oriented orthonormal frames.
  \end{example}

  \begin{example}[]\label{thm:14}
 For simplicity, we sometimes abbreviate `$\Ppm_n$-structure' to `pin
structure', just as `$\Spin_n$-structure' is abbreviated to `spin structure'.
The opposite of a pin structure is obtained by tensoring with the orientation
double cover; see Definition~\ref{thm:a21}, Remark~\ref{thm:a17}, and the
text following~\eqref{eq:11}.  One motivation for our general study of
symmetry groups~(\S\ref{subsec:12.1}) and involutions~(\S\ref{subsec:2.2})
is to explain the appearance of this opposite pin structure in the
formulation of reflection positivity for Wick-rotated quantum field theories
with fermions and time-reversal symmetry.
  \end{example}

We use the involution in Definition~\ref{thm:11} to construct an involution
of categories 
  \begin{equation}\label{eq:37}
     \bB=\beta \:\bne\to\bne .
  \end{equation}
In Appendix~\ref{sec:11} we explain that an involution on a category~$\sB$ is
a functor~$\beta \:\sB\to\sB$ and a natural transformation of functors $\eta
\:\id_{\sB}\to \beta ^2$.  The objects and morphisms in~$\bne$ are Riemannian
manifolds with $H_n$-structure: the functor~$\beta $ fixes the underlying
Riemannian manifold and flips the $H_n$-structure to its opposite.  The
equivalence~$\eta $ implements the canonical isomorphism indicated after
Definition~\ref{thm:11}.  We emphasize that the ``bar involution''~$\beta $
is covariant: a morphism $X\:Y_0\to Y_1$ maps to a morphism $\beta X\:\beta
Y_0\to \beta Y_1$.  Put differently, the arrows of time on objects are
unchanged under~$\beta $.

  \begin{remark}[]\label{thm:62}
 One can envisage other involutions on the bordism category, and so other
notions of reflection structure (Definition~\ref{thm:23} below), especially
for mathematical applications.  The heuristics in Remark~\ref{thm:43} are
meant to illustrate why we feel the involution defined here correctly models
Wick-rotated unitarity in relativistic field theories. 
  \end{remark}

  \subsection{Duals and opposites}\label{subsec:3.2}

An object~$Y$ in a symmetric monoidal category, such as $\bne$, may have a
dual~$Y\dual$, which is equipped with duality data; see
Definition~\ref{thm:19} for a quick review.  In a topological bordism
category every object has a dual.  The underlying smooth manifold of the
dual~$Y\dual$ equals that of~$Y$, but the arrow of time is reversed.  This
reversal is evident in the coevaluation and evaluation duality data.  For
example, evaluation is the bordism
  \begin{equation}\label{eq:43}
     e\mstrut _Y=[0,1]\times Y \:\,Y\dual\amalg Y\longrightarrow \emptyset
     ^{n-1} 
  \end{equation}
with the entire boundary incoming.  The $H_n$-structure is the same at the
two ends, but the arrows of time are opposite.  If the boundary at~$0\in
[0,1]$ is the object~$Y$, with its arrow of time, then the boundary at~$1\in
[0,1]$ is the object~$Y\dual$.  See Figure~\ref{fig:3}, where the
coevaluation~$c\mstrut _Y$ and the ``S-diagram''~\eqref{eq:27} are also
depicted.

  \begin{figure}[ht]
  \centering
  \includegraphics[scale=1]{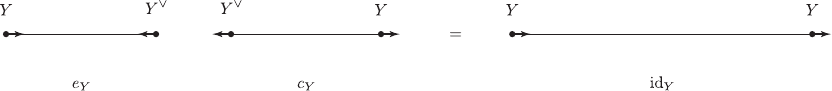}
  \caption{Evaluation, coevaluation, and the gluing to the
  identity}\label{fig:3} 
  \end{figure}
 
An object~$Y$ in a topological bordism category has a canonical product germ
(see~\S\ref{subsec:12.2}), namely the germ of~$\{0\}\times Y$ in
$X=(-\epsilon ,\epsilon )\times Y$, where we fix~$\epsilon >0$.  Let
$\sigma$~ be the diffeomorphism of~$X$ that reflects $t\mapsto -t$ and
fixes~$Y$.  The splitting in Proposition~\ref{thm:10} leads to an alternative
construction of the opposite $H_n$-structure and the following important
identification.

  \begin{proposition}[]\label{thm:44}
 For any object $Y$ in~$\bne$ there is an isomorphism
  \begin{equation}\label{eq:44}
     h\:\beta Y\xrightarrow{\;\;\cong \;\;} Y\dual  
  \end{equation}  
Also, $\beta h\dual=h$.
  \end{proposition}

\noindent
 Reversing the $H_n$-structure ($\beta Y$) is equivalent to reversing the
arrow of time~($Y\dual$).  Or, in the language of Definition~\ref{thm:59},
every object in~$\bne$ carries a \emph{hermitian structure}.

  \begin{proof}
 Set $X=(-\epsilon ,\epsilon )\times Y$.  The reflection  
  \begin{equation}\label{eq:54}
     \begin{aligned} \sigma \:(-\epsilon ,\epsilon )\times Y&\longrightarrow
      (-\epsilon ,\epsilon )\times Y \\ (t,y)&\longmapsto
      (-t,y)\end{aligned} 
  \end{equation}
lifts to the frame bundle~$\sBO(X)$.  We now construct a diagram of
principal $K$-bundles: 
  \begin{equation}\label{eq:55}
     \begin{gathered} \xymatrix{  Q' 
     \ar@{^{(}->}[r] \ar[d] &P'\,\ar[d]_{\pi'} \ar@{^{(}->}[r]& j_n(P)\ar[d] &
     P\ar@{_{(}->}[l]\ar[d]^{\pi } & Q\dual\ar@{_{(}->}[l]\ar[d]\\ \sB_Y
     \ar@{^{(}->}[r]& \sBO(X)
     \ar@{^{(}->}[r]^<<<<<{\;-1\times \id} & \pmo\times \sBO(X) & \sBO(X)
     \ar@{_{(}->}[l]_<<<<<{\;1\times \id} &
     \sB_Y\dual\ar@{_{(}->}[l]} \end{gathered}  
  \end{equation}
Let $\sB_Y\subset \sBO(X)$ be the $O_{n-1}$-subbundle of frames with first
vector~$\pm \partial /\partial t$, the sign chosen to align with the arrow of
time of the object~$Y$.  Let $\sB_Y\dual $~be the compatible frames with the
opposite arrow of time.  Then $\sigma $~induces an isomorphism
$\sB\mstrut_Y\to \sB\dual _Y $ which is realized inside~$\sBO(X)$ as
multiplication by the hyperplane reflection~$\sigma _1\in O_n$ in the
orthogonal complement to the vector~$e_1\in \RR^n$.  (Observe that $\sigma
_1$~centralizes $O_{n-1}\subset O_n$.)  Let $P\xrightarrow{\pi } \sBO(X)\to
X$ be the $H_n$-structure: the composition is a principal $H_n$-bundle and
the first map is a principal $K$-bundle over its image.  Set $Q\dual = \pi
\inv (\sB_Y\dual)$; then $Q\dual\to X$ is a principal $H_{n-1}$-bundle.  Let
$j_n(P),P'$~be as in Definition~\ref{thm:11}, so that $P'\xrightarrow{\pi'
}\sBO(X)\to X$ is the opposite $H_n$-structure.  Set $Q'={\pi '}\inv
(\sB_Y)$, so that $Q'\to X$ is an $H_{n-1}$-bundle that encodes the opposite
$H_n$-structure.  Let $\hhat_n\in \hH_n$ be the lift of~$\sigma _1\in O_n$, as
defined in~\eqref{eq:e20}; then $\hhat_n$~centralizes $H_{n-1}$ and has order
two.  The action of multiplication by~$\hhat_n$ on~$j_n(P)$ restricts to an
isomorphism of $H_{n-1}$-bundles $Q'\to Q\dual$.  (It covers multiplication
by~$(-1,\sigma _1)\in \pmo\times O_n$ on~$\pmo\times \sBO(X)$, which
restricts to an isomorphism $\sB\mstrut _Y\to \sB\dual _Y$.)
 
$\beta h\dual$~is the inverse of the involution~$\hhat_n$ on~$j_n(P)$,
restricted to the bar dual bundles.  Since $\hhat_n$~is its own inverse, we
conclude~$\beta h\dual=h$.
  \end{proof}

  \begin{remark}[]\label{thm:34}
 In a geometric bordism category not every germ admits a reflection which is
an isometry.  It is only for germs which do admit such a reflection that we
expect the associated topological vector space of a field theory to have a
Hilbert space structure; see~\cite{KS}.  This is the case for the
(noncompact) affine hyperplane in Figure~\ref{fig:2}, consistent
with~\eqref{eq:42}.
  \end{remark}

  \subsection{Reflection structures and positivity}\label{subsec:3.3}

Let 
  \begin{equation}\label{eq:40}
     \bC=\beta \:\Vect_{\CC}\longrightarrow \Vect_{\CC} 
  \end{equation}
be the involution of complex conjugation (Example~\ref{thm:16}).
Recall~\eqref{eq:25} that a topological field theory is a symmetric monoidal
functor $F\:\bne\to\Vect_{\CC}$. 

  \begin{definition}[]\label{thm:23}
 A \emph{reflection structure} on~$F$ is equivariance data for the
involutions~$\beta _{\sB},\beta _{\sC}$.
  \end{definition}

\noindent
 Equivariance data is spelled out in Definition~\ref{thm:17}.  For every
closed $(n-1)$-manifold~$Y$ with $H_n$-structure we have an isomorphism of
vector spaces
  \begin{equation}\label{eq:41}
     F(\beta Y)\xrightarrow{\;\;\cong \;\;}\overline{F(Y)}, 
  \end{equation}
the curved space analog of~\eqref{eq:42}.  Combining with the
isomorphism~\eqref{eq:44}, we see that $F(e\mstrut _Y)$~is a hermitian form
  \begin{equation}\label{eq:45}
     h_Y\:F(Y\dual)\otimes F(Y)\cong F(\beta Y)\otimes F(Y)\cong
     \overline{F(Y)}\otimes F(Y)\longrightarrow \CC ,
  \end{equation}
which by the usual ``S-diagram'' argument (Figure~\ref{fig:3}) is
nondegenerate.  Sesquilinearity is a consequence of the isomorphism 
  \begin{equation}\label{eq:230} 
   \begin{aligned}
      e\mstrut _Y&\longrightarrow \beta (e\mstrut _Y) \\
      (t,y)&\longmapsto (1-t,y) 
   \end{aligned}
  \end{equation}
where recall as a manifold $e\mstrut _Y=[0,1]\times Y$.

  \begin{definition}[]\label{thm:24}
 A reflection structure is \emph{positive} if the induced hermitian
form~$h_Y$ is positive definite for all~$Y\in \bne$. 
  \end{definition}

  \begin{remark}[]\label{thm:25}
 In a non-extended field theory reflection is \emph{data} and positivity is a
\emph{condition}.  In the extended case considered later, both reflection and
positivity are data.
  \end{remark}

  \begin{remark}[]\label{thm:175}
 There is also a notion of positivity if the domain is the category of
\emph{super} vector spaces; see Example~\ref{thm:159}. 
  \end{remark}

  \begin{example}[]\label{thm:26}
 To avoid trivialities, suppose the spacetime dimension~$n$ is even.  Fix a
nonzero complex number~$\lambda \in \CC$.  There is a simple invertible field
theory of unoriented manifolds ($H_n=O_n$) whose partition function on a
closed $n$-manifold~$X$ is $\lambda ^{\Euler(X)}$, where $\Euler(X)$~is the
Euler number of~$X$.  The vector space~$F_\lambda (Y)$ attached to any closed
$(n-1)$-manifold~$Y$ is the trivial line~$\CC$: the Euler characteristic of a
compact manifold with boundary is a well-defined number.  In the bordism
category we can write the closed manifold~$S^n$ as the composition $\emptyset
^{n-1}\xrightarrow{D^{n}}S^{n-1}\xrightarrow{D^n}\emptyset ^{n-1}$ of two
closed balls.  Denote the first arrow as~$X$ and apply the theory~$F_\lambda
$:
  \begin{equation}\label{eq:46}
     \lambda ^2 = F_\lambda (S^n) = h_{S^{n-1}}(F_\lambda (X),F_\lambda (X)).
  \end{equation}
Therefore, a necessary condition for positivity is that $\lambda$ be real.
  \end{example}

 A reflection structure imposes a curved space analog of~\eqref{eq:7}, which
for an $n$-dimensional $H_n$-bordism~$X$,  asserts that
  \begin{equation}\label{eq:47}
     F(\beta X)= \overline{F(X)}. 
  \end{equation}
For example, if $H=SO_n$ then the partition function complex conjugates when
the orientation of spacetime is reversed.  For a theory of unoriented
manifolds ($H_n=O_n$), condition~\eqref{eq:47} implies that every partition
function is real.  For theories of pin manifolds ($H_n=\Ppm_n$) the partition
function of the $w_1$-twisted pin structure (Definition~\ref{thm:a21}) is the
complex conjugate of the original partition function.

  \subsection{Doubles}\label{subsec:3.4}

The reflection-conjugation equation~\eqref{eq:47} also applies to manifolds
with boundary.  We use it to derive a necessary condition for reflection
positivity.

  \begin{definition}[]\label{thm:28}
 Let $X$~be a compact $H_n$-manifold with boundary, viewed as a bordism
$\emptyset ^{n-1}\to \partial X$.  The \emph{double} of~$X$ is the closed
$H_n$-manifold
  \begin{equation}\label{eq:48}
     \Delta X = e\mstrut _{\partial X}(\beta X,X). 
  \end{equation}
  \end{definition}

\noindent
 The double is illustrated in Figure~\ref{fig:4}.  In that picture~$Y=\bX$.

  \begin{figure}[ht]
  \centering
  \includegraphics[scale=1.1]{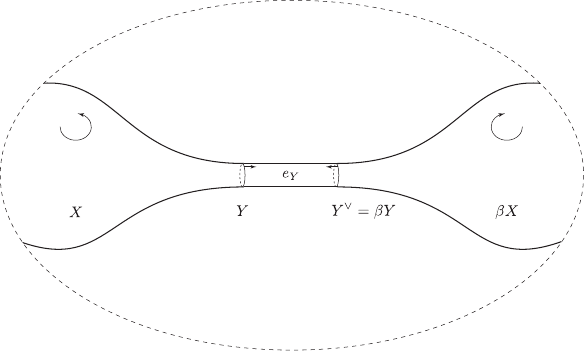}
  \caption{The double of~$X$}\label{fig:4}
  \end{figure}

  \begin{proposition}[]\label{thm:29}
 If a theory $F\:\bne\to\Vect_{\CC}$ admits a positive reflection structure,
then $F(\Delta X)\ge0$ for all compact $H_n$-manifolds~$X$ with boundary. 
  \end{proposition}

\noindent
 Note that the value of a theory on a closed $n$-manifold does not depend on
the reflection structure.  The necessary condition for positivity in
Proposition~\ref{thm:29} is the compact manifold analog of the usual
reflection positivity statement~\eqref{eq:9} in Euclidean space.

  \begin{proof}
 From~\eqref{eq:48} and~\eqref{eq:47} we deduce 
  \begin{equation}\label{eq:49}
     F(\Delta X) = F(e\mstrut _{\partial X})\bigl(F(\beta X),F(X)\bigr) =
     h\mstrut _{\bX}\bigl(\overline{F(X)},F(X) \bigr) = \|F(X)\|^2_{F(\bX)}
     \ge 0.  
  \end{equation}
\vskip-2.8em
  \end{proof}
\bigskip 
 
The double construction is standard for unoriented and oriented manifolds.
It is a bit trickier for spin and pin manifolds, so we give a recognition
principle and illustrate with some examples.  Observe that the double has an
obvious (anti-)involution $\Delta X\xrightarrow{\;\sigma \;}\beta \Delta X$
with fixed point set $Y=\{1/2\}\times \bX$, and $\sigma $~induces
multiplication by~$-1$ on the normal bundle.  Set
$X'=X\cup_{\bX}[0,1/2]\times \bX$ and cut along~$Y=\bX'$ to write
  \begin{equation}\label{eq:53}
     \Delta X=\beta X'\cup_{\bX'}X', 
  \end{equation}
which is the typical description of a double.  But we must account for the
$H_n$-structure as well. 

  \begin{proposition}[]\label{thm:35}
 Let $X$~be a closed $H_n$-manifold, $\sigma \:X\to\beta X$ an
anti-involution with fixed point set~$Y$ such that

 \begin{enumerate}[{\textnormal(}i{\textnormal)}]

 \item There exists a submanifold~$N\subset X$ with boundary~$Y$ such
that $X$~is the union of~$N$ and~$\sigma N$ along~$Y$ and $\sigma $~induces a
diffeomorphism $\beta N\cong \sigma N$ of $H_n$-manifolds; and

 \item $\sigma \res Y$~induces the hyperplane reflection isomorphism of the
$H_n$-structure on~$Y$ to its opposite.

 \end{enumerate}

Then $X\cong \Delta N$ as $H_n$-manifolds 
  \end{proposition}

\noindent
 The isomorphism in~(ii) is left multiplication by ~$\hhat_n\in \hH_{n}$;
see~\eqref{eq:e20}.

  \begin{proof}
 Use the tubular neighborhood theorem to replace~$Y$ with $[0,1]\times Y$ and
so construct the desired $H_n$-isomorphism. 
  \end{proof}

  \begin{corollary}[]\label{thm:37}
 The sphere~$S^n$ with $H_n$-structure $H_{n+1}\to H_{n+1}/H_n$ is a double 
  \end{corollary}

\noindent 
 We note from Remark~\ref{thm:33} that the homogeneous space~$H_{n+1}/H_n$ is
diffeomorphic to~$S^n$.

  \begin{proof}
 Reflection~$\sigma $ in the hyperplane perpendicular to~$e_1$ is an
involution of~$S^n$ with fixed point set the equatorial~$S^{n-1}$
perpendicular to~$e_1$.  The reflection lifts to an isomorphism of the
principal $H_n$-bundle $H_{n+1}\to H_{n+1}/H_n$ with the pullback of its
opposite.  (The isomorphism is globally left multiplication
by~$[e_1,1;1]\subset \hH_{n+1}$.) 
  \end{proof}

  \begin{example}[]\label{thm:38}
 For $H_m=\Spin_m$ the circle $\Spin_2/\Spin_1$~has the \emph{bounding} spin
structure: the $\Spin_1$-bundle $\Spin_2\to \Spin_2/\Spin_1$ is the
nontrivial double cover of the circle.  The nonbounding spin circle is
\emph{not} a double.  Indeed, there is a reflection positive invertible
1-dimensional spin topological field theory~$\alpha $ into super vector
spaces that attaches the odd line to a positively oriented spin point; it
follows that $\alpha (\cir_{\textnormal{nonbounding}})=-1$.  This does not
violate Proposition~\ref{thm:29} since $\cir_{\textnormal{nonbounding}}$~is
not a double.  Turning this argument around, since the oriented circle
\emph{is} a double, the 1-dimensional \emph{oriented} topological field
theory into super vector spaces that attaches the odd line to a positively
oriented point does \emph{not} admit a positive reflection structure.
  \end{example}

  \begin{remark}[]\label{thm:39}
 The group~$H_{n+1}$ acts as symmetries of the $H_n$-sphere in
Corollary~\ref{thm:37}.  Topologically, then, there is a universal family of
$H_n$-spheres parametrized by the classifying space~$BH_{n+1}$.  Field
theories may be evaluated on families of manifolds and bordisms; this family
of spheres enters our analysis in~\S\ref{subsec:6.2}.
  \end{remark}

   \section{Invertible topological field theories and stable homotopy theory}\label{sec:4}

We first recall that to fully implement locality in field theory we need to
use a bordism multicategory that encodes gluing laws in arbitrary
codimension.  Next we recount how invertible topological field theories lie
in the framework of homotopy theory: invertibility moves the discussion from
abstract multicategories to topological spaces.  Finally, we specify the
universal target that tracks deformation classes of invertible topological
theories.  The main result is Theorem~\ref{thm:56}, which is our point of
departure for implementing reflection positivity in invertible topological
theories.  We conclude in~\S\ref{subsec:4.4} with a discussion of
invertible non-topological theories and their role in low energy
approximations of gapped quantum systems.
 
The material in this section is covered in much more expository detail in
many references, so we only recount essentials.

  \subsection{Extended field theories}\label{subsec:4.1}

There are several physics motivations for extending an $n$-dimensional
Wick-rotated field theory to lower dimensional manifolds, and these are
hardly restricted to the topological case of interest here.  First, the
vector space of physical states attached to an $(n-1)$-manifold~$Y$ depends
locally on~$Y$.  This is familiar in $n=2$~dimensions, where a theory not
only has a vector space attached to a circle, but also to an interval with
boundary conditions; the gluing laws for intervals lie in codimension two,
since intervals are glued along 0-manifolds in this 2-dimensional theory.
The result is sometimes called an \emph{open-closed
theory}~\cite{MS}.\footnote{There is a difference between an open-closed
theory and a fully extended 2-dimensional theory~\cite[\S4.2]{L}.}  The
labels on the boundary are objects in a category, so it is natural to
associate that category to the 0-manifold consisting of a single point.  As
we are doing quantum mechanics, the category is linear and indeed the vector
space associated to the interval with boundary labels~$\beta _0,\beta _1$
is~$\Hom(\beta _0,\beta _1)$ in the category.  The objects are boundary
conditions, or \emph{D-branes}.  Another common example is 3-dimensional
Chern-Simons theory, in which a unitary modular tensor category is associated
to the 1-manifold~$S^1$, which is a manifold of codimension two in this
theory.
 
Let $X^n$~be a Riemannian $n$-manifold on which a theory~$F$ is defined, and
fix $x\in X$.  We explained in Remark~\ref{thm:40} that the vector
space~$F(S^{n-1}_x)$ attached to a small sphere around~$x$, in the limit of
small radius, is the space of point operators at~$x$.  A field theory also
has \emph{extended operators}, whose support may be a submanifold~ $W\subset
X$ of dimension~$k>0$.  An extended operator with~ $k=1$ is called a
\emph{line operator}, with~ $k=2$ a \emph{surface operator}, etc.  The link
of~$W$ at any~$x\in W$ is a sphere $S^{n-k-1}_x$.  In an extended field
theory~$F$ there is an invariant $F(S^{n-k-1}_x)$ which is a $k$-category
whose objects are the operators on~$W$.  Thus the line operators in a theory
form a 1-category, the surface operators a 2-category, etc.; see~\cite{Ka2}
for a thorough account.

We believe that every field theory of physical relevance should be fully
extended.  The mathematical implementation is most developed in the
topological case: a sampling of references is~\cite{F1,La,BD,L,F2,AF}.
Invariants of manifolds of increasing codimension are encoded in a higher
categorical structure of increasing complexity.  The modern framework also
includes invariants for families of manifolds; see~\cite{ST} for a
non-topological version.  The domain of an $n$-dimensional topological field
theory with symmetry group~$H_n$ is the bordism multicategory~$\be$ whose
objects are 0-manifolds; 1-morphisms are bordisms of 0-manifolds, which are
1-manifolds with boundary; 2-morphisms are bordisms of bordisms, which are
2-manifolds with corners; and so on until we reach $n$-manifolds with
arbitrary corners.  At that point we continue to $(n+\ell )$-morphisms which
are roughly $\ell $-dimensional families of $n$-manifolds, where $\ell $~is
an arbitrary positive integer.  The entire structure is an \emph{$(\infty
,n)$-category}~\cite{BM,L,BS,Ng,CS,S-P}.

  \begin{definition}[]\label{thm:46}
 Let $\sC$~be a symmetric monoidal $(\infty ,n)$-category.  A \emph{fully
extended $n$-dimensional topological field theory with Wick-rotated vector
symmetry group~$H_n$ and target~$\sC$} is a symmetric monoidal functor
  \begin{equation}\label{eq:56}
     F\:\be\longrightarrow \sC. 
  \end{equation}
  \end{definition}

\noindent
 We typically shorten this to `topological field theory'.  In general there
is no preferred choice of target~$\sC$, and it is an open issue to construct
suitable general targets.  In the very special invertible case we study here
there are two preferred targets; see~\S\ref{subsec:4.3}.

  \subsection{Invertible topological field theories}\label{subsec:4.2}

 There is a natural superposition of quantum systems which does not introduce
interactions between them.  In the framework of Wick-rotated field theories
on compact manifolds this is implemented by tensoring theories together, and
that tensor product makes sense for fully extended theories too.  There is a
unit for the tensor product: the trivial theory~$\bone$ in which the vector
space attached to any $(n-1)$-manifold is~$\CC$, all correlation functions
equal~1, and a similar triviality in higher codimension.  A theory~$F$ is
\emph{invertible} if there exists~$F'$ such that $F\otimes F'\cong \bone$.

  \begin{example}[]\label{thm:48}
 An $n=1$~theory~$F$ with $H_1=SO_1$ is determined by the vector
space~$F(\pt_+)$ attached to a point with positive orientation; it is
invertible if and only if this vector space is one-dimensional.  (A
one-dimensional vector space is called a \emph{line}.  A vector space~$V$ is
invertible if and only if there exists~$V'$ such that $V\otimes V'\cong \CC$,
and this happens if and only if $V$~is a line.)  In an $n$-dimensional
invertible field theory, the vector space attached to any $(n-1)$-dimensional
manifold is a line and all correlation functions between nonzero operators
are nonzero. 
  \end{example}

We first explain the transition to stable homotopy theory in the non-extended
case, as in Example~\ref{thm:48}.  The codomain, or target, of a non-extended
topological field theory (Definition~\ref{thm:42}) is the ordinary
category~$\Vect_{\CC}$ whose objects are complex vector spaces and whose
morphisms are linear maps.  To accommodate theories with fermionic states, we
use instead the codomain category~$\sVect$ of super vector spaces.  An
invertible theory~$F$ factors through the subcategory~$\sLC$ whose objects
are complex super lines\footnote{A $\zt$-graded line is either even or odd,
which means the single quantum state is either bosonic or fermionic.}  and
whose morphisms are invertible linear maps:
  \begin{equation}\label{eq:57}
     \begin{gathered}
     \xymatrix{\bne\ar[rr]^<<<<<<<<<<<<<F\ar@{-->}[dr]&&s\!\Vect_{\CC} \\ 
     &s\!\Line_{\CC}\ar@{^{(}->}[ur]} \end{gathered} 
  \end{equation}
The category~$\sLC$ is a \emph{groupoid}: every morphism is invertible.  Even
more, it is a \emph{Picard groupoid}: every object is invertible under tensor
product.  The main point is that groupoids and Picard groupoids come from
topology, as we quickly review.
 
One of the first constructions in algebraic topology goes in the opposite
direction:
  \begin{equation}\label{eq:58}
     \textnormal{Spaces}\xrightarrow{\;\;\pi
     _{\le1}\;\;}\textnormal{Groupoids} 
  \end{equation}
To any topological space~$S$ is attached a groupoid~$\pi _{\le1}S$ whose
objects are the points of~$S$; the set~$(\pi \mstrut _{\le1}S)(s_0,s_1)$ of
morphisms from~$s_0\in S$ to~$s_1\in S$ is the set of homotopy classes of
paths from~$s_0$ to~$s_1$.  If the space has no higher homotopy
information---$S$~is a \emph{homotopy 1-type}---then $\pi _{\le1}S$ captures
the homotopy type of~$S$ completely.  There is an inverse construction
that takes a groupoid~$\sG$ (or a category) and attaches a homotopy
1-type~$\|\sG\|$, the \emph{classifying space}~\cite{Se3}.

  \begin{example}[]\label{thm:50}
 Let $S=\|\sLC\|$.  Then $\pi _0S\cong \zt$, since there are two isomorphism
classes of super line; and $\pi _1S\cong \Cx$, since the automorphism group of
any super line is the group~$\Cx$ of nonzero complex numbers under
multiplication.   
  \end{example}

  \begin{remark}[]\label{thm:51}
 In Example~\ref{thm:50} the groupoid~ $\sLC$ is discrete: there is no
topology on objects or morphisms.  If we use the standard topology on the
morphism spaces of linear maps, then the geometric realization~$\|\sLC\|$ is
a homotopy 2-type with $\pi _0\cong \zt$, $\pi _1=0$, and $\pi _2\cong \ZZ$.
In other words, whereas in Example~\ref{thm:50} the discrete group~$\Cx$ of
morphisms gives rise to~$\pi _1\cong \Cx$, with the usual topology the
group~$\Cx$ deformation retracts to the circle ($\pi _0=0$, $\pi _1\cong
\ZZ$), and so its homotopy groups show up one degree higher in the geometric
realization.
  \end{remark}

A symmetric monoidal structure on a groupoid goes over to an \emph{infinite
loop structure} on the classifying space~$S$.  That is, there exists a
sequence $\sX=\{S_0, S_1, S_2, \dots \}$ of pointed spaces equipped with
homotopy equivalences $S_q\simeq \Omega S_{q+1}$, where $S_0=S$ and $\Omega
S_{q+1}$ is the based loop space.  These satisfy the condition that $S_q$~is
$(q-1)$-connected.  We call~$\sX$ a \emph{spectrum} and we call~$S$ its
\emph{0-space}.  See~\S\ref{sec:5} for a review of spectra.

  \begin{example}[]\label{thm:52}
 The classifying space $\|\Line_{\CC}\|$ has only one nontrivial homotopy
group $\pi _1\cong \Cx$, so it is an Eilenberg-MacLane space~$K(\Cx,1)$.  The
corresponding Eilenberg-MacLane spectrum is denoted~$\Sigma H\Cx$: the
0-space of the spectrum~$H\Cx$ is a~$K(\Cx,0)$, for which a simple model is
the discrete group~$\Cx$, and the~`$\Sigma $' indicates a shift. 
  \end{example}

The functor~\eqref{eq:58} is the first in a sequence of functors $\{\pi
_0,\pi _{\le1},\pi _{\le2},\dots \}$ in which the zeroth maps a space to its
set of path components and the higher ones map to higher groupoids.  The
classifying space construction also works in this context, and it produces a
space with potentially nonzero homotopy groups in any degree.  

A symmetric monoidal $\infn$-category~$\sB$ has a higher Picard groupoid
quotient~$\overline{\sB}$, obtained by formally adjoining inverses for every
object and morphism.  Also, a symmetric monoidal $\infn$-category~$\sC$ has a
maximal Picard subgroupoid~$\sCx\hookrightarrow \sC$ constructed by removing
the non-invertible objects and morphisms from~$\sC$.

  \begin{definition}[]\label{thm:53}
 A fully extended field theory $F\:\be\to\sC$ is invertible if it admits a
factorization 
  \begin{equation}\label{eq:59}
     \xymatrix@C+1em{\be\ar[r]^(.6)F \ar@{->>}[d]&\sC\\ \overline{\be
     }\ar@{-->}[r]^(.6){\widetilde F }& \sC^\times \ar@{^{(}->}[u]} 
  \end{equation}
  \end{definition}

\noindent
 Passing to classifying spaces, $\widetilde F$~is equivalent to an infinite
loop map 
  \begin{equation}\label{eq:60}
     \|F\|\:\|\be\|\longrightarrow \|\sCx\|, 
  \end{equation}
or equivalently a map of spectra.  The homotopy type of the
domain is given by the following variation of the celebrated
Galatius-Madsen-Tillmann-Weiss~\cite{GMTW} Theorem.

  \begin{theorem}[]\label{thm:54}
 $\|\be\|$~is the 0-space of the Madsen-Tillmann spectrum $\Sigma ^nMTH_n$.
  \end{theorem}

\noindent 
 One version of this theorem is proved in~\cite{BM}, though it is only for
unoriented manifolds and is carried out for ``$n$-uple categories'' rather
than $(\infty ,n)$-categories.  Proofs of Theorem~\ref{thm:54} in the context
of $(\infty ,n)$-categories have appeared in preprint form.  The theorem is
stated in~\cite[\S2.5]{L} as a corollary of the cobordism hypothesis.  A
preprint of Ayala-Francis~\cite{AF} proves the cobordism hypothesis and
Theorem~\ref{thm:54} for framed manifolds.  A preprint by
Schommer-Pries~\cite{S-P} contains a complete proof of Theorem~\ref{thm:54}
independent of the cobordism hypothesis.  Nonetheless, because there is
currently no published proof, in this paper we only use Theorem~\ref{thm:54}
as motivation and formally define an invertible field theory as a map of
spectra (Ansatz~\ref{thm:145} below).

See~\S\ref{subsec:6.1} for a review of Madsen-Tillmann spectra.

  \subsection{Universal targets}\label{subsec:4.3}

There are two universal targets for invertible topological field theories,
corresponding to the discrete and continuous topologies on~$\Cx$.  These
targets are spectra; there is no need to define an $\infn$-category~$\sC$ with
non-invertible morphisms and objects as we only consider invertible theories.
 
The first target is constructed so that invertible $n$-dimensional field
theories with that target are determined by the partition function.  The
spectrum~$\ICx$ is characterized in the homotopy category of spectra by a
functorial isomorphism
  \begin{equation}\label{eq:61}
     \pi _0\:[\sB,I\Cx] \longrightarrow  \Hom(\pi _0\sB,\Cx)
  \end{equation}
from the abelian group of homotopy classes of spectrum maps $\sB\to\ICx$ to
the character group of~$\pi _0\sB$, for any spectrum~$\sB$.  The shift~$\SCx$
satisfies a similar universal property with~$\pi _0$ replaced by~$\pi _n$.
The spectrum~$\ICx$ is closely related to the Brown-Comenetz dual to the
sphere spectrum~\cite{BC}.  Combining with the discussion
in~\S\ref{subsec:4.2} we arrive at the following.

  \begin{ansatz}[]\label{thm:145}
 A \emph{discrete invertible $n$-dimensional extended topological field
theory with symmetry group~$H_n$} is a spectrum map
  \begin{equation}\label{eq:188}
     F\:\Sigma ^nMTH_n\longrightarrow \SCx. 
  \end{equation}
The space of theories of this type is~$\Map(\Sigma ^nMTH_n,\SCx)$. 
  \end{ansatz}

\noindent
 Here `$\Map$'~indicates the \emph{space} of maps between the indicated
spectra; see~\eqref{eq:283} below.  The word `discrete' is meant to evoke the
choice~$\SCx$ for the codomain: $\Cx$~has the discrete topology.

  \begin{remark}[]\label{thm:164}
 The choice of codomain spectrum~$\SCx$, which implements the dictum `the
partition function determines the theory', holds magic derived from the first
few stable homotopy groups of spheres.  For example, the truncation to~$\pi
_{\langle n-1,n \rangle}$ is a non-extended theory, and it takes values in a
groupoid equivalent to the groupoid $\sLC$ of super lines: the homotopy groups
of spheres ``knows about'' the bosonic/fermionic grading of quantum states.
The next~ $\zt$ in the stable stem also has an interpretation in terms of
statistics of particles; see~\cite{GK} where objects with nontrivial
$\zt$-grading are termed `Majorana'.
  \end{remark}

The spectrum~$\Sigma ^nI\Cx$ is appropriate for classifying \emph{isomorphism
classes} of topological theories, but we are interested instead in
\emph{deformation classes}: we want to identify two theories if there is a
continuous path of theories connecting them.  For example, as maps
into~$\Sigma ^nI\Cx$ the Euler theories~$F_{\lambda _0},F_{\lambda _1}$ in
Example~\ref{thm:26} are nonisomorphic if $\lambda _0\not= \lambda _1$,
whereas they are always deformation equivalent.  The \emph{Anderson
dual}~$\SIZ$ is the appropriate codomain to compute deformation
classes.\footnote{$\Zo=2\pi \sqmo\ZZ\subset \CC$ avoids the choice of a
particular~$\sqmo\in \CC$.}  Roughly speaking, it results from~$\SCx$ by
taking the continuous topology on~$\Cx$.  Its universal property is expressed
in the short exact sequence
  \begin{equation}\label{eq:62}
     0 \longrightarrow \Ext^1(\pi _n\sB,\Zo) \longrightarrow
     [\sB,\SIZ]\longrightarrow \Hom(\pi _{n+1}\sB,\Zo)\longrightarrow 0 
  \end{equation}
which is non-canonically split.  The kernel is the torsion subgroup: 
  \begin{equation}\label{eq:185}
     [\sB,\SIZ]\tors\cong \Ext^1(\pi _n\sB,\Zo). 
  \end{equation}
There is a map
  \begin{equation}\label{eq:63}
     \phi \:[\sB,\Sigma ^nI\Cx]\cong \Hom(\pi _n\sB,\Cx)\longrightarrow
     \Ext^1(\pi _n\sB,\Zo)  
  \end{equation}
onto the kernel of~\eqref{eq:62}.  It sends a homomorphism $\pi _n\sB\to\Cx$
to the pullback of the exponential group extension 
  \begin{equation}\label{eq:92}
     1\longrightarrow
     \Zo\longrightarrow\CC\xrightarrow{\;\;\exp\;\;}\Cx\longrightarrow1.  
  \end{equation}
If we give~$\Cx$ its usual topology, then $\phi $~may be regarded as mapping
the topological space $\Hom(\pi _n\sB,\Cx)$ to its group of path components.
 
Intuitively, to define the notion of deformation equivalence of
theories~\eqref{eq:188} we want to consider a second topology on
$\Map(\Sigma ^nMTH_n,\SCx)$ induced from the continuous topology on~$\Cx$,
and then compute~$\pi _0$.  Instead we make use of the fibration
  \begin{equation}\label{eq:231}
     H\CC\xrightarrow{\;\;\exp\;\;}\ICx\longrightarrow \Sigma \IZ 
  \end{equation}
induced from~\eqref{eq:92} as follows. 

  \begin{definition}[]\label{thm:155}
 Theories  $\alpha _0,\alpha _1\in \Map(\Sigma ^nMTH_n,\SCx)$ are
\emph{deformation equivalent} if there exists $\xi \in H^n(\mth n;\CC)$ whose
image under exp is the difference $[\alpha _1]-[\alpha _0]$ of the
isomorphism classes $[\alpha _0],[\alpha _1]\in [\mth n,\SCx]$. 
  \end{definition}

We immediately conclude the following.

  \begin{theorem}[]\label{thm:56}
 There is a 1:1 correspondence 
  \begin{equation}\label{eq:64}
     \left\{\normalfont\vcenter{\hbox{deformation classes of discrete
     invertible}\hbox{$n$-dimensional extended topological
     field}\hbox{theories with 
     symmetry group~$H_n$}} \right\} \cong [\Sigma ^nMTH_n,\SIZ]\tors. 
  \end{equation}
  \end{theorem}

\noindent
 This appears, at least implicitly, in a joint paper~\cite{FHT1} of the
authors and Constantin Teleman; Theorem~\ref{thm:56} has been the basis of
many investigations since.  

It is natural to ask for a field theoretic interpretation of a map of spectra
$\Sigma ^nMTH_n\to\SIZ$ whose homotopy class is not torsion, so does not
factor through~$\SCx$.  We give one in the next subsection
(Ansatz~\ref{thm:156}).

  \subsection{Remarks on non-topological invertible theories and low energy
  approximations}\label{subsec:4.4} 

The main immediate application of Theorem~\ref{thm:110} in this paper is to
low-energy approximations of gapped unitary quantum systems in case that
approximation is invertible.  For the heuristic discussion in this section we
momentarily drop the invertibility hypothesis. 
 
A typical example of the phenomenon we wish to highlight is 3-dimensional
Yang-Mills theory with a Chern-Simons term.  The coupling constant of the
Chern-Simons term obeys an integrality constraint.  Then the low energy
effective theory is quantum ``topological'' Chern-Simons theory~\cite{W3}.
In fact, this low energy theory is \emph{not} topological; there is a mild
metric dependence~\cite{W2}.  One precise expression of the mildness is that
the energy-momentum tensor\footnote{The energy-momentum tensor is a multiple
of the Cotton tensor of the Riemannian 3-manifold.} is a multiple of the
identity operator, which is the only point operator in the theory anyhow.
(See the discussion in~\cite[\S1.1]{GK}.)  Witten observes that if one is
willing to introduce some sort of framing, then the long distance topological
Chern-Simons theory is the tensor product of a purely topological theory and
an invertible theory.  The invertible theory is analogous to a gravitational
Chern-Simons theory, but more precisely its partition function is the
exponential of the Atiyah-Patodi-Singer $\eta $-invariant.  The coupling
constant does not obey the usual integrality constraint, which is why the
framing is required for this global decomposition.  The full quantum
Yang-Mills theory with Chern-Simons term is a theory of oriented Riemannian
manifolds (the Wick rotated symmetry group is~$H_3=SO_3$), and so one expects
the same for the low-energy approximation.  That indeed holds; it is only to
make a global decomposition into topological $\times $ invertible that a
framing is introduced.
 
This example violates the physical principle~(ii) stated towards the
beginning of~\S\ref{sec:1}.  A more precise expectation is that the low
energy physics of a gapped system is well-approximated by a theory whose
energy-momentum tensor may depend on the the background fields, but as an
operator it is a multiple of the identity at each point.  Or, at least
locally we suppose the low energy theory is topological~$\times $~invertible.
If the low energy theory happens to be invertible, then we conclude that any
non-topological invertible theory can occur and that there is no shift of
symmetry group, e.g., no extra tangential structure is required.  We expect
that choices must be made in constructing the low energy effective theory, so
a potential `low energy approximation' map from gapped theories to theories
that are locally topological times invertible may only be defined up to
homotopy.  (See \cite[\S11.4]{F4} for another perspective on the
appearance of a possibly nontopological invertible theory.)
 
To illustrate the nature of the low energy approximation, we contemplate the
following three geometric objects associated to a smooth manifold~$M$: (a)~a
principal $\Cx$-bundle $P\to M$ with connection, (b)~a principal $\Cx$-bundle
$P\to M$ with flat connection, and (c)~a principal $\Cx$-bundle $P\to M$
(with no connection).  In particular, we track what information is induced on
the free loop space $LM=\Map(\cir,M)$ by integrating over the loop.  In~(a)
we obtain a smooth function $LM\to\Cx$, the holonomy, and if there is nonzero
curvature then it has nonzero derivative.  In~(b) the holonomy is a
\emph{locally constant} function $LM\to\Cx$, and therefore we can use the
\emph{discrete} topology on~$\Cx$: the holonomy represents a class
in~$H^0(LM;\Cx)$.  In~(c) there is no connection, so no holonomy, but
nonetheless we can extract a principal $\Zo$-bundle $E_P\to LM$, a fiber
bundle of $\Zo$-torsors.  Namely, an element~$\lambda \in \Cx$ determines a
$\Zo$-torsor $E_\lambda \subset \CC$ of all~$x\in \CC$ such that
$\exp(x)=\lambda $, and so the holonomy function $LM\to\Cx$ of a
connection~$\Theta \in \sA_P$ on $P\to M$ determines $E_{P,\Theta }\to LM$,
so a $\Zo$-torsor over $\sA_P\times LM$.  Since the affine space~$\sA_P$ of
connections is contractible, the principal $\Zo$-bundle over~$\sA_P\times LM$
descends to a principal $\Zo$-bundle $E_P\to LM$.  It may be regarded as the
homotopical information in a connection.  It determines a class in the sheaf
cohomology group~$H^0(LM;\Cxu)$ in which $\Cxu$~has the \emph{continuous}
topology.  Since $\Cxu$ is an Eilenberg-MacLane space with $\pi _1\cong \Zo$,
there is an isomorphism
  \begin{equation}\label{eq:232}
     H^0(LM;\Cxu)\xrightarrow{\;\;\cong \;\;} H^1(LM;\Zo) .
  \end{equation}
 
Returning to invertible field theories\footnote{Note that each of (a), (b),
and~(c) above determines the corresponding type of invertible 1-dimensional
field theory of oriented manifolds equipped with a map to~$M$.} we have the
following situations: (a)~a non-topological theory, as contemplated in
Remark~\ref{thm:142}; (b)~a discrete invertible topological theory, as in
Ansatz~\ref{thm:145}; and (c)~a topological field theory whose partition
``function'' is a $\Zo$-torsor rather than a complex number.  While (a)~and
(b)~have clear analogs for non-invertible field theories, it is unclear what a
non-invertible analog of~(c) would be.  In the invertible case we posit the
following definition of a type~(c) theory.  

  \begin{ansatz}[]\label{thm:156}
 A \emph{continuous invertible $n$-dimensional extended topological field
theory with symmetry group~$H_n$} is a spectrum map 
  \begin{equation}\label{eq:237}
     \varphi \:\Sigma ^nMTH_n\longrightarrow \Sniz. 
  \end{equation}
The space of theories of this type is~$\Map(\Sigma ^nMTH_n,\Sniz)$.  
  \end{ansatz}

  \begin{remark}[]\label{thm:157}
 In differential geometry a principal $\Cx$-bundle $P\to M$ has a
\emph{primary topological} invariant in~$H^2\bigl(M;\Zo \bigr)$, its Chern
class.  A connection gives a \emph{secondary geometric} invariant, its
holonomy.  If the connection is flat, the secondary invariant is also
topological (discrete), and in that case the Chern class lies in the torsion
subgroup of~$H^2\bigl(M;\Zo \bigr)$.  The \emph{stable} continuous invertible
field theories we encounter in~\S\ref{subsec:6.2} attach a primary
$\Zo$-valued invariant to closed $(n+1)$-manifolds. 
  \end{remark}

A discrete invertible topological field theory~$F$ (Ansatz~\ref{thm:145})
gives rise to a continuous invertible topological field theory~$\varphi $,
which retains the homotopical information in~$F$, in particular its
deformation class.  In this paper we do not develop the theory of
non-topological field theories, but in the invertible case we use instead
continuous topological theories, which represent the homotopical information
carried by a geometric theory. 

  \begin{remark}[]\label{thm:158}
 In the application to low energy approximations of gapped theories, we
expect that only this homotopical shadow of a geometric theory is
well-defined, due to the choices in constructing a low energy theory. 
  \end{remark}

   \section{Equivariant stable homotopy theory}\label{sec:5}

Reflection symmetry in invertible topological theories is expressed by a
$\Z/2$-action on the constituent spectra.  This requires working in
$\Z/2$-equivariant stable homotopy theory.  What we will use here is {\em
Borel equivariant} homotopy theory.  This is somewhat easier than the more
general theory, and at the moment is all that seems needed for our main
results.  There are many places to read about equivariant stable homotopy
theory.  The reader may wish to consult~\cite{Ad},~\cite{GM},
\cite[Chapter~2]{HHR},~\cite{Sch} and~\cite[Chapter~8]{tD}.

\subsection{Spectra}
\label{sec:mssspectra}

Let $\spaces$ be the category of pointed topological spaces, and for
$A,B\in \spaces$ write $\spaces(A,B)$ for the set of basepoint
preserving continuous functions from $A$ to $B$ and $\uspaces(A,B)$ for
the same set, regarded as a topological space with the compact open
topology.

A {\em spectrum} $X$ is a sequence $\{X_{0},X_{1},\dots \}$ of pointed
spaces, equipped with structure maps
$s_{n}:S^{1}\wedge X_{n}\to X_{n+1}$.  A map $X\to Y$ of spectra is a
sequence of maps $X_{n}\to Y_{n}$ making the diagrams
\[
\xymatrix{
S^{1}\wedge X_{n}  \ar[r]^{s^{X}_{n}}\ar[d]  & X_{n+1}
\ar[d] \\
S^{1}\wedge Y_{n}  \ar[r]_{s^{Y}_{n}}        & Y_{n+1}
}
\]
commute.   The set of spectrum maps from $X$ to $Y$ is a subset of 
\[
\prod_{n} \uspaces(X_{n},Y_{n})
\]
and so may be regarded as a topological space with the subspace
topology.   The space of maps between spectra $X$ and $Y$ will be
denoted $\uspectra(X,Y)$.

The {\em homotopy groups} $\pi_{n}X$ of a spectrum $X$ are
defined for $n\in\Z$ by 
\begin{equation}
\label{eq:mss12}
\pi_{n}(X) = \varinjlim_{k}\pi_{n+k}X_{n+k}
\end{equation}
in which the bonding maps are given by the suspension mapping 
\[
\pi_{n+k}X_{n+k}\xrightarrow{\Sigma}{} \pi_{n+k+1}\Sigma X_{n+k} \xrightarrow{s_{n+k}}{}\pi_{n+k+1}X_{n+k+1}.
\]
The group $\pi_{n+k}X_{n+k}$ is defined for any $n\in \Z$ as soon as
$k\ge -n$.  A map $X\to Y$ is a {\em weak equivalence} if it induces
an isomorphism of homotopy groups.  

Equipped with the weak equivalences, the category $\spectra$ of
spectra becomes a bona fide place for doing homotopy theory.  A
functor $\spectra\to \cat C$ to a category $\cat C$ is a {\em homotopy
functor} if it takes weak equivalences to isomorphisms.  There is a
universal homotopy functor $\spectra\to\ho\spectra$ characterized by
the property that the restriction mapping gives an equivalence between
the category of functors $\ho\spectra\to \cat C$ with the category of
homotopy functors $\spectra\to\cat C$.  The category $\ho\spectra$ is
the {\em homotopy category of spectra}, and the set (in fact abelian
group) $\ho\spectra(X,Y)$ is called the abelian group of {\em homotopy
classes of maps} from $X$ to $Y$.  We will use the common abbreviation
\[
[X,Y] = \ho\spectra(X,Y).
\]

\begin{example}
\label{eg:5}
The suspension spectrum $\Sigma^{\infty}Z$ of a space $Z$ is the
spectrum 
\[
\big(\Sigma^{\infty}Z\big)_{n}= S^{n}\wedge Z
\]
with the structure maps derived from the equivalence $S^{1}\wedge
S^{n}=S^{n+1}$.   When the context is clear it is customary to drop
the $\Sigma^{\infty}$ and not distinguish in notation between a space
nd its suspension spectrum.   
\end{example}

\begin{example}
\label{eg:4}
For a non-negative integer $k \ge 0$  let $S^{k}$ be the suspension
spectrum of the $k$-sphere and 
$S^{-k}$ be the spectrum
defined by 
\[
\big(S^{-k}\big)_{n} = \begin{cases} \ast & n < k \\ S^{n-k} & n\ge k\end{cases}.
\]
From the formula~\eqref{eq:mss12} one easily checks that for all $k\in\Z$ one
has an isomorphism
\[
[S^{k},X]\approx \pi_{k}X
\]
natural in $X$.   
\end{example}

\subsubsection{Smash product}

Suppose that $X=\{X_{n} \}$ is a spectrum and $Z$ is a space.
Define $X\wedge Z$ to be the spectrum with
\[ 
\big(X\wedge Z\big)_{n} = X_{n}\wedge Z
\]
and the structure maps derived from those of $X$.  This is the {\em
smash product} of the spectrum $X$ with the space $Z$.

\begin{example}
\label{eg:1}
The spectrum $S^{0}\wedge Z$ is the suspension spectrum of $Z$.
\end{example}

\begin{example}
\label{eg:6}
The spectrum $S^{-k}\wedge S^{k}$ consists of the spaces
\[
\big(S^{-k}\wedge S^{k}\big)_{m} = \begin{cases}
\ast &\quad m <k \\
S^{m} &\quad m\ge k.
\end{cases}
\]
There is an inclusion
\[
S^{-k}\wedge S^{k} \to S^{0}
\]
which is easily checked to be a weak equivalence.
\end{example}

For a spectrum $X=\{X_{n} \}$ there is a functorial weak equivalence 
\begin{equation}
\label{eq:mss1}
\ho\varinjlim S^{-n}\wedge X_{n} \xrightarrow{\approx}{} X.
\end{equation}
(See, for example~\cite[\S2.2.1]{HHR} where it is called the
{\em canonical homotopy presentation}.)

There is an enrichment $\dspectra$ of $\ho\spectra$ over the homotopy
category of spaces.   It is characterized by the existence of an isomorphism
\begin{equation}
\label{eq:mss2}
\ho\spaces(Z,\dspectra(X,Y)) \approx \ho\spectra(X\wedge Z,Y)
\end{equation}
functorial in CW complexes $Z$, and spectra $X$ and $Y$.   We will employ the
abbreviation   
  \begin{equation}\label{eq:283}
     \Map(X,Y)=\dspectra(X,Y). 
  \end{equation}
Taking $Z$  to be the {\em space} $S^{0}$ in~\eqref{eq:mss2} gives the isomorphism
\[
[X,Y] = \pi_{0}\Map(X,Y).
\]

When the spectrum $X=\{X_{n} \}$ has the property that each $X_{n}$ is
a CW complex and $Y$ has the property that each map 
\[
Y_{n}\to \Omega Y_{n+1}
\]
is a weak equivalence,  the homotopy type of $\Map(X,Y)$ is given by 
\begin{equation}
\label{eq:mss3}
\dspectra(X,Y) = \ho\varprojlim \mathcal M(X_{n},Y_{n}),
\end{equation}
with $\mathcal M(X_{n},Y_{n})$ is the homotopy limit of the diagram 
\[
\xymatrix@!C=2em{
\uspaces(X_{n},Y_{n})\ar[dr] && \ar[dl]^{\sim}
\uspaces(X_{n-1},Y_{n-1}) \ar[dr] && \ar[dl]^{\sim}
&\ar@{}[d]|{\dots}&\ar[dr] && \uspaces(X_{0},Y_{0})\ar[dl]^{\sim} \\ 
& \uspaces(X_{n-1},\Omega Y_{n})&& \uspaces(X_{n-2},\Omega Y_{n-1}) && &&
\uspaces(X_{0},\Omega Y_{1}) &
}
\]
in which the southeast arrows are given by  the compositions
\[
\uspaces(X_{m},Y_{m})\to \uspaces(S^{1}\wedge X_{m-1},Y_{m})\approx
\uspaces(X_{m-1}\Omega Y_{m}).
\]
Note that the projection map $\mathcal M(X_{n},Y_{n})\to
\uspaces(X_{n},Y_{n})$ is a weak equivalence, so that~\eqref{eq:mss3} can
heuristically be interpreted as giving a presentation of
$\dspectra(X,Y)$ as a homotopy inverse limit of the spaces
$\uspaces(X_{n},Y_{n})$.

A spectrum $Y$ with the property that for all $n$ the map $Y_{n}\to
\Omega Y_{n+1}$ is a weak equivalence is called an {\em $\Omega$-spectrum} (or a {\em loop
spectrum}).    Every spectrum $Y$ is naturally weakly equivalent to an
$\Omega$-spectrum.   Indeed, given $Y$ define $LY$ by 
\[
LY_{n}=\ho\varinjlim \Omega^{k}Y_{n+k}.
\]
Using the homeomorphism $\Omega(\Omega^{k}Y_{n+k})\approx
\Omega^{k}\Omega Y_{n+k}$ one sees that $LY$ has the structure
of an $\Omega$-spectrum and that the canonical map $Y\to LY$
is a weak equivalence.  

\subsubsection{Duality}
\label{sec:mssduality}

The operation $X\wedge Z$ extends to a symmetric monoidal
smash product on spectra.   In fact there is a unique extension having
the property that it commutes with
colimits in both variables, and for spaces $Z_{1}$ and $Z_{2}$ and
integers $k,\ell\ge 0$ one has
\[
\big(S^{-k}\wedge Z_{1}\big)\wedge \big(S^{-\ell}\wedge
Z_{2}\big) \simeq S^{-(k+\ell)}\wedge Z_{1}\wedge Z_{2}.
\]
The existence and uniqueness can be deduced from the canonical
homotopy presentation~\eqref{eq:mss1}.   

Equipped with the smash product the categories $\ho\spectra$ and
$\dspectra$ become symmetric monoidal categories.  By Example~\ref{eg:6} the
suspension spectra of spheres are dualizable (in fact invertible).  It
follows that the suspension spectrum of any finite CW complex is also
dualizable.

\subsubsection{Stability}
\label{sec:mssstability}

An easy check (or an appeal to the invertibility of spheres) shows
that for all $k$ and all $X$ the map
\[
\pi_{k}X\to \pi_{k+1}X\wedge S^{1}
\]
is an isomorphism.   This implies a map $A\to X$ gives rise to a long
exact sequence 
\[
\dots\to \pi_{k}A\to\pi_{k}X\to \pi_{k}X\cup CA \to \pi_{k-1}A\to\dots
\]
in which $X\cup CA$ is the spectrum 
\[
\big(X\cup CA \big)_{n} = X_{n}\cup CA_{n}
\]
with $CA=A\times [0,1]/A\times \{1 \}\cup \ast \times [0,1]$.   This, in turn, 
implies that the map from $A$ to the homotopy fiber of $X\to X\cup CA$
is a weak equivalence.

\subsubsection{Thom Spectra}
\label{sec:mssthom-spectra}

Let $X$ be a space.  Given a map $V:X\to BO$, define a sequence of
maps $V_{n}:X_{n}\to BO_{n}$ by the homotopy pullback squares
\begin{equation}
\label{eq:mss13}
\xymatrix{
X_{n}  \ar[r]\ar[d]_{V_{n}}  & X
\ar[d]^{V} \\
BO_{n}  \ar[r]        & BO\mathrlap{\ .}
}
\end{equation}
The map $V_{n}:X_{n}\to BO_{n}$ classifies a vector bundle of
rank $n$ over $X_{n}$ (which will also be denoted $V_{n}$).  By
construction, the pullback of $V_{n+1}\to X_{n+1}$ to $X_{n}$ comes equipped with
an isomorphism to $V_{n}\oplus\triv{}\to X_n$.  This give a map of Thom
spaces 
\[
\Sigma\thom(X_{n};V_{n})= \thom(X_{n};V_{n}\oplus 1)\to \thom(X_{n+1},V_{n+1})
\]
making the sequence of spaces $\{\thom(X_{n};V_{n}) \}$ into a
spectrum.   This is the {\em Thom spectrum } of $V$, denoted
$\thom(X;V)$.   The canonical homotopy presentation of $\thom(X;V)$
takes the form 
\[
\thom(X;V) = \ho\varinjlim S^{-n}\wedge \thom(X_{n};V_{n}).
\]

We will also encounter the Thom spectrum $\thom(X;-V)$ associated to a
map $V\:X\to BO$ by composing with the ``additive inverse'' map
$(-1):BO\to BO$ (see~\S\ref{subsec:6.1}).  With $X_{n}$ and $V_{n}$ defined as
in~\eqref{eq:mss13}, the isomorphism
\[
V_{n+1}\vert_{X_{n}} \approx V_{n}\oplus\triv{}
\]
becomes 
\[
-V_{n+1}\vert_{X_{n}} \approx -V_{n} -\triv{}.
\]
This leads to maps 
\[
\thom(X;-V_{n})\to S^{1}\wedge \thom(X_{n+1};-V_{n+1}),
\]
and an alternative presentation
  \begin{equation}\label{eq:thomlim}
  \thom(X;-V)=\ho\varinjlim S^{n}\wedge \thom(X_{n};-V_{n}).
  \end{equation}
If $V$ has virtual dimension $d$ then $V-\triv{d}$
has virtual dimension $0$ and one defines 
\[
\thom(X;V) = S^{d}\wedge\thom(X;V-\triv{d}).   
\]

The Thom spectrum construction is a functor on the category of spaces
over the classifying space $\Z\times BO$ of $KO$-theory.   
It is symmetric monoidal in the sense that for
$V:X\to \Z\times BO$ and $W:Y\to \Z\times BO$ there is a natural weak equivalence
\[
\thom(X\times Y;\pi_{X}^{\ast}V\oplus \pi_{Y}^{\ast}W)\approx \thom(X;V)\wedge \thom(Y;W),
\]
in which $\pi_{X}$ and $\pi_{Y}$ are the projections.

\subsection{Borel equivariant stable homotopy theory}
\label{sec:mssborel-equiv-stable}

Now suppose that $G$ is a compact Lie group (which in our case will be
$\Z/2$) and let $\bgspectra{G}$ be the category of spectra equipped with a
$G$ action, and equivariant maps.   An object of $\bgspectra{G}$
consists of a sequence $\{X_{n},s_{n} \}$ of left $G$-spaces $X_{n}$
and equivariant maps $S^{1}\wedge X_{n}\to X_{n+1}$ in which $S^{1}$
has the trivial $G$-action.   Sometimes what we are calling a
$G$-spectrum is called a {\em naive $G$-spectrum}.   

\begin{definition}
\label{def:1a}
A map $X\to Y$ in $\bgspectra{G}$ is a {\em Borel weak equivalence} if
it is a weak equivalence when regarded as a map in $\spectra$.   
\end{definition}

Equipped with the Borel weak equivalences, the category
$\bgspectra{G}$ becomes a category in which one can do homotopy
theory.  The homotopy category $\ho\bgspectra{G}$ is defined as the
target of the universal homotopy functor out of $\bgspectra{G}$.  
We will use the abbreviation 
\[
[X,Y]^{hG} = \ho\bgspectra{G}(X,Y).  
\]

The construction of the smash product goes through in a
straightforward way for the Borel equivariant spectra, and there is a
{\em derived equivariant mapping space} between two equivariant
spectra.  In fact, it follows from the expression~\eqref{eq:mss3} that
when $X$ and $Y$ are $G$-spectra, the space $\dspectra(X,Y)$ acquires
the homotopy type of a $G$-space.  The derived equivariant mapping
space works out to be homotopy fixed point space
\[
\Map^G(X,Y)=\Map(X,Y)^{hG},
\]
and the maps in the homotopy category of $G$-spectra are given by
\[
[X,Y]^{hG} = \pi_{0}\Map(X,Y)^{hG}.
\]

In Borel equivariant homotopy theory the suspension spectra of finite
$G$-sets (with a disjoint base point added) are self dual.  This
implies that the suspension spectra of finite $G$-CW-complexes are
dualizable and the suspension spectrum of the one point
compactification $S^{V}$ of a finite dimensional representation $V$ of
$G$ is invertible.  These facts are not quite immediate. If $X$ is a
finite $G$-set, then the evaluation map
\[
X_{+}\wedge X_{+}\to S^{0}
\]
is the map of suspension spectra induced by the map 
\[
X\times X\to S^{0}
\]
sending the diagonal to the non base point and the complement of the
diagonal to the base point.   It is not so straightforward to write
down the coevaluation map.   Nevertheless, for $G$-spectra $W$ and
$Z$, the composite
\[
\Map(Z,W\wedge X_{+})\to 
\Map(Z\wedge X_{+},W\wedge X_{+}\wedge X_{+})\to 
\Map(Z\wedge X_{+},W)
\]
is a $G$-equivariant map that is a weak equivalence of underlying
spaces, and so gives an equivalence
\begin{align*}
\Map(Z,W\wedge X_{+})^{hG} \approx \Map(Z\wedge X_{+},W)^{hG}
\end{align*}
and an isomorphism
\[
[Z,W\wedge X_{+}]^{hG} \approx [Z\wedge X_{+},W]^{hG}.
\]

Once one knows that the finite $G$-sets are dualizable it follows that
the suspension spectrum of any finite $G$-CW-complex is dualizable.  We
denote the dual of~$X$ as~$D(X)$.
This implies the invertibility of $S^{V}$ since the map
\[
D(S^{V})\wedge S^{V}\to S^{0}
\]
is a weak equivalence of underlying spectra.   It is customary to use
the notation
\[
S^{-V} = D S^{V}.
\]
For more on virtual representation spheres see Example~\ref{eg:3} of
\S\ref{sec:mssequiv-thom-spectra}.  

\subsubsection{Homotopy fixed points and homotopy orbits}

Regarding a non-equivariant spectrum as a $G$-spectrum with the trivial
action gives a functor
\[
\spectra\to \bgspectra{G}.
\]
This functor preserves weak equivalences  and so induces a functor 
on homotopy categories.   The homotopy orbit and fixed point functors
provide both a left and right adjoint to this induced functor.

Recall that the {\em homotopy orbit  space} of a pointed $G$-space $Z$ is the
space
\[
Z_{hG}=EG_{+}\underset{G}{\wedge}Z,
\]
and that the {\em homotopy fixed point space} is the space 
\[
Z^{hG} = \spaces(EG_{+},Z)^{G}
\]
of equivariant basepoint preserving maps from $EG_{+}$ to $Z$.   
These notions extend component-wise to equivariant spectra.  The {\em homotopy
orbit spectrum} of a $G$-spectrum $X=\{X_{n} \}$ is the spectrum
$X_{hG}=\{(X_{n})_{hG} \}$ and the {\em pre homotopy fixed point spectrum} is
the spectrum $X^{h'G}=\{(X_{n})^{hG}\}$.    

The functor $X_{hG}$ preserves weak equivalences and so directly
induces a functor  on homotopy categories.   The functor $X^{h'G}$
preserves weak equivalences between $\Omega$-spectra and so induces a
{\em homotopy fixed point}
functor 
\[
(\slot)^{hG}:\ho\bgspectra{G}\to \ho\spectra
\]
sending  $X$ to $(LX)^{h'G}$.   

These functors on the homotopy category
are adjoints to the inclusion 
\[
\ho\spectra\to \ho\bgspectra{G}
\]
in the sense that there are natural isomorphisms
\begin{align}
\label{eq:mss9}
[X,A]^{hG} &\approx [X_{hG},A] \\
\label{eq:mss10}
[A,Y]^{hG} &\approx [A,Y^{hG}] 
\end{align}
in which $X$ and $Y$ are $G$-spectra and $A$ is a spectrum with
trivial $G$-action.  Also, the fixed point spectrum~$A^{h\Z/2}$ is
computed as
 \begin{equation}\label{eq:126}
     \Map^{\Z/2}(S^{0},A)\;\simeq\; \Map(
     B\Z/2_+,A)\xleftarrow{\simeq}A\vee
     \Map(B\Z/2,A) \xrightarrow{\simeq}{} A\times
     \Map(B\Z/2,A),
  \end{equation}
in which the left pointing map  involves a choice of
a basepoint $x\in B\Z/2$ and is the sum of the map
\[
B\Z/2_{+} \to S^{0} 
\]
sending $B\Z/2$ to the non basepoint and the map
\[
B\Z/2_{+} \to B\Z/2 
\]
which is the identity map on $B\Z/2$ and sends the disjoint
base point on the left to the new basepoint on the right.

\subsubsection{Equivariant Thom spectra}
\label{sec:mssequiv-thom-spectra}

Suppose that $B$ is a space and $p:X\to B$ is a principal $G$-bundle.
A map $W:B\to BO$ leads, as above, to a sequence of maps
\[
\xymatrix{
B_{n}  \ar[r]\ar[d]^{W_{n}}  & B_{n+1}  \ar@{-->}[r]\ar[d]^{W_{n+1}}   &  B
\ar[d]^{W}\\ 
BO_{n}  \ar[r]        & BO_{n+1}  \ar@{-->}[r]         & BO
}
\]
and a Thom spectrum $\thom(B;W)=\big\{\thom(B_{n};W_{n})\big \}$.   Define
principal $G$-bundles $X_{n}\to B_{n}$ by the pullback square
\[
\xymatrix{
X_{n}  \ar[r]\ar[d]_{p_{n}}  & X
\ar[d]^{p} \\
B_{n}  \ar[r]        & B.
}
\]
The bundle $p_{n}^{\ast}W_{n}$ is a $G$-equivariant
vector bundle on $X_{n}$.  In fact, by descent,  the data of a $G$-equivariant
vector bundle on $X_{n}$ is equivalent to the data of a vector bundle
over $B_{n}$.   The $G$-action on $(X_{n},p^{\ast}W_{n})$ induces a $G$-action on the Thom
spectrum $\thom(X,p^{\ast} W) =\{\thom(X\mstrut _{n};p_{n}^{\ast} W_{n}) \}$
making it into an equivariant spectrum.  By construction the homotopy orbit
spectrum is given by
\begin{equation}
\label{eq:mss11}
\thom(X;p^{\ast}W)_{hG} = \thom(B;W).
\end{equation}

As in \S\ref{sec:mssthom-spectra}, equivariant Thom spectra for maps
$B\to \Z\times BO$ are defined by subtracting a suitable trivial
bundle and suspending the result.

\begin{example}[Representation spheres]
\label{eg:3}
An element $V\in KO^{0}(BG)$ is classified by a map 
\[
V:BG\to \Z\times BO
\]
and so gives rise to an equivariant Thom spectrum.   When $V$
corresponds to a representation of $G$ the equivariant Thom spectrum
is the spectrum $S^{V}$.   This construction sends sums of elements of
$KO^{0}(BG)$ to smash products of $G$-spectra.  Composing with the map 
\[
RO(G)\to KO^{0}(BG)
\]
gives a construction of a sphere $S^{V}$ associated to every virtual
representation $V$ of $G$.   This gives another approach to the
construction and invertibility of representation spheres in Borel
equivariant stable homotopy theory.
\end{example}

\subsubsection{The $\sigma$-sphere}
\label{sec:mssspecial-examples}

We now specialize to the case $G=\Z/2$, and write $\sigma$ for the
real sign representation.  The sphere $S^{\sigma}$ has an equivariant
cell decomposition with one non-basepoint fixed $0$-cell, and one free
$1$-cell as shown here.
  \begin{center}
\includegraphics[]{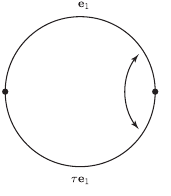}
  \end{center}
This gives a pushout square 
\[
\xymatrix{
\Z/2\times \partial D^{1}  \ar[r]\ar[d]  &\Z/2\times D^{1}
\ar[d] \\
S^{0}  \ar[r]        & S^{\sigma}
}
\]
leading to a cofibration sequence 
\begin{equation}
\label{eq:mss4}
\Z/2_{+}\to S^{0}\to S^{\sigma}
\end{equation}
of equivariant spectra.  Passing to duals and using the self-duality
of finite $G$-sets gives a cofibration sequence
\begin{equation}
\label{eq:mss5}
S^{-\sigma} \to S^{0}\to \Z/2_{+}.
\end{equation}
The map $S^{0}\to \Z/2_{+}$ is the {\em transfer map} and,
non-equivariantly, has degree $1$ on each summand of
$\Z/2_{+}=S^{0}\vee S^{0}$.

Write
\begin{align*}
\gamma &= 1- \sigma \\
\delta &= \sigma-1.
\end{align*}
For a $\Z/2$-spectrum $X$ we define 
\begin{equation}
\label{eq:mssh2}
\begin{aligned}
X^{\delta} &= S^{\delta}\wedge X \\
X^{\gamma} &= S^{\gamma}\wedge X.
\end{aligned}
\end{equation}

Smashing with~\eqref{eq:mss4} and~\eqref{eq:mss5} gives for any $X$, (co-)fibration sequences
  
\begin{align}\label{eq:mss189}
     X^{\delta}\to \Z/2_{+}&\wedge X\to X\quad\text{and}  \\
\label{eq:mss6}
X\to \Z/2_{+}&\wedge X\to X^{\gamma}.
\end{align}

\subsection{Real structures}
\label{sec:mssreal-structures}
Our next aim is to equip $\ICx$ and $\IZ$ with $\Z/2$-actions
corresponding to complex conjugation, in such a way that the
cofibration sequence (see~\eqref{eq:231})
  \begin{equation}
\IZ\to H\CC\xrightarrow{\exp} \ICx
  \end{equation}
is a cofibration sequence of $\Z/2$-equivariant spectra.  Though there no
mystery about the action on the abelian group-valued functor $[\slot,\ICx]$,
there are infinitely 
many refinements of this to an action on the spectrum $\ICx$.  Here we will
motivate a specific choice, and check it against three situations in which
there is a naturally occurring action.

\subsubsection{$\Z/2$-actions}\label{subsubsec:act} The space of
$\Z/2$-actions on a 
spectrum $X$ is the space of maps
\[
B\Z/2\to B\haut(X)
\]
from the classifying space of $\Z/2$ to the classifying space of the
monoid of self homotopy equivalences of $X$.    Smashing a map
$S^{0}\to S^{0}$ with the identity map of $X$ gives a map 
\[
B\haut(S^{0})\to B\haut(X).
\]
The maps $B\Z/2\to B\haut(S^{0})$ then correspond both to
(i)~$\Z/2$-actions on $S^{0}$ and (ii)~ $\Z/2$-actions on all spectra
which are 
natural in the sense that they commute with all maps and are homotopy
colimit preserving.  Put more succinctly, the ``natural'' $\Z/2$-actions are
homotopy colimit preserving sections of the forgetful functor
\begin{equation}
\label{eq:mss8}
\bgspectra{\Z/2} \to \spectra.
\end{equation}

Associating to a vector space its
one point compactification defines a map 
\[
BO\to B\haut(S^{0}),
\]
so that a virtual representation $V$ of $\Z/2$,
of virtual dimension $0$, determines a natural $\Z/2$-action via the composition
\[
B\Z/2\xrightarrow{V} BO\to B\haut(S^{0}).
\]
The corresponding section of~\eqref{eq:mss8} is the one sending a
spectrum $X$ to $S^{V}\wedge X$.   

\begin{remark}
\label{rem:3}   Because $S^{0}$ is the tensor unit in $\spectra$, the
space $B\haut(S^{0})$ is actually an infinite loop space.   The map
$BO\to B\haut(S^{0})$ also turns out to be an infinite loop map.
This means that ``natural'' $\Z/2$-actions may be composed, and that
the composition of actions corresponding to virtual representations
$V$ and $W$ is the natural action corresponding to $V\oplus W$.
\end{remark}

\begin{remark}
\label{rem:2a}
From the defining property of $\IZ$ one can check that the map
\begin{align*}
\Map(S^{0},S^{0}) &\to\Map(\IZ,\IZ)\\
f &\mapsto f\wedge \id
\end{align*}
is a weak equivalence.  Now the loop space of any component of the space of
maps $B\Z/2\to B\haut(S^{0})$ is the space of maps $B\Z/2\to
\haut(S^{0})$.  The homotopy type of this latter space falls within
the purview of the Segal conjecture, and consists of the path
components of $QB\Z/2_{+}\times QS^{0}$ whose first component is a
generator of
\[
\pi_{0}QB\Z/2_{+} \approx \Z.
\]
For this reason, one knows a lot about the space of actions of $\Z/2$
on $\IZ$, and in particular that there are infinitely many
inequivalent actions inducing the sign representation on $\pi_{0}\IZ$.	
\end{remark}

For the spectrum $H\CC$ one has $B\haut(H\CC)\approx K(\Aut(\CC),1)$,
in which $\Aut(\CC)$ is the group of abelian group automorphisms of
$\CC$.  In this case there is no difference between $\Z/2$-actions on
$H\CC$ and $\Z/2$-actions on $\CC$, and complex conjugation is
uniquely specified.

\subsubsection{Duality}
\label{sec:mssduality-1}

Spectra with no negative homotopy groups are modeled by (higher)
Picard groupoids.  Picard groupoids come equipped with a $\Z/2$-action
sending each object to its inverse.  This corresponds to a natural
$\Z/2$-action on spectra which we now determine. 

Let $\sC$ be a Picard category and  consider the category of pairs
$(x,y)$ equipped with an isomorphism $x\otimes y\to 1$.   The functor
$(x,y)\mapsto x$ is an equivalence of categories, so the $\Z/2$-action
sending $x$ to its inverse corresponds to the action on the category
of pairs sending 
\[
x\otimes y\to 1
\]
to 
\[
y\otimes x \to x\otimes y \to 1.   
\]
If $\sC$ corresponds to a
spectrum $X$ then the category of pairs corresponds to $X\vee X\approx
X\times X$, and the category of pairs $(x,y)$ equipped with an
isomorphism $x\otimes y\to 1$ is the homotopy fiber of the map 
\[
X\vee X\to X.
\]
Writing this in terms of equivariant spectra we are looking at the
homotopy fiber of 
\[
\Z/2_{+}\wedge X\to X,
\]
which by~\eqref{eq:mss189} is $ X^\delta $.

Summarizing, we have the following.

\begin{proposition}
\label{thm:mss1}
The natural $\Z/2$-action corresponding to ``duality'' is given by the
map 
\[
B\Z/2\xrightarrow{\delta}{}BO\to B\haut(S^{0})
\]
and associates to a spectrum $X$, the $\Z/2$-equivariant spectrum 
\[
X^{\delta} = S^{\delta}\wedge X = S^{\sigma-1}\wedge X.
\]
\end{proposition}

\subsubsection{Complex conjugation}
\label{sec:msscompex-conjugation}

A complex conjugation on $\IZ$ corresponds to a map 
\[
\nu:B\Z/2\to B\haut(\IZ)
\]
having at least the property that its effect on $\pi_{1}$ is the sign
representation of $\Z/2$ on $\Z(1)$.    Write 
\[
\spaces\bigl(B\Z/2,B\haut(\IZ)\bigr)_{c}
\]
 for the space of maps inducing this homomorphism on $\pi_{1}$.  The space
$\spaces\bigl(B\Z/2,B\haut(\IZ)\bigr)_{c}$ is a union of infinitely many path
components of $\spaces\bigl(B\Z/2,B\haut(\IZ)\bigr)$ (see
Remark~\ref{rem:2a}).

Similarly, complex conjugation on $\ICx$ corresponds to a map
\[
\nu':B\Z/2\to B\haut(\ICx),
\]
whose effect on $\pi_{1}$ corresponds to the action of $\Z/2$ by
complex conjugation on $\CC^{\times}$.   Write
$\spaces\bigl(B\Z/2,B\haut(\ICx)\bigr)_{c}$ for this space of maps.

Since the maps
\begin{align*}
\Map(\IZ,H\CC) &\to \Hom(\Z(1),\CC)\\
\Map(H\CC, \ICx) &\to \Hom(\CC,\CC^{\times})
\end{align*}
are weak equivalence, so are the maps 
\begin{align*}
\Map(\IZ,H\CC)^{h\Z/2} &\to \Hom(\Z(1),\CC)^{\Z/2} \\
\Map(H\CC, \ICx)^{h\Z/2} &\to \Hom(\CC,\CC^{\times})^{\Z/2}
\end{align*}
for any $\Z/2$-actions on~$\IZ$ and~$\ICx$.
It follows that any action $\nu$ as above extends uniquely to a
$\Z/2$-equivariant map 
\[
\IZ^{\nu}\to H\CC
\]
and so induces a $\Z/2$-action $\nu'$ on the cofiber $\ICx$.   Similarly an
action 
$\nu'$ as above induces a $\Z/2$-action $\nu$ on $\IZ$.   In this
way we have an equivalence 
\begin{equation}
\label{eq:mss7}
\spaces\bigl(B\Z/2,B\haut(\IZ)\bigr)_{c}\approx
\spaces\bigl(B\Z/2,B\haut(\ICx)\bigr)_{c}.
\end{equation}

The space of {\em real structures} on $\IZ$ and $\ICx$ will be defined to
be a single path component of the above spaces.   Before specifying
which one, we turn to a motivating example. 

  \begin{example}[Hermitian structures and positivity]\label{thm:177}
 Let $\fVC$ be the topological groupoid of finite dimensional complex vector
spaces and (complex) linear isomorphisms, endowed with the symmetric monoidal
structure of $\otimes$.  For $V\in\fVC$, let $V^{\ast}$ be the dual vector
space.  We define a {\em covariant} ``duality'' functor $V\mapsto V^{\vee}$
by
 \begin{align*}
V^{\vee} &= V^{\ast}\\
f^{\vee} &= \big(f^{\ast}\big)^{-1}.
\end{align*}
The canonical isomorphism $V^{\vee\vee}\approx V$ extends the functor
$V^{\vee}$ to a $\Z/2$-action on $\fVC$.  (See Appendix~\ref{sec:11}.)
There is another $\Z/2$-action 
\[
V\mapsto \overline{V}
\]
gotten by redefining scalar multiplication by $x\in \CC$ to be scalar
multiplication by $\bar{x}$.

Let $\VectPos$ be the topological groupoid of finite dimensional complex
vector spaces equipped with a positive definite Hermitian inner product, and
unitary transformations.  Since the inclusion $U(n)\subset GL_{n}(C)$ is a
homotopy equivalence, the functor \[ \VectPos\to \fVC \] is a weak
equivalence of topological categories.  On $\VectPos$ the Hermitian inner
product gives a natural isomorphism $\overline{V}^{*}\approx V$,
trivializing the composition ``bar star'' of the two $\Z/2$-actions defined
above.  This suggests that whatever complex conjugation is, on the categories
in which $\CC$ is regarded as having a topology, the combined action (in the
sense of Remark~\ref{rem:3}) of complex conjugation and duality should be
trivializable.  The trivialization is non-canonical, however.  One might have
chosen negative definite vector spaces, or, for each prime $p$ made a choice
of positive or negative definite Hermitian inner products on vector spaces of
dimension $p$ and then extend to all finite dimensional vector spaces by
tensoring.
  \end{example}

With Example~\ref{thm:177} as motivation, and in view of
Proposition~\ref{thm:mss1}, we propose the following.

\begin{definition}
\label{def:2a}
The space of {\em real structures} on $\IZ$ is the path component of the
space 
  \begin{equation}\label{eq:284}
     \spaces\bigl(B\Z/2,B\haut(\IZ)\bigr)_{c} 
  \end{equation}
containing the map $1-\sigma$.   The space of {\em real structures} on
$\ICx$ is the path component of the space
$\spaces\bigl(B\Z/2,B\haut(\ICx)\bigr)_{c}$ corresponding to the space of real
structures on $\IZ$ under the equivalence~\eqref{eq:mss7}.
\end{definition}

As above, we write $\IZ^{\nu}$ for the $\Z/2$-spectrum corresponding to a
real structure $\nu:B\Z/2\to B\haut(\IZ)$.  Any real structure fits
canonically into a cofibration sequence
 \begin{equation}\label{eq:mssexp}
 \IZ^{\nu}\longrightarrow 
  H\CC^{\nu '}\xrightarrow{\;\exp\;}{} (\ICx)^{\nu'}  
 \end{equation}
in which $\nu$ and $\nu'$ correspond under the equivalence~\eqref{eq:mss7};
the superscript on~$H\CC$ is the unique complex conjugation, explained at the
end of \S\ref{subsubsec:act}.  

\begin{remark} \label{rem:1a}  
 Since the space of real structures $\nu$ on $\IZ$ is connected, but not
contractible, any $\IZ^{\nu}$ is non-canonically equivariantly equivalent to
$\IZ^{\gamma}= S^{1-\sigma}\wedge \IZ$.
\end{remark}

  \begin{ansatz}[]\label{thm:182}
 We use the basepoint in~\eqref{eq:284} to fix once and for all~$\nu =\gamma
=1-\sigma $.  Under the equivalence~\eqref{eq:mss7} this determines a real
structure~$\npn$ on~$\ICx$.  Our choices render the cofibration
sequence~\eqref{eq:mssexp} as
 \begin{equation}\label{eq:mssexp2}
 \IZn\longrightarrow 
  \HCn\xrightarrow{\;\exp\;}{} \ICxn
 \end{equation}
  \end{ansatz}

  \begin{remark}[]\label{thm:187}
 The real structure~$\gamma $ on~$\IZ$ is the restriction of a natural action
of~$\Z/2$; the corresponding real structure~$\npn$ is not.  However, in
terms of the polar decomposition $\Cx=\TT\times \Rp$ we have 
  \begin{equation}\label{eq:308}
     \ICxn\approx \IT\wedge S^{1-\sigma }\;\vee\; H\Rp. 
  \end{equation}
The spectrum~$\IT$ is characterized in the homotopy category of spectra by a
functorial isomorphism 
  \begin{equation}\label{eq:310}
     [\sB,\IT]\xrightarrow{\;\;\cong \;\;} \Hom(\pi _0\sB,\TT) 
  \end{equation}
for all spectra~$\sB$, analogous to~\eqref{eq:61}.  The equivariant spectrum
$\IT^\gamma =\IT\wedge S^{1-\sigma }$ fits into a cofibration sequence
analogous to~\eqref{eq:mssexp2}:
  \begin{equation}\label{eq:309}
     \IZn\longrightarrow\HRon\xrightarrow{\;\exp\;}{} \ITn
  \end{equation}
  \end{remark}

  \begin{remark}[]\label{thm:mss146}
 This definition of real structure fits with the three cases in which
 one has an algebraic interpretation of $\IZ$ (see Remark~\ref{thm:164}).
The zeroth space of $\Sigma 
\IZ$ is modeled by the unit complex numbers with the usual topology;
that of $\Sigma ^2\IZ$ corresponds to the symmetric monoidal groupoid of
$\Z/2$-graded complex lines; and $\Sigma^{3}\IZ$ to the Brauer-Wall symmetric
monoidal 2-groupoid of $\Z/2$-graded simple algebras over $\CC$,
$\Z/2$-graded bimodules and intertwiners.  These three models come
equipped with natural real structures, coming from change of scalars.   
By direct computation one can show
that the homotopy fixed points of $\Sigma^{i}\IZ^\gamma $ is modeled by the
corresponding real versions of the three categories described above.
To check this it suffices to do so when $i=3$ as the other cases are
gotten from it by passing to loop spaces.   The real Brauer-Wall
category corresponds to a spectrum $B$ with the following 
homotopy groups
\begin{align*}
\pi_{i}B&=0  &&i\not\in [0,3] \\
\pi_{0}B &=\Z/8 &&\text{(the eight real Clifford algebras)} \\
\pi_{1}B &= \Z/2  &&\text{(the even and odd real line)} \\
\pi_{2}B &= \{\pm1 \} &&
\end{align*}
and has the property that the multiplication by $\eta$ maps 
\[
\pi_{0}B\to \pi_{1}B\to \pi_{2}B
\]
are non-zero.   A straightforward computation shows that any spectrum
$X$ with these properties is homotopy equivalent to $B$.   To verify
the claim it therefore suffices to show that the $(-1)$-connected
cover of $\big(\Sigma^{3}\IZn\big)^{h\Z/2}$ has these properties.   We
therefore need to know the groups 
\[
\pi_{i}\big(\Sigma^{3}\IZn\big)^{h\Z/2} \qquad i\ge 0
\]
and the effect of multiplication by $\eta$.   Now for the
real structure $\gamma = 1-\sigma $ one has
\begin{align*}
\Map(S^{0}, \Sigma^{3}\IZn)^{h\Z/2}  & \approx 
\Map(S^{0}, S^{(1-\sigma)}\wedge \Sigma^{3}\IZ)^{h\Z/2}  \\
&\approx \Map(S^{(\sigma-1)}, \Sigma^{3} \IZ)^{h\Z/2}  \\
&\approx \Map(S^{(\sigma-1)}_{h\Z/2}, \Sigma^{3} \IZ)  \\
&\approx \Map(\thom(B\Z/2;\sigma-1), S^{3}\wedge \IZ),
\end{align*}
by~\eqref{eq:mss9} and~\eqref{eq:mss11}.   We therefore need information about 
\[
[\Thom(B\Z/2;\sigma-1), S^{i}\wedge \IZ] \quad 1\le i\le 3
\]
or, from the defining property of $\IZ$, the character groups of 
\[
\pi_{i}\Thom(B\Z/2;\sigma-1)\quad 0\le i\le 2.
\]
As described in~\S\ref{sec:13}, these groups coincide with the same homotopy
groups of $MT\!\Pm$ and are shown in Figure~\ref{fig:m1} (the case $s=1$) to be 
the groups $\Z/2$, $\Z/2$, and $\Z/8$ with
both $\eta$-multiplications non-zero.
\end{remark}

\subsubsection{Terminology}
\label{subsubsec:6.3.4}
 
It will be convenient in the sequel to have names for the objects assigned to
closed manifolds of arbitrary codimension in an invertible field theory.  In
codimension~0 we have a complex number and in codimension~1 an object in the
category of complex $\zt$-graded lines with the monoidal structure of graded
tensor product and the Koszul sign in the symmetry.  We refer to such an
object as a `complex super line' or a `$\zt$-graded line'.  Hence in
codimension~$k$ we introduce the term `complex super
$k$-line'.\footnote{Kapranov~\cite[\S3.4]{Kap} suggests a higher use of super
based on the sphere spectrum.}

  \begin{definition}[]\label{thm:178}
 \ \begin{enumerate}[{\textnormal(}i{\textnormal)}]

 \item $\IZ$ is the spectrum of {\em higher complex super lines};

 \item $\bigl(\IZn\bigr)^{h\Z/2}$ is the spectrum of {\em higher real super
lines}; 

 \item $\IZ_{H} := (\IZn\wedge S^{\sigma-1})^{h\Z/2}$ is the spectrum of {\em
higher Hermitian super lines}; 
 
 \item $\ICx$ is the spectrum of {\em higher flat complex super lines};

 \item The $k^{\textnormal{th}}$~space in the spectrum~$\IZ$ is the space of
\emph{complex super~$k$-lines}.

 \end{enumerate}
  \end{definition}

\noindent
 Example~\ref{thm:177} is the motivation for~(iii).  There are analogs
of~(iv) and~(v) for real and Hermitian super lines.  For example, the fixed
point spectrum
  \begin{equation}\label{eq:288}
     \ICx_H:=(\ICxn\wedge S^{\sigma-1})^{h\Z/2} 
  \end{equation}
is the spectrum of \emph{higher flat Hermitian super lines}, and the
$k^{\textnormal{th}}$~space of that spectrum is the space of \emph{flat
Hermitian super $k$-lines}.  As for the fixed point spectrum in~(iii), since
$S^{1-\sigma }\wedge S^{\sigma -1}$ is the sphere spectrum with the trivial
$\Z/2$-action---the ``bar star'' involution---we deduce from~\eqref{eq:126} a
canonical identification
  \begin{equation}\label{eq:285}
     \IZ_{H} = \Map\bigl(B\Z/2_+,\IZ\bigr). 
  \end{equation}
Pulling back along $B\Z/2\to\pt$ we obtain a map 
  \begin{equation}\label{eq:286}
     I\Z(1)\longrightarrow I\Z(1)_{H}; 
  \end{equation}
the image is a summand, split by a choice of point in $B\Z/2$. 

  \begin{definition}[]\label{thm:179}
 The image~$\IZ_{\textnormal{pos}}$ of~\eqref{eq:286} is the \emph{spectrum
of higher positive definite Hermitian super lines}.
  \end{definition}

\noindent
 The $k^{\textnormal{th}}$~space in~$\IZ_{\textnormal{pos}}$ is the space of
\emph{positive definite Hermitian super $k$-lines}.  Define the spectrum of
\emph{higher \emph{flat} positive definite Hermitian super lines} as the
homotopy pullback
  \begin{equation}\label{eq:287}
  \begin{gathered}
     \xymatrix{ \ICx_{\text{pos}} \ar[r]\ar[d] & \Sigma
     \IZ_{\textnormal{pos}} \ar[d] \\ \ICx_{H} \ar[r] & \Sigma
     \IZ_{H}.  }  
  \end{gathered}
  \end{equation}

We examine this homotopy-theoretic definition of positivity by focusing on
the top piece, first in the ungraded case and then in the $\zt$-graded case.

  \begin{example}[Hermitian lines]\label{thm:i107} 
 Consider the spectrum $\Sigma ^2H\ZZ$.  Its zero-space represents the
ordinary groupoid of complex lines; morphisms have the continuous topology.
There is a contractible space of trivializable involutions, and we imagine a
point in it to represent bar star.  The analog of~\eqref{eq:285} implies that
the set of components of the fixed point spectrum of any such involution is
  \begin{equation}\label{eq:197}
     \pi_{0}\Map(\Si B\Z/2_{+},\Sigma ^2H\ZZ)= \pi_{0}\Sigma ^2H\ZZ\oplus
     \pi_{0}\Map(\Si B\Z/2,\Sigma ^2H\ZZ) = \{0\}\oplus \Z/2. 
  \end{equation}
The zero space of $\Map(\Si B\Z/2_{+},\Sigma ^2H\ZZ)$ represents the groupoid
of Hermitian lines, the~$\zt$ tracks the sign of the Hermitian form.  The
positive subspace, obtained by pulling back along $B\Z/2\to\pt$, picks out
the positive definite forms.
  \end{example}

  \begin{example}[super Hermitian lines]\label{thm:159}
  The zero-space of the spectrum~$\Sigma ^2\IZ$ represents the
groupoid of super lines~$L$ with continuous topology on morphisms.  We
compute the set of components of the fixed point spectrum of a trivializable
involution:
  \begin{equation}\label{eq:239}
     \pi_{0}\Map(\Si B\Z/2_{+},\Sigma ^2\IZ)= \pi_{0}\Sigma ^2\IZ\oplus
     \pi_{0}\Map(\Si B\Z/2,\Sigma ^2\IZ) = \Z/2\oplus \Z/2. 
  \end{equation}
This is the group of isomorphism classes of super Hermitian lines.  The
first~$\zt$ is the grading of the line, the second the ``sign'' of the form.
But the sesquilinearity condition 
  \begin{equation}\label{eq:330}
     \langle\bar\ell _1,\ell _2\rangle = (-1)^{|\ell _1||\ell
     _2|}\overline{\vphantom{M^M}\langle \bar\ell _2,\ell _1 \rangle},\qquad
     \ell _1,\ell _2\in L, 
  \end{equation}
implies that if $L$~is odd then $\langle \bar\ell ,\ell \rangle\in \sqmo\RR$
for all~$\ell \in L$.  (The form is a bilinear map $\overline{L}\times
L\to\CC$.)  The notion of positivity in this case chooses a ray
in~$\sqmo\RR$; there is no canonical choice.  In the literature,
e.g.~\cite[(4.4.2)]{DM}, an arbitrary choice is made.  In our homotopy
theoretic presentation, this choice lies in the identification of the space
of super Hermitian lines with the 0-space of~$\Sigma ^2\IZ$.  As we descend
deeper into extended field theories, there are further choices to be made;
see Remark~\ref{rem:2a}.
  \end{example}

   \section{Reflection structures and stability}\label{sec:6}

We begin in~\S\ref{subsec:6.1} by reviewing Madsen-Tillmann spectra;
see~\cite[\S3]{GMTW}.  They give a filtration~\eqref{eq:90} of Thom spectra,
which leads to an analysis of the obstructions to extending invertible field
theories to stable theories.  In~\S\ref{subsec:6.2} we develop the relation
between naive positivity and stability in two situations: non-equivariant
discrete theories and equivariant continuous theories.  In each case the only
obstruction in $n$~spacetime dimensions arises from the partition function of
the $n$-sphere.  But its positivity does not guarantee positive definite
metrics on the state spaces attached to arbitrary $(n-1)$-manifolds
(Proposition~\ref{thm:89}), consideration of which is deferred
until~\S\ref{sec:7}.  We conclude in~\S\ref{subsec:6.3} by analyzing the
obstruction to extending ``H-type'' theories to ``L-type'' theories.

  \subsection{Madsen-Tillmann and Thom spectra}\label{subsec:6.1}

The homomorphism $\rho _n\:H_n\to O_n$ in~\eqref{eq:14}, which defines the
symmetry type of a theory, produces a rank~$n$ vector bundle $V_n\to BH_n$
over the classifying space.  We refer to~ \S\ref{sec:mssthom-spectra} for the
general theory of Thom spectra.

  \begin{definition}[]\label{thm:71}
 The \emph{Madsen-Tillmann spectrum}~$MTH_n$ is the Thom spectrum of $-V_n\to
BH_n$. 
  \end{definition}

\noindent
 More natural for us is a suspension, the connective spectrum
  \begin{equation}\label{eq:88}
     \Sigma ^nMTH_n = \Thom(BH_n;\triv{n}-V_n) .
  \end{equation}
The general construction of Thom spectra is described
in~\S\ref{sec:mssthom-spectra}.  Here is a geometric description.  Let
$Gr_n(\RR^{n+q})$~denote the Grassmannian of $n$-dimensional subspaces
of~$\RR^{n+q}$.  It approximates~$BO_n$, and the pullback
  \begin{equation}\label{eq:89}
     \begin{gathered} \xymatrix{X_{n,n+q}\ar@{-->}[r]^{} \ar@{-->}[d]_{} &
     BH_n\ar[d]^{} \\ Gr_n(\RR^{n+q})\ar[r]^{} & BO_n} \end{gathered} 
  \end{equation}
is a finite dimensional approximation to~$BH_n$.  The
$q^{\textnormal{th}}$~space of the spectrum~\eqref{eq:88} can be taken to be
the Thom \emph{space}~$\Thom(X_{n,n+q};\,Q_q)$ of the vector bundle $Q_q\to
X_{n,n+q}$, which is the pullback of the rank~$q$ ``quotient bundle'' over
the Grassmannian: the fiber at a subspace~$W\subset \RR^{n+q}$ is~$W^\perp$.

  \begin{remark}[]\label{thm:75}
 The Pontrjagin-Thom construction provides the basic relationship to
$H_n$-manifolds.  If a map $S^{k+q}\to \Thom(X_{n,n+q};\,Q_q)$ is transverse to
the 0-section of $Q_q\to X_{n,n+q}$, then the inverse image of the 0-section
is a $k$-manifold $M\subset S^{k+q}$ whose stable tangent bundle is equipped
with an isomorphism to the pullback of the ``tautological
bundle''\footnote{The fiber of the tautological bundle at a point $W\subset
\RR^{n+q}$ in~$Gr_n(\RR^{n+q})$ is~$W$.} $V_n\to X_{n,n+q}$, which is
equipped with an $H_n$-structure.  Theorem~\ref{thm:54} implies that the
abelian group~$\pi _k\Sigma ^nMTH_n$ is generated by closed $k$-dimensional
$H_n$-manifolds under disjoint union.  The class of a closed manifold~$M^k$
is zero if and only if $M=\partial W$ where $W$~is a compact $(k+1)$-manifold
whose stable tangent bundle is isomorphic to a rank~ $n$ bundle with an
$H_n$-structure extending that of~$M$.  This bordism group was introduced by
Reinhart~\cite{R}; see also~\cite[Appendix]{E}.
  \end{remark}

  \begin{remark}[]\label{thm:72}
 Not every element of the homotopy group is represented by a manifold; group
completion of the semigroup of manifold classes is needed to obtain the
homotopy group.  For example, $\pi _0MTO_0\cong \ZZ$ but since a
0-dimensional manifold has a unique $O_0$-structure such manifolds only
realize the submonoid of nonnegative integers.  We also remark that the
sphere~$S^{2m}$ represents a nonzero element in $\pi _{2m}\Sigma
^{2m}MTSO_{2m}$, but is zero in the next group $\pi _{2m+1}\Sigma
^{2m+1}MTSO_{2m+1}$: the closed ball~$D^{2m+1}$ has nonzero Euler
characteristic so no $SO_{2m}$-structure.  As another illustration, the
2-sphere and the genus~2 surface represent opposite elements of $\pi _2\Sigma
^2MTSO_2$: a genus~2 handlebody with a 3-ball excised admits an
$SO_2$-structure.
  \end{remark}

The Stabilization Theorem~\ref{thm:6} provides a sequence of
spectra\footnote{That theorem supplies a \emph{stable tangential
structure}~$BH$ from which $BH_n$~is constructed by pullback;
recall~\eqref{eq:136}.}
  \begin{equation}\label{eq:90}
     \Sigma ^nMTH_n \longrightarrow \Sigma ^{n+1}MTH_{n+1} \longrightarrow
     \Sigma ^{n+2}MTH_{n+2} \longrightarrow \cdots 
  \end{equation}
whose colimit, denoted~$MTH$, is the Thom spectrum of the stable vector bundle 
  \begin{equation}\label{eq:277}
      -V\longrightarrow  BH 
  \end{equation}
which is the negative of the classifying map of~\eqref{eq:278}; see the
construction in~\S\ref{sec:mssthom-spectra}, especially the
presentation~\eqref{eq:thomlim} which is equivalent to~\eqref{eq:90}.  From the
geometric description in Remark~\ref{thm:75} the homotopy groups~$\pi
_k\Sigma ^nMTH_n$ stabilize once~$n>k$; then $\pi _kMTH$~is the bordism group
of $k$-dimensional manifolds with a \emph{stable tangential} $H$-structure.
We identify~$MTH$ with the Thom spectrum~$MH^{\perp}$ of the
perpendicular\footnote{\label{perp} The classifying space~$BH^{\perp}$ is the
pullback
  \begin{equation}\label{eq:233}
     \begin{gathered} \xymatrix{BH^\perp\ar[r]^{} \ar[d]_{} & BH\ar[d]^{} \\
     BO\ar[r]^{} & BO} \end{gathered} 
  \end{equation}
in which the bottom map classifies the negative of the universal bundle (of
rank zero).  There is a sequence of inclusions $\cdots
H_n^\perp\hookrightarrow H_{n+1}^\perp\hookrightarrow H_{n+2}^\perp\cdots$ of
compact Lie groups such that $BH^\perp$~is the colimit of~$BH_n^\perp$.
Namely, define $\tH_n^\perp$ as the pullback (see~\eqref{eq:226}) 
  \begin{equation}\label{eq:234}
     \begin{gathered} \xymatrix{
     1\ar[r]&K\ar[r]\ar@{=}[d]&\tH_n^\perp\ar@{->>}[d]\ar[r]& \Pm_n\ar@{->>}[d]
     \ar[r] & 1\\ 1\ar[r]&K\ar[r]&J\ar[r]& \pmo\ar[r] & 1} \end{gathered} 
  \end{equation}
and then set 
  \begin{equation}\label{eq:235}
     H_n^\perp\cong \tH_n^\perp \bigm / \langle(-1,k_0)\rangle. 
  \end{equation}
One checks that $BH_n^\perp$ is the pullback 
  \begin{equation}\label{eq:236}
     \begin{gathered} \xymatrix{BH_n^\perp\ar[r]^{} \ar[d]_{} & BH\ar[d]^{} \\
     BO_n\ar[r]^{} & BO} \end{gathered} 
  \end{equation}} 
 stable normal structure.  In many cases $H^\perp=H$; however, for example,
$(\Ppm)^{\perp}=\Pmp$.
 
Following Ansatz~\ref{thm:145} an invertible topological field theory is a
map with domain $\Sigma ^nMTH_n$.  To investigate extensions along the
sequence~\eqref{eq:90} we will use the following in~\S\ref{subsec:6.2}.

  \begin{proposition}[]\label{thm:73}
 The fiber of the map $\Sigma ^nMTH_n \longrightarrow \Sigma ^{n+1}MTH_{n+1}$
is $\Sigma ^n(BH_{n+1})_+$.  The map $\Sigma ^n(BH_{n+1})_+\to \Sigma
^nMTH_n$ is represented by the universal family $BH_n\to BH_{n+1}$ of
$H_n$-spheres.
  \end{proposition}

\noindent
 See~\cite[\S3.1]{GMTW}, \cite[Lemma 3.1]{FHT1} for a proof.  The universal
family of spheres was mentioned in Remark~\ref{thm:39}.  We remind that
spectra are built out of \emph{based} spaces; for a based space~$X$ the
spectrum ~$\Sigma ^nX_+$ is the one-point union of~$S^n$ and the suspension
spectrum~$\Sigma ^nX$, and the latter is $(n-1)$-connected if $X$~is connected.

Our final task in this section is to refine Ansatz~\ref{thm:145} and
Ansatz~\ref{thm:156}, which formulate invertible field theories as maps of
spectra, to include reflection structures.  Recall from ~\S\ref{sec:3} that
the reflection structure on the bordism category maps a manifold with
$H_n$-structure to the same manifold with the opposite $H_n$-structure, which
is defined using the group extension~\eqref{eq:30}.  Turning to bordism
spectra we observe that this group extension induces a $\Z/2$-action
on~$BH_{n}$ and makes the vector bundle $V_{n}\to BH_n$ into an equivariant
vector bundle $V_{n}^{\beta}\to BH_n^\beta $.  Applying the discussion
in~\S\ref{sec:mssequiv-thom-spectra} we refine the Thom
spectrum~\eqref{eq:88} to a $\Z/2$-equivariant spectrum we denote ~$\Sigma
^nMTH_n^\beta $.  There is an equivariant lift of~\eqref{eq:90}.  Recall the
involutions on~$\IZ,\ICx$ chosen after Remark~\ref {rem:1a}.

  \begin{ansatz}[]\label{thm:147}
 \ \begin{enumerate}[{\textnormal(}i{\textnormal)}]
 \item A \emph{\emph{discrete} invertible $n$-dimensional extended
topological field theory with symmetry group~$H_n$ and reflection structure}
is an equivariant map
  \begin{equation}\label{eq:191}
     F\:\Sigma ^nMTH_n^\beta \longrightarrow \Sigma ^n\ICxn,
  \end{equation}

 \item A \emph{\emph{continuous} invertible $n$-dimensional extended topological
field theory with symmetry group~$H_n$ and reflection structure} is an
equivariant map
  \begin{equation}\label{eq:aa238}
     \varphi \:\Sigma ^nMTH_n^\beta \longrightarrow \Snizn. 
  \end{equation}
The space of theories of this type is 
  \begin{equation}\label{eq:289}
     \fieldsRefl{H_n}{n} =\Map^{\Z/2}(\Sigma ^nMTH_n^\beta ,\Snizn). 
  \end{equation}
 \end{enumerate}

  \end{ansatz}

  \subsection{Naive positivity and stability}\label{subsec:6.2}

We first prove that the double of an $H_n$-manifold is null bordant through
an $H_{n+1}$-manifold.  Recall the evaluation bordism~\eqref{eq:43}, the
identification of duals and bars in Proposition~\ref{thm:44}, and
Definition~\ref{thm:28} of a double.

  \begin{proposition}[]\label{thm:76}
 Let $Y_0,Y_1$~be closed $(n-1)$-dimensional $H_n$-manifolds and $X\:Y_0\to
Y_1$ an $H_n$-bordism.  Then  
  \begin{equation}\label{eq:91}
     \beta X\amalg e\mstrut _{Y_1}\amalg X\:\beta Y_0\amalg
     Y_0\longrightarrow \emptyset ^{n-1} 
  \end{equation}
is $H_{n+1}$-bordant to $e\mstrut _{Y_0}$.
  \end{proposition}

  \begin{proof}
 The bordism\footnote{It is a bordism of manifolds with boundary, or better a
higher morphism in a multi-bordism category.  We only use $Y_0=\emptyset
^{n-1}$, as in Corollary~\ref{thm:77}, in which case $[0,1]\times X$ is a
null bordism of a closed manifold.} is~$[0,1]\times X$.
  \end{proof}

  \begin{corollary}[]\label{thm:77}
 The double~$\Delta X$ of a compact $H_n$-manifold with boundary is null
bordant through an $H_{n+1}$-manifold.  
  \end{corollary}

\noindent
 By Corollary~\ref{thm:37} this applies to~$S^n$ with its canonical
$H_n$-structure, and so every double is $H_{n+1}$-bordant to~$S^n$. 

  \begin{proof}
 Apply Proposition~\ref{thm:76} to $X\:\emptyset ^{n-1}\to \bX$ (and smooth
the corners of $[0,1]\times X$). 
  \end{proof}

  \begin{remark}[]\label{thm:80}
 If $X$~is the 2-dimensional disk, viewed as a bordism from the empty
1-manifold to the circle, then $\Delta X$~is the 2-dimensional sphere~$S^2$
and the null bordism~$[0,1]\times X$ is the 3-dimensional ball~$D^3$.  The
Euler characteristic obstructs the existence of an $H_2$-structure on~$D^3$
which restricts to the given $H_2$-structure on~$S^2$ (for any stable
tangential structure~$H$).
  \end{remark}

The sequence of bordism spectra~\eqref{eq:90} results in a special type of
invertible field theory.  The following applies to both discrete
(Ansatz~\ref{thm:145}) and continuous (Ansatz~\ref{thm:156}) invertible field
theories, possibly with reflection structure (Ansatz~\ref{thm:147}).

  \begin{definition}[]\label{thm:79}
 An $n$-dimensional invertible topological field theory with domain~$\Sigma
^nMTH_n$ is \emph{stable} if it is the restriction of a theory defined
on~$MTH$.  
  \end{definition}

Stability can be investigated one step at a time in the
sequence~\eqref{eq:90} using obstruction theory.  We first carry this out for
\emph{discrete} invertible topological field theories \emph{without}
reflection structure.  Recall that the sphere has a canonical $H_n$-structure
given by the principal bundle $H_{n+1}\to H_{n+1}/H_n$.

  \begin{theorem}[]\label{thm:78}
 A discrete invertible theory $F\:\Sigma ^nMTH_n\to\Sigma ^n\ICx$ is stable
if and only if $F(S^n)=1$.  The subspace of $\Map(\Sigma ^nMTH_n,\SCx)$
consisting of theories~$F$ with $F(S^n)=1$ is homotopy equivalent to the
mapping space~$\Map(MTH,\SCx)$.
  \end{theorem}

\noindent 
 By Corollary~\ref{thm:77} the condition is equivalent to $F(\Delta X)=1$ for
all compact~$X^n$ with boundary.

  \begin{proof}
 If $F$~is the restriction of $\tF\:MTH\to \Sigma ^n\ICx$, then
$F(S^n)=\tF(S^n)=1$ since $S^n$~is null bordant as an $H_{n+1}$-manifold.
Conversely, by Proposition~\ref{thm:73} the map~$F$ extends over~$\Sigma
^{n+1}MTH_{n+1}$ if and only if it evaluates trivially on the universal
family of $H_n$-spheres.  But that evaluation is the constant function
$BH_{n+1}\to\Cx$ with value~$F(S^n)$.  There is no further obstruction in the
sequence~\eqref{eq:90}, because the subsequent fibers have vanishing homotopy
groups in degrees~$\le n$ and $\pi _{q}\Sigma ^n\ICx=0$ for~$q>n$.

To analyze the space of discrete stable theories we note that the cofibration
sequence
  \begin{equation}\label{eq:209}
     \Sigma ^nMTH_n\longrightarrow \Sigma ^{n+1}MTH_{n+1}\longrightarrow
     \Sigma ^{n+1}(BH_{n+1})_+ 
  \end{equation}
of spectra induces a fibration sequence  
  \begin{multline}\label{eq:210}
     \Map\bigl(\Sigma ^{n+1}(BH_{n+1})_+,\SCx\bigr) \longrightarrow
     \Map(\Sigma ^{n+1}MTH_{n+1},\SCx)\\[3pt]\longrightarrow \Map(\Sigma
     ^{n}MTH_{n},\SCx) \longrightarrow \Map\bigl(\Sigma
     ^{n}(BH_{n+1})_+,\SCx\bigr) 
  \end{multline}
of mapping spaces.  The first space is contractible, since $\Sigma
^{n+1}(BH_{n+1})_+$~is $n$-connected.  The fiber of the last map is the
subspace indicated in the theorem, by the obstruction argument in the
previous paragraph.  To pass to stable maps make a similar argument with the
cofibration sequence
  \begin{equation}\label{eq:214}
     \Sigma ^{n+1}MTH_{n+1}\longrightarrow MTH\longrightarrow C 
  \end{equation}
and the induced fibration on mapping spaces.
 \end{proof}

  \begin{remark}[]\label{thm:81}
 If $X^n$ is a closed $H_n$-manifold, then $[0,1]\times X$ is a null bordism
of $\beta X\amalg X$.  Thus if $F$~is stable and has a reflection structure,
then $\|F(X)\|^2=1$.
  \end{remark}

Next, we turn to \emph{continuous} invertible field theories \emph{with}
reflection structure, which according to Ansatz~\ref{thm:147}(ii) are
$\zt$-equivariant maps
  \begin{equation}\label{eq:238}
     \varphi \:\Sigma ^nMTH_n^\beta \longrightarrow \Snizn . 
  \end{equation}
We investigate stability for these equivariant theories.

  \begin{remark}[]\label{thm:163}
 As explained after~\eqref{eq:232} a continuous invertible field theory
assigns a $\Zo$-torsor to a closed $H_n$-manifold, hence an equivariant
theory~\eqref{eq:238} assigns to a $\beta $-equivariant family $\sX\to S$ of
closed $H_n$-manifolds an equivariant $\Zo$-torsor over~$S$, where the action
on $\Zo$-torsors is that in Example~\ref{thm:83}; see also
Remark~\ref{thm:mss146}.  The universal model is the map $\exp\:\CC\to\Cx$,
equivariant for complex conjugation, with fibers $\Zo$-torsors.  Over the
fixed point set $\RR^\times =\RR^{>0}\amalg \RR^{<0}$ the fibers are
$\Zo$-torsors of Type~P and Type~N; see Example~\ref{thm:83}.  As discussed
in~\S\ref{subsec:4.4} a non-topological invertible field theory (type~(a) in
that discussion) has a homotopy class that is a continuous theory.  If we
have a reflection structure, then the partition function of a $\beta $-fixed
$H_n$-manifold is real, and if it is positive then the corresponding
$\Zo$-torsor has Type~P.  
  \end{remark}

  \begin{remark}[]\label{thm:176}
 A \emph{stable} continuous theory~$\tphi $ assigns an integer (better:
element of~$\Zo$) to a closed $(n+1)$-manifold.  The universal
property~\eqref{eq:62} of maps into the Anderson dual implies that the
topological field theory associated to~ $\tphi $ is determined by its
truncation to $n$-~and $(n+1)$-manifolds.
  \end{remark}

  \begin{theorem}[]\label{thm:82}
 An equivariant continuous invertible field theory $\varphi \:\Sigma
^nMTH_n^{\beta}\to\Sigma ^{n+1}\IZn$ is stable if and only if $\varphi
(S^n)$~has Type~P.  The subspace of $\Map^{\Z/2}(\Sigma ^nMTH_n^\beta
,\Snizn )$ consisting of equivariant continuous invertible field theories
with Type~P partition function on~$S^n$ is homotopy equivalent to the mapping
space~$\Map^{\Z/2}(MTH^\beta ,\Snizn )$.
  \end{theorem}

  \begin{proof}
 Since $S^n$~is diffeomorphic to~$\beta S^n$, the partition function $\varphi
(S^n)$ is a $\Zo$-torsor with involution.  The partition function of the
universal family of $n$-spheres is then a $\Zo$-torsor over~ $BH_{n+1}$ with
involution covering the trivial involution on the base.  It is classified by
a map $BH_{n+1}\to\RR^\times $ whose homotopy class
in~$H^0(BH_{n+1};\pmo)\cong \pmo$ encodes the Type (P or~N) of~$\varphi
(S^n)$.
 
Now use the stabilization sequence~\eqref{eq:90} as before.  If $\varphi $~is
stable, then it is trivial on the fiber~$\Sigma ^n(BH_{n+1})_+$ of the first
map, which is represented by the universal family of $n$-spheres.  The
argument in the preceding paragraph shows that $\varphi (S^n)$~has Type~P.
To prove the converse, if $\varphi (S^n)$~has Type~P then the first
obstruction vanishes, and so $\varphi $~is the restriction of a map $\Sigma
^{n+1}MTH_{n+1}^{\beta}\to \Sigma ^{n+1}\IZn$.  The obstruction at
the next stage is a map $\Sigma ^{n+1}(BH^{\beta}_{n+2})_+\to\Sigma
^{n+1}\IZn$.  But $\Sigma ^{n+1}(BH^{\beta}_{n+2})_+\simeq
S^{n+1}\vee \Sigma ^{n+1}BH^{\beta}_{n+2}$ with $\Z/2$~acting trivially on
the suspension~$S^{n+1}$ of the basepoint.  Since $\Sigma
^{n+1}BH^{\beta}_{n+2}$ is $(n+1)$-connected, the obstruction lies in
  \begin{equation}\label{eq:280}
  \begin{aligned}
  [S^{n+1},\Sigma ^{n+1}\IZn]^{\Z/2} &\cong 
  [S^{\sigma-1},
\IZ]^{\Z/2} \\ &\cong [E\Z/2_{+}\underset{\Z/2}{\wedge}S^{\sigma-1}, \IZ] \\
&\cong \Hom(\pi_{0}\,E\Z/2_{+}\underset{\Z/2}{\wedge}S^{\sigma-1},\ZZ(1)) =0,
  \end{aligned}
  \end{equation}
since 
\[
\pi_{0}\,E\Z/2_{+}\underset{\Z/2}{\wedge}S^{\sigma-1} =
\pi_{1}\RP^{\infty} = \Z/2.
\]  
 There are no further obstructions to extending to~$MTH$, because the fibers
have nonvanishing homotopy groups only in degrees greater than~$n+1$ and $\pi
_{q}\Sigma ^{n+1}\IZ=0$ for~$q>n+1$.

The equivariant version of~\eqref{eq:209} with the $\beta $-involution leads
to the fibration sequence 
  \begin{multline}\label{eq:211}
     \Map^{\Z/2}\bigl(\Sigma ^{n+1}(BH_{n+1}^\beta )_+,\Snizn \bigr)
     \longrightarrow 
     \Map^{\Z/2}(\Sigma ^{n+1}MTH_{n+1}^\beta ,\Snizn )\\[3pt]\longrightarrow
     \Map^{\Z/2}(\Sigma 
     ^{n}MTH_{n}^\beta ,\Snizn ) \longrightarrow \Map^{\Z/2}\bigl(\Sigma
     ^{n}(BH_{n+1}^\beta )_+,\Snizn \bigr) 
  \end{multline}
As in~\eqref{eq:210} the first space is contractible.  The obstruction
argument above identifies the fiber of the last map as equivariant continuous
theories with positive sphere partition function.  To pass to stable maps use
an equivariant version of~\eqref{eq:214}.
\end{proof}

  \begin{corollary}[]\label{thm:86}
 There is a 1:1 correspondence 
  \begin{equation}\label{eq:93}
     \left\{\normalfont \vcenter{\hbox{isomorphism classes of continuous
     invertible} 
     \hbox{$n$-dimensional extended topological
     field}\hbox{theories with \textnormal{(}i\textnormal{)}~symmetry
     group~$H_n$, }\hbox{\textnormal{(}ii\textnormal{)} reflection structure,
     and \textnormal{(}iii\textnormal{)} partition}\hbox{function on
     $S^n$ of Type~P}} \right\} \cong  
     [MTH^{\beta},\SIZn]^{\Z/2}.   
  \end{equation}
  \end{corollary}

  \begin{example}[]\label{thm:87}
 The restriction map\footnote{The involution
on $\pi _4MTSO$ and $\pi _4\Sigma ^3MTSO_3$ acts as~$-1$: both groups are
detected by the signature, which negates under orientation-reversal.}
  \begin{equation}\label{eq:95}
     [MTSO^\beta ,\Sigma ^4\IZn ]^{\Z/2}\longrightarrow [\Sigma
     ^3MTSO_3^\beta ,\Sigma ^4\IZn ]^{\Z/2}
  \end{equation}
is an index two inclusion of infinite cyclic groups.  It follows that there
exist continuous invertible 3-dimensional oriented theories~$\varphi $ with
reflection structure such that $\varphi (S^3)$~has Type~N.  In turn, this
suggests the existence of invertible non-topological theories with reflection
structure whose real-valued partition function on ~$S^3$ is negative;
see~\ref{subsec:4.4}.  Here is an explicit example.  The domain is the
geometric bordism category of oriented \emph{Riemannian} manifolds.  The
partition function is $F(X^3)=\exp(2\pi i\xi \mstrut _{X})$, where $\xi
\mstrut _X$~is the Atiyah-Patodi-Singer invariant~\cite{APS}.\footnote{of the
operator called `$B^{\textnormal{ev}}$' in their paper.}  To apply the
arguments in Theorem~\ref{thm:82} we need to use a Riemannian sphere that is
a double---the round sphere does nicely---in which case the spectrum of the
APS~operator is symmetric about zero and so the $\eta $-invariant vanishes.
The dimension of the kernel is one, $\xi \mstrut _X=1/2$, and so $F(S^3)=-1$.
We remark that the corresponding integer invariant of a closed oriented
4-manifold~$W$ is $(\Sign(W)\pm\Euler(W))/2$; either sign works.  Also, the
square of this theory, whose deformation class generates $[MTSO^\beta ,\Sigma
^4\IZn ]^{\Z/2} $, represents ``Kitaev's $E_8$-phase''~\cite{K5}.
  \end{example}

Let $F$~be a invertible topological $n$-dimensional theory, and suppose that
$F(S^n)>0$.  Then the hermitian form on~$F(S^{n-1})$ is positive definite;
see~\eqref{eq:49}.  The positivity holds for any null bordant
$(n-1)$-manifold, but on other manifolds there is no guarantee of positivity
(Definition~\ref{thm:24}), even for stable theories.

  \begin{proposition}[]\label{thm:89}
 Let $F$~be an invertible $n$-dimensional topological field theory of
$H_n$-manifolds with $F(S^n)>0$.  Suppose $F$~has a reflection structure.
Then the sign of the hermitian form~\eqref{eq:45} on a closed
$(n-1)$-manifold is a bordism invariant and determines a homomorphism
  \begin{equation}\label{eq:94}
     \pi _{n-1}\Sigma ^{n-1}MTH_{n-1}\longrightarrow \pmo. 
  \end{equation}
  \end{proposition}

  \begin{proof} 
 If $X\:Y_0\to Y_1$ is an $H_n$-bordism, then by reversing the arrow of time
on the incoming boundary we obtain $X'\:\emptyset ^{n-1}\to \beta Y_0\amalg
Y_1$.  Hence by Corollary~\ref{thm:77} and the remark which follows, we
deduce that the hermitian line $\overline{F(Y_0)}\otimes F(Y_1)$ is positive
definite.  Therefore, $F(Y_0)$~and $F(Y_1)$~are simultaneously positive or
simultaneously negative.
  \end{proof}

We conclude this section with a lemma we will use in~\S\ref{sec:7}. 

  \begin{lemma}[]\label{thm:152}
 The map $\Sigma ^nMTH_n\to MTH$ induces a surjection on $H_{n+1}(-;\RR)$. 
  \end{lemma}

\noindent
 We remark that $\pi _{n+1}(\sB)\otimes \RR\to H_{n+1}(\sB;\RR)$ is an
isomorphism for any spectrum~$\sB$. 

  \begin{proof}
 Arrange the stabilization~\eqref{eq:90} and cofibration
sequences~\eqref{eq:209} as follows:
  \begin{equation}\label{eq:220}
     \begin{gathered} \xymatrix{\Sigma ^n(BH_{n+1})_+\ar[d] & \Sigma
     ^{n+1}(BH_{n+2})_+\ar[d]^s\\ \Sigma ^nMTH_n\ar[r]^<<<<<<i\ar[d] & \Sigma
     ^{n+1}MTH_{n+1}\ar[r]^j\ar[d]^\chi  & \Sigma ^{n+2}MTH_{n+2}\ar[d]\\ \Sigma
     ^n(BH_n)_+ & \Sigma ^{n+1}(BH_{n+1})_+&\Sigma
     ^{n+2}(BH_{n+2})_+} \end{gathered}  
  \end{equation}
 The two compositions with shape
$\raisebox{18pt}{\xymatrix@R=3pt@C=3pt{\cdot\ar[d]\\\cdot\ar[r]&\cdot\ar[d]\\
&\cdot}}$ are cofibration sequences.  The map~ $s_*$ on~$\pi _{n+1}$ sends
the generator of the infinite cyclic group $\pi _{n+1}\Sigma
^{n+1}(BH_{n+2})_+$ to the class of $S^{n+1}$, and the map~ $\chi _*$ on~$\pi
_{n+1}$ sends the class of a closed $(n+1)$-manifold to its Euler number.
Also, $\pi _{n+1}\Sigma ^{n+2}(BH_{n+2})_+=0$.  It follows that $j_*$ on~$\pi
_{n+1}$ is surjective.  If $n$~is even, then $\chi _*=0$ on~$\pi _{n+1}$ and
by exactness $i_*$~is surjective.  If $n$~is odd, then $\chi _*\circ s_*$ is
multiplication by~2.  Working now on~$\pi _{n+1}\otimes \RR$ we can lift any
class in~$\pi _{n+1}\Sigma ^{n+2}MTH_{n+2}\otimes \RR$ through~$j_*$ to have
zero image under~$\chi _*$, hence by exactness to be in the image
of~$i_*\otimes \RR$.  In other words, $(j\circ i)_*\otimes \RR$~is
surjective.  Finally, the stabilization map $\Sigma ^{n+2}MTH_{n+2}\to MTH$
induces an isomorphism on~$\pi _{n+1}$.
  \end{proof}

  \subsection{H-type theories}\label{subsec:6.3}
 
Wen~\cite{Wen} and Morrison-Walker~\cite{MW} introduced the notion of
$n$-dimensional topological field theories defined only on $n$-manifolds with
an infinitesimal time direction.  These are of Hamiltonian type, or H-type,
and are the minimal expectation for the low energy effective theory
describing a Hamiltonian system.  In this paper we assume emergent
relativistic invariance, so do not engage with H-type theories in a serious
way.  Nonetheless, in this subsection we indicate briefly how to analyze
invertible theories of H-type.
 
The first issue is definitional: Do the $n$-manifolds in the bordism category
have (i)~an \emph{oriented} time direction or merely (ii)~a time direction?
In unoriented theories this means a reduction of~$O_n$ to either
(i)~$O_{n-1}$ or (ii)~$O_1\times O_{n-1}$.  We opt for~(i).  After all, a
Hamiltonian system does have a definite orientation of time, and even in
relativistic quantum field theory we assume a time orientation of Minkowski
spacetime~(\S\ref{subsec:12.1}).  Then a more general symmetry group~$H_n$ is
reduced to~$H_{n-1}$, and an invertible theory of H-type is a map out of the
spectrum $\Sigma ^{n-1}MTH_{n-1}$.
 
Now the extension question: Does an equivariant map $\varphi \:\Sigma
^{n-1}MTH_{n-1}^\beta \to\Snizn $ extend to an equivariant map $\Sigma
^nMTH_n^\beta \to\Snizn $?  (In Wen's language this is an extension from
H-type to L-type.)  The obstruction is the value of~$\varphi $ on the
universal family of $H_{n-1}$-spheres $S^{n-1}$ parametrized by~$BH_n$.
Without the equivariance the value\footnote{Parallel to the $\Zo$-torsors
attached to $n$-manifolds are graded gerbes attached to $(n-1)$-manifolds.
The construction of a line may depend on a choice of metric, for example, so
may be part of a non-topological theory.}  is a $\zt$-graded \emph{complex}
line bundle over~$BH_n$; the equivariance implies the value is a $\zt$-graded
\emph{real} line bundle.  (See Remark~\ref{thm:mss146} for the connective
cover of~$\Sigma ^2\IZ$ and its bar involution~$\gamma $.)  The first
obstruction is the grading: the single quantum state on~$S^{n-1}$ should be
bosonic.  If so, the remaining obstruction is a class in~$H^1(BH_n;\zt)\cong
\Hom(H_n,\zt)\cong \Hom(\pi _0H_n,\zt)$.  For example, if $H_n=O_n$ or
$H_n=\Ppm_n$, then a hyperplane reflection should act trivially on the
line~$\varphi (S^{n-1})$.

  \begin{example}[]\label{thm:88}
 Continuing Example~\ref{thm:87}, the restriction map 
  \begin{equation}\label{eq:96}
     [\Sigma ^3MTSO_3^\beta ,\Sigma ^4\IZn ]^{\Z/2} \longrightarrow [\Sigma
     ^2MTSO_2^\beta ,\Sigma ^4\IZn ] ^{\Z/2}
  \end{equation}
is an index two inclusion of infinite cyclic groups.  So there exists a
continuous invertible theory~$\varphi $ of H-type with reflection structure
that does not extend to all oriented 3-manifolds.  Here is an example
defined on the category of oriented Riemannian 2-manifolds: assign the
$\zt$-graded determinant line~$\varphi (Y)$ of the $\dbar$-operator to a
closed 2-manifold~$Y$.  Then $\inde\dbar_{S^2}=1$ implies that $\varphi
(Y)$~is odd.
  \end{example}

   \section{Positivity in extended invertible topological theories}\label{sec:7}

In this section we develop the theory of extended positivity in invertible
field theories.  We already introduced a homotopy-theoretic manifestation of
extended positivity for higher super lines in Definition~\ref{thm:178}.
Here, in~\S\ref{subsec:7.3}, we begin by introducing spaces of invertible
field theories leading up to the space of invertible reflection positive
theories.  Our main result, Theorem~\ref{thm:184}, identifies the homotopy
type of the space of invertible \emph{continuous} reflection positive
theories as the 0-space of the Anderson dual to a Thom spectrum.  The
homotopy type of the corresponding space in the \emph{discrete} case, worked
out in Theorem~\ref{thm:188}, is a corollary, as is Theorem~\ref{thm:110} in
the introduction.  The proof of Theorem~\ref{thm:184} appears
in~\S\ref{subsec:7.4} and~\S\ref{subsec:7.5}.

  \subsection{Spaces of invertible field theories, extended positivity, and
  stability}\label{subsec:7.3}  

\subsubsection{Preliminary: splitting off a reflection}\label{subsubsec:7.3.2} 
 Fix~$n>0$.  Recall that if $(H_n,\rho _n)$ is a symmetry type
(Definition~\ref{thm:153}), then we have a canonical
coextension~\eqref{eq:30} of~$H_n$ by~$\pmo$ to a group~$\hH_n$.  It is this
extension that determines the $\beta $-involution on the Madsen-Tillmann
spectrum~$MTH_n$, as in the discussion preceding Ansatz~\ref{thm:147}; the
homotopy quotient of~$MTH_n^\beta $ is~$MT\hH_n$.

The splitting of interest is contained in~\eqref{eq:192} (and is also
implicit in Proposition~\ref{thm:44}).  It exists whenever there is an
``auxiliary'' direction.  The middle vertical homomorphism
in~\eqref{eq:192} induces
 \[BH_{n-1}\times B\Z/2\to B\hat{H}_{n}, \]
which factors the projection 
\[
BH_{n-1}\times B\Z/2\to B\hat{H}_{n}\to B\Z/2.
\]
This, in turn, gives a sequence of equivariant maps 
\begin{equation}
\label{eq:p5}
\Sigma^{n-1}MTH_{n-1}\wedge S^{1-\sigma} \to 
\Sigma^{n}MTH_{n}^{\beta} \to MTH\wedge S^{1-\sigma}
\end{equation}
factoring the smash product of the identity map of $S^{1-\sigma}$
with the defining inclusion of $\Sigma^{n-1}MTH_{n-1}$ into $MTH$.

The stable form of the splitting implies the following.

  \begin{proposition}[]\label{thm:74}
 The $\Z/2$-equivariant spectra~$MTH^\beta $ and~$MTH^\gamma $ are
canonically equivariantly weakly equivalent. 
  \end{proposition}

\noindent
 We remind that, despite the similarity of notation, the $\beta $-involution
is defined by the group coextension whereas the $\gamma $-involution is
natural, obtained by smashing with~$S^{1-\sigma }$.

  \begin{proof}
 Take $n\to\infty $ in~\eqref{eq:p5}.  The colimit of the first term
is~$MTH\wedge S^{1-\sigma }$ and the composition is homotopic to the identity
map.
  \end{proof}

\subsubsection{Spaces of theories}\label{subsubsec:7.3.1}

Let $n>0$~be the spacetime dimension and fix a positive integer~$k\le n$.
Let $G$~be a Lie group equipped with a homomorphism $\rho:G\to O_{k}$.  The
map $\rho$ is used to form the Thom spectrum $MTG=\thom(BG;-\rho)$.  Define
the {\em space of continuous invertible $k$-truncated $n$-dimensional
topological field theories of symmetry type $(G,\rho)$} as\footnote{The
`$k$'~usually appears in the notation for~$G$, as in~\eqref{eq:p10} below, so
we do not adorn~`$\fieldSymbol$' with it.}
 \[
\fields{G}{n} =\fields{G,\rho }{n} = \Map(\Sigma^{k}MTG, \Sigma^{n+1}I\Z(1)).
\]
 Usually $\rho$ is understood in the notation.  A point of $\fields{G}{n}$ may
be thought of as a $k$-dimensional field theory that associates to a closed
$\ell $-manifold $M$, $\ell \le k$, a super $(n-\ell )$-line.

Different flavors of field theories are obtained by changing the target, as
in Definition~\ref{thm:178} and Definition~\ref{thm:179}.  We give the
definitions for continuous invertible theories; there are analogous
definitions for discrete invertible theories.

  \begin{definition}[]\label{thm:180}
  Fix integers $n>0$ and $k\le n$.

\begin{enumerate}[{\textnormal(}i{\textnormal)}]

 \item The space of \emph{continuous invertible $k$-truncated $n$-dimensional
{\em Hermitian} extended topological field theories with symmetry type
$(G,\rho)$} is
 \[ \fieldsHerm{G,\rho }{n} =\Map(\Sigma^{k}MTG,\Sigma^{n+1}I\Z(1)_{H}) \]

 \item The space of \emph{continuous invertible $k$-truncated $n$-dimensional
\emph{positive definite} extended topological field theories with symmetry
type $(G,\rho)$} is 
\[ \fieldsPos{G,\rho }{n} =
\Map(\Sigma^{k}MTG,\Sigma^{n+1}I\Z(1)_{\text{pos}}).  \]

 \end{enumerate}
  \end{definition}

\noindent
 Note that composition with the map $I\Z(1)_{\text{pos}}\to I\Z(1)_{H}$
induces a map 
  \begin{equation}\label{eq:305}
     \fieldsPos{G,\rho }{n} \longrightarrow \fieldsHerm{G,\rho }{n}. 
  \end{equation}

Assume the symmetry type is a pair~$(H_n,\rho _n)$ as in
Definition~\ref{thm:153}.  We recall the notation~\eqref{eq:289} for the
space of theories with reflection structure:
  \begin{equation}\label{eq:291}
     \fieldsRefl{H_n}{n} =\Map^{\Z/2}(\Sigma ^nMTH_n^\beta ,\Snizn). 
  \end{equation}
Composition with the first map in~\eqref{eq:p5} produces a map 
  \begin{equation}\label{eq:304}
     \fieldsRefl{H_n}n\longrightarrow \fieldsHerm{H_{n-1}}n. 
  \end{equation}
Therefore, the value of a theory with reflection structure on a closed
manifold of dimension~$\ell \le n-1$ is a \emph{Hermitian} super $(n-\ell
)$-line.  (The Hermitian line for~$\ell =n-1$ is described
in~\S\ref{subsec:3.3} for not necessarily invertible theories.)  Recall the
stabilization~$\rho \:H\to O$ in~\eqref{eq:278}, and define
  \begin{equation}\label{eq:296}
     \fieldsStable{H}{n} = \Map(MTH,\Sigma^{n+1}I\Z(1)),
  \end{equation}
the space of {\em stable} $n$-dimensional invertible topological field
theories of symmetry type $H$.  

We use the notations $\fieldsHermd{G}{n}$, $\fieldsPosd{G}{n}$,
$\fieldsRefld{H_n}{n} $ for the corresponding spaces of discrete field
theories, which are mapping spaces with codomain $\Sigma ^n\ICx_H$, $\Sigma
^n\ICx_{\textnormal{pos}}$, and $\SCxg$, respectively.  (See~\eqref{eq:288}
and~\eqref{eq:287}.)  

The main objects of interest are invertible reflection positive theories.  As
stated after~\eqref{eq:304}, an invertible theory with reflection structure
has values on closed manifolds of dimension~$\le (n-1)$ that are higher
Hermitian super lines.  The following definition uses~ \eqref{eq:305} to
impose positivity, which in dimension~$n-1$ is a \emph{condition}
(Definition~\ref{thm:24}) and in dimensions~$< (n-1)$ is a \emph{structure}.

  \begin{definition}[]\label{thm:p181}
 Fix $n>0$ and a symmetry type~$(H_n,\rho _n)$ in the sense of
Definition~\ref{thm:153}.  Define the spaces $\fieldsReflPos{H_n}{n}$ and
$\fieldsReflPosd{H_n}{n}$ of \emph{$n$-dimensional continuous
(resp.~discrete) invertible reflection positive topological field theories
with symmetry type $(H_{n},\rho _n)$} and maps out of these spaces so that
each square in the diagram
  \begin{equation}\label{eq:p10}
  \begin{gathered} \xymatrix{ \fieldsReflPosd{H_{n}}{n} \ar[r]\ar[d] &
\fieldsReflPos{H_{n}}{n} \ar[r]\ar[d] & \fieldsPos{H_{n-1}}{n} \ar[d] \\
 \fieldsRefld{H_{n}}{n} \ar[r] & \fieldsRefl{H_{n}}{n}
\ar[r]^>>>>>{\,\eqref{eq:304}\,\,} & 
\fieldsHerm{H_{n-1}}{n}
 }
  \end{gathered}
\end{equation}
is a homotopy pullback.
  \end{definition}

\noindent
 For the spaces of theories in the right hand column we use
Definition~\ref{thm:180} with~$k=n-1$, $G=H_{n-1}$, and $\rho =\rho _{n-1}$.
Our task is to determine the homotopy types of $\fieldsReflPos{H_{n}}{n}$ and
$\fieldsReflPosd{H_{n}}{n}$.

\subsubsection{Extended positivity structure}\label{subsubsec:7.3.4}

Definition~\ref{thm:p181} is natural given our homotopy-theoretic
implementation of higher positive definite Hermitian super lines in
Definition~\ref{thm:179}.  We now make a short digression to identify
extended positivity in an invertible $n$-dimensional field theory as a
structure that trivializes an associated invertible $(n-1)$-dimensional
field theory.  For this we need yet an additional space of invertible field
theories, based on the target spectrum of higher \emph{real} super lines
(Definition~\ref{thm:178}(ii)).

  \begin{definition}[]\label{thm:183}
 The space of \emph{continuous invertible $(n-1)$-dimensional \emph{real}
extended topological field theories with symmetry type~$(H_{n-1},\rho
_{n-1})$} is
  \begin{equation}\label{eq:294}
     \fieldsReal{H\mstrut _{n-1}}{n-1} = \Map\bigl(\Sigma ^{n-1}MTH_{n-1},(\Sigma
     ^n\IZ^\gamma)^{h\Z/2}\bigr). 
  \end{equation} 
  \end{definition}

\noindent
 The partition function on a closed $(n-1)$-manifold lies in~$\pmo$,
the value on a closed $(n-2)$-manifold is a real super line, etc.  (See
Remark~\ref{thm:mss146} for the top few homotopy groups of $\bigl(\IZ^\gamma
\bigr)^{h\Z/2}$.)

To begin, for any pointed space~$X$ there is an equivalence of spectra
$X_+\approx X\vee S^0$, which leads to a cofibration sequence 
  \begin{equation}\label{eq:292}
     X\longrightarrow X_+\longrightarrow S^0 .
  \end{equation}
Set $X=B\Z/2$, smash with~$\Sigma ^{n-1}MTH_{n-1}$, and apply
$\Map\bigl(-,\Sniz\bigr)$ to obtain the fibration sequence 
  \begin{equation}\label{eq:293}
     \fieldsPos{H_{n-1}}{n}\longrightarrow  \fieldsHerm{H_{n-1}}{n}
     \longrightarrow \fieldsReal{H\mstrut _{n-1}}{n-1} .
  \end{equation}
For the middle term use~\eqref{eq:285} and for the last the identification  
  \begin{align*}
  \Map(\Sigma^{n-1}MTH_{n-1}\wedge B\Z/2,\Sigma^{n+1}I\Z(1)) &\approx
\ztmap(\Sigma^{n}MTH_{n-1}\wedge S^{\sigma-1},\Sigma^{n+1}I\Z(1))\\ &\approx
\ztmap(\Sigma^{n}MTH_{n-1},\Sigma^{n+1}I\Z(1)^{\gamma}) \\ &\approx
\ztmap(\Sigma^{n-1}MTH_{n-1},\Sigma^{n}I\Z(1)^{\gamma})\\ &\approx
\Map(\Sigma^{n-1}MTH_{n-1},\Sigma^{n}(I\Z(1)^{\gamma})^{h\Z/2}).
  \end{align*}

Therefore, the space $\fieldsReflPos{H_{n}}{n}$ may also be defined as the
homotopy fiber of the composition
  \begin{equation}\label{eq:295}
     \kappa \:\fieldsRefl{H_{n}}{n} \longrightarrow
     \fieldsHerm{H_{n-1}}{n}\longrightarrow \fieldsReal{H\mstrut _{n-1}}{n-1}. 
  \end{equation}
This leads to the following definition. 
 
  \begin{definition}[]\label{thm:108}
 An \emph{\textnormal{(}extended\textnormal{)} positivity structure} on a
continuous $n$-dimensional field theory\newline $\varphi \in
\fieldsRefl{H_n}n$ is a trivialization of~$\kappa (\varphi )$.
  \end{definition}

\noindent
 That is, a positivity structure is a path from~$\kappa (\varphi )$ to the
basepoint in~$\fieldsReal{H\mstrut _{n-1}}{n-1}$.  This discussion identifies the
space of continuous reflection positive invertible field theories as the
space of continuous invertible field theories with both a reflection
structure and a positivity structure.

  \begin{remark}[]\label{thm:111}
 The partition function of the field theory~$\kappa (\varphi )\:\Sigma
^{n-1}MTH_{n-1}\to \Sigma^{n}(I\Z(1)^{\gamma})^{h\Z/2}$ is the homomorphism
  \begin{equation}\label{eq:131}
     \pi _{n-1}\Sigma ^{n-1}MTH_{n-1}\longrightarrow \pmo
  \end{equation}
induced on~$\pi _{n-1}$, and it agrees with the homomorphism~\eqref{eq:94}
which tracks the sign of the hermitian lines in the theory~$\varphi $.  The
highest piece of the positivity structure is therefore the standard
positivity \emph{constraint} in Definition~\ref{thm:24}.  The theory~$\kappa
(\varphi )$ assigns a real super line to a closed $(n-2)$-manifold and more
complicated objects in lower dimensions; their trivializations are
\emph{data}.  
  \end{remark}

\subsubsection{Main theorems}\label{subsubsec:7.3.3}

We apply the splitting of~\S\ref{subsubsec:7.3.2} to construct a map  
  \begin{equation}\label{eq:297}
     \fieldsStable Hn\longrightarrow \fieldsReflPos{H_{n}}{n} 
  \end{equation}
as follows.  (These spaces of invertible field theories are defined
in~\eqref{eq:296} and ~\eqref{eq:p10}.)   Map 
  \begin{equation}\label{eq:298}
     \Sigma ^{n-1}MTH_{n-1}\wedge B\Z/2_+\longrightarrow \Sigma
     ^{n-1}MTH_{n-1}\longrightarrow MTH 
  \end{equation}
into~$\Sniz$ to obtain a map of~$\fieldsStable Hn$ into the upper right
corner of~\eqref{eq:p10}.  Use equivariant maps of the sequence~\eqref{eq:p5}
into~$\Snizn$ to map~$\fieldsStable Hn$ into the middle of the bottom row
of~\eqref{eq:p10} .  The two compositions into the lower right corner are
canonically homotopic, so the fact that the right square in~\eqref{eq:p10} is
a homotopy pullback yields~\eqref{eq:297}.

  \begin{theorem}[]\label{thm:184}
 The map $ \fieldsStable Hn\longrightarrow \fieldsReflPos{H_{n}}{n} $
in~\eqref{eq:297} is a homotopy equivalence.
  \end{theorem}

\noindent
 We give the proof of Theorem~\ref{thm:184} in~\S\ref{subsec:7.4}
and~\S\ref{subsec:7.5}. 

  \begin{corollary}[]\label{thm:186}
 There is an isomorphism
  \begin{equation}\label{eq:300}
     \pi _0\,\fieldsReflPos{H_{n}}{n}\cong [MTH,\Sniz]. 
  \end{equation}
  \end{corollary}

Next, we turn to discrete invertible theories.  First, observe that the
$\Z/2$-action on~$\CC$ by complex conjugation is equivalent to the
$\Z/2$-action on $\Map(\Z/2,\RR)$, so for any $\Z/2$-spectrum~$X$ one has
  \begin{equation}\label{eq:311}
     \Map^{\Z/2}(X,\HCn)\approx  \Map(X,H\RR). 
  \end{equation}
The spectrum $\Map(X,\HR)$ carries a residual $\Z/2$-action, induced from the
$\Z/2$-action on~$X$; it splits as a wedge of the $(+1)$-~and
$(-1)$-eigenspaces. The exponential sequence~\eqref{eq:mssexp2} of
$\Z/2$-equivariant spectra implies that the left map in the bottom row
of~\eqref{eq:p10} extends to a fibration sequence
  \begin{equation}\label{eq:312}
   \begin{aligned}
     \Map^{\Z/2}(\mth n^\beta ,\SCxg)\longrightarrow &\Map^{\Z/2}(\mth
     n^\beta ,\Snizn)\\ &\qquad \longrightarrow \Map^{\Z/2}(\mth n^\beta
     ,\Sigma ^{n+1}\HCn).  
  \end{aligned}
  \end{equation}
Apply~\eqref{eq:311} to the last term and use the fact that the left hand
square in~\eqref{eq:p10} is a homotopy pullback to obtain a fibration
sequence
  \begin{equation}\label{eq:313}
     \fieldsReflPosd{H_{n}}{n}\longrightarrow
     \fieldsReflPos{H_{n}}{n}\longrightarrow \Map(\mth n,\Sigma
     ^{n+1}\HR). 
  \end{equation}

  \begin{proposition}[]\label{thm:185}
 The image of the homomorphism
  \begin{equation}\label{eq:301}
     \pi _0\,\fieldsReflPosd{H_{n}}{n}\longrightarrow \pi
     _0\,\fieldsReflPos{H_{n}}{n}  
  \end{equation}
is the torsion subgroup of $\pi _0\,\fieldsReflPos{H_{n}}{n}$.
  \end{proposition}

\noindent 
 Theorem~\ref{thm:110} in the introduction follows from
Proposition~\ref{thm:185} and~\eqref{eq:300}.  In Theorem~\ref{thm:188} below
we determine the homotopy type of the space of discrete invertible reflection
positive field theories.

  \begin{proof}
 Since \eqref{eq:313}~is a fibration sequence of spectra, applying~$\pi _0$
we obtain an exact sequence of abelian groups in which, after
applying~\eqref{eq:300}, the second map is\footnote{\label{eigen} The
map~\eqref{eq:314} is $\Z/2$-equivariant for the $\beta $-involution on~$MTH$
and~$\mth n$.  By Proposition~\ref{thm:74} the $\beta $-~and $\gamma
$-involutions on~$MTH$ agree, from which $\Z/2$-acts as~$-1$ on the domain.
It follows that the image is contained in the $(-1)$-eigenspace of the
codomain, which is why we write `$\HRo$' in place of~`$\HR$'.}
  \begin{equation}\label{eq:314}
     [MTH,\Sniz]\longrightarrow [\mth n,\Sigma ^{n+1}\HRo]. 
  \end{equation}
The construction following~\eqref{eq:298} implies that this map is pullback
along the defining inclusion of $\mth n$ into~$MTH$.  The proposition follows
if we prove \eqref{eq:314}~is injective after tensoring the domain
with~$\RR$.  This follows immediately from Lemma~\ref{thm:152}. 
  \end{proof}

We parlay~\eqref{eq:313} into a more useful expression for the homotopy type
of the space of discrete invertible reflection positive field theories.
Recall the spectrum~$\IT$ introduced in Remark~\ref{thm:187}.

  \begin{theorem}[]\label{thm:188}
 For $n$~odd there is a homotopy equivalence 
  \begin{equation}\label{eq:315}
     \Map(MTH,\Sigma ^n\IT)\xrightarrow{\;\;\approx \;\;}
     \fieldsReflPosd{H_n}n 
  \end{equation}
For $n$~even there is a fibration sequence 
  \begin{equation}\label{eq:316}
     \Map(MTH,\Sigma ^n\IT)\longrightarrow \fieldsReflPosd{H_n}n
     \xrightarrow{\;\;s\;\;} \Rp 
  \end{equation}
in which $\Rp$~has the discrete topology and $s$~maps a discrete theory~$F$
to~$F(S^n)$.  
  \end{theorem}

\noindent
 Compare with the more rigid Theorem~\ref{thm:78} in the absence of
reflection structures.  Also, note that for any $n$-manifold~$X$ the disjoint
union $\beta X\amalg X$~is null bordant, and so in a stable theory the
partition functions have unit norm, consistent with the appearance of~$\IT$
in~\eqref{eq:315} and~\eqref{eq:316}.  There is a canonical section of~$s$
given by Euler theories (Example~\ref{thm:26}): given $x\in \Rp$ define the
Euler theory as the composition
  \begin{equation}\label{eq:317}
     \mth n^\beta \longrightarrow \Sigma ^n(BH_n^\beta )\mstrut _+\longrightarrow
     \Sigma ^nS^0\xrightarrow{\;\;\sqrt x \;\;} \Sigma ^n\HRp\longrightarrow
     \SCxg  
  \end{equation}
The restriction to $\mth{n-1}^\beta $~is trivialized; using~\eqref{eq:p10} we
obtain a reflection positive theory.

  \begin{proof}
 For any pointed space~$C_n$ use the nonequivariant version of the
exponential sequence~\eqref{eq:309} and the fibration sequence~\eqref{eq:313}
to construct the diagram 
  \begin{equation}\label{eq:318}
     \begin{gathered} \xymatrix{\Map(MTH,\Sigma ^n\IT)\ar[r]\ar[d] &
     \Map(MTH,\Sniz)\ar[r]\ar[d] & \Map(MTH,\Sigma ^{n+1}\HRo)\ar[d]\\
     \fieldsReflPosd{H_{n}}{n}\ar[r]\ar[d] &
     \fieldsReflPos{H_{n}}{n}\ar[r]\ar[d] & \Map(\mth n,\Sigma
     ^{n+1}\HR)\ar[d] \\ \Omega C_n\ar[r]& \** \ar[r]& C_n} \end{gathered} 
  \end{equation}
in which the rows are fibration sequences, as is the middle column, by
Theorem~\ref{thm:184}.  We claim 
  \begin{equation}\label{eq:319}
     C_n=\begin{cases}
     *,&n~\textnormal{odd},\\K(\RR,1),&n~\textnormal{even},\end{cases} 
  \end{equation}
renders the last column a fibration sequence; it follows that the first
column is as well.  (Here $K(\RR,1)$~is an Eilenberg-MacLane space.)  There
is an exponential to pass from the third column to the first column
in~\eqref{eq:318}, and so naturally $\Omega C_n\approx \Rp$ with the discrete
topology.  
 
To prove the claim observe first that we can replace the upper right entry
of~\eqref{eq:318} with the homotopy equivalent space $\Map(\mth{n+2},\Sigma
^{n+1}\HRo)$, using arguments similar to those in~\S\ref{subsec:6.2}.  To
analyze the resulting right vertical map consider the composition 
  \begin{equation}\label{eq:320}
     \pi _q\mth n\otimes \RR\xrightarrow{\;\;i_*\;\;}\pi _q\mth{n+1}\otimes
     \RR\xrightarrow{\;\;j_*\;\;}\pi _q\mth{n+2}\otimes \RR. 
  \end{equation}
The composition~$j_*\circ i_*$ is an isomorphism for~$q<n$, and since we map
to~$\Sigma ^{n+1}\HR$ only~$q\le n+1$ is relevant.  Use~\eqref{eq:220} and
the exact sequence 
  \begin{equation}\label{eq:321}
     \pi _{m+1}\mth{m+1}\xrightarrow{\;\;\Euler\;\;}\ZZ
     \xrightarrow{\;\;[S^m]\;\;} \pi _m\mth m\longrightarrow \pi
     _m\mth{m+1}\longrightarrow 0 
  \end{equation}
to verify the following four assertions.  If $n$~is odd, then $j_*\circ i_*$
~is an isomorphism for~$q=n$ and~$q=n+1$.  If $n$~is even, then $j_*\circ
i_*$~is an isomorphism for $q=n+1$ and is surjective for~$q=n$ with kernel
generated by~$[S^n]$.  Observe that $[S^n]=[\hH_{n+1}/\hH_n]$~is fixed by the
$\beta $-involution.  It follows that the upper right arrow in~\eqref{eq:318}
is injective with image the $(-1)$-eigenspace of the $\beta $-involution; the
cokernel the $(+1)$-eigenspace generated by~$[S^n]$.  (Compare with the
discussion in footnote~\footref{eigen}.)  The claim, and so the theorem,
follows.
  \end{proof}

We conclude this subsection with a comment about our application of these
theorems to computations.  Namely, the considerations in~\S\ref{subsec:4.4}
lead to the following conjecture, which uses non-topological invertible
theories (for which we do not develop mathematical foundations in this
paper).

  \begin{conjecture}[]\label{thm:143}
 There is a 1:1 correspondence  
  \begin{equation}\label{eq:187}
     \left\{\vcenter{\normalfont\hbox{deformation classes of reflection positive
     }\hbox{invertible $n$-dimensional extended field}\hbox{theories with
     symmetry type~$(H_{n},\rho _n)$}} \right\} \cong [MTH,\SIZ]. 
  \end{equation}
  \end{conjecture}

\noindent
 We remark that since the rational cohomology of~$BH$ vanishes in odd degrees,
elements of infinite order in~\eqref{eq:187} occur only for $n$~odd.  

  \begin{remark}[]\label{thm:165}
 A restatement of Corollary~\ref{thm:186} is the 1:1 correspondence
  \begin{equation}\label{eq:248}
     \left\{ \vcenter{\hbox{isomorphism classes of reflection positive
     continuous}\hbox{invertible $n$-dimensional extended topological} 
     \hbox{field theories with symmetry type~$(H_{n},\rho _n)$}} \right\}
     \cong [MTH,\SIZ]. 
  \end{equation}
If we accept that the effective low energy theory of an invertible gapped
system is a continuous invertible topological field theory, as in
Remark~\ref{thm:158}, then we can apply~\eqref{eq:248} to the computations
in~\S\ref{sec:8} rather than~\eqref{eq:187}.  This has an advantage:
\eqref{eq:248}~is a theorem in the context of this paper.
  \end{remark}

  \begin{remark}[]\label{thm:144}
 A homotopy class of maps $MTH\to\Sniz$ leads to a canonical isomorphism
class of invertible field theories via the following sketch; the theories are
topological if and only if the homotopy class has finite order.  By the
twisted Thom isomorphism the homotopy classes are elements of~$\IZ^{\tau
+n+1}(BH)$, where $\tau $~is the canonical ``density twisting'': the pullback
to manifolds with tangential $H$-structure can be integrated.  According to
the main theorem in~\cite{FH1} there is a unique lift to the differential
cohomology group~$\widehat{\IZ}^{\tau +n+1}(B\mstrut _\nabla H)$.  Choose a
``cocycle'' representative.  Then on any manifold with a differential
$H$-structure we can integrate to construct an invariant, and these
invariants fit to an invertible field theory on~$\Bord_n^{\nabla }(H)$.
  \end{remark}

  \subsection{Proof of Theorem~\ref{thm:184}}\label{subsec:7.4}

We restate the theorem in the language of stable homotopy theory.

\begin{proposition}
\label{thm:p4}
The square
  \begin{equation}\label{eq:p1}
  \begin{gathered}
  \xymatrix{
\Map(MTH,\Sigma^{n+1}I\Z(1))  \ar[r]\ar[d]  & \Map(\Sigma^{n-1}MTH_{n-1},\Sigma^{n+1}I\Z(1))    
\ar[d] \\
\ztmap(\Sigma^{n}MTH_{n}^{\beta}, \Sigma^{n+1}I\Z(1)^{\gamma})   \ar[r]&
\ztmap(\Sigma^{n-1}MTH_{n-1}^{\gamma}, \Sigma^{n+1}I\Z(1)^{\gamma})
}
 \end{gathered}
 \end{equation}
is a homotopy pullback square of spaces.
\end{proposition}

\noindent
 The analysis of this square becomes cleaner if every term of the form
$\ztmap(X,\Sigma^{n+1}I\Z(1)^{\gamma })$ is replaced with $\Map((X\wedge
S^{\sigma-1})_{h\Z/2},\Sigma^{n+1}I\Z(1))$.  Doing so,
Proposition~\ref{thm:p4} becomes the assertion that the square
  \begin{equation}\label{eq:p11}
  \begin{gathered}
 \xymatrix{ \Sigma^{n-1}MTH_{n-1}\wedge B\Z/2_{+} \ar[r]\ar[d] &
\Sigma^{n}MT\hat{H}_{n}^{(\sigma-1)} \ar[d] \\ \Sigma^{n-1}MTH_{n-1} \ar[r] &
MTH}
 \end{gathered}
 \end{equation}
becomes a homotopy pullback square after applying $\Map(\slot,
\Sigma^{n+1}I\Z(1))$.  We use the notation
  \begin{equation}\label{eq:303}
     MT\hH_n^{(\sigma -1)} = \Thom(B\hH_n;-\hat\rho _n+\sigma -1). 
  \end{equation}
To clarify the argument we state this as as

\begin{proposition}
\label{thm:p6}
For any $m\ge n$, the square
\begin{equation}\label{eq:p20}
\begin{gathered}
\xymatrix{
\Sigma^{m-1}MTH_{m-1}\wedge B\Z/2_{+}  \ar[r]\ar[d]  &
\Sigma^{m}MT\hat{H}_{m}^{(\sigma-1)} 
\ar[d] \\
\Sigma^{m-1}MTH_{m-1}  \ar[r]        & MTH
}
\end{gathered}
\end{equation}
 becomes a homotopy pullback square after applying
$\Map(\slot,\Sigma^{n+1}I\Z(1))$.
\end{proposition}

The proof of Proposition~\ref{thm:p6} will make repeated use of the following
result, which follows from the universal property~\eqref{eq:62} of~$\IZ$.

\begin{lemma}
\label{thm:p5} Suppose $A$ is a spectrum having the property that
$\pi_{i}A=0$ for $i\le n$ and $\pi_{n+1}A$ is a torsion group.  If
$A\to X\to Y$ is a cofibration sequence then 
\[
\Map(Y,\Sigma^{n+1}I\Z(1)) \to \Map(X,\Sigma^{n+1}I\Z(1))
\]
is a weak equivalence of spaces.  
\end{lemma}

The proof of Proposition~\ref{thm:p6} is by decreasing induction on
$m$.  As $m\to\infty$ the square~\eqref{eq:p20} becomes
\[
  \begin{gathered}
\xymatrix{
MTH\wedge B\Z/2_{+}  \ar[r]\ar[d]  & MTH\wedge B\Z/2_{+}
\ar[d] \\
MTH  \ar[r]        & MTH
}
  \end{gathered}
\]
which is obviously a pushout.  On the other hand for $m>(n+2) $ the
maps 
\begin{gather*}
\Sigma^{m-1}MTH_{m-1}\to MTH \\
\Sigma^{m}MT\hat{H}^{(\sigma-1)}_{m} \to MTH\wedge B\Z/2_{+}
\end{gather*}
become equivalences after applying $\Map(\slot, \Sigma^{n+1}I\Z(1))$, so
the result is true for all $m>n+2$.   (Compare with the proof of
Theorem~\ref{thm:82}.) 

Since the homotopy fiber of the left vertical map in~\eqref{eq:p20} is
$\Sigma^{m-1}MTH_{m-1}\wedge B\Z/2$, Proposition~\ref{thm:p6} is
equivalent to the assertion that for all $m\ge n$, the sequence
\[
\Sigma^{m-1}MTH_{m-1}\wedge B\Z/2 \to
\Sigma^{m}MT\hat{H}^{(\sigma-1)}_{m} \to MTH
\]
becomes a fibration sequence after applying
$\Map(\slot,\Sigma^{n+1}I\Z(1))$.   The induction step therefore
follows from

\begin{proposition}
\label{thm:p8} For $m\ge n$, the square 
  \begin{equation}\label{eq:302}
  \begin{gathered}
     \xymatrix{ \Sigma^{m-1}MTH_{m-1}\wedge B\Z/2 \ar[r]\ar[d] &
     \Sigma^{m}MT\hat{H}^{(\sigma-1)}_{m} \ar[d] \\ \Sigma^{m}MTH_{m}\wedge
     B\Z/2\ar[r] & \Sigma^{m+1}MT\hat{H}^{(\sigma-1)}_{m+1} } 
  \end{gathered}\end{equation}
becomes a homotopy pullback square after applying
$\Map(\slot,\Sigma^{n+1}I\Z(1))$.
\end{proposition}

What is at stake in Proposition~\ref{thm:p8} is to prove that the induced map
 \begin{equation}\label{eq:p12}
\Sigma^{m-1}(BH_{m})_{+}\wedge B\Z/2 \to \Sigma^{m}\Thom(B\hat{H}_{m+1};\sigma-1)
 \end{equation}
of homotopy fibers of the vertical maps in~\eqref{eq:302} becomes a homotopy
equivalence after applying $\Map(-,\Sniz)$.  The following result will be
proved in~\S\ref{subsec:7.5}.

\begin{lemma} \label{thm:p7} 
 The map~\eqref{eq:p12} is the $(m-1)^{\textnormal{st}}$ suspension of the
map of Thom spectra \textnormal{(}of the bundle $(\sigma-1)$\textnormal{)}
associated to the map
  \begin{equation}\label{eq:p18}
BH_{m}\times B\Z/2\to B\hat{H}_{m+1}
\end{equation}
given by the choice of reflection in the last coordinate.  
\end{lemma}

Assuming Lemma~\ref{thm:p7} we can prove Proposition~\ref{thm:p8}.

\begin{proof}[Proof of Proposition~\ref{thm:p8}]
It suffices to show that the induced map~\eqref{eq:p12} becomes a weak
equivalence after applying $\Map(\slot,\Sigma^{n+1}I\Z(1))$.
The map~\eqref{eq:p18} fits
into a Cartesian square
\[
  \begin{gathered}
\xymatrix{
S^{m}  \ar@{=}[r]\ar[d]  & S^{m} \ar[d]   &  \\
BH_{m}  \ar[r]\ar[d]  & BH_{m}\times B\Z/2  \ar[r]\ar[d]   &  B\Z/2 \ar@{=}[d]\\
BH_{m+1}  \ar[r]        & B\hat{H}_{m+1}  \ar[r]         &
B\Z/2\mathrlap{\ ,}
}
  \end{gathered}
\]
so Lemma~\ref{thm:p7} implies that the cofiber of~\eqref{eq:p12} is
$2m$-connected.  Since $m\ge n\ge 1$, one has $2m\ge n $ and so
the cofiber is $n$-connected.  Both terms in~\eqref{eq:p12} are
rationally acyclic.  The result then follows from Lemma~\ref{thm:p5}.
\end{proof}

  \subsection{Transfers}\label{subsec:7.5} 
 
Suppose that $M\to X$ is a fiber bundle with fibers closed smooth
manifolds $M_{x}$ of dimension $n$.  Let $T_{M/X}$ be the vector
bundle over $M$ whose fiber at $a\in M_{x}$ is the tangent space
$T_{a}M_{x}$.  There is functorial stable map
\[
\Sigma^{\infty}X_{+}\to \thom(M,-T_{M/X})
\]
called the {\em transfer map}.   When there is an
embedding $M\subset X\times \R^{n}$ for some $n$ it can be constructed
from the Pontrjagin Thom collapse 
\[
\thom(X,\triv{n}) \to \thom(M,\triv{n}-T_{M/X})
\]
by passing to suspension spectra and desuspending $n$ times.   The
transfer map is constructed in the general case by passing to the
colimit over  the category of pairs
\begin{align*}
X_{\alpha}&\to X \\
i_{\alpha}:M_{\alpha}&\hookrightarrow
X_{\alpha}\times R^{N_{\alpha}}
\end{align*}
in which $M_{\alpha}\to X_{\alpha}$  is the pullback of $M\to X$ along
the map $X_{\alpha}\to X$. 

When there is an embedding $M\subset W$  over $B$, the Pontrjagin Thom
construction leads to a {\em twisted} transfer map 
\[
\thom(B;W) \to \thom(X;W-T_{M/X}).
\]
 The twisted transfer extends in the evident manner to the case of virtual
bundles $W$.

\begin{proposition}\label{thm:p3} 
 Suppose that $W$ is a vector bundle over $X$, that $f\:M\to
W$ is a map over $X$ transverse to the zero section and let $N$ be the
inverse image of $0$.  There is a commutative diagram
\[
\begin{gathered}
\xymatrix{
\thom(X;0)\ar[r]\ar[d]  & \thom(N;-T_{N/X})
\ar[d] \\
\thom(X;W)  \ar[r]        & \thom(M;W-T_{M/X})
}
\end{gathered}
\]
in which the left vertical map is derived from the zero section, and
the right is the natural map of Thom complexes coming from the
inclusion $N\subset M$ and the isomorphism 
\[
T_{M/X} \approx T_{N/X}\oplus W.
\]
\end{proposition}

\begin{proof}
It suffices to establish the case in which there is an embedding 
\[
\iota:M\hookrightarrow  \R^{n}.
\]
Applying the Pontrjagin-Thom constructions to the rows in the
transverse pullback square
\[
\begin{gathered}
\xymatrix{
N  \ar[r]\ar[d]  & X\times \R^{n}
\ar[d] \\
M  \ar[r]_-{(f,\iota)}        & W\times \R^{n} \mathrlap{\ .}
}
\end{gathered}
\]
gives a diagram
\[
\begin{gathered}
\xymatrix{
\thom(X;\triv{n})  \ar[r]\ar[d]  & \thom(N;\triv n-T_{N/X}) 
\ar[d] \\
\thom(X;W\oplus \triv{n})  \ar[r]        & \thom(M;W+\triv n-T_{M/X})
}
\end{gathered}
\]
in which the left vertical map is the inclusion of the zero section.
Desuspending, the claim follows easily from this.
\end{proof}

\begin{proof}[Proof of Lemma~\ref{thm:p7}]
The idea is to apply Proposition~\ref{thm:p3} to the left triangle in the
diagram  
  \begin{equation}\label{eq:307}
     \begin{gathered} \xymatrix{ S(\rho_{m})\times B\Z/2 \ar[r]\ar[dr] &
     S(\rho_{m}\oplus\sigma) \ar[r]\ar[d] & S(\hat{\rho}_{m+1}) \ar[d]\\ &
     BH_{m}\times B\Z/2 \ar[r] & B\hat{H}_{m+1} } \end{gathered} 
  \end{equation}
with 
\begin{align*}
X&=BH_{m}\times B\Z/2 \\
W &= \sigma\\
M& =S(\rho_{m}\oplus \sigma) \\
N&= S(\rho_{m})\times B\Z/2
\end{align*}
The diagram is written in order to clarify the relationship with
manifolds.  Note that there are equivalences
\begin{align*}
S(\hat{\rho}_{m+1})&\approx B\hat{H}_{m} \\
S(\rho_{m}) &\approx BH_{m-1}.
\end{align*}
Also, for a vector bundle $V\to X$ the relative tangent
bundle of $p:S(V)\to X$ is given by $T_{S(V)/X}\oplus \R =p^{\ast}V$.
Proposition~\ref{thm:p3} then gives the left square in the diagram 
\begin{equation}\label{eq:p21}
\begin{gathered}
\xymatrix{
\Sigma^{m-1}(BH_{m})_{+}\wedge B\Z/2_{+}  \ar[r]\ar[d]  & \Sigma^{m-1}(BH_{m})_{+}\wedge B\Z/2  \ar[r]\ar[d]   &  \Sigma^{m}\Thom(B\hat{H}_{m+1};\sigma-1) \ar[d]\\
\Sigma^{m-1}MTH_{m-1}\wedge B\Z/2_{+}  \ar[r]        & Y  \ar[r]         &   \Sigma^{m}MT\hat{H}_{m}^{(\sigma-1)}
}
\end{gathered}
\end{equation}
with 
\[
Y= \Sigma^{m}\thom(S(\rho_{m}\oplus \sigma);1-\rho_{m}-\sigma-1+\sigma);
\]
 the right square in~\eqref{eq:p21} is the pullback of transfer maps induced
from the pullback square in~\eqref{eq:307}.  The map~\eqref{eq:p12} is the
composition of
 \begin{equation}
\label{eq:p16}
\Sigma^{m-1}(BH_{m})_{+}\wedge B\Z/2\to \Sigma^{m-1}(BH_{m})_{+}\wedge
B\Z/2_{+} 
\end{equation}
with the top row of~\eqref{eq:p21}.  Lemma~\ref{thm:p7} now follows from
the fact that the composition
of~\eqref{eq:p16} with the left map in the top row of~\eqref{eq:p21} is
the identity.
\end{proof}

   \section{Fermionic theories with scalar internal symmetry group}\label{sec:8}

In this section we apply Theorem~\ref{thm:110} to some basic symmetry groups,
namely those whose subgroup~$K$ of internal symmetries is the group ~$O_1,
U_1, Sp_1$ of unit norm elements in the normed division algebras~$\RR, \CC,
\HH$, respectively.  (We use the names $\pmo, \TT, SU_2$ for these three
groups.)  The internal symmetry group~$K=\TT$ is the basic charge symmetry of
electromagnetism; in quantum mechanical models the presence of a so-called
particle-hole symmetry ``breaks''\footnote{We do not have any fundamental
understanding of this mechanism, especially the appearance of~$SU_2$.  In
\S\ref{subsec:8.1} we simply offer it as a storyline in relativistic theory
that matches the condensed matter literature.} it to either~$K=\pmo$ or
$K=SU_2$.  In~\S\ref{subsec:8.1} we classify the possible symmetry
groups~$H_n$ with these internal symmetries, and restricting to fermionic
symmetry groups we recover the 10-fold way; see Tables~\eqref{eq:101}
and~\eqref{eq:100}.  (Wang-Senthil~\cite{WS} list many of these groups---in a
nonrelativistic form~\eqref{eq:152}, \eqref{eq:153}---and the corresponding
``Cartan label''.  Metlitski~\cite{M} introduces the group~$\Pcp$, which
provided guidance for our treatment here.  This twisted form of~$\Pin^c$ also
appears implicitly in~\cite[\S A.4]{SeWi}.)  Lemma~\ref{thm:90} relates the
relativistic 10-fold way to the 10~real and complex Clifford algebras, thus
providing a link to other 10-fold ways.  

In~\S\ref{subsec:8.4} we sketch two ways in which a theory of free fermions
in Minkowski spacetime gives rise to a deformation class of reflection
positive invertible field theories, or to a reflection positive continuous
invertible topological field theory.  If one begins with an
$(n-1)$-dimensional free fermion theory, then there is an associated
$n$-dimensional invertible anomaly theory; if the original free fermion
theory admits a mass term, then the anomaly is trivializable.  In this paper
we do not attempt a complete treatment, so state the main result as a
conjecture, Conjecture~\ref{thm:124}.  It expresses the deformation class of
the anomaly theory as a composition of a twisted Atiyah-Bott-Shapiro map and
a Pfaffian map on real $K$-theory.  This $K$-theory interpretation depends on
Lemma~\ref{thm:119}, which expresses the existence of a mass in terms of
Clifford algebras.  

The second scenario is to begin with a massive free fermion theory in
$n$~dimensions, as we sketch in~\S\ref{subsubsec:8.4.6}.  The low energy
effective field theory is invertible, and \eqref{eq:148}~ is a formula for
its deformation class.  It is this scenario about gapped theories that is
relevant to this paper.

We carry out computations in low dimensions in~\S\ref{subsec:8.2}.  For each
of the 10~electron symmetry groups we list the groups of deformation classes
of reflection positive invertible topological theories and compute the map
from free fermions to it.  There is no further physical reasoning; we compute
directly from the results in Theorem~\ref{thm:110} and~\eqref{eq:148}.  The
techniques lie in stable homotopy theory, and in the next section we give
some details to illustrate how the computations are made.  As discussed
in~\S\ref{sec:1} these classification results apply to invertible topological
phases of condensed matter systems, often called SPT phases.  The fermionic
symmetry groups with~$K=\TT$ pertain to \emph{topological insulators}; those
with $K=\pmo$ and $K=SU_2$ pertain to \emph{topological superconductors}.

  \begin{remark}[]\label{thm:166}
 Most of the interacting groups we compute are torsion so are covered by
Theorem~\ref{thm:110}.  In the general case we interpret the computations as
theorems by using~\eqref{eq:248}, in which the interacting group is a group
of isomorphism classes of reflection positive \emph{continuous} invertible
topological field theories.  See~\S\ref{subsec:4.4} for a discussion of
expectations for low energy effective field theories.
  \end{remark}

In the theoretical discussions we assume~$n\ge 3$; in the computations we
apply the results to all~$n$.

  \subsection{Symmetry groups of fermionic systems}\label{subsec:8.1}

We already classified symmetry groups~$H_n$ with~$K=\pmo$ in
Proposition~\ref{thm:65}.  The \emph{fermionic} groups are the ones for which
$-1\in K$ is the distinguished element~$k_0$ of Theorem~\ref{thm:5} and
Corollary~\ref{thm:113}.\footnote{This implies the ``spin/charge relation''
of condensed matter physics, which is emphasized in~\cite{SeWi}: bosons have
even charge and fermions have odd charge.}  (The other possibility is
$k_0=1$, in which case the symmetry group is \emph{bosonic}.)  Those
fermionic groups are $\Spin_n$, $\Pp_n$, and $\Pm_n$.
 
Next, we classify symmetry groups with $K=\TT$.  These are group extensions 
  \begin{equation}\label{eq:97}
     1\longrightarrow \TT\longrightarrow SH_n\longrightarrow
     SO_n\longrightarrow 1 
  \end{equation}
if there is no time-reversal symmetry and
  \begin{equation}\label{eq:65}
     1\longrightarrow \TT\longrightarrow H_n\longrightarrow
     O_n\longrightarrow 1 
  \end{equation}
if there is time-reversal symmetry.  Recall the group~$E_n$ introduced before
Proposition~\ref{thm:65}.

  \begin{proposition}[$K=\TT$]\label{thm:60}
 Up to isomorphism there are two distinct group extensions~\eqref{eq:97}
with~$n\ge3$, and the groups~$SH_n$ that appear are $SO_n\times \TT$ and
$\Spin_n^c$.  Up to isomorphism there are six distinct group
extensions~\eqref{eq:65} with~$n\ge3$, and the groups~$H_n$ that appear are
mutually nonisomorphic.  Three of the groups have identity
component~$SO_n\times \TT$:
  \begin{align} 
  O_n&\times \TT\label{eq:71}\\
  O_n&\ltimes \TT\label{eq:72}\\
  E_n&\ltimes \TT\bigm/ \pmo \label{eq:73}
  \end{align}
The identity component of the remaining three groups is~$\Spin^c_n$:
  \begin{align} 
  \Pc_n\;\, &= \Pp_n\times \TT\bigm/ \pmo\label{eq:66}\\
  \Pcp_n &= \Pp_n\ltimes \TT\bigm/ \pmo\label{eq:67}\\
  \Pcm_n &= \Pm_n\ltimes \TT\bigm/ \pmo \label{eq:68}
  \end{align}
  \end{proposition}

\noindent
 The group~$\Pc_n$ is also isomorphic to $\Pm_n\times \TT\bigm/\pmo$.  It
sits in the complex Clifford algebra generated by~$\RR^n$ with a
nondegenerate symmetric bilinear form~\cite{ABS}.  In~$\Pcpm_n$ the action
of~ $\Ppm_n$ on~$\TT$ factors through $\pi _0\Ppm_n$ and is via inversion
$\lambda \mapsto\lambda \inv $.  In each case we divide out by the diagonal
subgroup~$\pmo$.  The groups with identity component~$\Spin^c_n$ are
fermionic. 

  \begin{proof}
 The extension~\eqref{eq:97} is central, so up to isomorphism classified by
the cohomology group 
  \begin{equation}\label{eq:98}
     H^2(BSO_n;\underline{\TT})\cong H^3(BSO_n;\ZZ)\cong \zt. 
  \end{equation}
The underline indicates the sheaf cohomology of continuous functions
into~$\TT$ with the standard topology.  It is well-known that $\Spin^c_n$
corresponds to the nonzero element.

The only nontrivial automorphism of~$\TT$ is inversion, so in the
extension~\eqref{eq:65} either $O_n$~acts trivially or it acts through its
components with elements of determinant~$-1$ acting by inversion.  In each
case the group extensions are classified by a cohomology group of the
classifying space~$BO_n$:
  \begin{align}
      H^2(BO_n;\underline{\TT})&\cong H^3(BO_n;\ZZ)\cong \zt
      \label{eq:69}\\ H^2(BO_n;\widetilde{\underline{\TT}})&\cong
      H^3(BO_n;\widetilde{\ZZ})\cong \zt \times \zt\label{eq:70}
  \end{align}
The tilde indicates coefficients twisted by inversion.  The
product~\eqref{eq:71} and semi-direct product~\eqref{eq:72} account
for~\eqref{eq:69} and the remaining four groups for~\eqref{eq:68}, as can be
seen from cohomological computations we omit.
  \end{proof}

According to the arguments in Appendix~\ref{sec:10}, the anti-Wick rotation
of~ $\Pcp$ contains a time-reversal symmetry~$T$ with $T^2=(-1)^F$ and the
anti-Wick rotation of~$\Pcm$ contains a time-reversal symmetry~$T$ with
$T^2=1$.  More precisely, the groups~\eqref{eq:66} and~\eqref{eq:71} are Wick
rotations of relativistic symmetry groups that include CT~symmetry; the
remaining groups are Wick rotations of relativistic symmetry groups that
include T~symmetry.\footnote{This is our interpretation of~\cite[\S3.7]{W1}.
There are more general possibilities with larger internal symmetry group~$K$.
This occurs in~\cite[\S3]{SeWi}, for example, in a theory with both~T and
CT~symmetry.} 

Finally, we classify symmetry groups with~$K=SU_2$.  Now we have possible
extensions  
  \begin{equation}\label{eq:99}
     1\longrightarrow SU_2\longrightarrow SH_n\longrightarrow
     SO_n\longrightarrow 1 
  \end{equation}
and
  \begin{equation}\label{eq:79}
     1\longrightarrow SU_2\longrightarrow H_n\longrightarrow
     O_n\longrightarrow 1 
  \end{equation}

  \begin{proposition}[$K=SU_2$]\label{thm:66}
  Up to isomorphism there are two distinct group extensions~\eqref{eq:99}
with~$n\ge3$, and the groups~$SH_n$ that appear are $SO_n\times SU_2$ and 
  \begin{equation}\label{eq:240}
      G_0=\Spin_n\times \mstrut _{\pmo}SU_2. 
  \end{equation}
Up to isomorphism there are four distinct group extensions~\eqref{eq:79}
with~$n\ge3$, and the groups~$H_n$ that appear are mutually nonisomorphic.
Two of the groups have identity component $SO_n\times SU_2$:
  \begin{align} 
  O_n&\times SU_2\label{eq:80}\\
  E_n&\times \mstrut _{\pmo}SU_2\label{eq:81}
  \end{align}
The identity component of the remaining two groups is $G_0$:  
  \begin{align} 
  \Gp_n &= \Pp_n \times \mstrut _{\pmo} SU_2\label{eq:82}\\
  \Gm_n &= \Pm_n \times \mstrut _{\pmo} SU_2 \label{eq:83}
  \end{align}
  \end{proposition}

\noindent
 The symmetry groups with identity component $G^0$ are fermionic.

  \begin{proof}
 The classification of the identity component~$SH_n$ follows from
Theorem~\ref{thm:5}(2): there are two central elements~$k_0\in SU_2$ with
$k_0^2=1$.  To classify the two-component group~$H_n$ we apply a useful
general result \cite[Corollary~7.3]{FHT2}.  Namely, for any compact Lie
group~$H$, let $H^0$~denote the component of the identity element,
$Z^0\subset H^0$ its center, and $\pi =\pi _0H$~the abelian group of
components.  Then there exists a group~$L$ that fits into the diagram
  \begin{equation}\label{eq:86}
     \begin{gathered} \xymatrix{1 \ar[r] & Z^0\ar[r]\ar[d] & L \ar[r]\ar[d] &
     \pi \ar[r]\ar@{=}[d] & 1\\ 1 \ar[r] & H^0 \ar[r]& H \ar[r]& \pi \ar[r]&1}
     \end{gathered} 
  \end{equation}
of group extensions.  Furthermore, the group~$L$ acts on~$H^0$ by
conjugation---the action descends to an action of~$\pi $ since $Z^0$~is
central, but it depends on the choice of~$L$---and the group~$H$ is
reconstructed from~$H^0$ and~$L$ as a semidirect product 
  \begin{equation}\label{eq:87}
     H \cong L\ltimes\mstrut _{Z^0}H^0 = L\ltimes H^0\bigm/ Z^0.
  \end{equation}

By the Stabilization Theorem~\ref{thm:6} we may assume that $n$~is odd, since
for $n$~even $H_n$~is obtained by pullback, so the center of~$SO_n$ is
trivial and the center of~$\Spin_n$ is~$\pmo$.  First, assume
$H^0=SH_n=SO_n\times SU_2$, so that $Z^0=\pmo$.  There are two possibilities:
$L\cong \pmo^{\times 2}$ or $L\cong \mu _4$.  We can take the image of~$L$
in~$O_n$ to be the central subgroup~$\pmo$.  The conjugation action on~$SO_n$
is trivial, and as all automorphisms of~ $SU_2$ are inner we can take the
entire action on~$H^0$ to be trivial.  Then \eqref{eq:87}~(with a direct
product in place of a semidirect product) yields the two groups~\eqref{eq:80}
and~\eqref{eq:81}.  The argument for~$H^0=\Spin_n\times \mstrut _{\pmo}SU_2$
is similar; again $Z^0\cong \pmo$.
  \end{proof}

  \subsection{Free fermions and twisted Dirac operators}\label{subsec:8.4}
 
In this section we take up the homotopy theory of relativistic free
fermions.  We treat the 10~fermionic symmetry groups simultaneously via
embeddings into Clifford algebras (\S\ref{subsubsec:8.4.1}).  For each we
define a twisted Atiyah-Bott-Shapiro map~(\S\ref{subsubsec:8.4.2}) that
encodes the index of twisted Dirac operators (\S\ref{subsubsec:8.4.3}) on
compact Riemannian manifolds.  The relativistic story begins on Minkowski
spacetime in Lorentz signature, where a free fermion theory is specified by a
real Clifford module for a Lorentz signature Clifford
algebra~(\S\ref{subsubsec:8.4.4}).  We develop that algebraic theory for the
fermionic symmetry groups, and in particular determine those theories that
admit a nondegenerate mass term (Lemma~\ref{thm:119}).  A \emph{massless}
theory has an anomaly, which is an invertible field theory, and we conjecture
its deformation class in~\S\ref{subsubsec:8.4.5}.  A formally similar setup
(\S\ref{subsubsec:8.4.6}) attaches an invertible field theory to a
\emph{massive} free fermion theory, and we conjecture that its deformation
class is given by the same formula.  It is this formula that we use in the
computations in~\S\ref{subsec:8.2}.

\bigskip
 \subsubsection{A relativistic 10-fold way}\label{subsubsec:8.4.1}

Proposition~\ref{thm:65}, Proposition~\ref{thm:60}, and
Proposition~\ref{thm:66} combine to yield $3+4+3=10$ fermionic symmetry
groups, which we arrange into two tables:
  \begin{equation}\label{eq:101}
     \begin{tabular}{ c@{\hspace{2em}} l@{\hspace{2em}} c@{\hspace{2em}}
     c@{\hspace{2em}} c@{\hspace{2em}}} 
     \toprule 
     $s$&$\;\;H^c$&$K$&Cartan&$D$\\ \midrule \\[-8pt] 
     $0$&$\Spin^c$&$\TT$&A&$\CC$\\
     $1$&$\Pin^c$&$\TT$&AIII&$\Cliff^{\CC}_{-1}$\\
     \bottomrule \end{tabular} 
  \end{equation}
  \begin{equation}\label{eq:100}
     \begin{tabular}{ c@{\hspace{2em}} l@{\hspace{2em}} c@{\hspace{2em}}
     c@{\hspace{2em}} c@{\hspace{2em}}} 
     \toprule 
     $\phantom{-}s$&$\quad\quad\; H$&$K$&Cartan&$D$\\ \midrule \\[-8pt] 
     $\phantom{-}0$&$\Spin$&$\pmo$&D&$\RR$\\
     ${-}1$&$\Pin^+$&$\pmo$&DIII&$\Cliff_{-1}$\\
     ${-}2$&$\Pin^+\ltimes\mstrut _{\{\pm1\}}\,\TT$&$\TT$&AII&$\Cliff_{-2}$\\
     ${-}3$&$\Pin^-\times\mstrut _{\{\pm1\}}SU_2$&$SU_2$&CII&$\Cliff_{-3}$\\ 
     $\phantom{-}4$&$\Spin\,\times\mstrut _{\{\pm1\}}SU_2$&$SU_2$&C&$\HH$\\
     $\phantom{-}3$&$\Pin^+\times\mstrut _{\{\pm1\}}SU_2$&$SU_2$&CI&$\Cliff_{+3}$\\ 
     $\phantom{-}2$&$\Pin^-\ltimes\mstrut _{\{\pm1\}}\,\TT$&$\TT$&AI&$\Cliff_{+2}$\\
     $\phantom{-}1$&$\Pin^-$&$\pmo$&BDI&$\Cliff_{+1}$\\
     \bottomrule \end{tabular} 
  \end{equation}
In addition to the fermionic symmetry group~$H$ or~$H^c$ and its internal
group~$K$, we list the Cartan label, an integer~$s$ called the ``type'', and
a super division algebra~$D$.  The type is defined mod~2 in~\eqref{eq:101}
and mod~8 in~\eqref{eq:100}; we choose a convenient integer representative.
We use notations $H(s), H^c(s), K(s), D(s)$ when we make the type explicit.
The Cartan label is used in the condensed matter literature, where this
10-fold way has many incarnations: see~\cite{D,AZ,HHZ,K6,SRFL,FM1,KZ,WS}.  In
those references the \emph{particle-hole symmetry} determines the internal
symmetry group~$K$: in its absence $K=\TT$; if particle-hole symmetry is
present and squares to~$+1$, then $K=\pmo$; and if particle-hole symmetry is
present and squares to~$-1$, then $K=SU_2$.  The existence (and square) of
\emph{time-reversal symmetry} in the references above matches that in our
account except for the entry~AIII, which is usually listed as not having
time-reversal symmetry (but see~\cite[\S III]{WS}).  The super division
algebra~$D$ is the unique super division algebra in the Morita class of the
Clifford algebra\footnote{The Clifford algebra $\Cliff_{\pm|s|}$ is generated
by $e_1,\dots ,e_{|s|}$ subject to $e_ae_b+e_be_a=\pm2\delta _{ab}$;
see~\cite{ABS}.}  $\Cliff_{s}$.  The groups~$\Spin^c$ and $\Pin^c$ in the
first table~\eqref{eq:101} are distinguished as having a central subgroup
isomorphic to~$\TT$, so are called \emph{complex}; the center of the groups
in~\eqref{eq:100} is~$\pmo$, and so they are called \emph{real}.

  \begin{remark}[]\label{thm:138}
 We would have found it more natural from a mathematical point of view in
several places to define $H(4)=\Spin\times \mstrut _{\pmo}\Spin_4$ rather
than $\Spin\times \mstrut _{\pmo}\Spin_3$, but we lack a physics motivation
to do so.  
  \end{remark}

The following embedding allows a uniform treatment of these symmetry groups,
and it opens a path to relating this relativistic 10-fold way to other
10-fold ways in the literature.  Fix~$n\ge0$.

  \begin{lemma}[]\label{thm:90}
 Fix a real type~$s$ as in~\eqref{eq:100}, and let $\hh ns$~denote the
$n$-dimensional version of the group~$H(s)$ of type~$s$ in
Table~\eqref{eq:100}.  Write $A_n(s)=\Cliff_{+n}\otimes D(s)$.  Then there is
an embedding
  \begin{equation}\label{eq:102}
     \iota \:\hh ns\longrightarrow A_n(s) 
  \end{equation}
such that the natural map 
  \begin{equation}\label{eq:106}
     c\:\RR^n\times A_n(s)\longrightarrow  A_n(s) 
  \end{equation}
is $\hh ns$-equivariant and graded commutes with right multiplication
by~$A_n(s)$. 
  \end{lemma}

\noindent
 Here $c$~is the extension of scalars of Clifford multiplication $\RR^n\times
\Cliff_{+n}\to \Cliff_{+n}$.  (Recall that $\RR^n\subset \Cliff_{+n}$.)  Note
that $A_n(s)$~is Morita equivalent to $\Clp {(n+s)}$; we specify a Morita
equivalence in~\S\ref{subsubsec:8.4.2}.  We regard~$\hh ns$ as an ungraded
group, and in fact $\iota (\hh ns)$~is contained in the even part of the
superalgebra~$A_n(s)$.  In the complex case~\eqref{eq:101} there is an
embedding $\iota ^{\CC}\:\Pin^c_n\hookrightarrow \Cliff^{\CC}_n\otimes
\Cliff^{\CC}_{-1}$ constructed using the same formulas as the real
case~$s=1$.  Of course, there is also the usual embedding $\iota
^{\CC}\:\Spin^c_n\hookrightarrow \Cliff^{\CC}_n$.

  \begin{proof}
 The case~$s=0$ requires no comment.  For~$s=4$ we use the fact that
$SU_2\cong Sp_1\subset \HH$.  The scalar~$-1$ passes between the factors in
the real tensor product $\Clp n\otimes \HH$, which explains the division
by~$\pmo$ in the group~$H$.  In the remaining six cases $D(s)$~is a Clifford
algebra on $|s|$~generators, and the group~$\Spin_{|s|} \subset \Cliff_s$ is
isomorphic to~$\pmo, \TT, SU_2$ for $|s| =1,2,3$, respectively.
For~$|s|=1,2$ fix a unit vector~$e\in \RR^{|s|}\subset D(s)$; for~$|s|=3$
define the volume form $\omega =e_1e_2e_3$ as the ordered product of the
generators of~$\Cliff_{|s|}$.  Define~$\iota $ by
  \begin{equation}\label{eq:103}
     \begin{aligned} g&\longmapsto g\otimes 1,\qquad &&g\in \Spin_n, \\
      g&\longmapsto \begin{cases} g\otimes e,&|s|=1,2,\\g\otimes
     \omega ,&|s|=3,\end{cases} 
     \qquad &&g\in \Pin^\pm_n\setminus \Spin_n, \\
      \lambda &\longmapsto 1\otimes \lambda ,\qquad &&\textnormal{$\lambda
      \in \TT$ or $SU_2$.}\end{aligned} 
  \end{equation}
A case-by-case check completes the proof.  To illustrate, we check the
equivariance of~$c$ for $g\in \Pin_n\setminus \Spin_n$ and $|s|=1,2$; it
suffices to take~$g=e_i$ for some standard basis element~$e_i\in \RR^n$.  For
$\xi \in \RR^n\subset \Clp n$, we have $e_i\cdot (\xi \otimes 1)=-e_i\xi
e_i\inv \otimes 1$.  For $\psi \in \Clp n$ homogeneous of parity~$|\psi |$
and $x\in D(s)$, we have $e_i\cdot (\psi \otimes x)=(-1)^{|\psi |}e_i\psi
\otimes ex$, since $e_i$~acts as left multiplication in~$A_n(s)$ by~$\iota
(e_i)$ and the Koszul sign rule applies in the superalgebra~$A_n(s)$.  Their
Clifford product is
  \begin{equation}\label{eq:104}
     -(-1)^{|\psi |}e_i\xi \psi \otimes ex = e_i\cdot (\xi \psi \otimes x), 
  \end{equation}
which proves the equivariance.  We leave the other checks to the reader.
  \end{proof}

  \begin{remark}[]\label{thm:136}
 In the condensed matter literature free fermion systems are often treated
non-relativistically and so are organized by nonrelativistic symmetry groups.
More specifically, they are organized by the subgroup~$I$ of internal vector
symmetries that fix the points of \emph{space}.  (The internal symmetry
group~$K$ in our account, which starts from a relativistic theory, is the
subgroup that fixes the points of \emph{spacetime}.)  We can easily compute
the group~$I_n$ in spacetime dimension~$n$ for a general group of symmetries,
as in~\S\ref{sec:1}.  Namely, let $\rho _n\:H_n\to O_n$ be a Wick-rotated
symmetry group.  Fix a splitting $\RR^n=\RR\times \RR^{n-1}$ of translations
of~$\EE^n$ into Wick-rotated-time translations cross spatial translations.
The subgroup $O_1\times O_{n-1}\subset O_n$ preserves that splitting, and
$O_1\times \{\id\}\subset O_1\times O_{n-1}$ is the vector subgroup of
transformations that fix space pointwise.  So for the symmetry group~$H_n$
we define the nonrelativistic internal subgroup~$I_n$ as the pullback
  \begin{equation}\label{eq:151}
     \begin{gathered} \xymatrix{I_n\ar@{^{(}->}[rr]\ar[d]_{} &&
     H_n\ar[d]^{\rho _n} \\ 
     O_1\times \{\id\}\ar@{^{(}->}[r]^{} & O_1\times O_{n-1}\ar@{^{(}->}[r] &O_n}
     \end{gathered} 
  \end{equation}
The inclusion $H_n\hookrightarrow H_{n+1}$ induces an
isomorphism~$I_n\xrightarrow{\;\cong \;}I_{n+1}$; denote the colimit of these
groups as~$I$.  We tabulate~$I$ for each of the ten fermionic symmetry groups
in Tables~\eqref{eq:101} and~\eqref{eq:100}: 
  \begin{equation}\label{eq:152}
     \begin{tabular}{ c@{\hspace{2em}} l@{\hspace{2em}} c@{\hspace{.5em}}
     l@{\hspace{2em}} c} 
     \toprule $s$&$\;\;H^c$&$I$&&Cartan\\
     \midrule \\[-8pt] $0$&$\Spin^c$&$\TT$&($\Spin^c_1$)&A\\
     $1$&$\Pin^c$&$\zt\times \TT$& ($\Pin^c_1$)&AIII\\ \bottomrule
     \end{tabular} 
  \end{equation}
  \begin{equation}\label{eq:153}
     \begin{tabular}{ c@{\hspace{2em}} l@{\hspace{2em}} c@{\hspace{.5em}}
     l@{\hspace{2em}} c} \toprule $\phantom{-}s$&$\quad\quad\;
     H$&$I$&&Cartan\\ \midrule \\[-8pt] $\phantom{-}0$&$\Spin$&$\pmo$&($\Spin_1$)&D\\
     $-1$&$\Pin^+$&$\zt\times \pmo$&($\Pp_1$)&DIII\\ $-2$&$\Pin^+\ltimes\mstrut
     _{\{\pm1\}}\,\TT$&$\zt\ltimes\TT$&($\Pp_2$)&AII\\ $-3$&$\Pin^-\times\mstrut
     _{\{\pm1\}}SU_2$&$\zmod4\times \mstrut _{\pmo}SU_2$&($\Pp_3$)&CII\\
     $\phantom{-}4$&$\Spin\,\times\mstrut 
     _{\{\pm1\}}SU_2$&$SU_2$&($\Spin_3$)&C\\ $\phantom{-}3$&$\Pin^+\times\mstrut
     _{\{\pm1\}}SU_2$&$\zt\times SU_2$&($\Pm_3$)&CI\\ $\phantom{-}2$&$\Pin^-\ltimes\mstrut
     _{\{\pm1\}}\,\TT$&$\zmod4\ltimes\mstrut _{\pmo}\TT$&($\Pm_2$)&AI\\[4pt]
     $\phantom{-}1$&$\Pin^-$&$\zmod4$&($\Pm_1$)&BDI\\ \bottomrule \end{tabular} 
  \end{equation}
In the physics literature a $\zt$ subgroup of~$I$ containing a time-reversal
symmetry, if it exists, is labeled~`$\zt_T$'.  The $\pmo$ subgroup is often
labeled~`$\zt_f$' where `$f$'~means `fermionic' since the nontrivial element
is the center of the spin group.  The groups in parentheses are abstractly
isomorphic to the group~$I$.
  \end{remark}

  \begin{remark}[]\label{thm:189}
 In the pullback~\eqref{eq:151} the group~$I_n$ has two extra pieces of
structure: the canonical central element $k_0\in K\subset I_n$ of order
dividing two (Theorem~\ref{thm:5}(2)) and a $\zt$-grading $\phi \:I_n\to
O_1=\pmo$ with $K=\ker\phi $.  In condensed matter models we are
given~$(I_n,k_0,\phi )$ and part of the determination of the low energy
effective field theory is the (re)construction of the symmetry
type~$(H_n,\rho _n)$.  We achieve this as follows.  If $\phi $~is
trivial then $I_n=K$, so set $\tSH_n\:=\Spin_n\times I_n$; then
define $H_n=SH_n$ by~\eqref{eq:16}.  If $\phi $~is surjective, consider the
commutative diagram  
  \begin{equation}\label{eq:329}
     \begin{gathered} \xymatrix{\Spin_1\ar[rr]\ar[dr]&&\Spin_n\ar[dr]\\
     &\tI_n\ar[rr]\ar[dl]\ar[dd]&&\tH_n\ar[dl]\ar[dr]\ar[dd]\\
     I_n\ar[dd]\ar[rr]&&H_n\ar[dd]&&J\ar[dd]\ar[dd]\\ 
     &\Pin_1^+\ar[dl]\ar[rr]&&\Pin_n^+\ar[dr]\\ O_1\ar[rr]&&O_n\ar[rr]&&\pmo}
     \end{gathered} 
  \end{equation}
in which every parallelogram is a pullback, the kernel of every vertical map
is~$K$, and the northeast diagonal composition is exact.
Given~$(I_n,k_0,\phi )$ define~$\tI_n$ by pullback, set~$K=\ker\phi $, set
$J=\tI_n/\Spin_1$, let~$\tH_n$ be the pullback~\eqref{eq:226}, and
define~$H_n$ using~\eqref{eq:227}. 
  \end{remark}

\bigskip 
\subsubsection{Twisted Atiyah-Bott-Shapiro map}\label{subsubsec:8.4.2}

Atiyah-Bott-Shapiro~\cite[\S11]{ABS} give a canonical construction of
$K$-theory elements on Thom complexes.  The universal
incarnation~\cite[\S6.1]{H} is a map of spectra 
  \begin{equation}\label{eq:137}
     \phi \:M\!\Spin\longrightarrow KO. 
  \end{equation}
Following their arguments we produce similar maps for the group~$H(s)$ of
type~s in Table~\eqref{eq:100}.  Fix a dimension~$n\in \ZZ^{\ge0}$.
 
As a first step we stipulate a Morita equivalence  
  \begin{equation}\label{eq:138}
     A_n(s)\;\underset{\textnormal{\tiny Morita}}{\approx }\Clp{(n+s)}. 
  \end{equation}
There is a sign at stake---for any Clifford algebra~$A$ the groupoid of
invertible $(A,A)$-bimodules is equivalent to the groupoid of $\zt$-graded
lines: the sign is the parity of the line.  Define the isomorphism
  \begin{equation}\label{eq:139}
     \Clp n\otimes \Clp s\xrightarrow{\;\;\;\cong
     \;\;\;}\Clp{(n+s)} 
  \end{equation}
as in \cite[(1.6)]{ABS}, and choose \cite[(6.9)]{ABS} a
$\Cliff_{\pm8}$-module $M=M^0\oplus M^1$ of dimension~$8|8$ such that the
volume form acts as~$+1$ on~$M^0$.  There result Morita
equivalences~\eqref{eq:138} for all cases except~$s=4$.  For that we fix a
\emph{quaternionic} $\Cliff_{\pm4}$-module $N=N^0\oplus N^1$ of quaternionic
dimension~$1|1$ such that the volume form acts as~$+1$ on~$N^0$.

Now to the twisted ABS construction.  Let $\pi \:V_n\to B\hh ns$ be the
universal bundle associated to $\rho _n\:\hh ns\to O_n$.  Define the spinor
bundle\footnote{Our choice of~$A\op$ in~\eqref{eq:122}, rather than~$A$, is
essentially a sign choice.  We use a geometric model~\cite{AS} in which a
class in~$KO^m(X)$ is represented by a $\zt$-graded vector bundle over~$X$
that is a \emph{left} module for~$\Cliff_m$ equipped with a family of
commuting odd \emph{skew}-adjoint (Fredholm) operators.\label{skewsign}}
  \begin{equation}\label{eq:122}
     \sS:=E\hh ns\times \mstrut _{\hh ns}A_n(s)\op\longrightarrow B\hh ns;
  \end{equation}
This is a vector bundle of right $A_n(s)\op$-modules, or equivalently of left
$A_n(s)$-modules.  Left Clifford multiplication~\eqref{eq:106} defines a
family of odd skew-adjoint endomorphisms of $\pi ^*\sS\to V_n$.  These
operators are invertible off the zero section, and they commute with the left
$A_n(s)$-module structure.  Therefore, using the Morita
equivalence~\eqref{eq:138}, they define an element in
$KO^{n+s}\bigl(\Thom(B\hh ns;{V_n})\bigr)$, where $\Thom(B\hh ns;{V_n})$ is
the Thom space of the universal bundle $\pi \:V_n\to B\hh ns$.  Take the
limit $n\to\infty $ after subtracting a trivial rank~$n$ bundle from~$V_n$ to
obtain
  \begin{equation}\label{eq:107}
     \phi \:MH(s)\longrightarrow \Sigma ^{s}KO 
  \end{equation}
out of the Thom spectrum associated to the stable \emph{normal}
structure~$H$.  For $s=0$ this is the Atiyah-Bott-Shapiro (ABS)
map~\cite[\S6.1]{H}.  We rewrite in terms of the stable \emph{tangential}
structure~$H$; see the comments following~\eqref{eq:90}.  That perp maneuver
exchanges~$\Pp$ and~$\Pm$, which in Table~\eqref{eq:100} exchanges
$s\leftrightarrow -s$.  Therefore, \eqref{eq:107} is a generalized ABS map
  \begin{equation}\label{eq:108}
     \phi \:MTH(s)\longrightarrow \Sigma ^{-s}KO. 
  \end{equation}
In the complex case we obtain a generalized ABS map 
  \begin{equation}\label{eq:184}
     \phi \:MTH^c(s)\longrightarrow \Sigma ^{-s}K. 
  \end{equation}

\bigskip 
\subsubsection{Twisted Dirac operators}\label{subsubsec:8.4.3}

Next, following~\cite[\S II.7]{LM}, we define twisted Dirac operators for
the structure groups in Table~\eqref{eq:100}.  Suppose $X$~is an
$n$-dimensional Riemannian manifold equipped with an $\hh ns$-structure $P\to
X$.  We assume given a connection on $P\to X$ compatible with the Levi-Civita
connection on the orthonormal frame bundle.  Use the embedding~\eqref{eq:102}
to form the $\zt$-graded spinor bundle
  \begin{equation}\label{eq:105}
      \sS':=P\times \mstrut _{\hh ns}A_n(s)\longrightarrow  X.
  \end{equation}
Clifford multiplication~\eqref{eq:106} defines a vector bundle map
$T^*X\otimes \sS'\longrightarrow \sS'$, and as usual the Dirac
operator~$\Dirac_X$ acts on smooth sections of~$\sS'$ as the covariant
derivative followed by Clifford multiplication.  The Dirac operator is odd
and skew-adjoint.  (See~\footref{skewsign} for our conventions.)  It commutes
with the {right} $A_n(s)$-module structure on~$\sS'$, or equivalently with
the {left} $A_n(s)^{\textnormal{op}}$-module structure.

There are topological and geometric indices of Dirac operators on compact
manifolds.  The topological index is defined using Fredholm
operators~\cite{AS}.  Namely, if $X$~is closed, then $\Dirac_X$~extends to a
Fredholm operator on Sobolev completions of the space of smooth sections
of~$\sS'$.  This construction works in families: from a fiber bundle $\sX\to
S$ of closed Riemannian $n$-manifolds with $\hh ns$-structure we obtain a
family of odd skew-adjoint Fredholm operators parametrized by~$S$.  Recalling
that $A_n(s)^{\textnormal{op}}$~is Morita equivalent to~$\Clm{(n+s)}$,
via~\eqref{eq:138}, we deduce that this family of operators has a
\emph{topological} index that lies in~$KO^{-(n+s)}(S)$.  For~$s=0$ this
reduces to the usual Clifford-linear Dirac operator definition of the
topological index.  The Atiyah-Singer index theorem equates this topological
index with an analytic index.  If $S$~is a smooth manifold and $\sX\to S$ a
smooth family of Riemannian manifolds with $\hh ns$-structure, then there is
a \emph{geometric} index that lies in the differential cohomology
group~$\widehat{KO}^{-(n+s)}(S)$; see~\cite{FL} for the differential complex
$K$-theory version as well as the Atiyah-Singer theorem in this differential
context.

  \begin{remark}[]\label{thm:137}
 For~$s=\pm1$ this discussion specializes to an effective approach to Dirac
operators and index theory on unoriented manifolds with a $\Ppm$-structure.
  \end{remark}

  \begin{remark}[]\label{thm:139}
 There is an analogous discussion in the complex case: replace~$H\to H^c$
and $KO\to K$. 
  \end{remark}

\bigskip 
\subsubsection{Free fermion theories on Minkowski
spacetime~$M^{n-1}$}\label{subsubsec:8.4.4}

As before we only treat the eight real fermionic symmetry groups.  Fix a
type~$s$ in Table~\eqref{eq:100}.  Let $\GG{1,n-2}s$~be the Lorentz signature
anti-Wick rotation of~$\hh{n-1}s$, as in~\eqref{eq:12}.  If $s=0$, which is
the basic case, then $\GG{1,n-2}s=\Spin_{1,n-2}$ is the Lorentz spin group.
The analog of~\eqref{eq:102} is an embedding (see~\eqref{eq:a9} for
$\Cliff_{p,q}$ conventions). 
  \begin{equation}\label{eq:140}
     \iota \:\GG{1,n-2}s\longrightarrow \Cliff_{n-2,1}\otimes\, D(s)=:B_{n-1}(s), 
  \end{equation}
and there is a Morita equivalence of superalgebras 
  \begin{equation}\label{eq:144}
     B_{n-1}(s)\;\underset{\textnormal{\tiny Morita}}{\approx }\Clp{(n-3+s)}. 
  \end{equation}
We use the conventions following~\eqref{eq:138} to define the Morita
equivalence.  The image of~$\iota $ lies in the even subalgebra $B_{n-1}(s)^0\subset
B_{n-1}(s)$.  A free fermionic field is specified by a real spinor representation
of~$\GG{1,n-2}s$, which by definition is an \emph{ungraded} real module~$\SS$
of~$B_{n-1}(s)^0$.  A spinor field is then a function $\psi \:M^{n-1}\to\SS$.

  \begin{remark}[]\label{thm:171}
 The CRT~theorem, which is reviewed in Appendix~\ref{sec:10}, implies that
the free fermion theory has a larger Lie group~$\GG{1,n-2}s^\beta \supset
\GG{1,n-2}s$ of symmetries; the non-identity component acts antilinearly on
the Hilbert space of states.  Proposition~\ref{thm:169}(3) implies that the
embedding~\eqref{eq:140} extends to $\GG{1,n-2}s^\beta $, and so
$\GG{1,n-2}s^\beta $~acts on the \emph{real} vector space~$\SS$, consistent
with Proposition~\ref{thm:170}(2).
  \end{remark}

We quickly summarize special facts about a real spinor representation~$\SS$
of the Lorentz spin group~$\Spin_{1,n-2}$; proofs may be found
in~\cite[\S6]{De}.  Fix a component~$C$ of timelike vectors~$\xi \in
\RMm$ with $|\xi|^2>0$.  The first special property is the existence of
\emph{symmetric} $\Spin_{1,n-2}$-invariant maps
  \begin{equation}\label{eq:141}
     \Gamma \:\SS\times \SS\longrightarrow \RMm .
  \end{equation}
If $\SS$~is irreducible, then $\Gamma $~is unique up to a real factor and
nonzero~$\Gamma $ are definite.  Choose~$\Gamma $ \emph{positive} definite in
the sense that $\Gamma (\psi ,\psi )\in \overline{C}$ for all~$\psi \in \SS$.
This fixes~$\Gamma $ up to a positive real factor.  There are two isomorphism
classes of real irreducible representations for $n-1\equiv 2,6\pmod8$ and a
unique irreducible in other cases.  Let~$S_1,S_2$ be representative
irreducibles (in dimensions with a unique irreducible, set $S_2=0$); let
$Z$~be the commutant of the spin action, so $Z=\RR$, $\CC$, or~$\HH$; and fix
positive definite~$\Gamma $ for ~$S_1,S_2$.  A general real spinor
representation~$\SS$ decomposes as
  \begin{equation}\label{eq:155}
     \SS\cong W_1\otimes _ZS_1\;\oplus \; W_2\otimes _Z S_2 
  \end{equation}
for right $Z$-modules $W_1,W_2$.  Then positive definite pairings~$\Gamma $
in~\eqref{eq:141} correspond to positive definite hermitian forms
on~$W_1,W_2$.  For each choice there is a unique compatible $\zt$-graded
$\Cliff_{n-2,1}$-module structure on~$\SS\oplus \SS^*$, where $\SS$~is in
even degree and $\SS^*$~in odd degree; in particular, the duality pairing
$\SS^*\otimes \SS\to\RR$ is $\Spin_{1,n-2}$-invariant.  Conversely, if
$\SS^0\oplus \SS^1$ is a $\Cliff_{{n-2,1}}$-module, then there is a duality
pairing $\SS^0\otimes \SS^1\to\RR$ that makes the resulting symmetric
form~\eqref{eq:141} positive definite.  (Deligne proves this for simple
modules in~\cite[(6.1)]{De}; any module is a sum of simples and the argument
applies to each summand.)  Observe that $\Gamma $~is a contractible choice.
 
The group~$\GG{1,n-2}s$ contains the spin group~$\Spin_{1,n-2}$ as a subgroup
and the quotient~$Q_{n-1}(s)$ is compact and independent of~$n$ up to
isomorphism.  An irreducible real representation of~$\GG{1,n-2}s$ decomposes
under the subgroup~$\Spin_{1,n-2}$ as~\eqref{eq:155}, and a central
extension~$\widehat{Q_{n-1}(s)}$ of~$Q_{n-1}(s)$ acts on each~$W_i$.  A
choice of $\widehat{Q_{n-1}(s)}$-invariant positive definite hermitian form
on~$W_i$ yields a $\GG{1,n-2}s$-invariant pairing~\eqref{eq:141}, and then a
$B_{n-1}(s)$-module $\SS\oplus \SS^*$.  Conversely, every $B_{n-1}(s)$-module
has this form.

  \begin{definition}[]\label{thm:140}
 The module~$\SS$ \emph{admits a mass term} if there is a
\emph{nondegenerate} skew-symmetric $\GG{1,n-2}s$-invariant bilinear form
  \begin{equation}\label{eq:142}
     m\:\SS\times \SS\longrightarrow \RR .
  \end{equation}
  \end{definition}

\noindent
 We call~$m$ the \emph{mass form}.  

  \begin{lemma}[]\label{thm:119}
 $\SS$~admits a mass term if and only if $\SS\oplus \SS^*$ extends to a
super module of the superalgebra~$B_{n-1}(s)[e]$, where $e$~is odd, $e^2=-1$, and
$e$~ (graded) commutes with the Clifford generators of~$B_{n-1}(s)$.
  \end{lemma}

\noindent
 If $s=4$ the hypothesis is that $e$~commutes with~$D=\HH$.  As always, the
commutation with Clifford generators obeys the Koszul sign rule.

  \begin{proof} 
 Given a $B_{n-1}(s)[e]$-module structure on~$\SS\oplus \SS^*$,
define~$m$ by
  \begin{equation}\label{eq:156}
     m(s_1,s_2)=\langle Es_1,s_2 \rangle,\qquad s_1,s_2\in \SS, 
  \end{equation}
where $E\:\SS\to\SS^*$ is part of the action of~$e=\left(\begin{smallmatrix}
0&-E\inv \\E&0 \end{smallmatrix}\right)$ on~$\SS\oplus \SS^*$.
Since~$e^2=-1$, the form~$m$ is nondegenerate, and since $e$~(graded)
commutes with $B_{n-1}(s)$, the form~$m$ is $H_{1,n-2}(s)$-invariant.  We
must prove that $m$~ is skew-symmetric.  It suffices to assume that
$\SS\oplus \SS^*$ is a simple $B_{n-1}(s)[e]$-module, since any module is a
direct sum of simples.  Then $m$~is either symmetric or skew-symmetric.  Let
$f\in \RMm\subset \Cliff_{n-2,1}\subset B_{n-1}(s)$ be the Clifford generator
with $f^2=-1$.  So $f$~ is a timelike vector, and we choose it to lie in~$C$.
Write $f=\left(\begin{smallmatrix} 0&-F\inv \\F&0 \end{smallmatrix}\right)$
for its action on~$\SS\oplus \SS^*$.  The positive definiteness of~$\Gamma $
implies that
  \begin{equation}\label{eq:157}
     (s_1,s_2)_{\SS}:=\langle Fs_1,s_2 \rangle,\qquad s_1,s_2\in \SS, 
  \end{equation}
is a positive definite inner product on~$\SS$.  The mass form is
$m(s_1,s_2)=(F\inv Es_1,s_2)_{\SS}$.  Set $A=F\inv E\in \End(\SS)$.  Since
$m$~is either symmetric or skew-symmetric, either $A^*=A$ or~$A^*=-A$, where
$*$~is with respect to the inner product~\eqref{eq:157}.  But $ef=-fe$
implies $A^2=-\id_{\SS}$, which rules out~$A^*=A$ since $A^*A$~is a
nonnegative operator. 
 
Conversely, let $m$~be a mass form.  Using the inner product~\eqref{eq:157}
write 
  \begin{equation}\label{eq:324}
     m(s_1,s_2)=(Bs_1,s_2)_{\SS}, \qquad s_1,s_2\in \SS,
  \end{equation}
for an invertible skew-symmetric operator $B\:\SS\to\SS$.  Define
$P=\sqrt{B^*B}$ and $A=P\inv B=BP\inv $.  Then set~$E=FA$ and let $e\in
B_{n-1}(s)[e]$ act on~$\SS\oplus \SS^*$ via $\left(\begin{smallmatrix}
0&-E\inv \\E&0 \end{smallmatrix}\right)$, where as above $f\in B_{n-1}(s)[e]$
acts as $\left(\begin{smallmatrix} 0&-F\inv \\F&0 \end{smallmatrix}\right)$.
We must check that this determines a well-defined action of $B_{n-1}(s)[e]$.
It is easy to verify that $e^2=-\id_{\SS\oplus \SS^*}$, and $ef=-fe$ follows
from $F\inv E=-E\inv F$, which in turn follows from $A=-A\inv $.  For later
use we observe the commutation relation $PF\inv E=F\inv EP$.  Let\footnote{We
leave the reader to give the appropriate modification for~$s=4$.} $c\in
\RMm\oplus \RR^{|s|}\subset B_{n-1}(s)$ be a vector perpendicular to~$f$, and
write its action on the module~ $\SS\oplus \SS^*$ as
$\left(\begin{smallmatrix} 0&\pm C\inv \\C&0 \end{smallmatrix}\right)$, the
sign determined according as $c^2=\pm1$ in~$B_{n-1}(s)$.  It remains to show
that $ec=-ce$ as operators on $\SS\oplus \SS^*$, or equivalently that
  \begin{equation}\label{eq:325}
     (EC\inv )^2=\pm\id_{\SS}. 
  \end{equation}
First, we use \eqref{eq:156}--\eqref{eq:324} to write 
  \begin{equation}\label{eq:326}
     m(s_1,s_2) = \langle FBs_1,s_2 \rangle= \langle EPs_1,s_2 \rangle,\qquad
     s_1,s_2\in \SS. 
  \end{equation}
Since $cf=-fc$ in~$B_{n-1}(s)$ we have $C\inv F=\pm F\inv C$.  Next, $cf\in
H_{1,n-2}(s)\subset B_{n-1}(s)$ preserves the duality pairing $\SS^*\otimes
\SS\to\RR$, from which
  \begin{equation}\label{eq:327}
     \langle CF\inv s^*,C\inv Fs \rangle=\mp\langle s^*,s \rangle ,\qquad
     s^*\in \SS^*,\quad  s\in \SS. 
  \end{equation}
Now since $m$~is $H_{1,n-2}(s)$-invariant, 
  \begin{equation}\label{eq:328}
     m(C\inv Fs_1,C\inv Fs_2) = m(s_1,s_2),\qquad s_1,s_2\in \SS. 
  \end{equation}
Use the first expression in~\eqref{eq:326} together with the previous
identities to conclude that $C\inv FB=-BC\inv F$.  It follows that $C\inv
F$~\emph{commutes} with~$P$.  Then rewrite~\eqref{eq:328} using the second
expression in~\eqref{eq:326} to deduce $FC\inv EPC\inv F=\mp EP$.  Apply the
foregoing to arrive at~\eqref{eq:325}. 
  \end{proof}

There is an abelian group law on free fermion theories: direct sum of
Clifford modules~$\SS$.  The relationship~\cite[(11.4)]{ABS},
\cite[p.~383]{A2} between Clifford modules and $K$-theory yields the
following.

  \begin{theorem}[]\label{thm:120}
 The abelian group of relativistic free fermion field theories in
dimension~$n-1$ with type~$s$, modulo those that admit a mass term, is
isomorphic to
  \begin{equation}\label{eq:143}
     KO^{n-3+s}(\pt)\cong \pi \mstrut _{3-s-n}(KO). 
  \end{equation}
  \end{theorem}

\noindent 
 Massive free fermions are anomaly-free; see~\cite[\S1.2]{W1} for a recent
exposition.  So the map from a free fermion theory to the isomorphism class
of its anomaly factors through the quotient~\eqref{eq:143}.

  \begin{remark}[]\label{thm:121}
 The nature of an irreducible real twisted spin representation~$\SS_0$
depends on the value of~$t=n-1+s\pmod8$.  We ask if it is self-conjugate---if
$\SS_0^*\cong \SS_0$---and if so whether the induced nondegenerate bilinear
form $\SS_0\otimes \SS_0\to\RR$ is symmetric ($\SS_0$~orthogonal) or
skew-symmetric ($\SS_0$~symplectic).  Also, the commutant is a real
division algebra, so is isomorphic to~$\RR$, $\CC$, or~$\HH$.  We list the
types.  If $t\equiv 3,4,7$, then $\SS_0$~is symplectic, and the commutant is
$\RR,\CC,\HH$, respectively.  If $t\equiv 0,1,5$, then $\SS_0$~is orthogonal
and the commutant is~$\CC,\RR,\HH$, respectively.  If $t\equiv 2,6$, then
there are two nonisomorphic irreducible spin representations that are each
others dual; the commutant is~$\RR,\HH$, respectively.  For $t\equiv 3,4,7$
the $K$-group~\eqref{eq:143} vanishes, as it must since there is always a
mass term.  For $t\equiv 0,1$ the $K$-group is isomorphic to~$\zt$---the
direct sum of two copies of the irreducible module admits a mass term---and
for~$t\equiv 5$ it vanishes.  For~$t\equiv 2,6$ the $K$-group is isomorphic
to~$\ZZ$.  These are the cases for which the anomaly theory is not
topological.
  \end{remark}

\bigskip 
\subsubsection{The anomaly theory and its deformation
class}\label{subsubsec:8.4.5} 

Our starting point is the $B_{n-1}(s)^0$-module~$\SS$ that defines a free
fermion theory on Minkowski spacetime~$M^{n-1}$ in $(n-1)$~dimensions, as in
~\S\ref{subsubsec:8.4.4}.  In this subsection we sketch the associated
$n$-dimensional anomaly theory, an invertible field theory in $n$~dimensions.
(See~\cite{F3}, \cite[\S11]{F4} for expositions of anomalies from this
viewpoint.)  The anomaly theory is not necessarily topological, but it has a
deformation class that is topological---or which can be regarded as a
continuous invertible topological theory---and we propose a general formula
for it.  See~\cite{W1} for a discussion of many special cases from a more
physical viewpoint.
 
First, the real representation~$\SS$ of~$\GG{1,n-2}s$ extends to a complex
representation~$\SS_{\CC}$ of the complexification~$\GG{1,n-2}s(\CC)$, which
then restricts to a complex representation of~$\hh{n-1}s$.  On a curved
Riemannian manifold~$X^{n-1}$ with differential $\hh{n-1}s$-structure $P\to
X$ there is an associated complex vector bundle $P\times
\mstrut_{\hh{n-1}s}\SC\to X$ whose sections are complex spinor fields.  There
is a Wick-rotated Dirac lagrangian, possibly with mass term, which is a
skew-symmetric form on the space of spinor fields.  If $X$~is closed, then
the fermionic functional integral over the space of spinor fields is the
pfaffian of the Dirac operator on~$X$.  In a smooth family $\sX\to S$ the
pfaffian is not a function, but rather is a section of the \emph{pfaffian
line bundle}
  \begin{equation}\label{eq:145}
     \Pfaff_{\sX/S}\longrightarrow S. 
  \end{equation}
The bundle $\Pfaff_{\sX/S}\to S$ carries a canonical hermitian metric and
compatible covariant derivative; it is $\zt$-graded by the mod~2 index.  It
is part of the anomaly theory associated to the module~$\SS$.
 
We now give a conjectural description of the entire anomaly theory.
Fix~$k\in \ZZ^{\ge 0}$, which is the codimension in the $n$-dimensional
theory.  Let $X^{n-k}$ be a closed $(n-k)$-dimensional Riemannian manifold
with differential $\hh{n-k}s$-structure.  The universal Dirac operator
(\S\ref{subsubsec:8.4.3}) acts on sections of a real vector bundle $\sS'\to
X$ of left $A_{n-k}(s)\op$-modules, where $A_{n-k}(s)=\Clp{(n-k)}\otimes
D(s)$ is Morita equivalent to $\Clp{(n-k+s)}$; see~\eqref{eq:138}.  Let
$\underline{\SS\oplus \SS^*}\to X$ be the constant vector bundle with fiber
$\SS\oplus \SS^*$.  Then $\sS'\otimes \mstrut _{\RR}(\underline{\SS\oplus
\SS^*})\to X$ is a real vector bundle of $\zt$-graded $A_{n-k}(s)\op\otimes
B_{n-1}(s)$-modules.  Our conventions in~\S\ref{subsubsec:8.4.2} give a
definite Morita equivalence $ A_{n-k}(s)\op\otimes
B_{n-1}(s)\underset{\textnormal{\tiny Morita}}{\approx }\Clm{(3-k)}$.  For a
family $\sX\to S$ the geometric index of the Dirac operator\footnote{Some
details of this construction appear in \cite[Appendix]{FH2}} with
coefficients in $\sS'\otimes \mstrut _{\RR}(\underline{\SS\oplus \SS^*})$
lies in the differential cohomology group~$\widehat{KO}^{-(3-k)}(S)$.  Notice
that it is independent of~$n$ and~$s$.  The anomaly picks off the lowest
piece of the index via the canonical \emph{Pfaffian homomorphism}
  \begin{equation}\label{eq:146}
     \Pfaff\:\widehat{KO}^{-(3-k)}(S)\longrightarrow \widehat{\IZ}^{1+k}(S). 
  \end{equation}
The invariants in differential~ $\IZ$ fit together into an invertible field
theory; see~\cite{HS}.

  \begin{example}[]\label{thm:122}
 For $k=0$, so $\sX\to S$ of relative dimension~$n$, there is an isomorphism
$\widehat{\IZ}^1(S)\cong \widehat{H}^1(S)\cong \Map(S,\TT)$.  The
corresponding lowest piece of the index is the partition function~$e^{2\pi
i(\xi /2)}$ of the anomaly theory on an $n$-manifold, where $\xi $~is the
Atiyah-Patodi-Singer invariant~\cite{APS}.  The division by~2 is due to the
skew-symmetry of the Dirac form, the same division by~2 that passes from
determinant to pfaffian.  The equality between the exponentiated $\xi
$-invariant and the integral in differential $K$-theory has only been proved
in a basic case~\cite{Klo,O,BuS,FL} as far as we know.
  \end{example}

  \begin{example}[]\label{thm:123}
 For $k=1$, so $\sX\to S$ of relative dimension~$n-1$, the group
$\widehat{\IZ}^2(S)$ is isomorphic to the group of isomorphism classes of
$\zt$-graded hermitian line bundles $L\to S$ with compatible covariant
derivative.  For the anomaly theory that element is the pfaffian line bundle
$\Pfaff_{\sX/S}\to S$.  The main theorem in~\cite{DF} is the gluing law in the
\emph{non-extended} invertible field theory in dimensions~$n-1,n$ with
partition function the exponentiated $\xi $-invariant. 
  \end{example}

The story continues to lower dimensional manifolds, on which the invariants
are graded gerbes~\cite{Lo,Bu} and higher analogs. 
 
The deformation class of an invertible field theory gotten from integration
in differential cohomology is the underlying topological cohomology theory.
In the background are techniques from~\cite{HS}, which lead to the
following.

  \begin{conjecture}[]\label{thm:124}
 Fix a type~$s$ in Table~\eqref{eq:100} and a dimension~$n$.  Fix an
isomorphism class of free fermion theories modulo those that admit a mass
term, i.e., an element $[\SS]\in \pi \mstrut _{3-s-n}(KO)$.  Then the
deformation class of the $n$-dimensional anomaly theory is the homotopy class
of the composition
  \begin{equation}\label{eq:148}
   MTH(s)\xrightarrow{\;\;\phi \wedge [\SS]\;\;} \Sigma ^{-s}KO\wedge \Sigma
   ^{-3+s+n}KO \xrightarrow{\;\;\mu \;\;}\Sigma
   ^{n-3}KO\xrightarrow{\;\;\Pfaff\;\;} \SIZ,
  \end{equation}
where $\phi $~is the Atiyah-Bott-Shapiro map~\eqref{eq:108}, $\mu $~is
multiplication in the ring spectrum~$KO$, and $\Pfaff$~is the topological
version of~\eqref{eq:146}. 
  \end{conjecture}

\noindent
 There is a similar conjecture in the complex case~\eqref{eq:101} with the
usual replacements $H\to H^c$ and $KO\to K$.  We hope to address this
conjecture in the future.  We use it in our computations below.

  \begin{remark}[]\label{thm:125}
 If the group~$\pi _{3-s-n}(KO)$ is finite, hence is isomorphic to~$\zt$,
then there is a reflection positive invertible \emph{topological} field
theory in the deformation class whose partition function is the mod~2 index.
If the group is free cyclic, hence isomorphic to~$\ZZ$, then the deformation
class is represented by a reflection positive invertible field theory whose
partition function is the exponentiated $\xi $-invariant of
Atiyah-Patodi-Singer, the secondary invariant for a $\ZZ$-valued topological
index in $n+1$~dimensions.  This is the case in which there are local
anomalies as well as global anomalies, and because of the shift~$s$ it
happens in both even and odd dimensions.
  \end{remark}

\bigskip \subsubsection{Massive free fermion theories}\label{subsubsec:8.4.6}

In~\S\ref{subsubsec:8.4.5} we explained how a free fermion theory in
$(n-1)$~dimensions has an associated $n$-dimensional invertible anomaly
theory, and Conjecture~\ref{thm:124} states its deformation class.  Here we
show that a second scenario leading to invertible $n$-dimensional theories
has the same starting data.  \emph{This} is the scenario we apply
in~\S\ref{subsec:8.2}.  Namely, begin with a \emph{massive} free fermion
theory in $n$~dimensions.  Because the theory has a mass gap its long-range
physics is described by a field theory, which naturally is also
$n$-dimensional.  As argued in~\S\ref{subsec:4.4} we expect that theory to
be, at least locally, the product of a topological theory and an invertible
theory.  But a massive free fermion theory has a unique vacuum on each
spatial manifold---the vacuum in the fermionic Fock space---so in fact the
long-range effective theory is invertible.

  \begin{remark}[]\label{thm:190}
 One must make choices to define the massive free fermion theory, and they
can be summarized as a trivialization of an anomaly; see~\cite[\S11]{F4} for
a general discussion.  There is a canonical choice for each fixed mass, and
it is implicitly used in the discussion below as well as
in~\S\ref{subsec:8.2}.  However, when the mass is a not necessarily constant
function then there is an anomaly; see~\cite{CFLS} for discussion and
details.
  \end{remark}

As in previous sections fix a type~$s$ in Table~\eqref{eq:100} and let $\GG
{1,n-1}s$~be the Lorentz signature anti-Wick rotation of the corresponding
group~$\hh ns$.  In the notation of~\eqref{eq:140} there is an embedding $\GG
{1,n-1}s\hookrightarrow B_{n-1}(s)[e']$, where $e'$~is an extra Clifford
generator with~$(e')^2=+1$.  By Lemma~\ref{thm:119} spinor representations
of~$\GG {1,n-1}s$ that admit a mass term are in bijection with super modules
over the superalgebra $B_{n-1}(s)[e',e]$, where $e$~is an extra Clifford
generator with~$e^2=-1$.  Observe that $B_{n-1}(s)[e',e]$~is Morita
equivalent to~$\Clp{(n-3+s)}$.  We speculate that
  \begin{equation}\label{eq:149}
     \textnormal{\vtop{\hbox{(i)~the resulting low energy theory is trivial
     if the 
     $B_{n-1}(s)[e',e]$-module is}\hbox{\quad\, extended to a module over the
     algebra 
     $B_{n-1}(s)[e',e,f]$ with $f^2=-1$.}}} 
  \end{equation}
The group of equivalence classes of $B_{n-1}(s)[e,f]$-modules modulo those
that extend is the $K$-group~\eqref{eq:143}.  The Morita equivalence to
massless theories in dimension~$n-1$ and the vanishing of the anomaly for
theories that admit a mass term are evidence in favor of~\eqref{eq:149}.
Furthermore, we speculate that
  \begin{equation}\label{eq:150}
     \textnormal{(ii)~the low energy theory is
     invertible and its deformation class
     is~\eqref{eq:148}.}  
  \end{equation}
As some evidence supporting~(ii) we point out that the partition function in
special cases is computed in~\cite[\S2.1.6, \S2.2.3, \S3.4, \S4.3, \S5]{W1}.
The universal part of the partition function of the low energy theory is an
exponentiated $\xi $-invariant, as in Example~\ref{thm:122}.

  \subsection{Phases of topological insulators and topological
  superconductors}\label{subsec:8.2}

We apply Conjecture~\ref{thm:143} to compute possible topological phases for
each of the 10~fermionic symmetry types~\eqref{eq:101} and~\eqref{eq:100}.
We remind that the fermionic symmetry groups with~$K=\TT$ pertain to
{topological insulators}; those with $K=\pmo$ and $K=SU_2$ pertain to
{topological superconductors}.  The abelian group of topological
phases---that is, the group of deformation classes of reflection positive
invertible topological field theories with symmetry group~$H$ in $n$
spacetime dimensions---is
  \begin{equation}\label{eq:111}
     \TP nH:=[MTH,\Sigma ^{n+1}\IZ].
  \end{equation}
It may be computed from the homotopy groups\footnote{These are Thom's bordism
groups, but for the perpendicular tangential structure on the stable normal
bundle (see footnote~\footref{perp}).  Note that $\Pp/\Pm$ and
$\Pcp/\Pcm$~exchange when passing from tangential to normal.}~$\pi _qMTH$;
see the universal property~\eqref{eq:62}.  As we are only interested in~$n\le
5$, we need only compute for~$q\le6$, and for~$q=6$ we only need to know
$\pi_6MTH/\textnormal{torsion}$, since that determines $\Hom(\pi _6MTH,\ZZ)$.
The abelian group~$\TP nH$ classifies deformation classes of
\emph{interacting} theories.  The abelian group of deformation classes of
massive (gapped) \emph{free} fermion theories in $n$~dimensions modulo those
with trivial long-range effective theory is given by Lemma~\ref{thm:119}
and~\eqref{eq:149}, at least conjecturally:
  \begin{equation}\label{eq:112}
      \FFF n{H(s)}:=\begin{cases} \pi \mstrut
      _{3-s-n}(K)\phantom{O}\cong [\Sigma 
      ^{-s}K,\SIZ], &\textnormal{$H^c(s)$ a
      \emph{complex} symmetry type}; \\ \pi \mstrut
      _{3-s-n}(KO)\cong  [\Sigma ^{-s}KO,\SIZ] 
      ,&\textnormal{$H(s)$ a \emph{real}  symmetry
      type},\end{cases}
  \end{equation}
where $s$~is the parameter in~\eqref{eq:101} or~\eqref{eq:100}.  (See
Remark~\ref{thm:121} for an enumeration of the $K$-theory groups in the real
case via the types of spin representation.)  According to \eqref{eq:150}
and~\eqref{eq:148}  the natural homomorphism  
  \begin{equation}\label{eq:110}
     \Phi \:\FFF nH\longrightarrow \TP nH
  \end{equation}
from the group of deformation classes of free fermion theories to the group
of all theories is the product with the ABS map~\eqref{eq:108}.  We
compute~$\Phi $ for each symmetry class.
 
The results are organized by internal symmetry group.  Some of the bordism
groups appear in the mathematics literature, whereas for the more exotic
symmetry groups the computations are new.  With the bordism groups in hand,
the classification of interacting theories is an immediate consequence of
Conjecture~\ref{thm:143} and the universal property expressed in the short
exact sequence~\eqref{eq:62}.  The free fermion computation
is~\eqref{eq:143}.  The map~ \eqref{eq:110} from massive free fermion phases
to interacting phases does not follow from the rest---it must also be
computed.  We give a uniform treatment based on Lemma~\ref{thm:90}
and~\S\ref{subsubsec:8.4.2}.  Manifold generators and formulas for partition
functions in 4~dimensions are worked out in~\cite{GPW}.

We check our computations against the condensed matter literature, where
groups of SPT phases are deduced using very different arguments.  There is
almost total agreement, and in the few places we differ we use the homotopy
computations to predict what should happen in the physics.  The computations
that we did not find in the physics literature should be considered
predictions.

 \bigskip
 \subsubsection{Internal symmetry group $K=\pmo$} 
 The symmetry groups are classified in Proposition~\ref{thm:65}.  The low
degree spin and pin bordism groups are described in a geometric way
in~\cite{KT1}.  The general structure of spin bordism is elucidated
in~\cite{ABP1}.  The computation of pin bordism groups in all degrees may be
found in~\cite{ABP2} and~\cite{KT2}.

  \begin{theorem}[]\label{thm:69}
  The low degree bordism groups for~$K=\pmo$ are:
  \begin{equation}\label{eq:84}
     \begin{tabular}{ c@{\hspace{2em}} c@{\hspace{2em}} c@{\hspace{2em}}
     c@{\hspace{2em}} } 
     \toprule q&$\pi _qMT\!\Spin$&$\pi _qMT\!\Pp$&$\pi _qMT\!\Pm$\\
     \midrule\\[-8pt] 
      $6$ &$0$&$0$ &$\zmod{16}$ \\ [3pt]
      $5$ &$0$&$0$ &$0$ \\ [3pt]
      $4$ &$\ZZ$&$\zmod{16}$ &$0$ \\ [3pt]
      $3$ &$0$&$\zt$ &$0$ \\ [3pt]
      $2$ &$\zt$&$\zmod2$ & $\zmod8$ \\ [3pt]
      $1$ &$\zt$&$0$ &$\zt$ \\ [3pt]
      $0$ &$\ZZ$&$\zt$ &$\zt$ \\ [3pt]
     \bottomrule \end{tabular} 
  \end{equation}
  \end{theorem}

  \begin{corollary}[Symmetry class D]\label{thm:70}
   The groups of deformation classes of free fermion theories and of reflection
positive invertible theories with symmetry group~$\Spin$ are isomorphic to:
  \begin{equation}\label{eq:85}
     \begin{tabular}{ c@{\hspace{2.5em}} c@{\hspace{4.25em}} 
	c@{\hspace{4.6em}} c@{\hspace{2.75em}} c}
     \toprule n&\multicolumn{4}{l}{\hspace{-1.2em}$\ker\Phi \longrightarrow
     \FFF n{\Spin}\xrightarrow{\;\;\Phi \;\;} \TP n{\Spin}\longrightarrow
     \coker\Phi$} \\
     \midrule\\[-8pt] 
      $5$ &$0$&$0$&$0$ &$0$ \\ [3pt]
      $4$ &$0$&$0$&$0$ &$0$ \\ [3pt]
      $3$ &$0$&$\ZZ$&$\ZZ$ &$0$ \\ [3pt]
      $2$ &$0$&$\zt$&$\zt$ & $0$ \\ [3pt]
      $1$ &$0$ &$\zt$ &$\zt$ &$0$ \\ [3pt]
      $0$ &$0$&$0$&$0$ &$0$ \\ [3pt]
     \bottomrule \end{tabular} 
  \end{equation}
  \end{corollary}

  \begin{lnote*}[]\label{thm:127}
 The groups~$\TP1{\Spin}$ and~$\TP2{\Spin}$ were computed by the ``group
super-cohomology theory'' in~\cite{GW}; see Table~II.  That theory is a
2-stage Postnikov truncation of~$\IZ$, so in general only computes a subgroup
of topological phases; it is the entire group in very low dimensions.  The
interacting classification~$\TP n{\Spin}$ appears in~\cite{QHZ}: see~\S IIA
for~$n=3$, \S IID for~$n=2$, and \S IIE for~$n=1$.  The group~$\TP3{\Spin}$
is discussed in~\cite[\S V A]{LV}, but their restriction to ``non-chiral''
phases means that the $E_8$~phases that generate~$\TP3{\Spin}$ were not
accounted for.  All of the groups in the table, but not the map from free
fermions to interacting theories, appear in~\cite{KTTW}.  Those authors
conjecture a cobordism classification of interacting fermionic SPT phases.
  \end{lnote*}

  \begin{proof}
 That $\Phi $~is an isomorphism in low dimensions follows since the ABS map
$\MSpin\to KO$ induces an isomorphism on~$\pi _{\le 7}$.
  \end{proof}

In the next example we meet a nontrivial kernel of~$\Phi $, which is to say
free fermion phases that become trivial when interactions are allowed.

  \begin{corollary}[Symmetry class DIII]\label{thm:93}
   The groups of deformation classes of free fermion theories and of reflection
positive invertible theories with symmetry group~$\Pp$ are isomorphic to:  
  \begin{equation}\label{eq:113}
     \begin{tabular}{ c@{\hspace{2.5em}} c@{\hspace{3.75em}} 
	c@{\hspace{4.6em}} c@{\hspace{2.75em}} c}
     \toprule n&\multicolumn{4}{l}{\hspace{-.7em}$\ker\Phi \longrightarrow
     \FFF n{\Pp}\xrightarrow{\;\;\Phi \;\;} \TP n{\Pp}\longrightarrow
     \coker\Phi$} \\
     \midrule\\[-8pt] 
      $5$ &$0$&$0$&$0$ &$0$ \\ [3pt]
      $4$ &$16\ZZ$&$\ZZ$&$\zmod{16}$ &$0$ \\ [3pt]
      $3$ &$0$&$\zt$&$\zt$ &$0$ \\ [3pt]
      $2$ &$0$&$\zt$&$\zt$ & $0$ \\ [3pt]
      $1$ &$0$ &$0$ &$0$ &$0$ \\ [3pt]
      $0$ &$2\ZZ$&$\ZZ$&$\zt$ &$0$ \\ [3pt]
     \bottomrule \end{tabular} 
  \end{equation}
  \end{corollary}

  \begin{lnote*}[]\label{thm:126}
 There are many arguments in the physics literature that 16~copies of the
basic free fermion theory in 4~dimensions has a trivial phase once
interactions are allowed, and that this does not occur with fewer copies.
(As noted in Remark~\ref{thm:144}, the group~$\TP4{\Pp}$ is torsion, hence
\emph{a priori} some multiple of the free theory necessarily becomes trivial
once interactions are allowed.)  A sample includes \cite{K1,FCV,WS,MFCV,K4}
and \cite[\S4]{W1}.  The interacting case in 3~dimensions is investigated
in~\cite[\S3]{W1}, and various aspects of the invertible field theory are
described explicitly.  It is also discussed in~\cite[\S V B]{LV}, but the
nonzero element is missed within the ``$K$-formalism'' as the authors
explain.  The groups~$TP_n(\Pp)$ as computed here also appear
in~\cite[Table~2]{KTTW}.
  \end{lnote*}

  \begin{corollary}[Symmetry class BDI]\label{thm:94}
 The groups of deformation classes of free fermion theories and of reflection
positive invertible theories with symmetry group~$\Pm$ are isomorphic to:
  \begin{equation}\label{eq:114}
     \begin{tabular}{ c@{\hspace{2.5em}} c@{\hspace{4.25em}} 
	c@{\hspace{4.6em}} c@{\hspace{2.75em}} c}
     \toprule n&\multicolumn{4}{l}{\hspace{-1.2em}$\ker\Phi \longrightarrow
     \FFF n{\Pm}\xrightarrow{\;\;\Phi \;\;} \TP n{\Pm}\longrightarrow
     \coker\Phi$} \\
     \midrule\\[-8pt] 
      $5$ &$0$&$0$&$0$ &$0$ \\ [3pt]
      $4$ &$0$&$0$&$0$ &$0$ \\ [3pt]
      $3$ &$0$&$0$&$0$ &$0$ \\ [3pt]
      $2$ &$8\ZZ$&$\ZZ$&$\zmod8$ & $0$ \\ [3pt]
      $1$ &$0$ &$\zt$ &$\zt$ &$0$ \\ [3pt]
      $0$ &$0$&$\zt$&$\zt$ &$0$ \\ [3pt]
     \bottomrule \end{tabular} 
  \end{equation}
  \end{corollary}

  \begin{lnote*}[]\label{thm:128}
 The breaking of the $\ZZ$~classification of free fermions in 2~spacetime
dimensions to the $\zmod8$~classification of interacting fermions is treated
in~\cite{FK1,FK2,TPB,YWOX}, and~\cite[\S5]{W1}.  The groups~$\TP n{\Pm}$,
$n=1,2$, are computed by the group super-cohomology in~\cite[Table~II]{GW}.
The vanishing of~$\TP3{\Pm}$ is argued in~\cite[\S V B]{LV}.  The
groups~$TP_n(\Pm)$ as computed here also appear in~\cite[Table~2]{KTTW}.

  \end{lnote*}

 \bigskip
 \subsubsection{Internal symmetry group $K=\TT$} The symmetry groups are
classified in Proposition~\ref{thm:60}.  $\Spin^c$~bordism groups are
computed in~\cite{ABP1}; cf.~\cite[Chapter~XI]{Sto}.  $\Pin^c$~bordism groups
are computed in~\cite{BG}.  The twisted $\Pin^c$ bordism computations are
new. 

  \begin{theorem}[]\label{thm:63}
 The low degree bordism groups for $K=\TT$ are:
  \begin{equation}\label{eq:74}
     \begin{tabular}{ c@{\hspace{2em}} c@{\hspace{2em}} c@{\hspace{2em}}
     c@{\hspace{2em}} c@{\hspace{2em}}} 
     \toprule q&$\pi _qMT\!\Spin^c$&$\pi _qMT\!\Pc$&$\pi _qMT\!\Pcp$&$\pi
     _qMT\!\Pcm$\\ 
     \midrule\\[-8pt] 
      $6$ &$\ZZ^2$&$\zmod{16}\times \zmod4$ &$\ZZ^2\times \zt$&$\ZZ^2\times
     \zt$ \\ [3pt] 
      $5$ &$0$&$0$ &$0$&$0$ \\ [3pt]
      $4$ &$\ZZ^2$&$\zmod8\times \zt$ &$(\zt)^3$ &$\zt$\\ [3pt]
      $3$ &$0$&$0$ &$\zt$ &$0$\\ [3pt]
      $2$ &$\ZZ$&$\zmod4$ & $\ZZ$ &$\ZZ\times \zt$\\ [3pt]
      $1$ &$0$&$0$ &$0$ &$0$\\ [3pt]
      $0$ &$\ZZ$&$\zt$ &$\zt$ &$\zt$\\ [3pt]
     \bottomrule \end{tabular} 
  \end{equation}
  \end{theorem}

  \begin{corollary}[Symmetry class A]\label{thm:95}
 The groups of deformation classes of free fermion theories and of reflection
positive invertible theories with symmetry group~$\Spin^c$ are isomorphic to:
  \begin{equation}\label{eq:115}
    \begin{tabular}{ c@{\hspace{2.8em}} c@{\hspace{4.5em}} 
	c@{\hspace{6.6em}} c@{\hspace{4em}} c}
     \toprule n&\multicolumn{4}{l}{\hspace{-1.4em}$\ker\Phi \longrightarrow
     \FFF n{\Spin^c}\xrightarrow{\;\;\Phi \;\;} \TP n{\Spin^c}\longrightarrow
     \coker\Phi$} \\
     \midrule\\[-8pt] 
      $5$ &$0$&$\ZZ$&$\ZZ^2$ &$\ZZ$ \\ [3pt]
      $4$ &$0$&$0$&$0$ &$0$ \\ [3pt]
      $3$ &$0$&$\ZZ$&$\ZZ^2$ &$\ZZ$ \\ [3pt]
      $2$ &$0$&$0$&$0$ & $0$ \\ [3pt]
      $1$ &$0$ &$\ZZ$ &$\ZZ$ &$0$ \\ [3pt]
      $0$ &$0$&$0$&$0$ &$0$ \\ [3pt]
     \bottomrule \end{tabular} 
   \end{equation}
  \end{corollary}

  \begin{lnote*}[]\label{thm:129}
 The vanishing of the group~$\TP 4{\Spin^c}$ is mentioned in~\cite{WPS} at
the end of Appendix~F. 

  \end{lnote*}

  \begin{corollary}[Symmetry class AIII]\label{thm:96}
 The groups of deformation classes of free fermion theories and of reflection
positive invertible theories with symmetry group~$\Pin^c$ are isomorphic to: 
  \begin{equation}\label{eq:116}
    \begin{tabular}{ c@{\hspace{2.8em}} c@{\hspace{4.5em}} 
	c@{\hspace{3.75em}} c@{\hspace{1.25em}} c}
     \toprule n&\multicolumn{4}{l}{\hspace{-.9em}$\ker\Phi \longrightarrow
     \FFF n{\Pin^c}\xrightarrow{\;\;\Phi \;\;} \TP n{\Pin^c}\longrightarrow
     \coker\Phi$} \\
     \midrule\\[-8pt] 
      $5$ &$0$&$0$&$0$ &$0$ \\ [3pt]
      $4$ &$8\ZZ$&$\ZZ$&$\zmod{8}\times \zt$ &$\zt$ \\ [3pt]
      $3$ &$0$&$0$&$0$ &$0$ \\ [3pt]
      $2$ &$4\ZZ$&$\ZZ$&$\zmod4$ & $0$ \\ [3pt]
      $1$ &$0$ &$0$ &$0$ &$0$ \\ [3pt]
      $0$ &$2\ZZ$&$\ZZ$&$\zt$ &$0$ \\ [3pt]
     \bottomrule \end{tabular} 
  \end{equation}
  \end{corollary}

  \begin{lnote*}[]\label{thm:130}
 The group~$\TP4{\Pin^c}$ and the map from free fermions is discussed
in~\cite[\S III]{WS}; see also~\cite[\S A.4]{SeWi} for the map from free
fermions.  The vanishing of the group~$\TP 3{\Pin^c}$ is discussed
in~\cite[\S V D]{LV} as well as in the last paragraph of~\cite[\S3.7]{W1}.
  \end{lnote*}

  \begin{corollary}[Symmetry class AII]\label{thm:97}
  The groups of deformation classes of free fermion theories and of reflection
positive invertible theories with symmetry group~$\Pcp$ are isomorphic to: 
  \begin{equation}\label{eq:117}
    \begin{tabular}{ c@{\hspace{2.8em}} c@{\hspace{4em}} 
	c@{\hspace{5em}} c@{\hspace{3em}} c}
     \toprule n&\multicolumn{4}{l}{\hspace{-1.2em}$\ker\Phi \longrightarrow
     \FFF n{\Pcp}\xrightarrow{\;\;\Phi \;\;} \TP n{\Pcp}\longrightarrow
     \coker\Phi$} \\
     \midrule\\[-8pt] 
      $5$ &$0$&$\ZZ$&$\ZZ^2$ &$\ZZ$ \\ [3pt]
      $4$ &$0$&$\zt$&$(\zt)^3$ &$(\zt)^2$ \\ [3pt]
      $3$ &$0$&$\zt$&$\zt$ &$0$ \\ [3pt]
      $2$ &$0$&$0$&$0$ & $0$ \\ [3pt]
      $1$ &$0$ &$\ZZ$ &$\ZZ$ &$0$ \\ [3pt]
      $0$ &$0$&$0$&$\zt$ &$\zt$ \\ [3pt]
     \bottomrule \end{tabular} 
  \end{equation}
  \end{corollary}

  \begin{lnote*}[]\label{thm:131}
 The $\zt$~invariant of free fermion systems in 3~and 4~spacetime dimensions
was introduced by Kane-Mele~\cite{KM} and Fu-Kane-Mele~\cite{FKM} and has
been further studied in many papers.  The interacting case in 4~dimensions is
investigated in~\cite{WPS} and in 3~dimensions in~\cite[\S3.7]{W1}; their
results agree with ours.  The initial computation in \cite[\S V C 2]{LV} of
$\TP3{\Pcp}\cong (\zt)^2$ was corrected in a subsequent erratum.  The
original argument in that paper asserts a $\zt$~subgroup of bosonic phases,
which would have symmetry group~$O\ltimes\TT$, as in~\eqref{eq:72}.  We
computed that $\pi _3\bigl(M(O\ltimes\TT) \bigr)\cong \zt$ and the natural
projection $\Pcp\to O\ltimes\TT$ induces the zero map on $\pi _3$~of the Thom
spectra.  This implies that the group of bosonic phases is~$\zt$, as claimed,
but that the lift of that bosonic phase to a fermionic phase is trivial.
This triviality of the pullback was not noticed initially; our homotopy
theoretic methods give a systematic approach, and we encounter this issue
again in Literature Note~\ref{thm:132}.  The physical results in 4~dimensions
were recounted in~\cite{M} at the end of~\S VI, where the question of
agreement with a bordism computation was raised.  This provided strong
motivation for the computations in this section.  We remark that the
description of the partition function of some phases in terms of
Stiefel-Whitney classes matches our bordism computations as well.  Also,
\S4.7~of \cite{W1} treats the invertible topological field theory in
4~dimensions defined by the free fermion theory, so only detects the image
of~$\Phi $ in~$\TP4{\Pcp}$.
  \end{lnote*}

  \begin{corollary}[Symmetry class AI]\label{thm:98}
 The groups of deformation classes of free fermion theories and of reflection
positive invertible theories with symmetry group~$\Pcm$ are isomorphic to: 
  \begin{equation}\label{eq:118}
    \begin{tabular}{ c@{\hspace{2.8em}} c@{\hspace{4em}} 
	c@{\hspace{5.6em}} c@{\hspace{3.3em}} c}
     \toprule n&\multicolumn{4}{l}{\hspace{-1.2em}$\ker\Phi \longrightarrow
     \FFF n{\Pcm}\xrightarrow{\;\;\Phi \;\;} \TP n{\Pcm}\longrightarrow
     \coker\Phi$} \\
     \midrule\\[-8pt] 
      $5$ &$0$&$\ZZ$&$\ZZ^2$ &$\ZZ$ \\ [3pt]
      $4$ &$0$&$0$&$\zt$ &$\zt$ \\ [3pt]
      $3$ &$0$&$0$&$0$ &$0$ \\ [3pt]
      $2$ &$0$&$0$&$\zt$ & $\zt$ \\ [3pt]
      $1$ &$0$ &$\ZZ$ &$\ZZ$ &$0$ \\ [3pt]
      $0$ &$0$&$\zt$&$\zt$ &$0$ \\ [3pt]
     \bottomrule \end{tabular} 
  \end{equation}
  \end{corollary}

  \begin{lnote*}[]\label{thm:132}
 The group~$\TP4{\Pcm}$ is discussed in detail in the erratum to~\cite{WS}.  
The group $\TP3{\Pcm}$ is asserted to be cyclic of order two in~\cite[\S V C
1]{LV} generated by a bosonic phase.  The bosonic phase is the same one
identified for the symmetry class~AII---see the Literature Note
following~\eqref{eq:117}---and again we compute that its lift to a fermionic
phase with symmetry group~$\Pcm$ vanishes, which explains the discrepancy.
  \end{lnote*}

\bigskip
 \subsubsection{Internal symmetry group $K=SU_2$} The symmetry groups~$G^0,
G^+, G^-$ are defined and classified in Proposition~\ref{thm:66}. 

  \begin{theorem}[]\label{thm:67}
  The low degree bordism groups for $K=SU_2$ are:
  \begin{equation}\label{eq:76}
     \begin{tabular}{ c@{\hspace{2em}} c@{\hspace{2em}} c@{\hspace{2em}}
     c@{\hspace{2em}} } 
     \toprule q&$\pi _qMTG^0$&$\pi
     _qMT\Gp$&$\pi _qMT\Gm$\\ 
     \midrule\\[-8pt] 
      $6$ &$\zt\times \zt$&$(\zt)^4$ &$\zt\times \zmod4\times \zmod{16}$ \\
     [3pt] 
      $5$ &$\zt\times \zt$&$\zt$ &$(\zt)^2$ \\ [3pt]
      $4$ &$\ZZ^2$ &$\zmod4\times \zt$ &$(\zt)^3$\\ [3pt]
      $3$ &$0$ &$0$ &$0$\\ [3pt]
      $2$ &$0$ & $\zt$ &$\zt$\\ [3pt]
      $1$ &$0$ &$0$ &$0$\\ [3pt]
      $0$ &$\ZZ$ &$\zt$ &$\zt$\\ [3pt]
     \bottomrule \end{tabular} 
  \end{equation}
  \end{theorem}

  \begin{corollary}[Symmetry class C]\label{thm:99}
 The groups of deformation classes of free fermion theories and of reflection
positive invertible theories with symmetry group~$G^0=\Spin\times \mstrut
_{\pmo}SU_2$ are isomorphic to:   
  \begin{equation}\label{eq:119}
    \begin{tabular}{ c@{\hspace{2.8em}} c@{\hspace{4em}} 
	c@{\hspace{3.6em}} c@{\hspace{2.75em}} c}
     \toprule n&\multicolumn{4}{l}{\hspace{-1.2em}$\ker\Phi \longrightarrow
     \FFF n{G^0}\xrightarrow{\quad\Phi \quad} \TP
     n{G^0}\xrightarrow{\qquad\;\;} \coker\Phi$} \\
     \midrule\\[-8pt] 
      $5$ &$0$&$\zt$&$\zt\times \zt$ &$\zt$ \\ [3pt]
      $4$ &$0$&$0$&$0$ &$0$ \\ [3pt]
      $3$ &$0$&$\ZZ$&$\ZZ^2$ &$\ZZ$ \\ [3pt]
      $2$ &$0$&$0$&$0$ & $0$ \\ [3pt]
      $1$ &$0$ &$0$ &$0$ &$0$ \\ [3pt]
      $0$ &$0$&$0$&$0$ &$0$ \\ [3pt]
     \bottomrule \end{tabular} 
  \end{equation}
  \end{corollary}

  \begin{lnote*}[]\label{thm:133}
 That $\TP4{G^0}=0$ was suggested in~\cite{WS} in the last paragraph
preceding \S V A.

  \end{lnote*}

  \begin{corollary}[Symmetry class CI]\label{thm:100}
 The groups of deformation classes of free fermion theories and of reflection
positive invertible theories with symmetry group~$G^+=\Pp\times \mstrut
_{\pmo}SU_2$ are isomorphic to:    
  \begin{equation}\label{eq:120}
    \begin{tabular}{ c@{\hspace{2.8em}} c@{\hspace{4.25em}} 
	c@{\hspace{3.3em}} c@{\hspace{1.6em}} c}
     \toprule n&\multicolumn{4}{l}{\hspace{-1em}$\ker\Phi \longrightarrow
     \FFF n{G^+}\xrightarrow{\;\;\Phi \;\;} \TP n{G^+}\longrightarrow
     \coker\Phi$} \\
     \midrule\\[-8pt] 
      $5$ &$0$&$0$&$\zt$ &$\zt$ \\ [3pt]
      $4$ &$4\ZZ$&$\ZZ$&$\zmod4\times \zt$ &$\zt$ \\ [3pt]
      $3$ &$0$&$0$&$0$ &$0$ \\ [3pt]
      $2$ &$0$&$0$&$\zt$ & $\zt$ \\ [3pt]
      $1$ &$0$ &$0$ &$0$ &$0$ \\ [3pt]
      $0$ &$2\ZZ$&$\ZZ$&$\zt$ &$0$ \\ [3pt]
     \bottomrule \end{tabular} 
  \end{equation}
  \end{corollary}

\noindent 
 Our computations prove $\Phi $~maps the generator of~$\FFF 4{G^+}$ to
an element of order~4 in~$\TP4{G^+}$.

  \begin{lnote*}[]\label{thm:134}
 Wang-Senthil~\cite[\S V]{WS} discusses the $n=4$~case and conjecture the
same group $\TP 4{G^+}\cong \zmod4\times \zt$ that we compute; the map from
free fermions also agrees.
  \end{lnote*}

  \begin{corollary}[Symmetry class CII]\label{thm:101}
 The groups of deformation classes of free fermion theories and of reflection
positive invertible theories with symmetry group~$G^-=\Pm\times \mstrut
_{\pmo}SU_2$ are isomorphic to:     
  \begin{equation}\label{eq:121}
    \begin{tabular}{ c@{\hspace{2.8em}} c@{\hspace{3.5em}} 
	c@{\hspace{3.8em}} c@{\hspace{2.25em}} c}
     \toprule n&\multicolumn{4}{l}{\hspace{-1em}$\ker\Phi \longrightarrow
     \FFF n{G^-}\xrightarrow{\;\;\Phi \;\;} \TP n{G^-}\longrightarrow
     \coker\Phi$} \\
     \midrule\\[-8pt] 
      $5$ &$0$&$\zt$&$(\zt)^2$ &$\zt$ \\ [3pt]
      $4$ &$0$&$\zt$&$(\zt)^3$ &$(\zt)^2$ \\ [3pt]
      $3$ &$0$&$0$&$0$ &$0$ \\ [3pt]
      $2$ &$2\ZZ$&$\ZZ$&$\zt$ & $0$ \\ [3pt]
      $1$ &$0$ &$0$ &$0$ &$0$ \\ [3pt]
      $0$ &$0$&$0$&$\zt$ &$\zt$ \\ [3pt]
     \bottomrule \end{tabular} 
  \end{equation}
  \end{corollary}

  \begin{lnote*}[]\label{thm:135}
 The 4-dimensional case is treated in~\cite[\S VI]{WS}; the answer they
obtain for~$\TP4{G^-}$ is $(\zt)^5$, which disagrees with the corresponding
entry in~\eqref{eq:121}, but it may be a different symmetry group they are
considering. In any case, in the note following Corollary~\ref{thm:97}, we
compute the group of bosonic phases with symmetry group $O\times \mstrut
_{\pmo}SU_2$ and find~$(\zt)^4$, but the lift to fermionic phases kills a
$(\zt)^2$ subgroup.  
  \end{lnote*}

   \section{Computations}\label{sec:13}

The computations in \S\ref{subsec:8.2} involve finitely generated
abelian groups having no odd torsion, so it suffices then to make
them after completing at $2$.  This can be done using the
Adams spectral sequence 
  \begin{equation}\label{eq:m159}
     \ext^{s,t}_{A}(H^{\ast}(MTH),\Z/2) \Rightarrow \pi_{t-s}MTH, 
  \end{equation}
where $A$ is the mod $2$ Steenrod algebra, and, though not indicated
in the notation,  the homotopy groups have been completed at $2$.   

What makes this approach tractable is an identification\footnote{Remark:
Corollary~\ref{thm:113} implies that for \emph{any} symmetry type~$(H,\rho
)$, the spectrum~$MTH$ is an $\MSpin$-module.} of the spectrum
$\Sigma^{s}MTH(s)$ with \begin{equation} \label{eq:m6} \begin{aligned}
&M\spin\wedge MTO_{|s|} &\qquad -3\le s \le 0& \\ &M\spin\wedge MO_{|s|}
&\qquad 0 \le s \le 3& \\ &\Sigma M\spin\wedge MSO_{3} &\qquad s =
4&\mathrlap{,}
\end{aligned}
\end{equation}
and in the complex case, of $\Sigma^{s}MTH^{c}(s)$ with  
  \begin{equation}\label{eq:m160}
     M\spinc \wedge MO_{s}\approx \Sigma^{-2}M\spin\wedge
     MU_{1}\wedge MO_{s}. 
  \end{equation}
Let $A_{1}\subset A$ be the sub algebra generated by $\sq^{1}$ and
$\sq^{2}$.  Anderson, Brown, and Peterson~\cite{ABP1} give
an isomorphism 
  \begin{equation}\label{eq:m161}
     H^{\ast}M\spin \approx A\underset{A_{1}}{\otimes} \{\Z/2\oplus M\} 
  \end{equation}
in which $M$ is a graded $A_{1}$-module with $M_{i}=0$ for
$i<8$.   This means that for $t-s <8$ one can identify the
$E_{2}$-term of the Adams spectral sequence for\footnote{Here only we use the
notation `$H(d)$' in place of~`$H(s)$' to avoid the conflict with Adams'
homological grading index~`$s$'.} $\pi_{\ast}MTH(d)$ with 
\begin{align*}
&\ext_{A_{1}}^{s,t}(H^{\ast-d}MTO_{|d|},\Z/2) & -3\le d\le 0& \\
&\ext_{A_{1}}^{s,t}(H^{\ast+d}MO_{|d|},\Z/2) & -0\le d \le 3& \\
&\ext_{A_{1}}^{s,t}(H^{\ast+3}MSO_{3},\Z/2) & d=4&\mathrlap{\ ,} \\
\intertext{and of $\pi_{\ast}MTH^{c}(d)$ with}
&\ext_{A_{1}}^{s,t}(H^{\ast+2+d}MU_{1}\wedge MO_{d},\Z/2) &
d=0,1&\mathrlap{\ .}
\end{align*}

These groups are computed by standard methods, and the computation, as
well as the spectral sequences (which collapse) are described
Figure~\ref{fig:m1} and give the results described in
tables~\eqref{eq:84},~\eqref{eq:74}, and~\eqref{eq:67}.

The relationship with the free fermion theories is given by maps of
spectra 
\begin{align}
MTH(s) &\to \Sigma^{-s}KO \\
MTH^{c}(s) &\to \Sigma^{-s}K 
\end{align}
or, under the above identifications, maps  
  \begin{equation}\label{eq:m162}
     \begin{aligned}      M\spin\wedge MTO_{|s|}&\to KO &\quad -3&\le s\le 0 \\
M\spin \wedge MO_{|s|}&\to KO &\quad 3&\ge s\ge 0 \\
\Sigma M\spin\wedge 
     MSO_{3} &\to KO &\quad &\;s=4 \\ M\spinc\wedge MO_{s} &\to K
     &&\;s=0,1\mathrlap{\ .}
     \end{aligned} 
  \end{equation}
These are all maps of $M\spin$ (or $M\spinc$) modules, in which $KO$
and $K$ are into $M\spin$ and $M\spinc$-modules using the
Atiyah-Bott-Shapiro orientation.  They are therefore determined by
their restrictions
\begin{equation}
\label{eq:m7}
\begin{aligned}
MTO_{|s|}&\to KO &\quad -3&\le s\le 0 \\
MO_{|s|}&\to KO &\quad 3&\ge s\ge 0 \\
\Sigma MSO_{3} &\to KO &\quad &\;s=4 \\
MO_{s} &\to K &\quad &\;s=0,1\mathrlap{\ .}
\end{aligned}
\end{equation}
These are described in Propositions~\ref{thm:m1}, \ref{thm:m3}, and~\ref{thm:m4}  below, and
using them, the assertions about the maps in
tables~\eqref{eq:85}, \eqref{eq:113}, \eqref{eq:114}, \eqref{eq:115}, \eqref{eq:116}, \eqref{eq:117}, \eqref{eq:118}, \eqref{eq:119}, \eqref{eq:120},
and~\eqref{eq:121}
can be verified.  The details
are summarized in the charts in Figure~\ref{fig:m1}.   The complex
case is easier and left to the reader.  See~\cite{C,BeC} for a detailed
account of the computations.

  \begin{figure}
  \includegraphics[width=\textwidth]{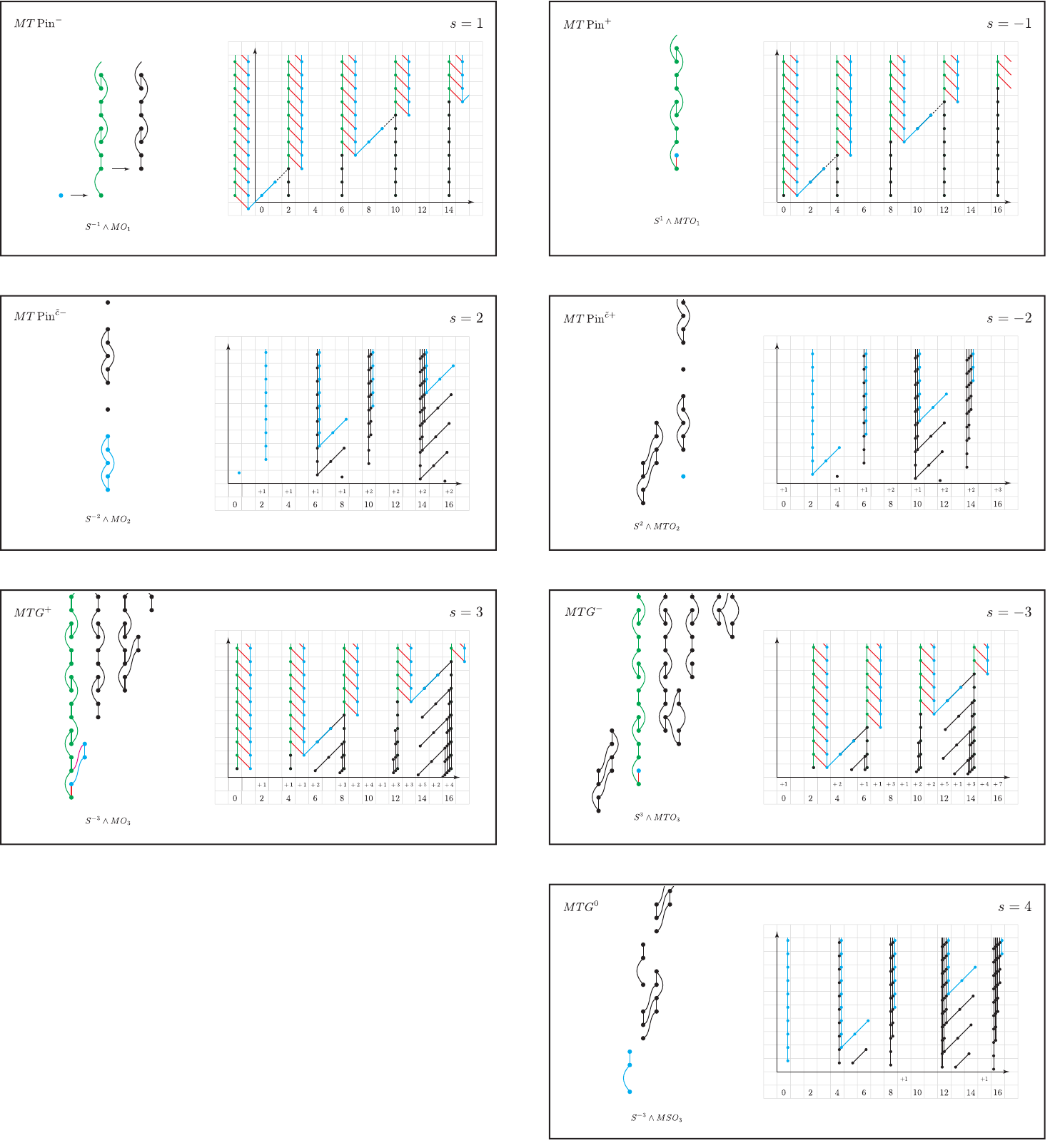}
  \caption{The Adams spectral sequences}
\label{fig:m1}
  \end{figure}

For the identifications~\eqref{eq:m6} and the maps~\eqref{eq:m7}
we begin with a uniform description of the groups $BH(\pm s)$ (for $s\ne 4$).
Write  
  \begin{equation}\label{eq:m163}
     \K=K(\Z/2,1)\times K(\Z/2,2) 
  \end{equation}
with the group structure  
  \begin{equation}\label{eq:m164}
     (x_{1},x_{2})*(y_{1},y_{2}) = (x_{1}+y_{1}, x_{2}+y_{2}+x_{1}y_{1}) 
  \end{equation}
in which $x_{i},y_{i}\in H^{i}(\slot,\Z/2)$.  With this choice the map  
  \begin{equation}\label{eq:m165}
     BO\xrightarrow{(w_{1},w_{2})}{} \K 
  \end{equation}
is a group homomorphism.

For $s\ge 0$ define a map $B\tH(s)\to BO$ by the homotopy pullback square 
  \begin{equation}\label{eq:m166}
     \xymatrix@C=54pt{ B\tH(s) \ar[r]\ar[d] & BO_{s} \ar[d]^{(w_{1},w_{2})} \\
     BO \ar[r]_-{(w_{1},w_{2}+w_{1}^{2})} & \K } 
  \end{equation}
and set $B\tH(-s)\to BO$ to be the composite  
  \begin{equation}\label{eq:m167}
     B\tH(s)\to BO \xrightarrow{-\id}{} BO. 
  \end{equation}
The space $B\tH(-s)\to BO$ fits into a homotopy pullback square 
  \begin{equation}\label{eq:m168}
     \xymatrix@C=54pt{ B\tH(-s) \ar[r]\ar[d] & BO_{s} \ar[d]^{(w_{1},w_{2})}
     \\ BO \ar[r]_-{(w_{1},w_{2})} & \K\mathrlap{\ .}  } 
  \end{equation}
For later reference we note 

\begin{remark}
\label{rem:m1}
The homotopy fiber of $B\tH(\pm s)\to BO$, being the same as the homotopy
fiber of $BO_{s}\to \K$ is  
  \begin{equation}\label{eq:m169}
     \begin{aligned} B\spin_{s} &\qquad s \ge 1 \\ 
     \Z/2\times B\Z/2 &\qquad s =0\mathrlap{\ .}  \end{aligned} 
  \end{equation}
For $-3\le s\le 3$ one may identify $B\tH(s)\to BO$ with $BH(s)\to BO$.
The map $BH(4)\to BO$ fits into a homotopy pullback diagram  
  \begin{equation}\label{eq:m170}
     \xymatrix{ BH(4) \ar[r]\ar[d] & BSO_{3} \ar[d]^{w_{2}} \\ BO
     \ar[r]_{(w_{1},w_{2})} & \K \mathrlap{\ .}} 
  \end{equation}
We leave the verification of these assertions to the reader.
\end{remark}

With $s\ge 0$, the maps $B\tH(\pm s)\to BO$ and $B\tH(\pm s)\to BO_{s}$ 
can also be expressed in terms of the diagrams of homotopy pullback squares
\begin{equation}
\label{eq:m2}
\xymatrix@C=54pt{
B\tH(s)  \ar[r]\ar[d]  & B\spin\ar[d]   & \\
BO\times BO_{s}  \ar[r]_-{-\id - (V_{s}-s)}        &  BO \ar[r]_-{(w_{1},w_{2})}         & \K
}
\end{equation}
and 
\begin{equation}
\label{eq:m3}
\xymatrix@C=54pt{
B\tH(-s)  \ar[r]\ar[d]  & B\spin\ar[d]    & \\
BO\times BO_{s}  \ar[r]_-{\id-(V_{s}-s)}        &  BO
\ar[r]_-{(w_{1},w_{2})}         & \K\mathrlap{\ .}
}
\end{equation}
A map $X\to B\tH(s)$ therefore classifies a pair $(V,V_{s})$ consisting of a
stable vector bundle $V$ (of virtual dimension $0$), a vector bundle
$V_{s}$ of dimension $s$ and a $\spin$ structure on $-V-(V_{s}-s)$.
Writing $W=-V-(V_{s}-s)$, so that $V=-W-(V_{s}-s)$, one sees that
$B\tH(s)$ classifies pairs $(W,V_{s})$ in which $W$ is a stable $\spin$
bundle of virtual dimension zero.  Thus $B\tH(s)\to BO$ may be
identified with the map
\begin{align*}
B\spin\times BO_{s} &\to BO \\
(W,V_{s})&\mapsto -W - (V_{s}-s).
\end{align*}
Similarly $B\tH(-s)\to BO$ may be identified with 
\begin{align*}
B\spin\times BO_{s} &\to BO \\
(W,V_{s})&\mapsto -W + (V_{s}-s),
\end{align*}
and  $BH(4)\to BO$  with 
\begin{align*}
B\spin\times BSO_{3} &\to BO \\
\intertext{via either of the maps} \\
(W, V_{3}) & \mapsto -W + (V_{3}-3) \quad\text{or} \\
(W, V_{3}) & \mapsto -W - (V_{3}-3).
\end{align*}
This leads to the identifications 
\begin{equation}
\label{eq:m5}
\begin{aligned}
MT\tH(s) &\approx \Sigma^{-s}M\spin\wedge MO_{s} \\
MT\tH(-s) &\approx \Sigma^{s}M\spin \wedge MTO_{s} \\
MTH(4) &\approx \Sigma^{-3}M\spin \wedge MSO(3) \\ 
&\approx  \Sigma^{3}M\spin \wedge MTSO(3).
\end{aligned}
\end{equation}

We define $B\tH(\pm s)_{n}\to BO_{n}$ by the pullback square
\begin{equation}
\label{eq:m8}
\xymatrix{
B\tH(\pm s)_{n}  \ar[r]\ar[d]  & B\tH(\pm s)
\ar[d] \\
BO_{n}  \ar[r]        & BO\mathrlap{\ .}
}
\end{equation}
The space $B\tH_{n}(s)$ classifies a pair $(V_{n},V_{s})$ consisting
of vector bundles of dimension $n$ and $s$ and a $\spin$ structure on
$-V_{n}-V_{s}$ (or, equivalently on $V_{n}+V_{s}$), while
$B\tH(-s)_{n} $ classifies pairs $(V_{n},V_{s})$ a $\spin$ structure
on $-V_{n}+V_{s}$.  For $s\ge 0$ there is therefore a pullback square
\begin{equation}
\label{eq:m4}
\xymatrix{
B\tH_{n}(s)  \ar[r]\ar[d]  & B\spin_{n+s}
\ar[d] \\
BO_{n}\times BO_{s}  \ar[r]        & BO_{n+s}\mathrlap{\ .}
}
\end{equation}

\begin{proposition}
\label{thm:m2}
The space $B\tH(\pm s)_{n}$ is the classifying space of a compact Lie
group $\tH(\pm s)_{n}$.   The group $\tH_{n}(s)$ is the stabilizer in $\spin_{n+s}$
of a $s$-plane in $\R^{n+s}$.   
\end{proposition}

\begin{proof}
The first assertion is a consequence of the pullback
square~\eqref{eq:m8} and Remark~\ref{rem:m1}.  The second is immediate
from~\eqref{eq:m4}
\end{proof}

The construction of~\S\ref{subsubsec:8.4.2} leads to maps
\begin{align*}
MT\tH(s) &\to  \Sigma^{-s}KO \\
\end{align*}
and so, by~\eqref{eq:m5}, to
\begin{align*}
M\spin\wedge MTO_{s} &\to KO \\
M\spin\wedge MO_{s} &\to KO \\
\Sigma M\spin\wedge MSO_{3} &\to KO.
\end{align*}

These are maps of $M\spin$-modules, so to describe them it suffices
the restricted maps
\begin{align*}
MO_{s} &\to  KO \\
MTO_{s} &\to KO \\
\Sigma MSO_{3} &\to KO\mathrlap{\ .}
\end{align*}

\begin{proposition}
\label{thm:m1} Let $V\to BO_{s}$ be the universal vector bundle.  The
map $MO_{s}\to KO$ corresponds to the element of $KO(V,V-0)$
given by applying the difference bundle construction to
\begin{align*}
V\times \Lambda^{\ast}(V) &\to 
\Lambda^{\ast}(V) \\
(v,\omega)&\mapsto v\wedge \omega.
\end{align*}
\end{proposition}

\begin{proof}
In the notation of~\ref{thm:90}, the algebra $A(s)$ is
$\cliff_{+s}\otimes \cliff_{-s}$, so that $A^{\text{op}}$ is also
$\cliff_{+s}\otimes \cliff_{-s}$, but with left Clifford
multiplication by $v\in \R^{s}$ sending $x\otimes y$ to
$(-1)^{|x|}x\otimes vy$.  The composed embedding $O_{s}\to
H_{s}\to A^{\text{op}}$ is the map 
  \begin{equation}\label{eq:m172}
     O_{s}\to \cliff_{+s}\otimes \cliff_{-s} 
  \end{equation}
sending reflection through the hyperplane perpendicular to $v\in
\R^{s}$ to $v\otimes v$.  

Let $P\to BO_{s}$ be the universal principal $O_{s}$ bundle.  The
$K$-theory class described in~\ref{subsubsec:8.4.2} is the difference
bundle on $(V,V-0)$ associated to the $O_{s}$-equivariant ``Clifford
multiplication'' map 
  \begin{equation}\label{eq:m173}
     \R^{s}\times (A^{\text{op}}\underset{A^{\text{op}}}{\otimes}M)\to
     (A^{\text{op}}\underset{A^{\text{op}}}{\otimes}M) 
  \end{equation}
in which $M=\cliff_{s}$ is the left $A^{\text{op}}$ bimodule specified
in \S\ref{subsubsec:8.4.2}, and giving the Morita equivalence of $A^{\text{op}}$ with
$\R$.  Passing to associated bundles, this works out to be
\begin{align*}
V\times \cliff(V) &\to \cliff(V) \\
(v,\omega) &\mapsto (-1)^{|\omega|}\omega v.
\end{align*}
The anti-automorphism of $\cliff(V)$ extending the identity map of $V$
gives an isomorphism of this with 
\begin{align*}
V\times \cliff(V) &\to \cliff(V) \\
(v,\omega) &\mapsto v\omega.
\end{align*}
The claim now follows from the standard way of ``wrapping up'' the
complex $V\times \Lambda(V)\to \Lambda(V)$ using $v\pm\iota_{v}$
(see~\cite[Proposition~11.6]{ABS} and the surrounding discussion for the complex case).
\end{proof}

\begin{proposition}
\label{thm:m3}
The map $MTO_{s}\to KO$ factors as  
  \begin{equation}\label{eq:m174}
     MTO_{s}\to (BO_{s})_{+}\to KO 
  \end{equation}
in which the first map is the map  
  \begin{equation}\label{eq:m175}
     \thom(BO_{s},-V)\to \thom\big(BO_{s},(-V)\oplus V\big) 
  \end{equation}
and the second corresponds to the trivial line bundle $1\in KO^{0}(BO_{s})$.
\end{proposition}

\begin{proof}
Write $\grass_{s}(\R^{n+s})$ for the Grassmannian of $s$-planes in
$(n+s)$-space, and let $V_{n}$ and $V_{s}$ be the universal $n$-plane
and $s$-plane bundles.   These bundles come equipped with a
trivialization  
  \begin{equation}\label{eq:m176}
     V_{s}\oplus V_{n}\approx \grass_{s}(\R^{n+s})\times \R^{n+s}. 
  \end{equation}
From the identification $\grass_{s}(\R^{n+s})=\spin_{n+s}/H_{n}$ of
Proposition~\ref{thm:m2} it
follows that the bundle $V_{n}$ comes equipped with an
$H_{n}$-structure.   The construction of~\ref{subsubsec:8.4.2} gives
an element $U\in KO^{n+s}(\thom(\grass_{s}(\R^{n+s}),V_{n}))$.   The
assertion is that this pulled back from the canonical generator (the
suspension of $1\in KO^{0}(\text{pt})$) of $\tilde KO^{n+s}(S^{n+s})$
along the map 
\begin{align*}
\thom(\grass_{s}(\R^{n+s});V_{n})&\to
\thom(\grass_{s}(\R^{n+s});V_{s}\oplus V_{n}) \\ 
&\approx
S^{n+s}\wedge \grass_{s}(\R^{n+s})_{+}\to S^{n+s}.
\end{align*}
This is immediate from the construction.  The algebra $A(s)^{\text{op}}$
is $\cliff_{-s}\otimes \cliff_{-n}$.  The class $U$ is the complex of
left $A$-modules (which come as right $A^{\text{op}}$-modules)
obtained by applying 
  \begin{equation}\label{eq:m177}
     \spin_{s+n}\underset{H_{n}}{\times}(\slot) 
  \end{equation}
to the the $H_{n}$-equivariant Clifford
multiplication map 
  \begin{equation}\label{eq:m178}
     \R^{n}\times \cliff_{-s}\otimes \cliff_{-n} \to \cliff_{-s}\otimes
     \cliff_{-n}. 
  \end{equation}
This map evidently extends to the $\spin_{s+n}$ equivariant Clifford
multiplication map
\begin{equation}
\label{eq:m1}
\R^{s}\oplus \R^{n}\times \cliff_{-s}\otimes \cliff_{-n} \to \cliff_{-s}\otimes \cliff_{-n}
\end{equation}
so the class $U$ is pulled back from the bundle of left $A$-modules on $(\R^{s+n},\R^{s+n}-\{0
\})$ obtained by applying 
  \begin{equation}\label{eq:m179}
     \spin_{n+s}\underset{\spin_{n+s}}{\times}(\slot) 
  \end{equation}
to~\eqref{eq:m1}.   This class represents the suspension of
$1$. 
\end{proof}

For the case $s=4$ what we require is the following
\begin{proposition}
\label{thm:m4}
The restriction of the map 
\[
S^{1}\wedge MSO_{3}\to KO
\]
to $S^{4}\to KO$ is the generator of $\tilde KO^{0}(S^{4})$.
\end{proposition}

\begin{proof}
From the diagram~\eqref{eq:m170} a map to $BH(4)$ can be thought of as
consisting of a stable vector
bundle $V$, an oriented $3$-plane bundle $V_{3}$ and a
$\spin$-structure on $V\oplus V_{3}$.   We map $BSO(4)\to BH(4)$ by
taking $V$ to corresponding to the defining representation and $V_{3}$
to be one of the two irreducible representations of dimension $3$.
The construction of \S\ref{subsubsec:8.4.2} then leads to the bundle on $MSO(4)$
corresponding to the $SO(4)$-equivariant map 
\[
\R^{4}\times N \to N
\]
where $N$ is the irreducible quaternionic $\cliff_{4}$-module specified in
\ref{subsubsec:8.4.2} with $SO(4)$-action from the embedding above.
This restricts to the generator of $KO(\R^{4},\R^{4}-\{0 \})$, by~\cite[Theorem~11.5]{ABS}.
\end{proof}

The two complex cases are handled similarly, using either the pullback
squares 
  \begin{equation}\label{eq:m180}
     \xymatrix@C=54pt{ BH^{c}(s) \ar[r]\ar[d] & BO_{s} \ar[d]^{(w_{1},
     \beta w_{2})} \\ BO
     \ar[r]_-{(w_{1},\beta w_{2})} & K(\Z/2,1)\times K(\Z,3) } 
  \end{equation}
for the identification  
  \begin{equation}\label{eq:m181}
     MTH^{c}(s)\approx \Sigma^{-s}M\spinc\wedge MO_{s} 
  \end{equation}
or  
  \begin{equation}\label{eq:m182}
     \xymatrix@C=54pt{ HB \ar[r]\ar[d] & BO_{s}\times BU(1) \ar[d]^{(w_{1},
     w_{2}+c_{1})} \\ BO
     \ar[r]_-{(w_{1},w_{2})} & \K } 
  \end{equation}
for the identification  
  \begin{equation}\label{eq:m183}
     MTH^{c}(s)\approx \Sigma^{-s-2}M\spin\wedge MU_{1}\wedge MO_{s}. 
  \end{equation}

   \section{A topological spin-statistics theorem}\label{sec:9}

In a relativistic quantum field theory the spin-statistics theorem states
that the central element of the Lorentz spin group acts on the Hilbert space
of the theory as~$(-1)^F$, where $F$~is the $\zt$-valued grading
operator;\footnote{$F$~vanishes on bosonic states and is the identity on
fermionic states.  In a \emph{free} theory there is a dense Fock space of
states with a finite number of particles on which $F$~counts the number of
fermionic particles modulo two.  In any theory $(-1)^F$~is the grading
operator on the $\zt$-graded Hilbert space of states.} see~\cite{SW,GJ,Kaz}
for proofs in the framework of Wightman quantum field theory.  In this
section we prove the analog for reflection positive \emph{non-extended}
invertible topological theories.  We do not know a version for fully extended
theories.  See~\cite{J-F} for another account of spin-statistics in
topological field theory, but without positivity.  A topological version of
spin-statistics also enters into~\cite{GK} in the context of fermionic
lattice models.
 
To formulate the statement we Wick rotate the central element of the Lorentz
spin group to the central element of the Euclidean spin group.  On a curved
Riemannian spin manifold~$M$, it acts as the \emph{spin flip}: the identity
diffeomorphism of~$M$ covered by the action of~$-1$ on the spin frames.  For
a general symmetry group~$H_n$ it is the action of the distinguished central
element~$k_0\in K$ in the internal symmetry group; see
Corollary~\ref{thm:113}.  Let $\sVect$~be the symmetric monoidal category of
super vector spaces; the symmetry incorporates the Koszul sign rule.  Recall
the notation (Remark~\ref{thm:142}) for the domain of a not necessarily
topological field theory.

  \begin{definition}[]\label{thm:114}
 Let $F\:\Bord_{\langle n-1,n \rangle}^\nabla (H_n)\to \sVect$ be a field
theory.  We say \emph{$F$~satisfies spin-statistics} if it maps the spin flip
on every $(n-1)$-manifold~$Y$ to the exponentiated grading operator~$(-1)^F$
on the super vector space~$F(Y)$.
  \end{definition}

  \begin{example}[]\label{thm:112}
 The spin-statistics connection fails without reflection positivity.
Consider a 1-dimensional invertible topological theory~$F$ of spin manifolds
with values in the category of $\zt$-graded complex lines.  There are
4~theories up to isomorphism:\footnote{We compute using
Theorem~\ref{thm:56}: $[\Sigma ^1MT\!\Spin_1,\Sigma ^1\ICx]\cong \Hom(\pi
_1\Sigma ^1MT\!\Spin_1,\Cx)$, the Thom spectrum~$\Sigma ^1MT\!\Spin_1$ is the
suspension spectrum of~$\RP^{\infty}_+$, and $\pi _1\RP^{\infty}_+\cong
\zt\times \zt$.  By contrast, $\pi _1MT\!\Spin\cong \zt$, hence
$[MT\!\Spin,\Sigma ^1\ICx]\cong \zt$, and so by Theorem~\ref{thm:110} there
are only two reflection positive theories.} $F(\pt_+)$~is either even or odd,
the spin flip acts as either~$+1$ or~$-1$, and these choices are independent.
Half of these theories satisfy spin statistics, and they are precisely the
ones for which $F(S^1_{\textnormal{bounding}})=+1$, which by
Theorem~\ref{thm:78} is the condition for stability, and so for reflection
positivity.
  \end{example}

  \begin{theorem}[]\label{thm:115}
 Let $F\:\Bord_{\langle n-1,n \rangle}(H_n)\to\sLine$ be a reflection
positive invertible topological field theory.  Then $F$~satisfies
spin-statistics.
  \end{theorem}

  \begin{figure}[ht]
  \centering
  \includegraphics[scale=1.1]{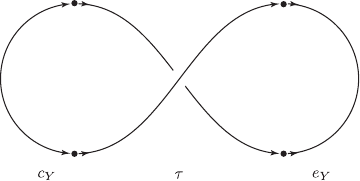}
  \caption{The composition $e_Y\circ \tau \circ c_Y$}\label{fig:7}
  \end{figure}

  \begin{proof}
 We first treat the case $H_n=\Spin_n$.  Let $Y$~be a closed $H_n$-manifold
and set~$L=F(Y)$.  Recall from~\S\ref{subsec:3.2} and Definition~\ref{thm:19}
the coevaluation $c_Y\:\emptyset ^{n-1}\to Y\amalg Y\dual$ and the evaluation
$e_Y\:Y\dual\amalg Y\to\emptyset ^{n-1}$.  Let $\tau \:Y\amalg Y\dual\to
Y\dual\amalg Y$ be the symmetry map.  The composition $e_Y\circ \tau \circ
c_Y$ is $S^1_{\textnormal{nonbounding}}\times Y$ (see~Figure~\ref{fig:7}),
and under~$F$ it maps to the composition $\CC\to L\otimes L^*\to L^*\otimes
L\to\CC$.  The Koszul sign rule in the symmetry gives
  \begin{equation}\label{eq:154}
     F(S^1_{\textnormal{nonbounding}}\times Y) = \tr_s\id_L = \tr
     (-1)^F=\begin{cases}
     +1,&\text{$L$~even},\\-1,&\text{$L$~{odd}},\end{cases}  
  \end{equation}
where $\tr_s$ is the supertrace.  The nonbounding circle is obtained by
cutting the bounding circle at two points and regluing using the spin flip
diffeomorphism of one of the points and the identity of the other.  In other
words, it is a triple composition of coevaluation, the indicated
diffeomorphism, and evaluation.  Take Cartesian product with~$Y$ and
apply~$F$ to conclude that the ratio of~\eqref{eq:154} with
$F(S^1_{\textnormal{bounding}}\times Y)$ is the supertrace of the spin flip
on~$Y$, and since the spin flip has order two this ratio equals~$\pm1$.  But
$S^1_{\textnormal{bounding}}\times Y$ is the spin double of~$c_Y$ (see
Example~\ref{thm:38}), so by reflection positivity we conclude from
Proposition~\ref{thm:29} that $F(S^1_{\textnormal{bounding}}\times Y)=1$,
hence the spin flip acts as~$(-1)^F$.
 
In the general case we use Corollary~\ref{thm:113} to construct an $H_{k+\ell
}$-structure on the Cartesian product of a $\Spin_k$-manifold and an $H_{\ell
}$-manifold.  Then the argument in the preceding paragraph goes through for
$Y$~an $H_{n-1}$-manifold and the same spin circles.
  \end{proof}

\appendix

   \section{The CRT theorem for general symmetry types}\label{sec:10}

In~\S\ref{subsec:a3.2} we take as our starting point a relativistic quantum
field theory in Minkowski spacetime.  Positivity of energy gives analytic
correlation functions for which the Minkowski correlation functions are
boundary values; Euclidean correlation functions are the restriction to a
suitable subdomain.  This leads to the CRT theorem
(Theorem~\ref{thm:172}),\footnote{It is usually called the CPT theorem, but
we follow the nomenclature in~\cite{W1}, which is more appropriate for
arbitrary dimensions: the~`P' in~`CPT' is understood to be the parity
transformation that acts as~$-1$ on space and so is orientation-preserving
if the dimension of spacetime is odd; by contrast, the~`R' in~`CRT' denotes
reflection in a single spatial direction and is orientation-reversing in all
dimensions.  The `C'~is best read as `complex conjugation'.}  and we outline
Jost's proof~\cite{J}, extended to general symmetry types.  Recall that the
symmetry group~$H\ono$ of a relativistic quantum field theory acts by
time-orientation preserving transformations; see~\eqref{eq:12}.  The CRT
theorem asserts that a larger symmetry group, including time-orientation
reversing transformations, also acts; the time-reversing elements act
antilinearly.  There is a subtlety in the Lorentz spin central extensions,
flagged in~\cite{GT},\footnote{The setting of~\cite{GT} is ``formal field
theory'' as opposed to that in the Wightman axioms.} which we elucidate and
generalize to arbitrary symmetry types in~\S\ref{subsec:10.1}.  This subtlety
is present even in the spin case without time-reversal symmetry.  It implies,
for example, that the ten Lorentz signature symmetry groups for free fermion
theories~(\S\ref{sec:8}) embed in Clifford algebras, a fact which is implicit
in~\S\ref{subsubsec:8.4.4}.  In this appendix we work in the framework of
Wightman quantum field theory.  One consequence of our discussion
(Remark~\ref{thm:174}) is a justification of the correspondence between the
alternatives
  \begin{equation}\label{eq:a48}
     \textnormal{$\pinp$-structure \quad vs.\ \quad  $\pinm$-structure} 
  \end{equation}
in Wick-rotated field theory and the alternatives 
  \begin{equation}\label{eq:a49}
     \textnormal{$T^2=(-1)^F$ \quad vs.\ \quad $T^2=1$ \phantom{MM}} 
  \end{equation}
for the action of time-reversal~$T$ on the Hilbert space~$\sH$ of states.  We
begin in~\S\ref{subsec:a3.1} with a review of pin groups and pin manifolds,
which also serves to fix some conventions about Clifford algebras..

We assume the dimension of spacetime is~$n\ge3$.

  \subsection{Pin groups and pin manifolds}\label{subsec:a3.1}

References for this section include~\cite{ABS,BDGK,KT1}.  While we assume the
dimension~$n\ge3$, with minor modifications the discussion goes through for~
$n=1,2$ as well.

\medskip
 \subsubsection{Pin groups and Clifford algebras}\label{subsubsec:a.1.1} 
 We take Lorentz signature as our starting point.  Let $\RR^{1,n-1}$ be the
standard vector space with basis~$e_0,e_1,\dots ,e_{n-1}$ and the standard
inner product: $\langle e_0,e_0 \rangle=1$, $\langle e_i,e_i
\rangle=-1,\;i=1,\dots ,n-1$, and $\langle e_\mu ,e_\nu \rangle=0,\;\mu \not=
\nu $.  Its isometry group is the orthogonal group~$O_{1,n-1}$.  The group of
components of~$O_{1,n-1}$ is isomorphic to~$\pmo\times \pmo$; an orthogonal
transformation either preserves or exchanges the two components of timelike
vectors~$\xi $ (vectors with $\langle \xi ,\xi \rangle>0$), and it either
preserves or reverses the orientation of any spacelike codimension~1
subspace.  In terms of the block matrix~$\left(\begin{smallmatrix} a&\alpha
\\\eta &A \end{smallmatrix}\right)\in O_{1,n-1}$ the first question is the
sign of the real number~$a$ and the second the determinant of the
$(n-1)\times (n-1)$ matrix~$A$.  The identity component of~$O_{1,n-1}$ has a
unique (up to isomorphism) connected double covering group~$\Spin_{1,n-1}$.
It is contained in the even subalgebra of a real Clifford algebra, and there
are two equally good choices for the signs:
  \begin{equation}\label{eq:a9}
     \begin{aligned} \Cliff_{1,n-1}:&\quad e_0^2=+1, \quad e_i^2=-1,\quad
     i=1,\dots 
      ,n-1, \\ \Cliff_{n-1,1}:&\quad e_0^2=-1, \quad e_i^2=+1,\quad i=1,\dots
     ,n-1. \\ 
      \end{aligned} 
  \end{equation}

The Lorentz orthogonal group~$O_{1,n-1}$ has a complexification~$O_n(\CC)$
consisting of complex $n\times n$~orthogonal matrices.  This complex group
has two components distinguished by the determinant, which is~$\pm1$.  The
identity component~$SO_n(\CC)$ has a subgroup that is the union of the two
components of~$O_{1,n-1}$ of matrices with determinant~1.  Also,
$SO_n(\CC)$~has a unique connected double covering group~$\Spin_n(\CC)$,
which contains~$\Spin_{1,n-1}$ as a subgroup.  The complex Lie
group~$O_n(\CC)$ deformation retracts onto its maximal compact
subgroup~$O_n$, which is the group of orthogonal symmetries of the real
vector space spanned by
  \begin{equation}\label{eq:a11}
     \te_0=i\,e_0,\;\te_1=e_1,\;\dots ,\;\te_{n-1}=e_{n-1} 
  \end{equation}
with its inherited negative definite inner product.  Here $i$~is a choice of
complex number with~$i^2=-1$.  The identity component~$SO_{n}$ has a unique
connected double covering group~$\Spin_n$, which is the maximal compact
subgroup of~$\Spin_n(\CC)$.  It is contained in the even subalgebra of a real
Clifford algebra, and again there are two equally good choices for the signs:
  \begin{equation}\label{eq:a10}
     \begin{aligned} \Clm n:&\quad \te_\mu ^2=-1,\quad
     \mu =0,\dots ,n-1, \\ \Clp n:&\quad \te_\mu^2=+1,\quad
     \mu =0,\dots ,n-1. \\ \end{aligned} 
  \end{equation}

The four-component orthogonal group~$O_{1,n-1}$ has many double cover groups
with identity component $\Spin_{1,n-1}$; we discuss two of them
in~\S\ref{subsec:10.1}.  In the remainder of this subsection we focus on the
two-component compact orthogonal group~$O_n$, which has two double
covers~$\Ppm_n$ with identity component~$\Spin_n$.  Each is a subgroup of
invertible elements in a real Clifford algebra: $\Ppm_n\subset \Cliff_{\pm
n}$.  They are group extensions
  \begin{equation}\label{eq:a12}
     1\longrightarrow \pmo\longrightarrow \Ppm_n\longrightarrow
     O_n\longrightarrow 1 
  \end{equation}
Observe that $\Pp_1\cong \zt\times \zt$ and $\Pm_1\cong \zmod4$.

\medskip 
 \subsubsection{Pin manifolds} A Riemannian manifold~$X$ has a principal
$O_n$-bundle of frames $\sBO(X)\to X$ whose points represent orthonormal
bases of the tangent spaces to~$X$.  The following is a special case of
Definition~\ref{thm:45}.

  \begin{definition}[]\label{thm:a5}
 A \emph{$\ppm$-structure} on~$X$ is a pair~$(P,\theta )$ consisting of a
principal $\Ppm_n$-bundle $P\to X$ and an isomorphism
$\sBO(X)\xrightarrow{\;\theta \;}P/\pmo$ of principal $O_n$-bundles. 
  \end{definition}

\noindent
 Pin structures, as spin structures, do not necessarily exist.  The
obstructions are given by Stiefel-Whitney classes: a $\pinp$-structure exists
on~$X$ if and only if\footnote{These are Stiefel-Whitney classes of the
\emph{tangent} bundle: $w_q(X)=w_q(TX)$.  There is a potential confusion with
Stiefel-Whitney classes of the \emph{stable normal} bundle, which is what
appears naturally in bordism theory.} $w_2(X)=0$ and a $\pinm$-structure
exists if and only if $(w_1^2+w_2)(X)=0$.  Double covers of~$X$ act on pin
structures as follows.  If $Q\to X$ is a double cover, viewed as a principal
$\pmo$-bundle, and $(P,\theta )$~is a $\ppm$-structure, then $Q\times \mstrut
_XP\to X$ is a principal $(\pmo\times \Ppm_n)$-bundle.  The $\Ppm_n$-bundle
$(Q\times \mstrut _XP)\,/\,\pmo\to X$ associated to the homomorphism
$\pmo\times \Ppm_n\to \Ppm_n$ (multiplication in~$\Ppm_n$ with first argument
restricted to the central subgroup in~\eqref{eq:a12}), along with a canonical
isomorphism of underlying $O_n$-bundles obtained from~$\theta $, is a
$\ppm$-structure.  The set of isomorphism classes of $\ppm$-structures, if
nonempty, is a torsor over the abelian group~$H^1(X;\zt)$; that is, this
group acts freely and transitively on the set of isomorphism classes.  There
is a canonical double cover of~$X$, the orientation double cover, whose
points represent orientations of the tangent spaces to~$X$.

  \begin{definition}[]\label{thm:a21}
 The \emph{$w_1$-involution} is the action of the orientation double cover on
pin structures. 
  \end{definition}

\noindent
 Recall that the equivalence class of the orientation double cover is
$w_1(X)\in H^1(X;\zt)$.

  \begin{remark}[]\label{thm:a17}
 Let $\ha $~be the automorphism of $\Ppm_n$ that is the identity
on~$\Spin_n$ and multiplication by the central element~$-1$ on the
complement; it covers the identity automorphism of~$O_n$.  An alternative
description of the $w_1$-transform~$(P^\sim,\theta )$ of a pin-structure
$(P,\theta )$ is the same manifold~$P$ with the same map~ $\theta $, but with
the $\Ppm_n$-action altered by precomposition with~$\ha $.  (To see this,
write the orientation double cover as $P/\Spin_n$ and construct the
isomorphism of $\Ppm_n$-bundles
  \begin{equation}\label{eq:a50}
     P/\Spin_n\times P\longrightarrow P^\sim
  \end{equation}
which maps $(\mathfrak{o},p)\mapsto p$ if $p\in \mathfrak{o}$ and
$(\mathfrak{o},p)\mapsto p\cdot (-1)$ if $p\notin\mathfrak{o}$.  Here
$\mathfrak{o}\subset P$ is a $\Spin_n$-orbit.)
  \end{remark}

  \subsection{Lorentz signature symmetry groups}\label{subsec:10.1}
 
This section is an exposition and elaboration of ideas in~\cite{GT}.  We
continue with the hypothesis~$n\ge3$, largely for convenience of exposition;
with minor modifications the discussion goes through for~ $n=1,2$ as well.
 
\medskip
 \subsubsection{Complex pin groups} The complex orthogonal group~$O_n(\CC)$
has two components.  The identity component~$SO_n(\CC)\subset O_n(\CC)$ has a
unique isomorphism class of nontrivial double cover groups, any
representative of which is called~$\Spin_n(\CC)$.
 
  \begin{proposition}[]\label{thm:a3}
 There are unique complex Lie groups~$\Ppm_n(\CC)$ with identity
component~$\Spin_n(\CC)$, which double cover~$O_n(\CC)$, and which
contain~$\Ppm_n$ as maximal compact subgroups.  Furthermore, any complex Lie
group that double covers~$O_n(\CC)$ and has identity component isomorphic
to~$\Spin_n(\CC)$ is isomorphic to either~$\Pp_n(\CC)$ or~$\Pm_n(\CC)$.
  \end{proposition}

  \begin{remark}[]\label{thm:a4}
 We warn that $\Ppm_n(\CC)$ are complex Lie groups, whereas the
group~`$\Pin^c_n$', which is defined in~\cite[\S3]{ABS} as a subgroup of the
complex Clifford algebra, is a compact real Lie group; it and twisted
variants appear in~\S\ref{sec:8}.
  \end{remark}

  \begin{proof}
 Up to isomorphism there is a unique double covering space $X\to O_n(\CC)$
whose inverse image over each component of~$O_n(\CC)$ is connected.  The
restriction over~$O_n\subset O_n(\CC)$ is isomorphic as a double covering
space to $\Ppm_n\to O_n$.  Choose an isomorphism of double covers and
transport the group structure, then extend the group structure on the
identity component~$\Spin_n$ to that of~$\Spin_n(\CC)$ on the entire
component~$X_+\subset X$ containing~$\Spin_n$.  Now use covering space theory
to extend the group structure to all of~$X$.  For example, setting
$X_-=X\setminus X_+$, lift the map $ X_+\times X_-\to O_n(\CC)_-$ to a map
$X_+\times X_-\to X_-$ using basepoints in the compact pin group.  In fact,
the extension of the group structure is determined by the square of a lift of
a single hyperplane reflection, for which there are two choices, and this
implies the last assertion.
  \end{proof}
 
\medskip
 \subsubsection{Double covers of Lorentz isometry
groups}\label{subsubsec:a.2.2} The two-component group $SO\ono\subset O\ono$
consists of isometries that preserve the overall orientation of~$\RM$.  Let
$\mu _m\subset \Cx$ be the group of $m^{\textnormal{th}}$~roots of unity.
Using the diagram
  \begin{equation}\label{eq:250}
     \begin{gathered}
     \xymatrix@C-20pt{\Spin_n(\CC)\ar@{^{(}->}[rr]\ar[dr]_{\pi
     _2}\;&&\Spin_n(\CC)\times \mstrut _{\mu _2}\mu _4\ar[dl]^{\pi _4}\\ 
     &SO_n(\CC)} \end{gathered} 
  \end{equation}
set $\tSO\ono^\alpha =\pi _2\inv (SO\ono)$, and let $\tSO\ono^\beta \subset
\Spin_n(\CC)\times \mstrut _{\mu _2}\mu _4$ be the union of~$\Spin\ono$ and
the complement of~$\pi _2\inv (\SOmt)$ in~$\pi _4\inv (\SOmt)$, where
$\SOmt$~is the non-identity component of~$SO\ono$.  For the pin groups let
$\tO\noo^\alpha $ and $\tO\ono^\alpha $ be the inverse image of $O\ono\subset
O_n(\CC)$ under the double cover homomorphisms $\Pp_n(\CC)\to O_n(\CC)$ and
$\Pm_n(\CC)\to O_n(\CC)$, respectively.  Finally, using the diagram 
  \begin{equation}\label{eq:251}
     \begin{gathered}
     \xymatrix@C-20pt{\Ppm_n(\CC)\ar@{^{(}->}[rr]\ar[dr]_{\pi
     _2}\;&&\Ppm_n(\CC)\times \mstrut _{\mu _2}\mu _4\ar[dl]^{\pi _4}\\
     &O_n(\CC)} \end{gathered} 
  \end{equation}
let $\tO\noo^\beta $~and $\tO\ono^\beta $~be the union of $\pi _2\inv (\OMt)$
and the complement of $\pi _2\inv (\OMmt)$ in $\pi _4\inv (\OMmt)$, where we
use the $+$~and $-$~pin groups, respectively.  Here $\OMmt$~is the complement
of~$\OMt\subset O\ono$, the components of time-reversing linear isometries.

  \begin{proposition}[]\label{thm:169}
 \ 
 \begin{enumerate}[{\textnormal(}1{\textnormal)}]

 \item Every double cover group of~$SO\ono$ whose identity component is
isomorphic to~$\Spin\ono$ is isomorphic to either~$\tSO\ono^\alpha $
or~$\tSO\ono^\beta $.

 \item The double cover group~$\tSO\ono^\beta $ of $SO\mstrut \ono$ is a
subgroup of the even subalgebras of~$\Cliff\noo$ and~$\Cliff\ono$.

 \item The double cover groups~$\tO\noo^\beta $ and~$\tO\ono^\beta $
of~$O\mstrut \ono$ are subgroups of~$\Cliff\mstrut \noo$ and~$\Cliff\mstrut
\ono$, respectively.

 \end{enumerate}
  \end{proposition}

\noindent 
 Summary: the $\alpha $-double covers are subgroups of complex (s)pin groups;
the $\beta $-double covers are subgroups of Lorentz signature Clifford
algebras.

  \begin{proof}
 For~(1), let $g\in S\OMmt$ be the diagonal matrix $\diag(-1,-1,+1,\dots
,+1)$.  Then the square of a lift of~$g$ to a double cover of~$SO\ono$ has
square the identity~$+1$ or the central element~$-1$ of~$\Spin\ono$.  By
covering space theory, as in the proof of Proposition~\ref{thm:a3}, we can
deduce that this dichotomy determines the group structure on the double
cover.  

The element~$e_0e_1$ in the Clifford algebra (of either signature~$(n-1,1)$
or~$(1,n-1)$) acts on~$\RM$ as~$g$ and squares to~$+1$.  On the other hand,
$g$~lies in $SO\ono\cap SO_n\subset SO_n(\CC)$, so a lift of~$g$ to~
$\Spin_n(\CC)$ lies in the compact spin group~$\Spin_n$ where it squares
to~$-1$, as we compute in the Clifford algebra~$\Cliff_{\pm n}$.  This is the
essential point in the proof of~(2).
 
As for~(3) there are double covers $\Pin\noo\subset \Cliff\noo$ and
$\Pin\ono\subset \Cliff\ono$ of~$O\ono$, as defined in~\cite{ABS},
\cite[\S1.2]{LM}.  By~(2) the restriction over~$SO\ono$ is isomorphic
to~$\tSO\ono^\beta $.  The element $\diag(-1,+1,\dots ,+1)\in \OMmt$ lifts
to~$e_0$ in the Clifford algebra, and its square is given in~\eqref{eq:a9}.
Arguing as above with the compact pin groups we deduce that this is opposite
the square of a lift in the corresponding complex pin group.  This is the
new step in proving the isomorphisms
  \begin{equation}\label{eq:252}
     \begin{aligned} \Pin\mstrut \noo&\cong \tO\noo^\beta \\ \Pin\mstrut
     \ono&\cong \tO\ono^\beta \\ \end{aligned} 
  \end{equation}
  \vskip-2.8em
  \end{proof}

\medskip 
 \subsubsection{General Lorentz signature symmetry groups} There are analogs
of the $\alpha $~and $\beta $-extensions of the Lorentz signature vector
symmetry group~$H\ono$ for an arbitrary symmetry type, which as
in~\S\ref{subsec:12.1} is the quotient of the full symmetry group of a
relativistic quantum field theory by translations.  It comes equipped with a
homomorphism $\rho _n\:H\ono\to\OMt$.  We use the Structure
Theorem~\ref{thm:5}, and in particular~\eqref{eq:16}, \eqref{eq:226},
and~\eqref{eq:227} to define the $\alpha $~and $\beta $-extensions~$H\ab\ono$
of~$H\mstrut \ono$ simultaneously.  Set
  \begin{equation}\label{eq:253}
     SH\ab\ono\cong \tSO\ono\ab\times K\bigm / \langle(-1,k_0)\rangle.
  \end{equation}
If the image of~ $\rho _n$ is~$S\OMt$,  set $H\ab\ono=SH\ab\ono$.  If
$\rho _n$~is surjective,  define~$\tH\ono\ab$ by pullback
  \begin{equation}\label{eq:254}
     \begin{gathered} \xymatrix{
     1\ar[r]&K\ar[r]\ar@{=}[d]&\tH\ono\ab\ar@{->>}[d]\ar[r]^{}&
     \tO\noo\ab\ar@{->>}[d] \ar[r] & 1\\
     1\ar[r]&K\ar[r]&J\ar[r]^{}&
     \pmo\ar[r] & 1} \end{gathered} 
  \end{equation}
where the right vertical map is the determinant homomorphism.  Then let 
  \begin{equation}\label{eq:255}
     H\ono\ab\cong \tH\ono\ab \bigm / \langle(-1,k_0)\rangle. 
  \end{equation}
We observe that $H\onoa $~is a real subgroup of the complex Lie
group~$H_n(\CC)$, the inverse image of~$O\ono$ under the homomorphism $\rho
_n\:H_n(\CC)\to O_n(\CC)$ in~\eqref{eq:13}.  Also, our notation is set up so
that $\Spin\ono\ab\cong \tSO\ono\ab $.

\medskip 
 \subsubsection{Extensions of real representations} As just remarked, the
$\alpha $-extension sits as a subgroup of the complex symmetry group.  One
key feature of the $\beta $-extension is the following. 

  \begin{proposition}[]\label{thm:170}
 Let $R=R^0\oplus R^1$ be a $\zt$-graded real representation of~$H\ono$ such
that $k_0\in K\subset H\ono$ acts as the grading operator.  Let
$R_{\CC}:=R\otimes \mstrut _{\RR}\CC$ denote the complexification, which
carries an action of the complex Lie group~$H_n(\CC)$, hence of the
subgroup~$H\onoa $.

 \begin{enumerate}[{\textnormal(}1{\textnormal)}]

 \item If $h\in H\onoa \setminus H\mstrut \ono $, then $h(R^0)=R^0$ and
$h(R^1)=\sqmo R^1$.

 \item There is a canonical extension of the action of~$H\ono$ on~$R$ to an
action of~$H\ono^\beta $.

 \end{enumerate}

  \end{proposition}

\noindent
 All Lie groups that appear are ungraded, so act by even transformations
of~$R$.  The conclusion is that the $\beta $-extension acts on \emph{real}
representations of~$H\ono$.

  \begin{proof}
 For~(1) it suffices to check for a single element $h\in H\onoa \setminus
H\mstrut \ono$.  By Corollary~\ref{thm:113}, anti-Wick rotated to Lorentz
signature, we choose~$h$ to be the image in~$H\onoa $ of a lift of
  \begin{equation}\label{eq:257}
     \begin{pmatrix} -1&0\\0&-1 \end{pmatrix} \in SO_{1,1}\cap SO_2\subset
     SO_2(\CC)\subset SO_n(\CC) 
  \end{equation}
to~$\Spin_n(\CC)$.  In the compact spin group $\Spin_2\subset \Spin_2(\CC)$
the element~$h$ is represented as~$f_0f_1$ and is connected to the identity
by the curve $\cos t/2 + \sin t/2 \,f_0f_1$, $0\le t\le\pi $, where we embed
$\Spin_2\subset \Cliff_{-2}$; see~\eqref{eq:a10}.  Complex conjugation,
defined so that $\Spin_{1,1}\subset \Spin_2(\CC)$ is real, takes this curve
to the curve $\cos t/2 - \sin t/2\,f_0f_1$, $0\le t\le\pi $
in~$\Spin_2\subset \Spin_2(\CC)$.  In particular, the complex conjugate
of~$f_0f_1$ is~$-f_0f_1$.  Since $-1$~maps to~$k_0$ and acts as the grading
operator, $f_0f_1$~is a real operator on~$R^0_{\CC}$ and a purely imaginary
operator on~$R^1_{\CC}$.  This proves~(1).
 
Consider the diagram 
  \begin{equation}\label{eq:258}
     \begin{gathered}
     \xymatrix@C-20pt{H\onoa \ar@{^{(}->}[rr]\ar[dr]_{\pi
     _2}\;&&H\ono^{\alpha\phantom{\beta }} \times \mstrut _{\mu _2}\mu
     _4\ar[dl]^{\pi _4}\\ &\quad O\ono} \end{gathered} 
  \end{equation}
in which $\mu _2\subset H\onoa $ is generated by~$k_0$.  Then $H\ono^\beta
\subset H\ono^{\alpha\phantom{\beta }} \times \mstrut _{\mu _2}\,\mu _4$ is the
union of~$H\ono\mstrut $ and the complement of $\pi _2\inv (\OMmt)$ in $\pi
_4\inv (\OMmt)$.  Let $\mu _4\subset \Cx$ act on~$R^1_{\CC}$ via scalar
multiplication and on~$R^0_{\CC}$ trivially.  Then by~(1) the restriction to
$H\ono^\beta \subset H\ono^{\alpha\phantom{\beta }} \times \mstrut _{\mu
_2}\,\mu _4$ is real, i.e., preserves~$R\subset R_{\CC}$.  This proves~(2).
  \end{proof}

  \subsection{Wick rotation and the CRT theorem}\label{subsec:a3.2}

In this section we sketch a rigorous argument for the CRT theorem in
relativistic quantum field theory.  We use the analytic continuation of
correlation functions, working in the framework of Wightman quantum field
theory~\cite{SW,GJ,Kaz}.  Our purpose is to treat general symmetry types.
Even for theories with Lorentz symmetry group~$H\ono=\Spin\ono $ there is a
subtlety: the group~$\tSO\ono^\alpha $ acts on the holomorphic correlation
functions, whereas the group~$\tSO\ono^\beta $ acts on the Minkowski
spacetime correlation functions.  (See~\S\ref{subsubsec:a.2.2} for the
definitions of these Lie groups.)  This argument also demonstrates why only
the ``Cliffordian''~\cite{BDGK} Lorentz signature pin groups~$\Pin_{n-1,1}$
and~$\Pin_{1,n-1}$ can be symmetries of a relativistic quantum field theory
instead of more general possible double covers of~$O_{1,n-1}$; see
Remark~\ref{thm:174}.  We assume~$n\ge3$.
 
Recall from~\S\ref{subsec:12.1} that Minkowski spacetime~$M^n$ is an
$n$-dimensional affine space whose vector space~$V=\RM$ of translations is
equipped with an inner product of signature~$(1,n-1)$ and a choice of
component~$V_+$ of the space $\{\xi :\langle \xi ,\xi \rangle>0\}$ of
timelike vectors.\footnote{The latter choice is required in order to
formulate the positivity of energy.}  To Wick rotate to imaginary time, fix
an orthogonal splitting $V= U\oplus \Up$ with $U$~a 1-dimensional timelike
subspace.  Then the Euclidean translation group is $V_E=\sqmo\,U\oplus \Up$
and the corresponding Euclidean space is $E=M\times \mstrut _VV_E$, an affine
space over~$V_E$.  Complexified Minkowski spacetime is $M_{\CC}=M\times
\mstrut _VV_{\CC}$, where $V_{\CC}$~is the complexification of~$V$.  The
symmetry group~$H\ono$ of a relativistic quantum field theory acts on~$M^n$
by time-orientation preserving transformations via a homomorphism $\rho
_n\:H\mstrut \ono\to\OMt$, as in~\eqref{eq:12}.

  \begin{theorem}[CRT Theorem]\label{thm:172}
 Let $\sQQ$ denote a relativistic quantum field theory with symmetry
group~$H\ono$.  Then the symmetry extends to $H\ono^\beta $; elements of
$H\ono^\beta \setminus H\mstrut \ono$ act antilinearly.
  \end{theorem}

\noindent
 Here $\sQQ$~is a quantum field theory in the Wightman axiomatic framework.
It is determined by its \emph{correlation functions}, called \emph{Wightman
functions}; see~\cite[\S1.3]{Kaz}.  For simplicity of notation we only
discuss 2-point functions in this account.  A precise version of
Theorem~\ref{thm:172} is~\eqref{eq:275} below.

The \emph{fields} in~$\sQQ$ are defined by a finite dimensional $\zt$-graded
real representation
  \begin{equation}\label{eq:259}
     \sigma \:H\ono\longrightarrow \Aut(R).
  \end{equation}
We write~$R=R^0\oplus R^1 $ according to the grading; elements of~$H\ono$
preserve the grading.  The spin-statistics theorem, which we assume in this
account, asserts that the special element $k_0\in K\subset H\ono$ defined in
Theorem~\ref{thm:5}(2) acts as the grading operator on~$R$.  Write
$\RC=R\otimes \mstrut _{\RR}\CC$ for the complexification.  Classical fields
are functions $M^n\to R$.  Quantum fields are $R$-valued operator-valued
distributions~$\Phi=\Phi ^0+\Phi ^1$ on~$M^n$.  The 2-point ``function'' is a
complex distribution whose value on Schwartz functions $f_i\:M^n\to R^*$ is
written
  \begin{equation}\label{eq:260}
     \langle \Phi (f_1)\Phi (f_2) \rangle = \int_{M^2}dp_1dp_2\,
     f_1(p_1)f_2(p_2)\,\langle \Phi (p_1)\Phi (p_2) \rangle,
  \end{equation}
where $\langle \Phi (p_1)\Phi (p_2) \rangle$ denotes the kernel of the
$\RC^{\otimes 2}$-valued distribution on~$M^{\times 2}$.  The theory~$\sQQ$ has
a $\zt$-graded Hilbert space~$\sH=\sH^0\oplus \sH^1$ of states, constructed
from the correlation functions, and a distinguished vacuum vector~$\Omega \in
\sH^0$.  The \emph{field operators}~$\Phi (f)$ act as unbounded operators
on~$\sH$, and the 2-point function is the vacuum expectation value of the
product of the field operators:
  \begin{equation}\label{eq:261}
     \langle \Phi (f_1)\Phi (f_2) \rangle = \langle \Omega \,,\,\Phi
     (f_1)\Phi (f_2)\Omega \rangle_{\sH}\mstrut . 
  \end{equation}
There is a unitary representation of the affine extension of~$H\ono$
on~$\sH$---all symmetries preserve the $\zt$-grading.  The vacuum vector and
2-point function are invariant under that action, in particular under
translations.  Hence there is an $\RC^{\otimes 2}$-valued distribution on~$V$
with kernel
  \begin{equation}\label{eq:262}
     W(\xi ):=\langle \Phi (p)\Phi (p+\xi ) \rangle,\qquad p\in M^n,\quad \xi
     \in V, 
  \end{equation}
which is independent of~$p$.

The important step in Jost's proof is the construction of holomorphic
correlation functions from which the Wightman functions are recovered as
boundary values~\cite[\S2.1]{Kaz}.  This is a consequence of the positivity
of energy and geometric arguments.  The holomorphic 2-point function 
  \begin{equation}\label{eq:263}
     \WC\:\sD\longrightarrow \RC^{\otimes 2} 
  \end{equation}
has domain $\sD\subset V_{\CC}$ that is connected and $H_n(\CC)$-invariant.
Define the \emph{backward tube} $\sT = V - iV_+\subset V_{\CC}$, where $i$~is
a choice of square root of~$-1$.  Then\footnote{Note $SO_n(\CC)(\sT) =
O_n(\CC)(\sT)$.}
  \begin{equation}\label{eq:264}
      \sD=SO_n(\CC)(\sT)\cup -SO_n(\CC)(\sT). 
  \end{equation}
An important feature of~$\sD$ is that it contains \emph{Jost
points},\footnote{Here we use~$n\ge3$.}  which in this case of 2-point
functions are the real spacelike vectors~$\xi \in V\subset V_{\CC}$ that
satisfy $\langle \xi ,\xi \rangle>0$.  From~\eqref{eq:264} we see $\sT\subset
\sD$, and as stated $W$~is a boundary value of~$\WC$:
  \begin{equation}\label{eq:265}
     W(\xi ) = \lim\limits_{\;\;\epsilon \to 0^+}\WC(\xi -\epsilon i
     \eta ),\qquad \xi \in V,\quad \eta \in V_+, 
  \end{equation}
and the limit is independent of~$\eta $.  We also have $V_E\setminus
\{0\}\subset \sD$, and the Wick-rotated Euclidean 2-point function is the
restriction of~$\WC$ to~$V_E\setminus \{0\}$.
 
We collect some properties of the holomorphic correlation functions.  First,
since the inner product on~$\sH$ is even, it follows that 
  \begin{equation}\label{eq:266}
     W_{\CC}^{\vphantom{0}} = \WC^0+\WC^1 
  \end{equation}
where $\WC^q$~takes values in~$(\RC^q)^{\otimes 2}$, $q=0,1$.  Note that both
$\WC^0$~and $\WC^1$~are even.  Next, as already stated, $\WC$~ is
$H_n(\CC)$-invariant, hence invariant under the subgroup $H\ono^\alpha
\subset H_n(\CC)$:
  \begin{equation}\label{eq:267}
     \WC(\zeta ) = \sigma (h^\alpha )^{\otimes 2}\,\WC\bigl(\rho _n(h^\alpha )\zeta
     \bigr),\qquad h^\alpha \in H\ono^\alpha ,\quad \zeta \in \sD. 
  \end{equation}
Now if $\xi $~is real and spacelike, then since field operators at
spacelike separated points commute (in the graded sense), and since real
spacelike (Jost) points are in the domain~$\sD$, we have 
  \begin{equation}\label{eq:268}
     \begin{aligned} \WC^0(-\xi )&= \phantom{-}\WC^0(\xi ) \\ \WC^1(-\xi )&=
      -\WC^1(\xi ) \\ \end{aligned} 
  \end{equation}
Continuing with $\xi $~real and spacelike, we claim 
  \begin{equation}\label{eq:269}
     \begin{aligned} \overline{\WC^0(\xi )}&= \phantom{-}\WC^0(\xi ) \\
     \overline{\WC^1(\xi )}&= -\WC^1(\xi ) \\ \end{aligned} 
  \end{equation}
Since such~$\xi $ lie in~$\sD$, and $\sD$ is connected, we deduce a Schwarz
reflection formula valid for all~$\zeta \in \sD$:
  \begin{equation}\label{eq:273}
     \begin{aligned} \overline{\WC^0(\zeta )}&= \phantom{-}\WC^0(\bz ) \\
     \overline{\WC^1(\zeta )}&= -\WC^1(\bz ) \\ \end{aligned} 
  \end{equation}
The manipulation that justifies~\eqref{eq:269} is, for any~$p\in M^n$ and
$\xi \in V$, 
  \begin{equation}\label{eq:270}
     \overline{\WC(\xi )} = \langle \Phi (p)\Phi (p+\xi
      )\Omega \,,\,\Omega \rangle = \langle \Omega \,,\,\Phi (p+\xi )\Phi
      (p)\Omega \rangle = \WC(-\xi );
  \end{equation}
then we apply~\eqref{eq:268}.  The middle step is straightforward in the even
case: $\Phi^0 (q)$~is self-adjoint for $q$~real.  The corresponding
manipulation in the odd case uses the adjoint of the odd operator~$\Phi
^1(q)$, which involves a tricky sign\footnote{We thank Greg Moore for help
straightening this out.} as we explain in the following remark.

  \begin{remark}[]\label{thm:173}
 The usual physics conventions are: the norm square of an odd vector in~$\sH$
is real and positive; for any two operators~$A,B$ we have
$(AB)^*=B^*A^*$---there is no sign even if both $A$~and $B$~are odd; and the
odd field operator~$\Phi ^1(q)$ is self-adjoint in the usual sense.  However,
the Koszul sign rule demands that the first two of these be modified to: the
norm square of an odd vector in~$\sH$ is purely imaginary and lies on one of
the two rays of nonzero purely imaginary numbers, the choice of which is a
convention (Example~\ref{thm:159}); if $A,B$ are operators that have
definite parities~$|A|,|B|$, then~\cite[\S4.4]{DM}
  \begin{equation}\label{eq:271}
     (AB)^*=(-1)^{|A||B|}B^*A^*. 
  \end{equation}
If we use these conventions, then the odd field operator~$\Phi ^1(q)$ is not
self-adjoint, but rather 
  \begin{equation}\label{eq:272}
     \Phi ^1(q)^*=i\,\Phi ^1(q) 
  \end{equation}
One justification for~\eqref{eq:272} is to consider the $*$-structure on the
complex operator algebra, and to note that \eqref{eq:271}~implies that the
square of an odd self-adjoint operator is even skew-adjoint, and so if $\Phi
^1(q)$~were self-adjoint we would contradict expectations for the
quantization of real fields.  We remark that the factor~$i$ in~\eqref{eq:272}
already occurs in quantum mechanics; see~\cite[(4.10)]{FM1}.  The middle step
in~\eqref{eq:270} is valid with either the standard physics conventions or
the Koszul-compatible notion of adjointness supplemented with~\eqref{eq:272}.
  \end{remark}

  \begin{proof}[Proof of Theorem~\ref{thm:172}]
 Fix~$h^\alpha \in H\ono^\alpha \setminus H\mstrut \ono$.  Then $h^\alpha
$~reverses the time orientation, in other words, $H^\alpha (V_+)=-V_+$.
Hence for~$\xi \in V$ we use~\eqref{eq:265}, \eqref{eq:267},
and~\eqref{eq:273} to deduce that for~$\xi \in V$ and~$q=0,1$ we have 
  \begin{equation}\label{eq:274}
     \begin{aligned} \overline{W^q(\xi )} &= \lim\limits_{\;\;\epsilon \to0^+}
      \overline{\WC^q(\xi -\epsilon i\eta )} \\ &=
      \lim\limits_{\;\;\epsilon \to 0^+} \overline{\sigma (h^\alpha
      )^{\otimes 2}\,\WC^q\bigl(\rho _n(h^\alpha )\xi -\epsilon i\rho
     _n(h^\alpha 
      )\eta \bigr)} \\ &= \lim\limits_{\;\;\epsilon \to 0^+}(-1)^q\sigma
      (h^\alpha )^{\otimes 2}\,\WC\bigl(\rho _n(h^\alpha )\xi + \epsilon i\rho
      _n(h^\alpha )\eta \bigr) \\ &= \; (-1)^q\sigma (h^\alpha )^{\otimes
     2}\,W\bigl(\rho      _n(h^\alpha )\xi 
      \bigr).\end{aligned} 
  \end{equation}
To pass to the third equation we use the fact that $\sigma (h^\alpha )$~is
real on even vectors (Proposition~\ref{thm:170}(1)).  The construction that
proves Proposition~\ref{thm:170}(2) combines with~\eqref{eq:274} to yield
  \begin{equation}\label{eq:275}
     \overline{W^q(\xi )} = \sigma (h^\beta )^{\otimes 2}\,W\bigl(\rho
     _n(h^\beta )\xi \bigr),\qquad  h^\beta \in H\ono^\beta
\setminus H\ono\mstrut, \quad \xi \in V.
  \end{equation}
This is the precise statement that the Minkowski spacetime 2-point function
is antilinear-invariant under elements of~$H\ono^\beta \setminus H\ono\mstrut
$.
  \end{proof}

  \begin{remark}[]\label{thm:174}
 If $\sQQ$~is a relativistic quantum field theory with fermionic states and
time-reversal symmetry, and no other internal symmetries, then $H\ono$ is a
double cover of~$S\OMt$ whose identity component is isomorphic
to~$\Spin\ono$.  The complex Lie group~$H_n(\CC)$ is then a double cover
of~$O_n(\CC)$ whose identity component is isomorphic to~$\Spin_n(\CC)$.
Proposition~\ref{thm:a3} implies that $H_n(\CC)$~is isomorphic
to~$\Pp_n(\CC)$ or~$\Pm_n(\CC)$.  The construction with~\eqref{eq:251}
and~\eqref{eq:252} tells that the group~$H\ono^\beta $ is~$\Pin\mstrut \noo$
and~$\Pin\mstrut \ono$, respectively.  Recalling the sign
convention~\eqref{eq:a9} for Clifford algebras, this proves the
correspondence between~\eqref{eq:a48} and~\eqref{eq:a49} and also limits the
possible symmetry groups on relativistic quantum field theories to the
Cliffordian pin groups.
  \end{remark}

   \section{Involutions on categories and duality}\label{sec:11}

  \begin{definition}[]\label{thm:15}
 Let $\sC$~be a category.

      \begin{enumerate}[{\textnormal(}1{\textnormal)}]

      \item An \emph{involution} of~$\sC$ is a pair~$(\tau
,\eta )$ of a functor $\tau \:\sC\to\sC$ and a natural isomorphism $\eta
\:\id_{\sC}\to \tau ^2$ such that for any~$x\in \sC$ we have $\tau \eta
_x=\eta _{\tau x}$ as morphisms $\tau x\to \tau ^3x$.

 \item A \emph{fixed point} of~$\tau $ is a pair~$(x,\theta )$ of an
object~$x\in \sC$ and an isomorphism $x\xrightarrow{\;\theta \;}\tau x$ such
that $\tau \theta \circ \theta =\eta _x$ as morphisms $x\to \tau ^2x$.

 \end{enumerate}
  \end{definition}

\noindent
 If $\sC$~is a symmetric monoidal category, then the involution~$\tau $ is
required to be a symmetric monoidal functor: for $x,y\in \sC$ there is given
an isomorphism $\tau x\otimes \tau y\xrightarrow{\;\cong \;}\tau (x\otimes
y)$ and these isomorphisms are compatible with the symmetry and with~$\eta $.

  \begin{example}[]\label{thm:16}
 Let $\sC=\Vect_{\CC}$ be the category of complex vector spaces and linear
maps.  Define $\tau \:\sC\to\sC$ to be the functor that takes complex vector
spaces and linear maps to their complex conjugates.  (The complex conjugate
vector space is the same underlying real vector space with the sign of
multiplication by $\sqrt{-1}\in \CC$ reversed; the complex conjugate of a
linear map is the same map of sets.)  Then there is a canonical
identification of $\tau ^2$ with $\id_{\sC}$.  A fixed point is a complex
vector space with a real structure.  As a variation, if $\sC=s\!\Vect_{\CC}$
is the category of super ($\zt$-graded) vector spaces and $\tau $~complex
conjugation as above, but now $\eta $~is composed with the exponentiated
grading automorphism (denoted `$(-1)^F$' in the physics literature), then a
fixed point is a super vector space with a real structure on its even part
and a quaternionic structure on its odd part.  If we restrict to the
subgroupoid~$\sC^\times $ of super lines and isomorphisms, then all fixed
points are even.
  \end{example}

  \begin{definition}[]\label{thm:18}
 Let $(\tau ,\eta )$~be an involution on a category~$\sC$.  The \emph{fixed
point category}~$\sC^\tau $ has as objects fixed points~$(x,\theta )$, and a
morphism $(x,\theta )\to(x',\theta ')$ in~$\sC^\tau $ is a morphism
$(x\xrightarrow{f}x')\in \sC$ such that the diagram
  \begin{equation}\label{eq:39}
     \begin{gathered} \xymatrix{x\ar[r]^{f} \ar[d]_{\theta } &
     x'\ar[d]^{\theta '} \\ \tau x\ar[r]^{\tau f} & \tau x'} \end{gathered} 
  \end{equation}
commutes.  There is a \emph{forgetful functor} $\sC^\tau \to\sC$ that maps
$(x,\theta )\mapsto x$. 
  \end{definition}

  \begin{example}[]\label{thm:83}
 Let $\sC$~be the \emph{groupoid} of $\Zo$-torsors:\footnote{Recall that
$\ZZ(1)=2\pi \sqmo\ZZ\subset \CC$.} an object~$T$ is a set with a simply
transitive action of the additive group~$\Zo$ and a morphism $T\to T'$ is an
isomorphism that commutes with the $\Zo$-actions.  Let $\tau $~be the
involution that sends a torsor~$T$ to its dual $\Hom\mstrut _{\Zo}(T,\Zo)$
and sends a morphism to its inverse transpose.  The dual of~$T$ may be
identified with~$T$ as a set; the dual $\Zo$~action by~$\zeta \in \Zo$ is the
original action by~$\bar\zeta $.  The fixed point category~$\sC^\tau $ is
equivalent to the set~$\zt$: there are two isomorphism classes of objects and
no nontrivial automorphisms.  The first, which we call `Type~P', is the
torsor~$\Zo$ with complex conjugation~$\theta $ as a map to the dual torsor.
The second, which we call `Type~N', is the torsor $\pi \sqrt{-1}+\Zo$ with
complex conjugation~$\theta $.  Observe that in the Type~P case the
involution~$\theta $ has a fixed point whereas in the Type~N case it does
not.  Also, $\Zo$-torsors form a Picard groupoid, as do torsors for any
abelian group, and the fixed point category is a Picard groupoid as well.
The Type~P torsor is the tensor unit; the square of a Type~N torsor has
Type~P.  The names derive from the family $\exp\:\CC\to\Cx$ of $\Zo$-torsors
with complex conjugation acting.  There are two components~$\RR^{>0}$
and~$\RR^{<0}$ of fixed points in the base.  The fiber of~$\exp$ has Type~P
over positive real numbers and Type~N over negative real numbers; the
representatives described above  are $\exp\inv (+1)$ and $\exp\inv (-1)$,
respectively.
  \end{example}

  \begin{definition}[]\label{thm:17}
 Let $\sB,\sC$~be categories with involutions and $F\:\sB\to\sC$ a functor.
Then \emph{equivariance data} for~$F$ is an isomorphism $\phi \:F\tau
_{\sB}\xrightarrow{\;\cong \;}\tau _{\sC}F$ of functors $\sB\to\sC$ such that
for every object $x\in \sB$ the
diagram
  \begin{equation}\label{eq:38}
     \begin{gathered} \xymatrix@C+2em@R+.7em{Fx\ar[r]^{F\eta _{\sB}}\ar[dr]_{\eta
     _{\sC}}&F\tau ^2_{\sB}x\ar[d]^{\phi ^2} \\  &\tau _{\sC}^2Fx}
     \end{gathered} 
  \end{equation}
commutes. 
  \end{definition}

\noindent
 There are additional compatibilities for a symmetric monoidal functor
between symmetric monoidal categories; we do not spell them out.  We often
loosely say that ``$F$~is an equivariant functor'', but it is important to
remember that equivariance is data+condition, not simply a condition.
 
Next, we review duality in a symmetric monoidal category.  Let $\sC$~be a
symmetric monoidal category and $x\in \sC$.  Denote the tensor unit
by~$1\in \sC$.  (The tensor unit in~$\bne$ is the empty set as an
$(n-1)$-dimensional manifold; the tensor unit in~$\Vect_{\CC}$ is the trivial
1-dimensional vector space~$\CC$.)   

  \begin{definition}[]\label{thm:19}
 Let $x$~be an object in a symmetric monoidal category~$\sC$.  \emph{Duality
data} for~$x$ is a triple~$(x^\vee,c,e)$ consisting of an object~$x^\vee\in
\sC$ together with morphisms $c\:1\to x\otimes x^\vee$ and $e\:x^\vee\otimes
x\to 1$ such that the compositions
  \begin{equation}\label{eq:27}
     \begin{aligned} &x\xrightarrow{\;\;c\otimes \id\;\;}\phantom{j}x\otimes
     x^\vee\otimes 
      x\xrightarrow{\;\;id\otimes e\;\;}\phantom{^\vee}x \\
     &x^\vee\xrightarrow{\;\;\id\otimes 
      c\;\;}x^\vee\otimes x\otimes x^\vee\xrightarrow{\;\;e\otimes
      \id\;\;}x^\vee\end{aligned} 
  \end{equation}
are identity maps.  If $x_0\xrightarrow{\;f\;}x_1$ is a morphism, then the
dual morphism is the composition
  \begin{equation}\label{eq:28}
     f\dual\:x_1^\vee\xrightarrow{\;\;\id\otimes c_{x_0}\;\;}x_1^\vee\otimes x_0
     \otimes x_0^\vee \xrightarrow{\;\;\id\otimes f\otimes \id\;\;}
     x_1^\vee\otimes x_1\otimes x_0^\vee \xrightarrow{\;\;e_{x_1}\otimes
     \id\;\;} x_0\dual  
  \end{equation}
  \end{definition}

\noindent
 The morphism~$c$ is called \emph{coevaluation} and $e$~is called
\emph{evaluation}.  We say that $x^\vee$~is ``the'' dual to~$x$ since any two
triples of duality data are uniquely isomorphic.  Assuming all objects have
duals, we can make choices of duality data for all objects at once and so
obtain a duality involution~$\delta $ on~$\sC$, but $\delta $~ does not
satisfy Definition~\ref{thm:15} since the direction of morphisms is
reversed~\eqref{eq:28}; in other words, $\delta $~is a functor to the
\emph{opposite} category.

  \begin{definition}[]\label{thm:20}
 Let $\sC$~be a category. 

 \begin{enumerate}[{\textnormal(}1{\textnormal)}]

\item A \emph{twisted involution} of~$\sC$ is a pair~$(\delta ,\eta )$ of a
functor $\delta \:\sC\to\sC\op$ and a natural isomorphism $\eta \:\id_{\sC}\to
\delta \op\circ \delta $ such that for any~$x\in \sC$ we have $\delta \eta
_x\circ \eta _{\delta x}=\id_{\delta x}$. 

 \item A \emph{fixed point} of~$\delta $ is a pair~$(x,\theta )$ of an
object~$x\in \sC$ and an isomorphism $x\xrightarrow{\;\theta \;}\delta x$ such
that $\delta \theta\circ  \eta_x = \theta$ as morphisms $x\to \delta x$.

 \end{enumerate}  
  \end{definition}

\noindent
 Definition~\ref{thm:18} applies with a single change: the direction of
the bottom arrow in~\eqref{eq:39} is reversed. 

  \begin{example}[]\label{thm:21}
 For $\sC=\fVC$ the category of finite dimensional complex vector spaces, the
duality involution $\delta \:\sC\to\sC\op$ maps a vector space~$V$ to its
dual~$V^*$ and a linear map $f\:V\to W$ to $f^*\:W^*\to V^*$.  A fixed point
of~$\delta $ is a vector space~$V$ equipped with a nondegenerate symmetric
bilinear form; a linear map $f\:V\to W$ in~$\sC^\delta $ preserves the
bilinear forms.  A fixed point for the composite of duality and complex
conjugation (Example~\ref{thm:16}) is a complex vector space~$V$ with a
nondegenerate hermitian form; a linear map $f\:V\to W$ in the fixed point
category is a partial isometry---an injective map that preserves the
hermitian forms.
  \end{example}

  \begin{remark}[]\label{thm:22}
 There is a higher categorical context for Definition~\ref{thm:20}.  Let
$\Cat$~denote the 2-category of categories.  There is an involution $\alpha
\:\Cat\to\Cat$ that sends a category~$\sC$ to its opposite~$\sC\op$.  (There
is an extra categorical layer over Definition~\ref{thm:15}: there is a
triple~$(\alpha ,\eta _1,\eta _2)$ of data and a single condition.)  A twisted
involution in the sense of Definition~\ref{thm:20} is fixed point data
for~$\alpha $.  
  \end{remark}

  \begin{definition}[]\label{thm:59}
 Let $(\tau ,\eta )$~be an involution on a symmetric monoidal category~$\sC$.
A \emph{hermitian structure} on an object $x\in \sC$ is an isomorphism $h\:\tau
x\to x\dual$ such that the composition 
  \begin{equation}\label{eq:323}
     \tau x\cong \tau \bigl((x\dual)\dual \bigr)\xrightarrow{\;\;\tau
     (h\dual)\;\;} 
     \tau \bigl((\tau x)\dual \bigr)\cong \tau ^2(x
     \dual)\xrightarrow{\;\;\eta \inv \;\;}x\dual 
  \end{equation}
is equal to~$h$.
  \end{definition}

\noindent 
 Proposition~\ref{thm:44} asserts that every object in a bordism category
carries a hermitian structure.  Observe that if $F\:\sB\to\sC$ is an
equivariant symmetric monoidal functor between symmetric monoidal categories
with involution, as in Definition~\ref{thm:17}, then the image of a hermitian
structure on an object $b\in \sB$ is a hermitian structure on~$Fb$.

   \section{Noncompact Wick-rotated vector symmetry groups}\label{sec:14}
 
Let $(H_n,\rho _n)$ be a symmetry type, as in Definition~\ref{thm:153}. 

  \begin{proposition}[]\label{thm:161}
 Assume~$n\ge 3$.

 \begin{enumerate}[{\textnormal(}1{\textnormal)}]

 \item There exist a canonical noncompact Lie group~$\uH_n$, a homomorphism
$\uH_n\to GL_n\RR$ with kernel~$K$, and an inclusion $H_n\hookrightarrow
\uH_n$ such that (i)~$H_n\subset \uH_n$ is a maximal compact Lie subgroup,
(ii)~the inclusion induces an isomorphism on~$\pi _0$, and (iii)~the diagram
  \begin{equation}\label{eq:241}
     \begin{gathered} \xymatrix{H_n\ar@{^{(}->}[r]^{} \ar[d]_{\rho _n} &
     \uH_n\ar[d]^{} \\ 
     O_n\ar@{^{(}->}[r]^{} & GL_n\RR} \end{gathered} 
  \end{equation}
commutes. 

 \item There exists a canonical Lie group~$\uhH_n$ that fits into the
diagram
  \begin{equation}\label{eq:242}
     \begin{gathered} \xymatrix{1\ar[r]&H_n\ar[r]^{j_n}\ar@{^{(}->}[d]^{}
     &\hH_{n}\ar[r]\ar@{^{(}->}[d]^{} & \pmo
     \ar@{=}[d]\ar[r] & 1\\ 1 \ar[r]& \uH_{n} \ar[r]^{}& \uhH_{n}
     \ar[r]& \pmo\ar[r]&1} \end{gathered} 
  \end{equation}
of group extensions, as well as a canonical homomorphism $\uhH_n\to
\pmo\times GL_n\RR$ that fits into a pullback square 
  \begin{equation}\label{eq:243}
     \begin{gathered} \xymatrix{\uH_n\ar[r]^{} \ar[d]_{} & \uhH_n\ar[d]^{} \\
     GL_n\RR\ar[r]^{} & \pmo\times GL_n\RR} \end{gathered} 
  \end{equation}
and a commutative cube built from~\eqref{eq:35} and \eqref{eq:243}. 
 \end{enumerate} 
   \end{proposition}

These noncompact groups are used to define topological bordism categories
(\S\ref{subsec:12.2}). 

  \begin{proof}
 First define~$\uSpin_n$ and $\uPp_n$ as follows.  Choose a lift
$P\xrightarrow{\;\rho \;}GL_n\RR\xrightarrow{\;\pi \;}GL_n\RR/O_n$ of the
homogeneous principal bundle~$\pi $ to a principal $\Pp_n$-bundle~$\pi \circ
\rho $; it is unique up to isomorphism since $GL_n\RR/O_n$~is contractible.
Define $\uPp_n$ as the group of automorphism of~$\rho $ that cover the
action of left multiplication of~$GL_n\RR=\underline{O}_n$, and $\uSpin_n\in
\uPp_n$ the subgroup covering left multiplication by
$GL_n^+\RR=\underline{SO}_n$.  Then set 
  \begin{equation}\label{eq:244}
     \underline{SH}_n=\uSpin_n\times K\bigm / \langle(-1,k_0)\rangle, 
  \end{equation}
analogous to~\eqref{eq:16}.  If $\rho _n(H_n)=SO_n$, set
$\uH_n=\underline{SH}_n$.  If $\rho _n$~is surjective,
define~$\widetilde{\uH}_n$ as the pullback (see~\eqref{eq:226})
  \begin{equation}\label{eq:245}
     \begin{gathered} \xymatrix{
     1\ar[r]&K\ar[r]\ar@{=}[d]&\widetilde{\uH}_n\ar@{->>}[d]\ar[r]^{}&
     \uPp_n\ar@{->>}[d] \ar[r] & 1\\ 1\ar[r]&K\ar[r]&J\ar[r]&
     \pmo\ar[r] & 1} \end{gathered} 
  \end{equation}
and then 
  \begin{equation}\label{eq:246}
     \uH_n\cong \widetilde{\uH}_n \bigm / \langle(-1,k_0)\rangle. 
  \end{equation}
It is straightforward to check the properties in~(1).

For~(2) imitate the proof of Proposition~\ref{thm:10} with $\uSpin_n$
and~$\uPp_n$ replacing $\Spin_n$ and~$\Pp_n$, respectively. 
  \end{proof}

   \section{Computations with $A_1$-modules}\label{sec:15}
 
The computations described in \S\ref{sec:13} depend on knowledge of the mod $2$
cohomology   of the spectra
\[
\begin{aligned}
MTO_{|d|}& &\qquad  0\le &d \le 3 &\\
MO_{|d|}& &\qquad  -3 \le &d \le 0& \\
MSO_{3}& &\qquad &&
\end{aligned}
\]
as modules over the subalgebra $\aone$ of the mod $2$ Steenrod algebra
generated by $\sq^{1}$ and $\sq^{2}$.    The purpose of this appendix
is to describe these computations and the methods for arriving at
them.

We thank Meng Guo for her careful reading and astute corrections.

\subsection{Cell diagrams}
\label{sec:hadcell-diagrams}

It is common practice to depict an $\aone$ module $M$ as a graph with
nodes corresponding to a chosen homogeneous basis for $M$, at a height
corresponding to grading, and with an edge drawn with a straight line
between $e$ and $e'$ if the coefficient of $e'$ in $\sq^{1}(e)$ is
non-zero, and an edge drawn with a curved line if they are analogously
related by $\sq^{2}$.  This works best when a basis can be chosen
so that the operations $\sq^{1}$ and $\sq^{2}$ send basis elements to
basis elements.  This is the case with all of the $\aone$ modules
needed in this paper.   Here are three examples:

\begin{center}
\includegraphics{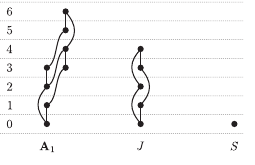}
\end{center}
For clarity the degrees of the basis elements have been indicated in
this example, though we will not usually do this.  Topologists call these graphs
``cell diagrams.''  The one on the left is the free $\aone$ module on
one generator (of degree $0$) and the one on the right is just
$\Z/2=H^{\ast}(S^{0})$, concentrated in degree $0$.  The one in the
middle right comes up frequently and was deemed the {\em Joker} by
Adams.  It is the cohomology of a spectrum also called $J$.

As explained in \S\ref{sec:13} the mod $2$ cohomology $H^{\ast}\MSpin$ was
show by Anderson, Brown and Peterson~\cite{ABP1} to have the form \[
A\underset{\aone}{\otimes}N \] for some $\aone$ module $N$ (which they
determined).  Figure~\ref{fig:9} is a cell diagram of $N$ through
dimension~ $28$.  The modules to the right (in gray) are free, and the modules
to the left (in black) are either $S$ or $J$.

\begin{figure}[ht]
\centering
\includegraphics[scale=.85]{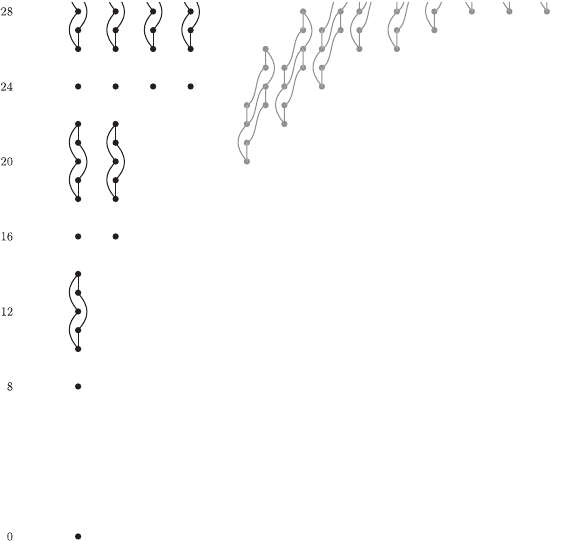}
\caption{The cell diagram for $\MSpin$}\label{fig:9}
\end{figure}

How does one use this in practice?    Suppose $X$ is a connective
spectrum of finite type and one wishes to determine the localization
at $2$ of $\pi_{\ast}\MSpin\wedge X$.   One makes three computations,
(in which the abutments, though not indicated, have been completed at $2$)
\begin{align*}
\ext_{\aone}^{s,t}(H^{\ast}X,\Z/2)&\Rightarrow \pi_{t-s}\ko\wedge X \\
\ext_{\aone}^{s,t}(J\otimes H^{\ast}X,\Z/2)& \Rightarrow \pi_{t-s} \ko\wedge J\wedge X =: M_{J}(X)\\
\ext_{\aone}^{s,t}(\aone\otimes H^{\ast}X,\Z/2)& = H_{\ast}X\ . 
\end{align*}
The two spectral sequences often
collapse (they do in the cases studied in this paper).   Write 
\begin{align*}
M_{S}(X) &=\pi_{\ast}\ko\wedge X  \\
M_{J}(X) &= \pi_{\ast}\ko\wedge J\wedge X\ .
\end{align*}
The result of Anderson-Brown-Peterson~\cite{ABP1} is that
after localizing at $2$, $\pi_{\ast}\MSpin\wedge X$ is isomorphic to a
sum of copies of $M_{S}(X)$, $M_{J}(X)$ and $H_{\ast}X$, shifted
according to the location of the corresponding summands in the cell
diagram of $X$:
\[
\pi_{\ast}\MSpin\wedge X = M_{S}(X) \oplus \Sigma^{8}M_{S}(X)\oplus
\Sigma^{10}M_{J}(X)\oplus \cdots \oplus\Sigma^{20}H_{\ast}X
\oplus\cdots .
\]

One further comment about the spectral sequences above.   If $M$ is a
free $\aone$-module then 
\begin{align*}
\ext_{\aone}^{s,t}(M,\Z/2) &= \ext_{\aone}^{s,t}(J\otimes M,\Z/2) = 0
\qquad s>0\\
\ext_{\aone}^{0,t}(M,\Z/2) &= \Hom_{\aone}(M,\Z/2)  \\
\ext_{\aone}^{0,t}(J\otimes M,\Z/2) &= \Hom_{\aone}(J\otimes M,\Z/2) 
\end{align*}
In these cases the display of the spectral sequences are all on the line
$s=0$, and the spectral sequences collapse.    

More generally if $M$
is of the form $M'\oplus F$ with $F$ a free $\aone$ module, then 
\[
\ext_{\aone}^{s,t}(M,\Z/2) \approx 
\ext_{\aone}^{s,t}(M',\Z/2) \oplus
\ext_{\aone}^{s,t}(F,\Z/2) 
\]
and the spectral sequence is the sum of two spectral sequences, one of
which collapses for trivial reasons.  The analogous statement holds
for the second spectral sequence.  For this reason it is useful to
omit free summands from the cell diagrams and keep track of them in
some other way.

\subsection{The charts}

We can now explain in more detail what is shown in Figure~\ref{fig:m1}.  In
each case we are interested in $\pi_{\ast}\MSpin\wedge X$ for some
appropriate spectrum $X$.  A cell diagram for $X$, modulo free $\aone$
summands is shown on the left, with $X$ labeled below it.  The chart to the
right depicts $\ext_{\aone}^{s,t}(H^{\ast}(X);\Z/2)$ as a module over
$\ext_{\aone}^{s,t}(\Z/2,\Z/2)$.  Following standard convention the
horizontal axis is the $(t-s)$-axis and the vertical axis is the $s$-axis.
Each dot represents a basis element.  The contributions from the free
summands contribute only to $\ext^{0,t}$ and to keep the picture uncluttered
they are indicated below the table.  For example in the case $s=3$, in
dimension $(t-s)=8$, there is a $\Z/2$ not indicated in graphical notation,
but only by the $+1$.  The group in that case is the sum of that $\Z/2$ and
$\Z/2\oplus \Z/8\oplus\Z/32$.

The color coding allows one to read off the effect of the twisted Dirac
operators of \S\ref{subsec:8.4} as described in homotopy theoretic terms in
\S\ref{sec:13}.  Consider, for example, the case $s=3$.  One needs to know
the effect of the map \[ \pi_{\ast}\MSpin\wedge S^{-3}\wedge MO_{3}\to
S^{-3}\wedge KO.  \] The $(-1)$-connected cover of $S^{-3}\wedge KO$ is
equivalent to $\ko\wedge W$, in which $W$ is the finite spectrum whose cell
diagram is depicted below \begin{center}
\includegraphics{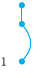}
\end{center}
The effect in cohomology of the twisted Dirac operator corresponds to
the inclusion of the blue cells, and the cokernel of this map, in the
relevant summand, is displayed in green.    The $\ext$ charts are
correspondingly color coded and the red line indicates the connecting
homomorphism in the long exact sequence.   The $\ext$ computation of
interest is  built from the kernel and cokernel of this connecting
homomorphism.   For example the connecting homomorphism is a
monomorphism from the column $(t-s)=1$ to the column $(t-s)=0$, and
the only non-zero $\ext$ group in this range is 
\[
\ext^{0,0}_{\aone}(H^{\ast}S^{-3}MO_{3},\Z/2) = \Z/2.
\]
In dimension $6$, the group is the sum of $(\Z/2)^{2}$ (coming from the
free summands) and another $\Z/2\oplus \Z/2$.   The fact that the dot
in filtration $s=2$ is blue indicates that the corresponding basis
element maps non trivially under the map to $\pi_{6}\Sigma^{-3}KO$.

\subsection{The cases $s=\pm1$}
\label{sec:hadcases-s=pm1}

The cell diagrams for  $\Sigma^{-1} MO(1)$ and $\Sigma^{1}MTO(1)$ are
easily derived from the Thom isomorphism and Wu formula 
\[
\sq^{n}(U)=w_{n}\cdot U
\]
for the action of the Steenrod operations on the Thom class of a
(virtual) vector bundle.  The diagrams work out to be
\begin{center}
\includegraphics{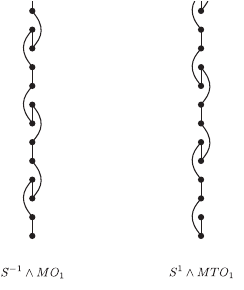}
\end{center}
and continue infinitely far upward, repeating the evident pattern of
Steenrod operations.   There are no additional free summands in these cases.

\subsection{The case $s=4$}
\numberwithin{equation}{section}
\label{sec:hadcohomology-bso3}

The next easiest case to understand is the case $s=4$.  To derive it
requires a useful technique introduced by Adams and
Margolis~\cite{AM}, and developed considerably further by
Margolis~\cite{Ma}.  The subalgebra $\aone$ contains two of the
Milnor operators
\begin{align*}
Q_{0} & = \sq^{1} \\
Q_{1} &= [\sq^{2},\sq^{1}]
\end{align*}
and together they generate an exterior algebra
\[
E[Q_{0},Q_{1}]\subset \aone.
\]

\begin{definition}
\label{def:1b}
Suppose that $M$ is an $\aone$ module.   For $i=0,1$ the
$i^{\text{th}}$ Margolis homology of $M$ is 
\[
H_{\ast}(M;Q_{i})= \ker Q_{i}/\image Q_{i}.
\]
The Margolis homology of a space or spectrum $X$ is the Margolis
homology of $H^{\ast}X$
\[
H_{\ast}(X;Q_{i}) = H_{\ast}(H^{\ast}(X);Q_{i}).
\]
\end{definition}

\begin{remark}
\label{rem:2}
The Milnor elements are primitive, and the Kunneth isomorphism holds:
\[
H_{\ast}(M\otimes N;Q_{i}) \approx
H_{\ast}(M;Q_{i}) \otimes
H_{\ast}(N;Q_{i}).
\]
\end{remark}

The following theorem of Adams and Margolis~\cite[Theorem~3.1]{AM}
(attributed by Adams and Margolis to Wall, in this particular case) is
one reason the Margolis homology groups are important.

\begin{theorem}[Adams-Margolis]
\label{thm:had1} A  connected $\aone$-module $M$ is free if and only if 
\[
H_{\ast}(M;Q_{0}) = H_{\ast}(M;Q_{1}) = 0.
\]
\end{theorem}

The action of the Milnor
operators on
\[
H^{\ast}(BSO_{3};\Z/2) = \Z/2[w_{2},w_{3}].
\]
is given by
\begin{align*}
Q_{0}(w_{2}) &=w_{3} \\
Q_{0}(w_{3}) &=0.
\end{align*}
This implies that the Margolis homology with respect to  $Q_{0}$ is
\[
H_{\ast}(BSO_{3};Q_{0}) \approx  \Z/2[w_{2}^{2}].
\]

Write $U$ for the Thom class in $H^{\ast}MO_{3}$.   
Since $Q_{0}(U)=w_{1}U=0$ the Thom isomorphism commutes with $Q_{0}$,
and the Margolis homology of $MSO_{3}$ with respect
to $Q_{0}$ is 
\[
U\cdot\Z/2[w_{2}^{2}].
\]

For the $Q_{1}$ homology note that
\begin{align*}
Q_{1}(w_{2}) &= w_{2}w_{3} \\
Q_{1}(w_{3}) &= w_{3}^{2} \\
Q_{1}(U) &=  U w_{3}.
\end{align*}
It follows that $H^{\ast}MSO(3)$, as a module over the exterior
algebra $E[Q_{1}]$, is a sum of 
\[
UF_{j} =\{U w_{2}^{j}, U w_{2}^{j} w_{3}, U w_{2}^{j}w_{3}^{2},U w_{2}^{j}w_{3}^{3},\dots\}.
\]
Using this one sees that the Margolis homology with respect to $Q_{1}$ of
$MSO(3)$ has basis $\{U w_{2}^{2j+1}  \}$.  

Now let $M$ and $N$ be the $\aone$-modules 
\begin{center}
\includegraphics{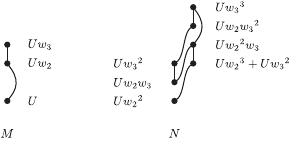}
\end{center}
and consider the map 
\begin{equation}
\label{eq:had1}
(M\oplus N)\otimes \Z/2[w_{2}^{4}] \to H^{\ast}(MSO_{3}).
\end{equation}
The map~\eqref{eq:had1} is an inclusion.  Together with the Kunneth
formula, the computation just described implies that it induces an
isomorphism of Margolis homology with respect to both $Q_{0}$ and
$Q_{1}$.  By the Theorem of Adams and Margolis its cokernel is free,
and there is an isomorphism
\[
H^{\ast}(MSO_{3})\approx (M\oplus N)\otimes \Z/2[w_{2}^{4}]\oplus\text{ free modules}.
\]
The cell diagram in box $s=4$ in Figure~\ref{fig:m1} depicts $(M\oplus
N)\otimes \Z/2[w_{2}^{4}]$. 

One can work out the disposition of the free modules by computing
Poincar\'e series.   The Poincar\'e series for the indecomposables of the free modules (with
$U$ placed in degree $0$) is the quotient of 
\[
\frac{1}{(1-t^{2})(1-t^{3})}-\frac{(1+t^{2}+t^{3}+t^{4}(1+t+2t^{2}+t^{3}+t^{4}+t^{5}))}{(1-t^{8})}
\]
by the Poincar\'e series $(1+t)(1+t^{2})(1+t^{3})$ of $\aone$.   This
works out to be 
\[
\frac{t^{9}}{(1-t^{6})(1-t^{8})} = t^{9}+t^{15}+t^{17}+O[t]^{21}.
\]

Most of the time this is enough information.  However for some
purposes it is useful to have a basis for the generators of the free
modules.  In this case one can work out that the summand of 
free modules is
\[
\aone[w_{3}^{2},w_{2}^{4}]\cdot U w_{2}^{3}w_{3},
\]
and that
\begin{equation}
\label{eq:had2}
(M\oplus N)\otimes Z/2[w_{2}^{4}]  \oplus \aone[w_{3}^{2},w_{2}^{4}]\otimes U
w_{2}^{3}w_{3} \to H^{\ast}(MSO_{3})
\end{equation}
is an isomorphism.   We now digress to describe a technique for
verifying this.   The technique applies to modules over any connected
graded Hopf algebra and exploits the fact that such an algebra is a
Frobenius algebra.   We will describe it explicitly for $\aone$.

Let $b(x)=\sq^{2}\sq^{2}\sq^{2}(x)$ (this is the operation that goes
from the bottom dot to the top dot in the cell diagram for $\aone$).
If $F$ is a free $\aone$-module, and $x\in F$ there are elements
$a\in\aone$ and $y\in F$ with $a\cdot x = b(y)\ne 0$.  This proved by
reducing to the case $F=\aone$ and either checking directly or
appealing to the fact that $\aone$ is a Frobenius algebra.

\begin{lemma}
\label{thm:had3} 
Suppose that $F$ and $M$ are $\aone$ modules and that
$F$ is free.  A map $F\to M$ is a monomorphism if and only if the
induced map $b(F)\to b(M)$ is a monomorphism.
\end{lemma}

\begin{proof}
The only if statement is clear.   For the converse, suppose that
$b(F)\to b(M)$ is a monomorphism and $x\in F$.   By the
above there are $a\in \aone$ and $y\in F$ with $a \cdot x = b(y)\ne 0$.
Since $b(F)\to b(M)$ is a monomorphism the image of $b(y)$ is
non-zero, hence so is the image of $a(x)$ and hence so is the image of $x$.
\end{proof}

\begin{remark}
\label{rem:1} Since $\aone$ is a finite dimensional Hopf algebra, it
is also injective as a module over itself.  This means that if
$F\subset M$ is a free submodule of finite type (finite rank in each
degree) then there is a decomposition $M\approx M'\oplus F$.  This
leads to a fairly quick way of locating the free summands in an
$\aone$-module $M$.  They are generated by any subset $B\subset M$
with the property that $b(B)\subset b(M)$ is a basis.
\end{remark}

\begin{lemma}
\label{thm:had4}
For an $\aone$ module $N$ the following are equivalent 
\begin{thmList}
\item If $F$ is a free module and $F\subset N$ then $F=0$.
\item $b(x)=0$ for all $x\in N$.   
\end{thmList}
\end{lemma}

\begin{proof}
Suppose that $F\subset N$ is a free submodule.  If $F$ is non-zero
then there is an $x\in F$ with $b(x)\ne 0$, so $b(N)\ne 0$.
Conversely if there is an $x\in N$ with $b(x)\ne0$ then the map
\begin{align*}
\Sigma^{|x|}\aone &\to N \\
a &\mapsto a\cdot x
\end{align*}
is a monomorphism by Lemma~\ref{thm:had3}.
\end{proof}

\begin{definition}
\label{def:2}
An $\aone$ module $N$ {\em has no free submodules} if it has the
equivalent properties above.
\end{definition}

By Remark~\ref{rem:1} having a free submodule is equivalent to having
a free summand.

\begin{lemma}
\label{thm:had2} 
Suppose that $H$ is an $\aone$-module, and $N\subset H$
a summand having no free submodules.  If $F$ is a free module and
$F\to H$ is a monomorphism, then $F\to H/N$ is a monomorphism.
\end{lemma}

\begin{proof}
By Lemma~\ref{thm:had3} it suffices to show that $b(F)\to b(H/N)$ is a
monomorphism.   Since $b(N)=0$ and $N$ is a summand, the map $b(H)\to
b(H/N)$ is an isomorphism.
\end{proof}

Returning to the cohomology of $MSO_{3}$, we now use these ideas to
show that~\eqref{eq:had2} is an isomorphism of $\aone$ modules.  Both
sides have the same Poincar\'e series so it suffices to show that the
map is a monomorphism, or equivalently that the map
\[
\aone[w_{3}^{2},w_{2}^{4}]\otimes U
w_{2}^{3}w_{3} \to H^{\ast}(MSO_{3})/\big((M\oplus N)\otimes Z/2[w_{2}^{4}]\big)
\]
is a monomorphism.  Since $M$ and $N$ visibly have no free submodules,
neither does $(M\oplus N)\otimes Z/2[w_{2}^{4}]$, so by
Lemma~\ref{thm:had2} it suffices to show that 
\[
\aone[w_{3}^{2},w_{2}^{4}]\otimes U
w_{2}^{3}w_{3} \to H^{\ast}(MSO_{3})
\]
is a monomorphism.   This is done with the aid of Lemma~\ref{thm:had3}.
Since
\begin{align*}
\sq^{1}(w_{2}^{4}) &= \sq^{2}(w_{2}^{4})= 0 \\
\sq^{1}(w_{3}^{2}) &= \sq^{2}(w_{3}^{2})= 0 \\
\end{align*}
and 
\[
\sq^{2}\sq^{2}\sq^{2}(U w_{2}^{3}w_{3}) = U w_{3}^{5}
\]
the assertion comes down to checking that 
\[
\{U w_{3}^{5} w_{2}^{4k}\, w_{3}^{2\ell} \},
\]
is linearly independent, which is easy.

\subsection{The case $s=\pm2$}
\label{sec:hadcase-s=pm2}

We begin with the formulas
\begin{align*}
Q_{0}(w_{1}) &=w_{1}^{2} \\ 
Q_{0}(w_{2}) &=w_{1} w_{2} \\
Q_{1}(w_{1}) &=w_{1}^{4} \\
Q_{1}(w_{2}) &=w_{1}^{3} w_{2}+w_{1}w_{2}^{2}.
\end{align*}
For both $MO_{2}$  and $MTO_{2}$
\begin{align*}
Q_{0}(U) &= w_{1}U \\
Q_{1}(U) &= (w_{1}^{3}+w_{1}w_{2})U,
\end{align*}
so the Thom isomorphism 
\[
H^{\ast}(MO_{2})\approx H^{\ast}(MTO_{2})
\]
induces an isomorphism of Margolis homology.   

Restricting attention to $MO_{2}$, let
\[
F_{n}\subset H^{\ast}MO_{2}
\]
be the subspace with basis 
\[
\{U\, w_{1}^{i}\, w_{2}^{j}\mid j\le n \}
\]
and $\bar F_{n}$ the subspace with basis 
\[
\{U\, w_{1}^{i}\, w_{2}^{n}\},
\]
so that there is a vector space isomorphism
\[
F_{n}\approx\bigoplus_{j\le n}\bar F_{j}.
\]

The Milnor operator $Q_{0}$ preserves the decomposition into the spaces
$\bar F_{j}$ and from the formulas above one concludes that 
\[
H_{\ast}(\bar F_{2n};Q_{0})=0
\]
and
\[
H_{\ast}(\bar F_{2n+1};Q_{0}) = \Z/2\{U\,w_{2}^{2n+1} \}.
\]
This shows that the $Q_{0}$ Margolis homology of $H_{\ast}MO_{2}$ has
basis $\{U\, w_{2}^{2n+1} \}$.

The Milnor operator $Q_{1}$ maps $F_{n-1}$ to $F_{n}$.   We can
determine the Margolis homology from the associated spectral
sequence.   Identifying $F_{n}/F_{n-1}\approx \bar F_{n}$ and using
the formulas above, one easily checks that the first differential in
this spectral sequence is the $\Z/2[w_{1}]$-linear map 
\begin{gather*}
\bar F_{2n} \xrightarrow{\cdot w_{1}w_{2}}{} \bar F_{2n+1} \\
\bar F_{2n+1}\xrightarrow{0}{} \bar F_{2n+2}.
\end{gather*}
It follows that the $Q_{1}$ Margolis homology of $H^{\ast}(MO_{2})$
also has basis $\{U\, w_{2}^{2n+1} \}$.

Let $M$ and $N$ be the $\aone$
modules below
\begin{center}
\includegraphics{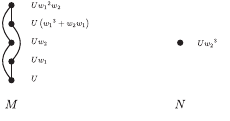}
\end{center}
The map 
\[
\Z/2[w_{2}^{4}]\otimes \big(M\oplus N\big) \to H^{\ast}(MO_{2})
\]
is then an inclusion and induces an isomorphism of Margolis homology.  If follows that 
\[
H^{\ast}MO_{2}\approx \Z/2[w_{2}^{4}]\otimes \big(M\oplus N\big)
\oplus\text{free}.
\]

The location of the free modules can be determined from the 
Poincar\'e series.  The Poincar\'e series for the generators is the
quotient of
\[
\frac{1}{(1-t)(1-t^{2})}-\frac{(1+t+t^{2}+t^{3}+t^{4}+t^{6})}{(1-t^{8})}
\]
by the Poincar\'e series $(1+t)(1+t^{2})(1+t^{3})$ of $\aone$.   This
works out to be 
\[
\frac{t^{2}}{(1-t^{2})(1-t^{8})} = 
\frac{t^{2}+t^{4}}{(1-t^{4})(1-t^{8})}.
\]

In fact the subspace of free modules is a free module over
$\aone[w_{1}^{4},w_{2}^{4}]$ and has 
\[
\{U w_{1}^{2}, U w_{2}^{2} \}
\]
as a basis.  As before, it suffices from the Poincar\'e series above
to check that the map
\[
\aone[w_{1}^{4},w_{2}^{4}]\{U w_{1}^{2},U w_{2}^{2} \} \to
H^{\ast}(MO_{3}) 
\]
is a monomorphism, and for this to check that the set
\[
\{\sq^{2}\sq^{2}\sq^{2} \big(U w_{1}^{2}w_{1}^{4 k}w_{2}^{4\ell}\big),
\sq^{2}\sq^{2}\sq^{2} \big(U w_{2}^{2}w_{1}^{4 k}w_{2}^{4\ell}\big)\}
\]
is linearly independent.  This is easily deduced from the fact that
$\sq^{2}\sq^{2}\sq^{2}$ is linear over $\Z/2[w_{1}^{4},w_{2}^{4}]$ and
\begin{align*}
\sq^{2}\sq^{2}\sq^{2}(U w_{1}^{2}) &=  U w_{1}^{6}w_{2} \\
\sq^{2}\sq^{2}\sq^{2}(U w_{2}^{2}) &=  U w_{1}^{4}w_{2}^{3}.
\end{align*}

The situation with $MTO_{2}$ is similar, the variations being the
use of the modules
\begin{center}
\includegraphics{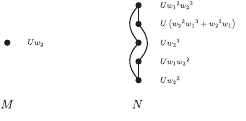}
\end{center}
and the Poincar\'e series 
\[
\frac{1+t^{6}}{(1-t^{4})(1-t^{8})}
\]
for the generators of the free modules, from which one can
conclude that the subspace of free modules is the sub
$\aone[w_{1}^{4},w_{2}^{4}]$-module with basis 
\[
\{U, U w_{1}^{2} w_{2}^{2} \}
\]
on which the operator $\sq^{2}\sq^{2}\sq^{2}$ takes the value
\[
U w_{1}^{4}w_{2}, U w_{1}^{6}w_{2}^{3}.
\]

\subsection{The case $s=\pm3$}
\label{sec:hadcase-s=pm3}

We now turn to the case of $MO_{3}$.   This is the most complicated of
the cases and the specific determination of the free summands was
carried out with the aid of Mathematica.   

It will be helpful to use the
equivalence 
\[
BO_{1}\times BSO_{3}\to BO_{3}
\]
classifying the tensor product of the defining vector bundles.   Write
\begin{align*}
w_{i} & \in H^{i}(BO_{3})\\
v_{i} & \in H^{i}BSO_{3} \\
v_{1} & \in H^{1}BO_{1} \\
\end{align*}
for the corresponding Stiefel-Whitney classes, so that under the
equivalence above we have 
\begin{align*}
w_{1}  &= v_{1} \\
w_{2} &= v_{2}+ v_{1}^{2} \\
w_{3} &= v_{3}+ v_{2}v_{1} + v_{1}^{3}.
\end{align*}
and
\begin{align*}
v_{1} &= w_{1} \\
v_{2} &= w_{1}^{2}+w_{2} \\
v_{3} &= w_{1} w_{2}+w_{3}.
\end{align*}

Now note that 
\begin{align*}
Q_{0}U &=U(v_{1}) \\
Q_{1}U &=U(v_{3}+v_{1}^{3}) \\
\end{align*}
so that as far as the Minor operators are concerned there is an
isomorphism
\[
H^{\ast}(MO(3))\approx H^{\ast}(MSO_{3})\otimes H^{\ast}(MO_{1}).
\]
From this one concludes that 
\[
H^{\ast}(MO_{3};Q_{0}) = 0
\]
and that the Margolis homology $H^{\ast}(MO_{3};Q_{1})$ has basis $\{U
v_{1} v_{2}^{2j+1}\}$.  

As in the case of $MSO(3)$ let $M$ and $N$ be the $A_{1}$-modules
depicted below (in which the blue dot indications the location of the
Margolis homology group)
\begin{center}
\includegraphics{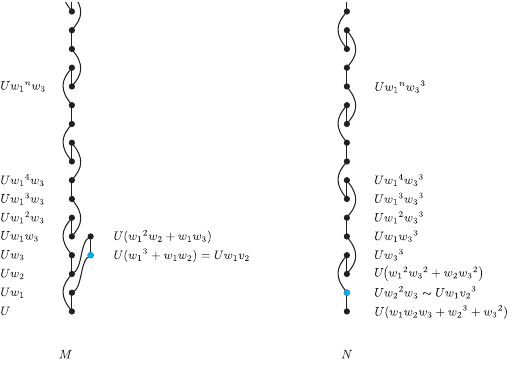}
\end{center}
Then the map 
\[
(M\oplus N) \otimes \Z/2[v_{2}^{4}] \to H^{\ast}(MO_{3})
\]
is a monomorphism and induces an isomorphism of Margolis homology
groups.   It follows that
\[
H^{\ast}(MO_{3}) \approx (M\oplus N) \otimes \Z/2[v_{2}^{4}] \oplus\text{free}.
\]
The Poincar\'e series for the indecomposables of the free modules (with
$U$ placed in degree $0$) is the quotient of 
\[
\frac{1}{(1-t)(1-t^{2})(1-t^{3})}-\frac{(1-t)^{{-1}}+t^{3}+t^{4}+t^{6}(1-t)^{-1}}{(1-t^{8})}
\]
by the Poincar\'e series $(1+t)(1+t^{2})(1+t^{3})$ of $\aone$.   It
works out to be 
\[
\frac{t^2}{\left(1-t^4\right)
   \left(1-t^8\right)}+\frac{t^4+t^5+t^6+t^9+t^{10}+t^{11}+t^{12}+t^{15}}{\left(1-t^4\right)
   \left(1-t^8\right) \left(1-t^{12}\right)}\ .
\]

The free modules correspond to the sum of 
\[
\aone[w_{1}^{4}, w_{2}^{4}]\{U w_{1}^{2} \}
\]
and the free $\aone[w_{1}^{4}, w_{2}^{4}, w_{3}^{4}]$-module on
\[
\big\{U w_2^2,U w_2 w_3,U w_3^2,U w_2^3 w_3,U w_2^2 w_3^2,U w_1^2 w_2^3 w_3,U
   w_1^2 w_2^2 w_3^2,U w_2^3 w_3^3\big\}
\]
To see that these are linearly independent, one applies
$\sq^{2}\sq^{2}\sq^{2}$ to reduce the problem to showing that the
union of 
\[
\left\{U \left(w_1^6w_{2}+w_1^5 w_{3}\right) w_{1}^{4k}w_{2}^{4\ell} \right\}
\]
and the set consisting of the products of
$w_{1}^{4k}w_{2}^{4\ell}w_{3}^{4m}$ with the elements of 
\begin{multline*}
\big\{U\left( w_1^4 w_2^3+w_1^3 w_2^2
   w_3+w_1^2 w_2 w_3^2+w_1 w_3^3\right), U\left(
   w_1^4 w_2^2 w_3+w_1^2
   w_3^3\right), \\ U\left( w_1^4 w_2
   w_3^2+w_1^3 w_3^3\right), U\left(
   w_1^2 w_2^2
   w_3^3+w_3^5\right),  U\left( w_1^2 w_2
   w_3^4+w_1 w_3^5\right), U\left( w_1^6
   w_2^4 w_3+w_1^2
   w_3^5\right),  \\U\left(
   w_1^6 w_2^3 w_3^2+w_1^5
   w_2^2 w_3^3+w_1^4 w_2
   w_3^4+w_1^3
   w_3^5\right), \\ U\left(
   w_1^4 w_2^4
   w_3^3+w_3^7\right)\big\}
\end{multline*}
is linearly independent.   A couple of maneuvers will make this
obvious.   First of all, let's apply the Thom isomorphism to get rid
of the appearance of $U$.   Next regard everything as a module over
$\Z/2[w_{1}^{4},w_{2}^{4}]$ and look at the associated graded of the
increasing filtration by powers of $w_{3}$.  Doing so reduces the
problem to showing that the map from the free
$\Z/2[w_{1}^{4},w_{2}^{4}]$-module on 
\[
\big\{w_1^5 w_{3} , w_1 w_3^{3+4k},  w_1^2
   w_3^{3+4k},   w_1^3 w_3^{3+4k}, 
   w_3^{5+4k},  w_1 w_3^{5+4k},
  w_1^2 w_3^{5+4k},   w_1^3
   w_3^{5+4k},  w_3^{7+4k}\big\}
\]
to $H^{\ast}(BO_{3})$ is a monomorphism, which is easy.

The analysis is similar for $MTO_{3}$.  The Margolis
homology is the same as that for $MO_{3}$ since the ratio of the two
Thom classes is $w_{3}^{2}$ which is annihilated by the Milnor
operators.  The basic modules for $MTO_{3}$ are as below.
\begin{center}
\includegraphics{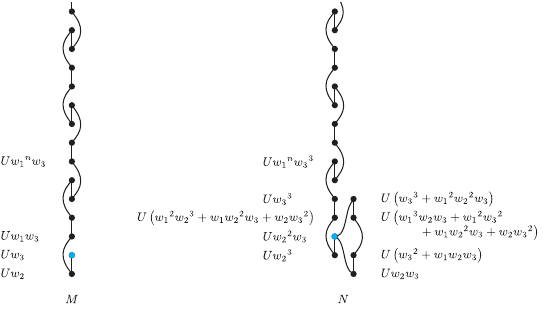}
\end{center}

The Poincar\'e series for the free modules as
the quotient of
\[
\frac{1}{(1-t)(1-t^{2})(1-t^{3})}-\frac{t^{2}(1-t)^{{-1}}+t^{6}(1-t)^{-1}+t^{5}+t^{6}+t^{8}+t^{9}}{(1-t^{8})}
\]
by the Poincar\'e series $(1+t)(1+t^{2})(1+t^{3})$ of $\aone$.   This
can be written as
\[
\frac{t^7}{\left(1-t^4\right)
   \left(1-t^8\right)}+\frac{1+t^4+t^6+t^9+t^{10}+t^{11
   }+t^{15}+t^{17}}{\left(1-t^4\right)
   \left(1-t^8\right) \left(1-t^{12}\right)}
\]

The inclusion of the free summands turns out to be the sum of the
$\aone[w_{1}^{4},w_{2}^{4},w_{3}^{4}]$ module map
\begin{multline*}
\aone[w_{1}^{4},w_{2}^{4},w_{3}^{4}]\big\{U,U w_2^2,U w_1^2 w_2^2,U w_2^3 w_3,U
   w_2^2 w_3^2, \\ U w_2 w_3^3,U w_2^3 w_3^3,U
   w_1^2 w_2^3 w_3^3\big\}
\to H^{\ast}(MTO_{3})
\end{multline*}
and the $\aone[w_{1}^{4},w_{2}^{4}]$-module map 
\[
\aone[w_{1}^{4},w_{2}^{4}]\{U w_{1}^{2}w_{2}w_{3} \}\to H^{\ast}(MTO_{3}).
\]
As above, to check this it suffices to apply
$\sq^{2}\sq^{2}\sq^{2}$ to the generators above and show that the map
from the sum of the free $\Z/2[w_{1}^{4},
w_{2}^{4},w_{3}^{4}]$-module on
\begin{multline*}
\big\{U \left(w_1^4 w_2+w_1^3 w_3\right),U \left(w_1^2 w_2 w_3^2+w_1
w_3^3\right), \\ U \left(w_1^6 w_2^3+w_1^5 w_2^2 w_3+w_1^4 w_2
w_3^2+w_1^3 w_3^3\right), U \left(w_1^4 w_2^4 w_3+w_3^5\right), \\ U
\left(w_1^4 w_2^3 w_3^2+w_1^3 w_2^2 w_3^3+w_1^2 w_2 w_3^4+w_1
w_3^5\right), \\ U \left(w_1^4 w_2^2 w_3^3+w_1^2 w_3^5\right),U
\left(w_1^2 w_2^2 w_3^5+w_3^7\right),U \left(w_1^6 w_2^4 w_3^3+w_1^2
w_3^7\right)\big\}
\end{multline*}
and the free $\Z/2[w_{1}^{4},w_{2}^{4}]$-module on 
\[
U \left(w_1^6 w_2^2 w_3+w_1^4 w_3^3\right)
\]
to $H^{\ast}(MTO_{3})$ is a monomorphism.  Again, by filtering by
powers of $w_{3}$, using the Thom isomorphism, and looking at the
associated graded, it suffices to check that the map from
\[
\Z/2[w_{1}^{4}, w_{2}^{4}]\{w_1^4 w_3^3, w_1^3 w_3^{1+4k}, w_1 w_3^{3+4k},  w_1^3 w_3^{3+4k},  w_3^{5+4k},  w_1^{2}
w_3^{5+4k},  w_3^{7+4k}, w_1^2 w_3^{7+4k}  \}
\]
to $H^{\ast}(BO_{3})$ is a monomorphism, which
is obvious.

   \section{Construction of~$\hH_n$ via classifying spaces}\label{sec:16}

 For the analysis in Section~\ref{sec:13} it is useful to describe the
construction of the spaces $B\widehat{H}_{n}$ from the point of view of
homotopy theory.  This constitutes an alternative proof of
Theorem~\ref{thm:10}.

  \subsection{Preliminary}\label{subsec:e0.9}

We begin with a left, pointed $BSpin$-module $BH$ and a
$BSpin$-module map $p:BH\to BO$.  We will write the base point as $0\in BH$,
and the action
\[
BSpin\times BH\to BH
\]
as
\[
(V,W) \mapsto V\oplus W.
\]

We define $BH_{n}$ and maps $BSpin_{n}\to BH_{n}\to BO_{n}$ by the pullback diagram 
\[
\xymatrix{
BSpin_{n}  \ar[r]\ar[d]  & BSpin
\ar[d] \\
BH_{n}  \ar[r]\ar[d]  & BH
\ar[d] \\
BO_{n}  \ar[r]        & BO
}
\]

\begin{remark}
\label{rem:e1} We will be interested in the case in which the spaces
$BH_{n}$ are the classifying spaces of a family of compact Lie groups
$H_{n}$.
\end{remark}

  \subsection{A family of spaces~$\bhath nm$}\label{subsec:e0.10}

Our aim is to define extensions
\[
H_{n} \to \widehat H_{n}\to \Z/2
\]
equipped with a splitting for each choice of a hyperplane reflection $\sigma\in
O_{n}$, and normalized so that 
\begin{equation}
\label{eq:e2}
B \widehat Spin_{n}= BPin_{n}^{+}.
\end{equation}
We will actually construct extensions 
\[
H_{n} \to \widehat H_{n}^{(m)}\to \Z/2
\]
for every $m\in \Z$, and show that these extensions are $4$-fold periodic in
$m$.   The case $m\equiv 1\mod 4$ is satisfies the normalization
condition~\eqref{eq:e2} above.  

We first construct what will end up being the classifying
spaces of~$\widehat H_{n}^{(m)}$.  Namely,
for $m\in \Z$ define
$\bhath{n}{m}$ by the pullback square 
\begin{equation}
\label{eq:e1}
\begin{gathered}
\xymatrix{
\bhath{n}{m}  \ar[r]\ar[d]  & BH\times B\Z/2
\ar[d]^-{p\,\oplus \,m\overline{\O(1)}} \\
BO_{n}  \ar[r]        & BO
}
\end{gathered}
\end{equation}
In the above, $\O(1)$ is the tautological line bundle on $B\Z/2$ and for a vector space~ $V$ we write 
\[
\bar{V} = V-\dim V.
\]

\begin{example}
\label{eg:e1} Suppose that $BH=BSpin$ so that $BH_{n}=BSpin_{n}$.  If
$m\equiv 1\mod 4$, then the pullback square then fits into a diagram
\[
\xymatrix{
B\widehat{Spin}^{(m)}_{n}  \ar[r]\ar[d]  & BSpin\times B\Z/2
\ar[d]^{p\,\oplus\,m\overline{\O(1)}}& 
 \\
BO_{n}  \ar[r]        & BO \ar[r]_-{w_{2}} & K(\Z/2,2)
}
\]
in which the corner of maps on the right is a fibration sequence.  It follows
that $B\widehat{Spin}^{(m)}_{n}\simeq B\!\Pp_n$.
\end{example}

We next show that the maps $\bhath{n}{m}\to BO_{n}$
essentially depend only on $m$ mod $4$, so we may take any convenient
value of $m\equiv 1\mod 4$ as our definition of $B\widehat{H}_{n}\to
BO_{n}$.  To do this choose a lift
\[
\xymatrix@C=6em{
  &  BSpin
\ar[d] \\
B\Z/2 \ar@{-->}[ur]^-{W}\ar[r]_-{4\overline{\O(1)}} & BO\mathrlap{\ .}  }
\]
Then for any integer~$\ell$ the map 
\begin{align*}
\rho_{\ell}:B\Z/2\times BH &\to B\Z/2\times BH\\
(x,y) &\mapsto (x,\ell W(x)\oplus y)
\end{align*}
is a homotopy equivalence, and fits into a diagram
\[
\xymatrix{
BH\times B\Z/2  \ar[rr]^{\rho_{\ell}}\ar[dr]_{p\,\oplus\, ((4\ell+k)\overline{\O(1)})\;\;\;\;}  && BH\times B\Z/2 \ar[dl]^{p\,\oplus \,k\overline{\O(1)}} 
 \\
& BO\mathrlap{ .}&
}
\]
for any integer~$k$.
Pulling back along $BO_{n}\to BO$ gives 
\[
\xymatrix{
B\widehat{H}_{n}^{(k+4\ell)}  \ar[rr]^{\rho_{\ell}}\ar[dr]  && B\widehat{H}_{n}^{(k)} \ar[dl] \\
 &BO_{n}&
}
\]
in which the top map is an equivalence.   

\begin{remark}
\label{rem:e2}
The maps $\rho_{\ell}$ depend on the choice of $W$ which is not unique.  In fact there are two lifts of $4\overline{\O(1)}$.    In terms of formulas, a choice of lift corresponds to choosing an element of $Spin_{4}$ lying over $-I_{4}\in SO_{4}$.   The two choices are 
\[
\pm e_{1}e_{2}e_{3}e_{4}.
\]
To be definite we choose the lift given by 
\[
 e_{1}e_{2}e_{3}e_{4}.
\]
\end{remark}

  \subsection{The Lie groups~$\widehat H_{n}^{(m)}$}\label{subsec:e0.11}

The pullback~\eqref{eq:e1} can be rearranged in many ways.   Note that if 
\[
\xymatrix{
F  \ar[r]^{j}\ar[d]_{i}  & Y
\ar[d]^-{g} \\
X  \ar[r]_-{f}        & BO
}
\]
is a pullback diagram, then so is 
\[
\xymatrix{
F  \ar[r]^-{(i,j)}\ar[d]  & X\times Y
\ar[d]^-{-f\oplus g} \\
\ast  \ar[r]        &  BO
}
\]
Using this \eqref{eq:e1}~ can be rewritten as 
\[
\xymatrix@C=6em{
\bhath{n}{m}  \ar[r]\ar[d]  & BH \ar[d] \\
BO_{n}\times B\Z/2  \ar[r]_-{\overline{V}_{n}-m\, \overline{\O(1)}}\ar[d]  &   BO \\
BO_{n}       &
}
\]
where $\overline{V}_n\to BO_n$ is the universal bundle.
When  $m=-k$, with $k\ge 0$, the above pullback can be factored as 
\[
\xymatrix@C=6em{
\bhath{n}{m}  \ar[r]\ar[d]  & BH_{n+k}  \ar[r]\ar[d]   &  BH \ar[d]\\
BO_{n}\times B\Z/2 \ar[r]_-{id\oplus k\,\O(1)}        &   \ar[r] BO_{n+k}        & BO{\mathrlap\ .}
}
\]
This has the advantage of showing that if for all $n$, $BH_{n}$ is the
classifying space of a compact Lie group $H_{n}$ then for $m\le 0$,
$\bhath{n}{m}$ is the classifying space of a compact Lie group
$\widehat{H}_{n}^{(m)}$.  By the $4$-fold periodicity described above
this is actually true for all $m$.

In~\S\ref{subsec:2.3} we fix the value of $m=-3$ and define \[
\widehat{H}_{n}=\widehat{H}_{n}^{(-3)}.  \] This matches the construction
in~\eqref{eq:e12}.  For the bordism computations in~\S\ref{sec:13} it is
useful to use the equivalence \[ \rho_{1}:\bhath{n}{-3}\to \bhath{n}{1} \]
and obtain the pullback square \[ \xymatrix{ B\widehat{H}_{n} \ar[r]\ar[d] &
B\Z/2\times BH \ar[d]^-{\overline{\O(1)}\oplus p} \\ BO_{n} \ar[r] & BO } \]

 \bigskip\bigskip
\providecommand{\bysame}{\leavevmode\hbox to3em{\hrulefill}\thinspace}
\providecommand{\MR}{\relax\ifhmode\unskip\space\fi MR }
\providecommand{\MRhref}[2]{%
  \href{http://www.ams.org/mathscinet-getitem?mr=#1}{#2}
}
\providecommand{\href}[2]{#2}

  \end{document}